\documentclass[12pt,reqno]{amsart}
\usepackage{graphics,amsmath,amssymb,epsfig,stmaryrd,color,amsaddr}
\usepackage{subfigure,amsfonts,geometry,array}
\usepackage{xcolor}
\usepackage{titletoc}
\usepackage{graphicx}
\usepackage{amsthm}
\usepackage{flafter}
\usepackage{float}
\usepackage{setspace}
\usepackage{bbm}
\usepackage{lscape}
\usepackage{pdflscape}
\usepackage{booktabs, topcapt}
\usepackage{threeparttable}
\usepackage{lscape}
\usepackage{pdflscape}
\usepackage{amsmath}
\usepackage{mathtools}
\usepackage{amsmath}
\usepackage{arydshln}
\usepackage{mathtools}
\usepackage{natbib}
\usepackage{enumerate}

\newtheorem{theorem}{Theorem}
\theoremstyle{plain}

\newtheorem{claim}{Claim}
\newtheorem{assumption}{Assumption}

\newtheorem{corollary}{Corollary}

\newtheorem{definition}{Definition}
\newtheorem{example}{Example}

\newtheorem{lemma}{Lemma}

\newtheorem{remark}{Remark}

\newcommand\independent{\protect\mathpalette{\protect\independenT}{\perp}}
\def\independenT#1#2{\mathrel{\rlap{$#1#2$}\mkern2mu{#1#2}}}
\numberwithin{equation}{section}

\usepackage[hypertexnames=false]{hyperref}
\hypersetup{
    colorlinks,
    linkcolor={blue!30!black},
    citecolor={blue!30!black},
    urlcolor={blue!30!black}
}
\begin{document}
\title[Regulatory Policies and Social Welfare]{Evaluating the Impact of Regulatory Policies on Social Welfare in Difference-in-difference Settings}
\thanks{\scriptsize{We are grateful to the editor and three anonymous referees for helpful comments and suggestions. We also thank Manuel Arellano, Dmitry Arkhangelsky, David Autor, Bocar Ba, Brendan Beare, St\'ephane Bonhomme, Joachim Freyberger, Bulat Gafarov, Bo Honor\'e, Guido Imbens, Simon Lee, Michal Kolesar, Kory Kroft, Attila Linder, Patrick Kline, Matthew Masten, Claudia Noack, Christoph Rothe, Pedro Sant'Anna, Andres Santos, Jesse Shapiro, Liyang Sun, and Kaspar W\"uthrich as well as participants at the NBER Summer Institute 2024, Canadian Econometrics Study Group 2024, Triangle Econometrics Conference 2023, Southern Economics Association 2022, and the econometrics seminars at Boston University, Columbia, Duke University, Kentucky, Notre Dame, NYU, Ohio State, Ottawa, Princeton, Queen's, Stanford, Tulane, UNC Chapel Hill, Virginia, and Yale for helpful discussions. We are also grateful to \'Alvaro S\'anchez Leache for excellent research assistance. Dalia Ghanem is grateful to the Center for Monetary and Financial Studies (CEMFI)  for its generous hospitality during her sabbatical visit.\\
$^\dagger$Department of Agricultural \& Resource Economics, University of California, Davis. One Shields Ave, Davis CA, 95616, U.S.A., \href{dghanem@ucdavis.edu}{dghanem@ucdavis.edu}.\\
$^\ddagger$Department of Economics, University of North Carolina, Chapel Hill, Gardner Hall CB3305, Chapel Hill, NC 27599, U.S.A. \href{dkedagni@unc.edu}{dkedagni@unc.edu}.\\ 
$^\ast$Department of Economics,
  Washington University in St. Louis \& NBER. Address: One Brookings Drive
St. Louis, MO 63130-4899, USA. Email: \href{ismaelm@wustl.edu}{ismaelm@wustl.edu}.}}

\author{Dalia Ghanem$^\dagger$ \quad D\'esir\'e K\'edagni$^\ddagger$ \quad
Ismael Mourifi\'e$^\ast$}
\date{\scriptsize{The present version is of \today. 
}}
\maketitle
\pagenumbering{gobble}
\begin{center}{\footnotesize{\begin{minipage}{0.9\textwidth} \textsc{Abstract}. Quantifying the impact of regulatory policies on social welfare generally requires the identification of counterfactual distributions. Many of these policies (e.g. minimum wages or minimum working time) generate mass points and/or discontinuities in the  outcome distribution. Existing approaches in the difference-in-difference literature cannot accommodate these discontinuities while accounting for selection on unobservables and non-stationary outcome distributions. We provide a unifying partial identification result that can account for these features. Our main identifying assumption is the stability of the dependence (copula) between the distribution of the untreated potential outcome and group membership (treatment assignment) across time. Exploiting this copula stability assumption allows us to provide an identification result that is invariant to monotonic transformations. We provide sharp bounds on the counterfactual distribution of the treatment group suitable for any outcome, whether discrete, continuous, or mixed. Our bounds collapse to the point-identification result in \citet{AtheyImbens2006} for continuous outcomes with strictly increasing distribution functions. We illustrate our approach and the informativeness of our bounds by analyzing the impact of an increase in the legal minimum wage using data from a recent minimum wage study \citep{Cengizetal2019}. 
\end{minipage}}}\end{center}

\setcounter{page}{0}
\newpage

\pagenumbering{arabic}
\section{Introduction}
Government's regulatory role and its impact on social welfare has been a critical question for economists. These regulatory policies often restrict the budget or choice sets for certain agents in the market by imposing floors  or quotas, such as minimum wages, minimum/maximum working time, wage floors for different occupation groups as well as action, reporting and notification thresholds in environmental monitoring. Those types of policies tend to induce behavioral responses that can generate mass points in the outcome of interest. For instance, an important question in the labor economics literature is the effect of an increase or introduction of minimum wages on low-wage jobs or overall employment, see for instance \citet{CK1994}, \citet{NeumarkWascher2008}, \citet{Cengizetal2019}, among many others. 
The figure below (taken from \citet{Cengizetal2019}) illustrates that an increase in the minimum wage will shift jobs that were previously paying below the minimum wage $MW$, and then will create “excess jobs” at and slightly above the minimum wage.
\begin{figure}\caption{Figure 1 from \citet{Cengizetal2019}}\label{fig:Ceng}
\begin{tabular}{c}
\includegraphics[width=8cm]{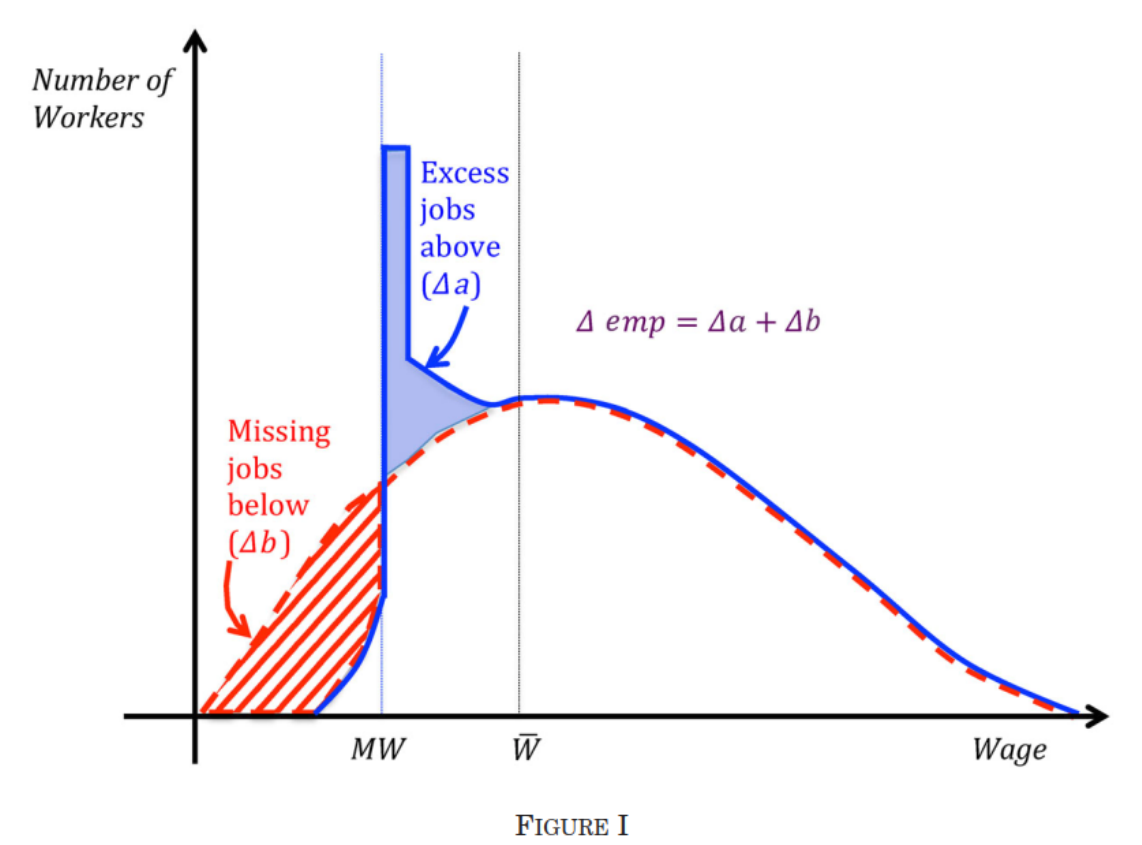}\\
\multicolumn{1}{c}{\parbox{0.8\textwidth}{\scriptsize\emph{Notes}: It is paramount to highlight that the counterfactual in Figure \ref{fig:Ceng} assumes no pre-existing minimum wage, and therefore its distribution is continuous. In practice, however, there is a pre-existing minimum wage which typically leads to a discontinuity in the counterfactual outcome distribution under the pre-existing minimum wage (see Appendix \ref{app:num_illustration_identification} for a numerical illustration).}}\\
\end{tabular}
\end{figure}
This figure also shows the heterogeneous effect of such a policy; it is expected to only affect the wage of low-wage workers and not have an effect on the upper tail of the distribution. In sum, those types of policies have two main features. First, the potential outcomes of interest are likely to exhibit some mass points. Second, the causal effect of the policy is expected to affect only a part of the distribution of the outcomes of interest. As a result, to adequately analyze the impact of these policies, a distributional treatment effect analysis is key, as in \citet{Cengizetal2019} for instance; see, also, \citet{Almondetal2011,assuncao2022}. Furthermore, measuring the impact of such policies on social welfare requires recovering the counterfactual distribution of the outcome of interest.

While these types of policies are widely studied in economics, the existing econometrics methods are not necessarily adequate to recover distributional causal effects in these settings.  In the presence of data before and after a new policy, one of the most widely used techniques to assess its impact is the  difference-in-differences (DiD) method. Its main drawbacks, however, are two-fold: (1) it does not identify the counterfactual distribution, (2) it is not invariant to monotonic transformations. While there are several methods to identify the counterfactual distribution in difference-in-difference settings \citep{AtheyImbens2006,bonhommesauder2011,CallawayLi2019,HavnesMogstad2015}, to the best of our knowledge,  the distributional DiD and changes-in-changes (CiC) are the only two approaches that are invariant to monotonic transformations.\footnote{The distributional DiD method relies on a parallel trends assumption in the cdfs as opposed to the expectations \citep[e.g.][]{HavnesMogstad2015,RothSantanna2021}. }

\citet{RothSantanna2021} show that distributional DiD requires that the distribution of the untreated potential outcome is independent of policy adoption, is stationary across time (within each group), \emph{or} consists of a mixture of two subpopulations each obeying one of the two restrictions.     Such conditions are unlikely to be valid for the policy evaluation questions we are interested in.
 Indeed, the independence assumption (random assignment) is implausible in our context since the decision to implement a new minimum wage policy is a response to the unsatisfactory features of the pre-policy outcome distribution, such as large wage inequalities, high proportion of workers under poverty, etc. When the policy is not randomly assigned, the validity of the distributional DiD essentially rests on the stationarity assumption, which is restrictive in many practical settings.\footnote{The stationarity assumption can be tested using the control group. \citet{RothSantanna2021} provide a sharp specification test of the validity of the distributional DiD assumption in general.}

 While  the CiC approach introduced in the seminal work by \citet{AtheyImbens2006} can accommodate endogenous policy (treatment) assignment as well as time-varying potential outcome distributions, their identification result does not apply to the case where the potential outcomes exhibit some mass points (mixed distributions), as in Figure \ref{fig:Ceng}.\footnote{Mass points are common for a wide range of economic outcomes resulting from censoring \citep{bottomcoding1,topcoding1} or bunching \citep{bunching3,bunching4,bunching2,bunching5,bunching6,bunching7,bunching8}.} In fact, \citet{AtheyImbens2006} introduce the CiC approach for \emph{either} continuous \emph{or} discrete outcomes that are monotonic (time-varying) functions of a scalar unobservable with a time-invariant distribution across time. In sum, the CiC approach introduced in \citet{AtheyImbens2006} should not be applied to evaluate the policies described above.

The current paper provides an alternative, unifying identification result that applies to any type of outcome distribution, is invariant to monotonic transformations, allows for endogeneity of the policy assignment, and does not restrict the evolution of the marginal distribution, nor treatment effect heterogeneity. Our identification result exploits the stability of the dependence (copula) between treatment assignment and the untreated potential outcome across time without imposing restrictions on the structural function that generates the potential outcomes. 

Exploiting our copula stability (CS) assumption, we provide a unifying partial identification result for the counterfactual distribution of the treatment group.  We then extend our analysis to the case where multiple pre-treatment periods are available. In this case, we show that if copula stability holds for multiple pre-treatment periods, then our multi-period CS bounds exploit the information from the pre-treatment periods to provide tighter bounds.\footnote{We use multi-period CS bounds to refer to CS bounds that use multiple pre-treatment periods.} The presence of multiple pre-treatment periods also allows us to provide a testable restriction of our model assumptions. We demonstrate our theoretical results numerically in Section \ref{sec:illustration_MP}.

Our CS bounds apply to any type of outcome distribution, whether it is continuous, mixed, or discrete. They shrink to the point-identification result in \citet{AtheyImbens2006} for continuous outcomes. Indeed, we show that in this case our copula stability assumption is equivalent to the CiC conditions. For discrete outcomes, we show that our copula stability assumption can be compatible with  an underlying production function featuring  multi-dimensional unobserved heterogeneity, whereas the CiC bounds for discrete outcomes require a scalar unobservable. For mixed outcomes, we demonstrate that a na\"ive implementation of the CiC approach may lead to a point-estimand that does not coincide with the true counterfactual, whereas our CS bounds will include it.\footnote{We refer to this implementation as na\"ive since \citet{AtheyImbens2006} did not provide identification results for mixed outcomes. Nonetheless, an empirical researcher might ignore the mixed-nature of this outcome and implement their point-identification result. }

We also examine the connection between our main identifying assumption and the parallel trends assumption required by DiD. The parallel trends assumption can be equivalently stated as a covariance stability assumption. It is specifically a time invariance assumption on the covariance between treatment assignment and the untreated potential outcome, whereas our assumption maintains the stability of the copula between these two variables. As a result, there are several differences between our copula stability assumption and covariance stability (parallel trends). First, the parallel trends assumption restricts the joint variability of treatment assignment and the untreated potential outcome over time, whereas our copula stability assumption only restricts their dependence structure.  Second, while the parallel trends assumption restricts the evolution of the marginal distribution of the untreated potential outcome across time, copula stability does not restrict the  evolution of the marginal distribution, nor treatment effect heterogeneity. Last but not least, parallel trends is not invariant to monotonic transformations except under strong conditions on heterogeneity \citep{RothSantanna2021}. These conditions specifically rule out the existence of a subpopulation that selects into treatment based on unobservables \emph{and} exhibits changes in its potential outcome distribution. By contrast, our copula stability condition does not rule out such a subpopulation.

Since the motivation behind policies, such as increases in the legal minimum wage, is often to reduce inequality and/or target a specific part of the outcome distribution, we introduce a broad class of social welfare treatment effect parameters that can accommodate the policymaker's objective. While this class includes the average treatment effect on the treated (ATT) as a special case, the ATT corresponds to a social welfare function that is inequality-neutral and gives equal weight to all individuals in the population. As a result, if a policymaker is averse to inequality, then the ATT would be an inadequate causal parameter to judge the policy's effectiveness. In general, the social welfare function adequate to evaluate a specific regulatory policy can be highly context-specific and may depend on the policymaker's preference and/or objective.\footnote{Please see the discussion in \citet{Bergeretal2022} which illustrates how the 
quantitative analysis of the effect of the minimum wage could highly differ depending on the social welfare weights, which are usually unknown to the researcher.}  We therefore introduce a broad class of treatment effect parameters that take into account the policy objectives. This broad class specifically includes the class of generalized Gini social welfare functions \citep[e.g.][]{Mehran1976,Weymark1981}. These social welfare functions can take into account measures of inequality by putting higher weight on individuals with lower-ranked outcomes. In addition, we include a class of parameters that can capture the welfare of individuals at the lower tail or a specific interquantile range of the distribution. Bounds on these social welfare treatment effect parameters can be easily computed using our bounds on the counterfactual distribution. We illustrate the usefulness of this broad class of parameters and compare it to the ATT in the context of our empirical application examining the impact of a minimum wage policy (Section \ref{sec:empirical}).

We organize the rest of the paper as follows. Section \ref{Sec: Main-results} introduces the analytical framework and presents our main identification results. Section \ref{sec:swtt} introduces the class of social welfare treatment effect parameters. Section \ref{sec:empirical} provides an empirical illustration examining the impact of minimum wage increases on the wage distribution revisiting \citet{Cengizetal2019}.

\subsection*{Related Literature}

A comparison between our identifying assumption and some of the related approaches in the literature is warranted. \citet{bonhommesauder2011} exploit a separable model of the potential outcome to identify the entire counterfactual distribution of the treatment group in a DiD design. By relying on restrictions on the outcome model, it is therefore similar in spirit to the identification approach in \citet{AtheyImbens2006}. \citet{botosarumuris2023} propose identification of counterfactual parameters for a class of semiparametric panel models, whereas our approach can accommodate both repeated cross-sections and panel data and is fully nonparametric. \citet{CallawayLi2019} also provide a fully nonparametric identification result exploiting a copula stability restriction on  different objects than the ones used in this paper. 
They require the copula between changes and levels of the untreated potential outcome to be invariant across time for the treatment group, while our copula stability assumption does not restrict the evolution of the marginal distribution of the untreated potential outcome (Remark \ref{rem:CL2019}). Furthermore, our approach can be applied to repeated cross-sections or panel data and only requires two time periods, whereas \citet{CallawayLi2019} require at least three periods of panel data. \citet{wooldridge2023simple} proposes alternative parallel trends assumptions that are more suitable for binary, fractional and count outcome data. The approach in \citet{wooldridge2023simple} requires the specification of a parametric transformation model of a linear index for each type of outcome and point-identifies the average treatment effect on the treated, whereas our approach applies to any outcome, is fully nonparametric and partially identifies the counterfactual distribution. 

Finally, this paper contributes to a strand in the microeconometrics literature that relies on copula theory. For cross-sectional settings with exogenoeus regressors, \citet{rothe2012} provides identification results for partial distributional effects, which hold the copula of the covariates constant, but vary their marginal distributions. \citet{mourifie2015} relies on copula theory to provide sharp bounds on the average treatment effect in a binary triangular system. \citet{arellanobonhomme2017} propose a method to correct for sample selection in quantile models, where the conditional copula of the error terms in the outcome and selection equations is a key ingredient in their approach.

\section{Analytical Framework and Main Identification Results}\label{Sec: Main-results}

Following \citet{Abadie2005}, we consider the following potential outcomes model:\footnote{Note that this model implicitly assumes that there are no anticipatory effects of the treatment, that is, $Y_{00}=Y_{01}=Y_0$.}
\begin{eqnarray}\label{seq1}
	\left\{ \begin{array}{lcl}
		Y_0 &=& Y_{00} \\ \\
		Y_1&=&Y_{11}D+Y_{10}(1-D)
	\end{array} \right.
\end{eqnarray}
where $Y_{t}$ denotes the observed outcome at period $t$ and $Y_{td}$ denotes the potential outcome at period $t\in\{0,1\}$ and treatment status $d\in\{0,1\}$. In the two-group, two-period case, $D$ denotes both group membership and the treatment status in period 1.

We use the following shorthand notation: $p\equiv \mathbb P(D=1)$, $q=1-p$, $\overline{\mathbb R}\equiv\mathbb R \cup \{-\infty, \infty\}$, $Ran H \equiv\{ H(y):y \in \mathbb R\}$, $\overline{\operatorname{Ran}}  F\equiv Ran F \cup \{\inf Ran F, \sup Ran F\}$, and $Dom H$ denotes the domain of the function $H$.
 We consider the following mappings $Q^{\mathbb  T,-}_X: [0,1] \rightarrow \mathbb T$,  and $Q^{\mathbb  T,+}_X: [0,1] \rightarrow \mathbb T$, where $Q^{\mathbb  T,-}_X(u)\equiv \inf \{x \in \mathbb T \cup \{\infty\}: F_X(x)\geq u\}$ for all $u \in [0,1]$, 
  $Q^{\mathbb  T,+}_X(u)\equiv \sup \{x \in \mathbb T \cup \{-\infty\}: F_X(x)\leq u\}$ for all $u \in [0,1]$.
 We call  $Q^{\mathbb  T,+}_X$ and $Q^{\mathbb  T,-}_X$ \textit{generalized quantile functions} whenever $F_X(.)$ is a well-defined cumulative distribution function (cdf). We denote by $\mathbb F$ the space of all well-defined cdfs. 
  $Supp X=\mathbb X$ denotes the support of $X$, and  $\mathbb X_{s|d}$ denotes the support of $X_s|D=d$ for $d \in \{0,1\}$. Finally, we define $F_X(x-)\equiv \mathbb{P}(X<x)$.

\subsection{Identifying Assumptions}
Our main identification result  relies on restrictions imposed on the dependence structure across time. To do so, we rely on copula theory. 
Copulas are functions that enable us to separate the marginal distributions from the (scale-free) dependence structure of a
given multivariate distribution. In our context, we are interested in the subcopula between the untreated potential outcome and group membership across time. Working with copulas in our case will allow us to avoid restricting the type of marginal distribution of the potential outcomes as well as its heterogeneity across time.  
To fix ideas, let us first provide a formal definition of the (sub)copula.  
\begin{definition}[\cite{Nelsen2006}]
A two-dimensional subcopula is a function $C$ with the following properties: 
\begin{enumerate}
\item $Dom C=S_1\times S_2$, where  $S_1$ and $S_2$ are subsets  of $[0,1]$ containing $0$ and $1$;
\item For all $u,u' \in S_1$, and $v,v' \in S_2$ such that $u\leq u'$, and $v\leq v'$, we have:
$$ C(u',v') +  C(u,v)\geq C(u',v) +  C(u,v');$$
\item $C(0,v)=C(u,0)=0$ for all $(u,v) \in S_1\times S_2$, and $C(1,v)=v$, $C(u,1)=u$ for all $(u,v) \in S_1\times S_2$.
\end{enumerate}
\end{definition}
A copula is a special case of a subcopula where $S_1=S_2=[0,1]$. For a fixed $v \in S_2$, $u \mapsto C(u,v)$ is usually called the horizontal subcopula.
The link between the joint distribution and the subcopula has been established by the well-known Sklar (1959) theorem, which provides the following lemma  when applied to our context.
\begin{lemma}[Sklar, 1959]
There exists a unique  subcopula $C: \overline{\operatorname{Ran}}  F_{Y_{t0}}\times \{0,q,1\} \rightarrow [0,1]$ such that
\begin{eqnarray}
\mathbb P(Y_{t0}\leq y, D=0)=C_{Y_{t0},D}(F_{Y_{t0}}(y),q),\;\;\; \text{ for }  y \in \overline{\mathbb{R}}. \label{eq:sklar}
\end{eqnarray} 
\end{lemma}\label{lem:sklar}

To provide intuition for the role of the horizontal subcopula at $q$, it is helpful to divide each side of Equation\ \eqref{eq:sklar} by $q$, which yields the following for $y\in \overline{\mathbb{R}}$
\begin{eqnarray}
   F_{Y_{t0}|D=0}(y)&=&C_{Y_{t0},D}(F_{Y_{t0}}(y),q)/q.\label{eq:rank_mapping}\end{eqnarray}
Now, let us assume that the copula is strictly increasing in its first argument such that its inverse $C_{Y_{t0},D}^{-1}(\cdot~;q)$ is well-defined.\footnote{Note that a horizontal copula is by definition Lipschitz continuous.}  We can then show that $C_{Y_{t0},D}^{-1}(\cdot~;q)$ is the main ingredient in the \textit{rank mapping} between the treatment and control group's untreated potential outcome distribution in period $t$, which we denote by $\Gamma_t(\cdot)$:
\begin{eqnarray}
F_{Y_{t0}|D=1}(y)
&=& \Gamma_t\left( F_{Y_{t0}|D=0}(y) \right)\quad \text{for }y\in\overline{\mathbb{R}},\label{opt-map}
\end{eqnarray}
where $\Gamma_t(u)\equiv \frac{1}{p}\left( C_{Y_{t0},D}^{-1}\left(uq;q \right)- uq \right)$ for $u\in RanF_{Y_{t0}|D=0}$. The mapping governs the relationship of the rank that a given value $y$ has in the control group's distribution of the untreated potential outcome $F_{Y_{t0}|D=0}$ (factual at each period) onto its rank in the treatment group's distribution of the untreated potential outcome, $F_{Y_{t0}|D=1}$.

Next, we introduce our main assumption.

\begin{assumption}[Copula stability]\label{stab} The following condition holds:
$C_{Y_{00},D}(u,q)=C_{Y_{10},D}(u,q)$ for all $u \in [0,1]$.
\end{assumption} 
In the following, we will refer to Assumption \ref{stab} as ``copula stabilty'' for brevity,  but we emphasize that it only requires the stability of the horizontal copula between $Y_{t0}$ and $D$ at $q$, $C_{Y_{t0},D}(u,q)$ for $u\in[0,1]$. There are multiple advantages to our copula stability assumption. First, it is invariant to strictly monotonic transformations. Specifically, for any right-continuous function $g$, that is strictly increasing on $\mathbb Y_{td}$, we have:\footnote{See \citeauthor{Embrechts_al2013} (\citeyear{Embrechts_al2013}, Proposition 4(2)) for a formal proof.} $$C_{g(Y_{td}),D}(u,q)=C_{Y_{td},D}(u,q) \;\; \forall u \in RanF_{Y_{td}}.$$
Second, it does not impose any restrictions on the variability of the marginal distribution $F_{Y_{t0}}$ across time.  Last, but not least, it does not restrict the type of marginal distribution $F_{Y_{t0}}$, whether it is continuous, discrete or mixed.

Assumption \ref{stab} is the key assumption behind our identification approach.  It implies that the rank mapping $\Gamma_t(\cdot)$ is stable across periods, i.e. $ \Gamma_1(\cdot)=\Gamma_0(\cdot)\equiv \Gamma(\cdot)$. In the presence of multiple pre-treatment periods, $\Gamma(\cdot)$ can be recovered from each pre-treatment period. As a result, analogous to pre-trend testing in difference-in-differences designs, the time-invariance of $\Gamma(\cdot)$ can also be tested as we demonstrate in Section \ref{sec:multiple_periods}.

Given the wide use of difference-in-differences, it is also helpful to clarify the relationship between our copula stability assumption and the parallel trends assumption. 
The parallel trends assumption can be equivalently rewritten as a covariance stability assumption as we show in Appendix \ref{Appen:covstability},
\begin{eqnarray}\label{eq-PT}
    \mathbb E[Y_{10}-Y_{00}|D=1]=\mathbb E[Y_{10}-Y_{00}|D=0] \iff Cov(Y_{00},D)=Cov(Y_{10},D).
\end{eqnarray}
This equivalence result provides, first, an intuition for why the parallel trends assumption is not invariant to a monotonic transformation since the covariance is not invariant to monotonic transformations. Second, it allows us to observe that the parallel trends assumption jointly restricts the evolution of the marginal distribution of $Y_{t0}$ across time and the dependence between $Y_{t0}$ and $D$. Unlike the parallel trends assumption, our copula stability assumption
does not constrain the evolution of the marginal distribution across time, yet it relies only on the stability
of the horizontal copula that governs the relationship between $Y_{t0}$ and $D$. 
As can be seen in the following equation, the two assumptions are non-nested in general:
\begin{eqnarray*}
Cov(Y_{td},D)=\int \left[C_{Y_{td,D}}(F_{Y_{td}}(y),q)-F_{Y_{td}}(y)q\right]dy.
\end{eqnarray*}
Indeed, copula stability may hold while $Cov(Y_{10},D)\neq Cov(Y_{00},D)$ because $F_{Y_{10}}\neq F_{Y_{00}}$; and the covariance stability may hold 
while  the copula stability is violated.
 
 In the following, we provide several examples to illustrate the restrictions imposed by our key assumption, and how it compares to some existing assumptions.
 
\begin{example}[Roy selection]\label{ex:1}
Consider the following data generating process (DGP) in which the treatment is received when its gain (treatment effect) is bigger than or equal to a threshold, say 0 for simplicity. This is a simple Roy model where selection into treatment is on the gain. 
\begin{eqnarray}\label{eq:example1}
\left\{ \begin{array}{lcl}
     Y_0 &=& U_0\\ \\
     Y_{1} &=& \eta D+ U_1  \\ \\
     D &=& \mathbbm{1}\{\eta\geq 0\}
     \end{array} \right.
\end{eqnarray}
where $\left(\begin{array}{c}U_0\\ U_1 \\ \eta \end{array}\right)\sim N(0,\Sigma)$, $\Sigma=\left(\begin{array}{ccc}\sigma_0^2 & \delta \sigma_0 \sigma_1 &  \rho_0 \sigma_0\\ \delta \sigma_0 \sigma_1 & \sigma_1^2 &  \rho_1 \sigma_1\\ \rho_0 \sigma_0&  \rho_1 \sigma_1 & 1\end{array}\right)$, and $\rho_t \neq 0$.
In this case, we have the following:
\begin{enumerate}
\item [(a)] {\bf Copula  stability:} $\rho_0=\rho_1 \Leftrightarrow   Corr(\eta,Y_{00})=Corr(\eta,Y_{10})$.
\item [(b)] {\bf Parallel trends:} $\rho_0 \sigma_0=\rho_1 \sigma_1 \Leftrightarrow   Cov(\eta,Y_{00})=Cov(\eta,Y_{10})$.
\item [(c)] {\bf Distributional DiD:} $\rho_0=\rho_1$ and $\sigma_0^2=\sigma_1^2$ $\Leftrightarrow$ $Y_{00}|D=d\sim Y_{10}|D=d$, for $d=0,1$.
\end{enumerate}

As can be seen, the copula stability assumption is equivalent to $\rho_0=\rho_1$, meaning that the correlation between the policy effect $\eta$ and $Y_{t0}$ is stable over time.  
It does not restrict any moment of the marginal distribution of the potential outcomes $Y_{t0}$. The parallel trends assumption, however, restricts the variances of the potential outcomes $Y_{00}$ and $Y_{10}$, since it is equivalent to $\rho_0 \sigma_0=\rho_1 \sigma_1$. The validity of the  distributional DiD in this setting is implausible, since it requires stationarity of $Y_{t0}|D=d$. This could be easily checked using the observed distribution of the control group.

Note that while the copula stability condition, result (a), does not rely on the Gaussianity assumption imposed on the marginal distribution, results (b) and (c), which involve the parallel trends and distributional DiD assumptions, are heavily dependent on this distributional assumption. For further details,  see Appendices \ref{sec:example} and \ref{Appen:example1_nonnormal}.

\end{example}

The above example demonstrates the copula stability assumption in the context of selection on the gains from the treatment. We next consider selection on untreated potential outcomes. This example shows that copula stability requires comonotonicity between the untreated potential outcomes in the pre- and post-treatment periods.

\begin{example}[Selection on untreated potential outcomes]

Consider the following model, where selection into treatment is a function of the pre-treatment outcome, such as in the Ashenfelter dip,
\begin{eqnarray*}
\left\{ \begin{array}{lcl}
     Y_0 &=& Y_{00},\\ 
     Y_{1} &=& Y_{11} D + Y_{10} (1-D),  \\ 
     D &=& \mathbbm{1}\{Y_{00}> c\}.
     \end{array} \right.
\end{eqnarray*}
Assume that $Y_{00}$ has a continuous and strictly increasing cdf. It can be shown that 
$$
 C_{Y_{t0},D}(u,q)=\mathbb P\left(F_{Y_{t0}}(Y_{t0}) \leq u, F_{Y_{00}}(Y_{00}) \leq q\right).
$$
By construction, we can see that we always have $C_{Y_{00},D}(u,q)=\min(u,q)$, while 
$C_{Y_{10},D}(u,q)=\min(u,q)$ if $Y_{00}$ and $Y_{10}$ are comonotone.
Thus, in this model with selection on the lagged outcome, the copula stability assumption holds  if $Y_{10}=h_1(U)$, $Y_{00}=h_0(U)$ for some non-decreasing functions $h_0$ and $h_1$.

Notice that when selection is on lagged outcomes, the parallel trends assumption fails in general (as $\mathbb E[Y_{10}-Y_{00} \vert Y_{00} > c]\neq \mathbb E[Y_{10}-Y_{00} \vert Y_{00} \leq c]$), unless $Y_{10}-Y_{00}$ is independent of $Y_{00}$; that is, $Y_{00}$ follows a martingale process.\footnote{Relatedly, \citet{gsw2022} show that for parallel trends to hold under selection on pre-treatment unobservables, a martingale-type restriction on the untreated potential outcome is necessary.} Therefore, in this framework, even when parallel trends fail, copula stability may still hold under a particular mapping between $Y_{10}$ and $Y_{00}$. It is important to note however that the comonotonicity assumption may not be plausible in some applications, and therefore copula stability would fail.

Finally, from the above arguments, it is straightforward to show that if selection was on post-treatment (untreated) potential outcomes, $D=\mathbbm{1}\{Y_{10}>c\}$, then copula stability would also require comonotonicity between pre- and post-treatment untreated potential outcomes.
\end{example}

We next consider selection on time-varying shocks, an example that will not be compatible with copula stability in general.

\begin{example}[Selection on time-varying shocks]
Consider a setting where selection into treatment depends only on the post-treatment shock, specifically:
\begin{eqnarray*}
\left\{ \begin{array}{lcl}
     Y_{t0} &=& f_t + \lambda + \varepsilon_{t},\\  \\
     D &=& \mathbbm{1}\{\varepsilon_{1}>c\},
     \end{array} \right.
\end{eqnarray*}
where $\lambda$ and $\varepsilon_{t}$ are time-invariant and time-varying unobservables, respectively, and $f_t$ captures nonstochastic time trend. Suppose further that  $\left(\begin{array}{c}\lambda\\ \varepsilon_{0} \\ \varepsilon_{1} \end{array}\right)\sim N(0,\Sigma)$, $\Sigma=\left(\begin{array}{ccc}\sigma_{\lambda}^2 & 0 &  0\\ 0 & \sigma_{\varepsilon_0}^2 &  0\\ 0&  0 & \sigma_{\varepsilon_1}^2 \end{array}\right)$. In this model, the copula $C_{Y_{t0},D}(u,q)$ is Gaussian, and $Corr(\varepsilon_{1},Y_{00})=0$ while $Corr(\varepsilon_{1},Y_{10})=\frac{\sigma_{\varepsilon_1}}{\sqrt{\sigma^2_{\lambda}+\sigma^2_{\varepsilon_{1}}}}$.  Hence, our CS assumption fails to hold. Note, however, that in this DGP, both CiC and PT assumptions fail to hold as well.

Again, in this example, we can demonstrate that if selection was on pre-treatment shocks, $\varepsilon_0$, then copula stability would be violated by similar arguments.
\end{example}

\begin{example}[Firm-specific wage increases]
 Consider an individual working for a firm $f$. Let $Y_{00}$ and $Y_{10}$ be, respectively, the worker's wage in periods 0 and 1 in the absence of a minimum-wage increase $D$. The worker's wage in period 1 in the absence of the minimum-wage increase would be her wage in period 0 plus any increase that firm $f$ provides to its workers. Suppose there is an increase of $100R_f\%$ in workers' salaries in firm $f$. Then, we can write $Y_{10}=(1+R_f)Y_{00}$. If the wage increase rate $R_f$ is jointly independent of the baseline salary $Y_{00}$ and the policy $D$, i.e., $R_f \perp (Y_{00},D)$, then copula stability holds.\footnote{See proof in Appendix \ref{proof:mwexample}. We allow $D$ to depend on $Y_{00}$ and the treatment effect.} On the other hand, parallel trends as well as distributional parallel trends fail to hold unless $D$ is (mean) independent of $Y_{00}$ (i.e., unless random assignment holds).

In practice, the wage increase rate $R_f$ may depend on some firm-level characteristics $X_f$ that can explain $R_f$, $Y_{00}$, and $D$. Then, our horizontal copula stability assumption will hold conditional on firm characteristics $X_f$, i.e., $C_{Y_{10},D\vert X_f}(u,q)=C_{Y_{00},D\vert X_f}(u,q)$ for all $u$, but not unconditionally. 

\label{ex:firm_wages}\end{example}

Next, we proceed to our second identifying assumption, which requires the strict monotonicity of the horizontal copula.
\begin{assumption}[Strictly increasing horizontal copula]\label{inc} The function 
 $u \mapsto C_{Y_{10},D}(u,q)$  is strictly increasing on $[0,1]$.
\end{assumption} 
While Assumption \ref{inc} is less critical for our bounding approach, it allows us to simplify the expression of our bounds. It is essentially a restriction on the type of dependence between the potential outcomes and group membership. 
Many well-known parametric classes of copulas satisfy this assumption, e.g. Frank, Gumbel, Joe, or Gaussian copulas among many others. It excludes, however, extreme types of dependence captured by the 
Fr\'echet-Hoeffding copula bounds, i.e. $C(u,v)=\min\{u,v\}$ and $C(u,v)=\max\{u+v-1,0\}$. 
It is worth noting that this assumption is implied by  some support conditions on the potential outcome distributions, as we show in the following result.
\begin{lemma}\label{cop:mon} 
If $\mathbb Y_{t0|1} \subseteq \mathbb Y_{t0|0}$, then $u \mapsto C_{Y_{t0},D}(u,q)$  is strictly increasing on $Ran F_{Y_{t0}}$ for $t \in \{0,1\}$.
\end{lemma}
The main implication of the above lemma is that for continuous potential outcome distributions, we have $Ran F_{Y_{t0}}=[0,1]$, and the strict monotonicity of the copula (Assumption~\ref{inc}) is implied by a condition on the support of $Y_{t0}$, $\mathbb Y_{t0|1} \subseteq \mathbb Y_{t0|0}$. That is, the support of the untreated potential outcome of the treatment group is included in the support of the untreated potential outcome of the control group. The support condition imposed in \citet{AtheyImbens2006} on the scalar unobservable in the CiC model implies this support condition on the untreated potential outcome.

\begin{remark}\label{rem:CL2019}
    Here, we formally compare our copula stability assumption with the one introduced in \citet{CallawayLi2019}. To see this, let us define $\Delta Y_{t0}=Y_{t0}-Y_{(t-1)0}$, \citet{CallawayLi2019} require $C_{\Delta Y_{t0},Y_{(t-1)0}|D=1}(\cdot,\cdot)=C_{\Delta Y_{(t-1)0},Y_{(t-2)0}|D=1}(\cdot,\cdot)$.
As can be seen, their assumption imposes a dependence  stability on  different objects than ours, and it requires at least three time periods of panel data.  In addition, unlike us,  their identification results require an additional independence condition between the change in the untreated potential outcome and treatment assignment, $\Delta Y_{t0}\perp D$.  
\end{remark}

\subsection{Main Identification Result}

We next state our main identification result:   

\begin{theorem}\label{Main:theorem}
  Suppose that $\mathbb Y_{t0|1} \subseteq \mathbb Y_{t0|0}$ for $t \in \{0,1\}$, then 
 under Assumptions \ref{stab} and \ref{inc},  the bounds on the unobserved counterfactual distribution $F_{Y_{10}|D=1}(.)$ are: 

{\begin{eqnarray*}
&&\lim_{\tilde y \downarrow y}\sup\left\{F^{LB}(t): t\leq \tilde y \; \& \;  t \in \mathbb Y_{10|0} \cup\{-\infty \} \right\}  \\&&\qquad \leq  F_{Y_{10}|D=1}(y)\leq \nonumber\\&&  \lim_{\tilde y \downarrow y}\sup\left\{ F^{UB}(t): t\leq \tilde y \; \& \;  t \in \mathbb Y_{10|0} \cup\{-\infty \} \right\}
\end{eqnarray*}
 for all   $y \in \mathbb R$, where 
\begin{eqnarray*}
    F^{LB}(t)&=&F_{Y_0|D=1}\left(Q^{\mathbb R,+}_{Y_0|D=0}\left(F_{Y_{1|D=0}}(t)\right)-\right)\; \\\; F^{UB}(t)&=&F_{Y_0|D=1}\left(Q^{\mathbb R,-}_{Y_0|D=0}\left(F_{Y_{1|D=0}}(t)\right)\right).
\end{eqnarray*}}
The above bounds are shown to be sharp when $\overline{\operatorname{Ran}}  F_{Y_{0}}$ is closed.\footnote{We conjecture that the sharpness statement remains valid without this closure requirement, but it requires a more involved construction of the subcopula that rationalizes the data.}

\end{theorem}

Theorem \ref{Main:theorem} provides a general (partial) identification result on the counterfactual distribution of the treatment group for any type of potential outcome variables (discrete, continuous, or mixed). Our result neither imposes any restriction on the heterogeneity of potential outcomes within a period nor across periods. 
We specifically do not impose restrictions on individual treatment effects, $Y_{11}-Y_{10}$, or the evolution of the distribution of the untreated potential outcome across time, $F_{Y_{t0}}, t \in\{0,1\}$. The formal proof is relegated to Appendix \ref{Appen:Proof}. The derived bounds may look involved since we aim to  provide a general formulation that covers any type of distribution and want to ensure that our bounds are indeed right-continuous.\footnote{As recognized by \citet{AtheyImbens2006}, their  upper bound in the discrete outcome case may be left-continuous, and therefore may not satisfy the properties of a cdf.} The bounds simplify for some special cases as we will illustrate in Corollary \ref{cor:Iden} below.

The intuition behind our (partial) identification result is very simple and can be summarized as follows: In the first period, we identify the joint distribution $\mathbb P(Y_{00}\leq y, D=0)$ and both marginal distributions, $\mathbb P(Y_{00}\leq y)$ and $q$.  Using the Sklar result, we can recover the horizontal subcopula $C_{Y_{00},D}(u,q)$ on $RanF_{Y_{0}}$, and thereby the rank mapping $\Gamma(\cdot)$ on $Ran F_{Y_0|D=0}$. Then, since we assume the rank mapping to be stationary across time, we can then carry it over from the pre-treatment period to the post-treatment period to recover the 
treatment group's distribution of the untreated potential outcome, $F_{Y_{10}|D=1}$, as follows:
\begin{eqnarray}
F_{Y_{10}|D=1}(y)
&=& \Gamma\left( F_{Y_{1}|D=0}(y) \right)\quad \text{for } F_{Y_{1}|D=0}(y) \in Ran F_{Y_{0}|D=0}. 
\end{eqnarray}
 The main reason behind the partial identification is that in the first period we recover the subcopula  $C_{Y_{00},D}(\cdot,q)$ only on $Ran F_{Y_0}$ ($\Gamma(\cdot)$ only on $Ran F_{Y_0|D=0}$), and we do not know the rank mapping outside this range. We provide a graphical illustration of these functions as well as our bounds in the context of a minimum-wage numerical example in Appendix \ref{app:num_illustration_identification}.

In the case of continuous potential outcomes, $Ran F_{Y_{0}}=[0,1]$, our bounds shrink to a point because the pre-treatment period allows us to recover the entire rank mapping that we carry over to the post-treatment period, as we show in the following corollary of  Theorem \ref{Main:theorem}. 
\begin{corollary}\label{cor:Iden} Under Assumption \ref{stab}, whenever $\mathbb Y_{t0|1} \subseteq \mathbb Y_{t0|0}$ for $t \in \{0,1\}$ and the cdfs $F_{Y_{t0}\vert D=d}(.)$, $t, d \in \{0,1\}$ are continuous, we have:
$$F_{Y_{10}|D=1}(y)=F_{Y_0|D=1}\left(Q^{\mathbb R,-}_{Y_0|D=0}\left(F_{Y_{1|D=0}}(y)\right)\right)$$ for all  $ y \in \mathbb R$. 
\end{corollary} 
The proof of this corollary is in Appendix \ref{apx:corident}.
Corollary \ref{cor:Iden} recovers the point-identification result obtained in \citet{AtheyImbens2006}. \citet{AtheyImbens2006} provide (partial) identification results for two types of potential outcomes relying on different assumptions for each of the two cases: (i) continuous outcomes that are strictly monotonic in a scalar unobservable, (ii) discrete outcomes that are monotonic in a scalar unobservable. By contrast, Theorem \ref{Main:theorem} establishes a unifying identification result for any type of outcome under consideration. In addition to the connection to our identification result, there is a link between the CiC assumptions and our copula stability condition for continuous outcomes. We provide details on this connection and compare the two identification approaches in Section \ref{sec:cic_connection}.

\begin{remark}
The bounds in Theorem \ref{Main:theorem} may cross, indicating that at least one of our key assumptions does not hold. We present the formal testable implication in the following subsection. 
\end{remark}

\subsection{Multiple pre-treatment periods}\label{sec:multiple_periods}

In this section, we characterize our bounds in the presence of multiple pre-treatment periods. Suppose we have the following model with $T_0+1$ pre-treatment periods: 

\begin{eqnarray}\label{seq1}
	\left\{ \begin{array}{lcl}
		Y_t &=& Y_{t0}, \;\; t=-T_0,\dots,0 \\ \\
		Y_1&=&Y_{11}D+Y_{10}(1-D)
	\end{array} \right.
\end{eqnarray}

We impose the following stability restriction on the horizontal copula at $q$ over multiple pre-treatment periods $t=-T_0,\dots,0$. 
\begin{assumption}[Dependence stability over multiple periods]\label{Gstab} For all $t=-T_0,\dots,0$ and $u \in [0,1]$,
$$C_{Y_{t0},D}(u,q)=C_{Y_{10},D}(u,q).$$ 
\end{assumption}

The following theorem generalizes Theorem \ref{Main:theorem} to the multiple-period case under Assumption \ref{Gstab}. Corollary \ref{result:fals-test} then provides testable restrictions of our model assumptions. 
\begin{theorem}\label{GMain:theorem}
  Suppose that $\mathbb Y_{t0|1} \subseteq \mathbb Y_{t0|0}$ for $t \in \{-T_0,\dots,0\}$.
 If  Assumptions \ref{inc} and \ref{Gstab} hold, then  the bounds on the unobserved counterfactual distribution $F_{Y_{10}|D=1}(.)$ are: 

{\begin{eqnarray*}
&&\lim_{\tilde y \downarrow y}\sup\left\{ \max_{t \in \{-T_0,\dots,0\}}F_t^{LB}(s): s\leq \tilde y \; \& \;  s \in \mathbb Y_{10|0} \cup\{-\infty \} \right\}  
\\&& \qquad \leq  F_{Y_{10}|D=1}(y) \leq \nonumber\\&&
 \lim_{\tilde y \downarrow y}\sup\left\{ \min_{t \in \{-T_0,\dots,0\}}F_t^{UB}(s): s\leq \tilde y \; \& \;  s \in \mathbb Y_{10|0} \cup\{-\infty \} \right\},
\end{eqnarray*}
 for all   $y \in \mathbb R$, where for $t\in\{-T_0,\dots,0\}$
\begin{eqnarray*}
    F_t^{LB}(s)&=&F_{Y_t|D=1}\left(Q^{\mathbb R,+}_{Y_t|D=0}\left(F_{Y_{1|D=0}}(s)\right)-\right)\; \\\; F_t^{UB}(s)&=&F_{Y_t|D=1}\left(Q^{\mathbb R,-}_{Y_t|D=0}\left(F_{Y_{1|D=0}}(s)\right)\right).
\end{eqnarray*}}
\end{theorem}

\begin{corollary}[Model's Testable Restriction]\label{result:fals-test}
Suppose that $\mathbb Y_{t0|1} \subseteq \mathbb Y_{t0|0}$ for $t \in \{-T_0,\dots,0\}$, then 
if  Assumptions \ref{inc} and \ref{Gstab} hold, the following inequalities must be satisfied:
\begin{eqnarray*}
&&\Delta(y)\leq 0 \quad \forall y \in \mathbb{Y}_{10|0},~\text{where }\end{eqnarray*}
$\Delta(y)\equiv \max_{t \in \{-T_0,\dots,0\}}F_{Y_t}\left(Q^{\mathbb R,+}_{Y_t|D=0}\left(F_{Y_{1|D=0}}(y)\right)-\right)- \min_{t \in \{-T_0,\dots,0\}}F_{Y_t}\left(Q^{\mathbb R,-}_{Y_t|D=0}\left(F_{Y_{1|D=0}}(y)\right)\right)$.

\end{corollary}
We illustrate the arguments in Theorem \ref{GMain:theorem} and Corollary \ref{result:fals-test} in a numerical example motivated by our minimum wage setting in the presence of multiple pre-treatment periods in Section \ref{sec:illustration_MP}. 
\begin{figure}[htbp]\caption{CS bounds in the minimum-wage numerical example with CS holding for two pre-treatment periods}
    \vspace{0.3cm}
    {\scriptsize{\begin{tabular}{ccc}
   (a) Using $t\in\{-1,1\}$&(b) Using $t\in\{0,1\}$&(c) Using $t\in\{-1,0,1\}$\\\\
    \includegraphics[width=5cm]{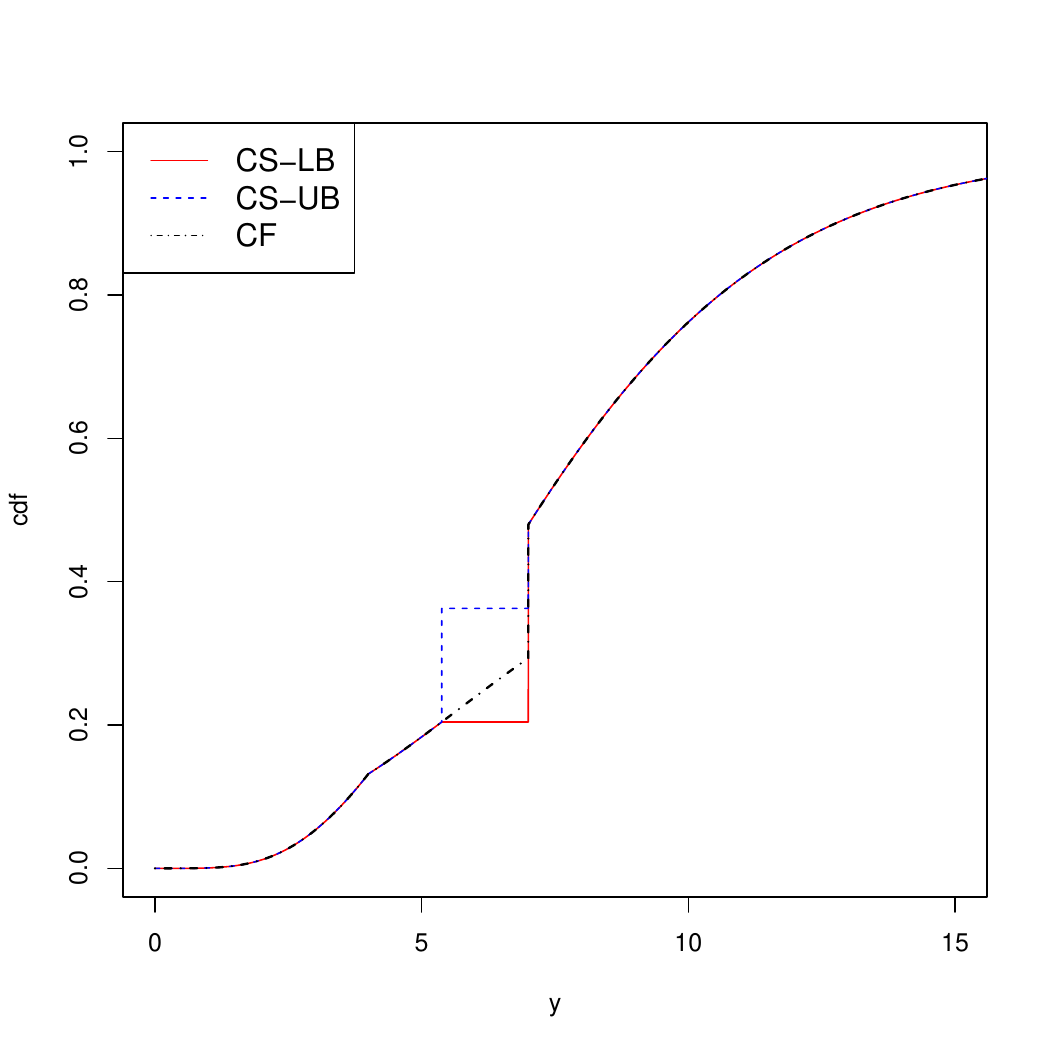}&  \includegraphics[width=5cm]{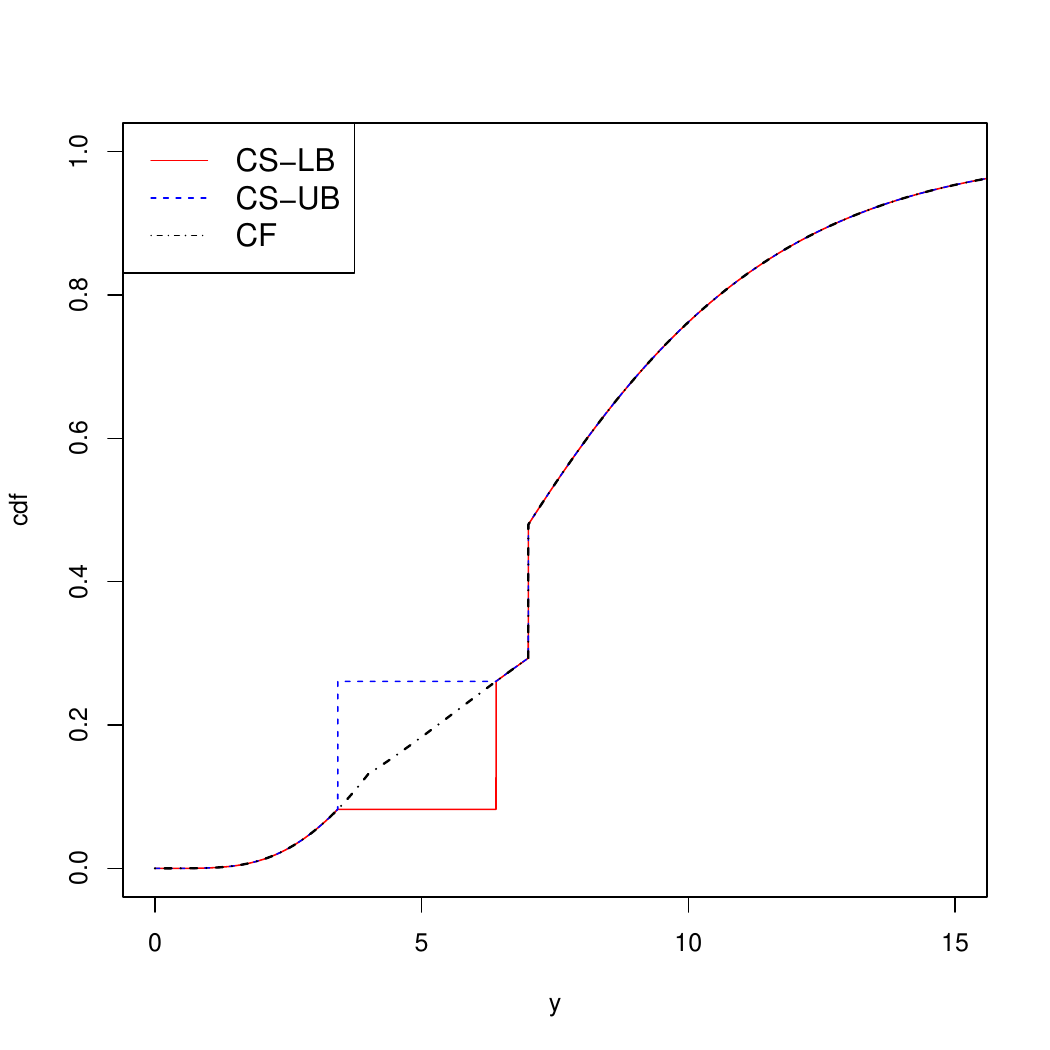}&    
    \includegraphics[width=5.25cm]{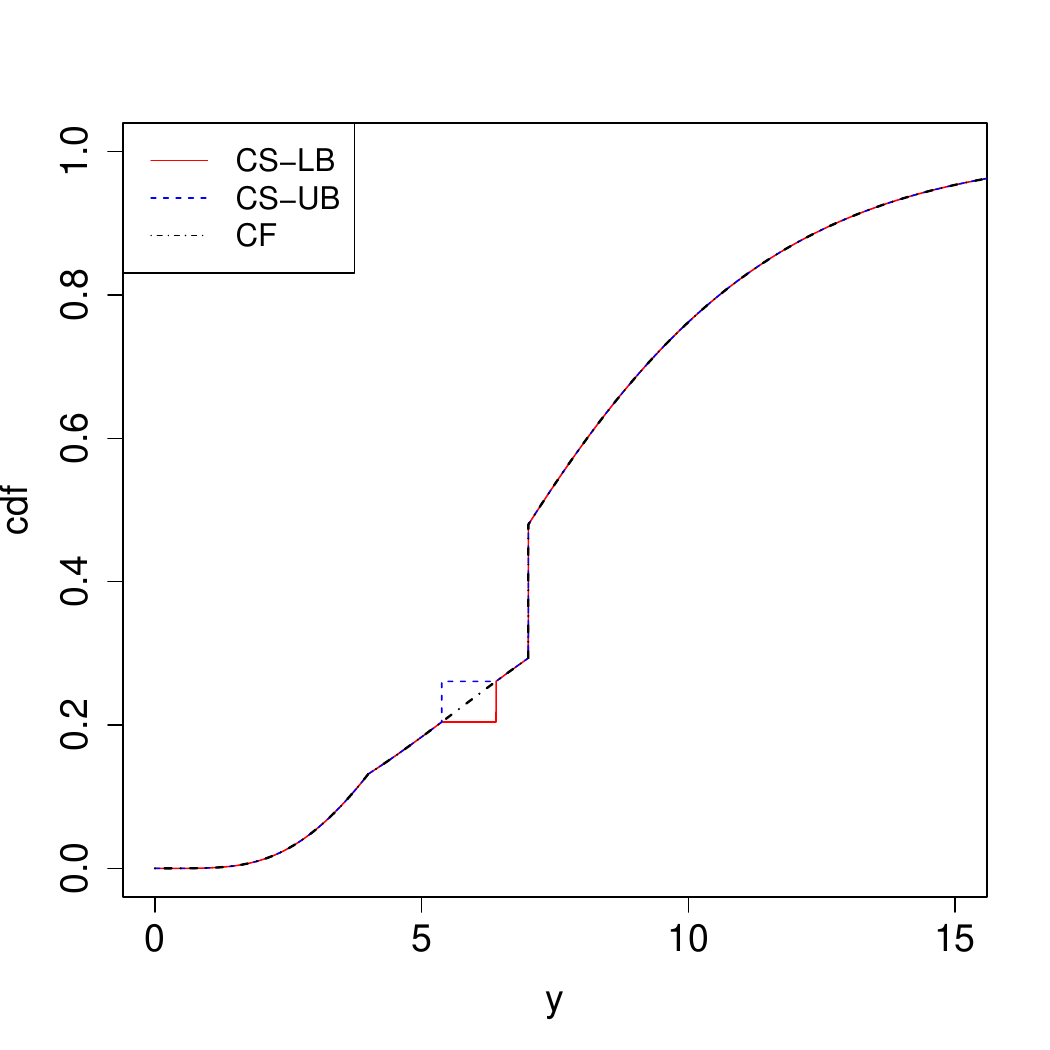}\\
    (d) $C_{Y_{t0},D}(\cdot,q)$ for $t=-1,0$&(e) $\Gamma_t(\cdot)$ for $t=-1,0$&(f) Model Testable Restriction ($\Delta$)\\
\includegraphics[width=5cm]{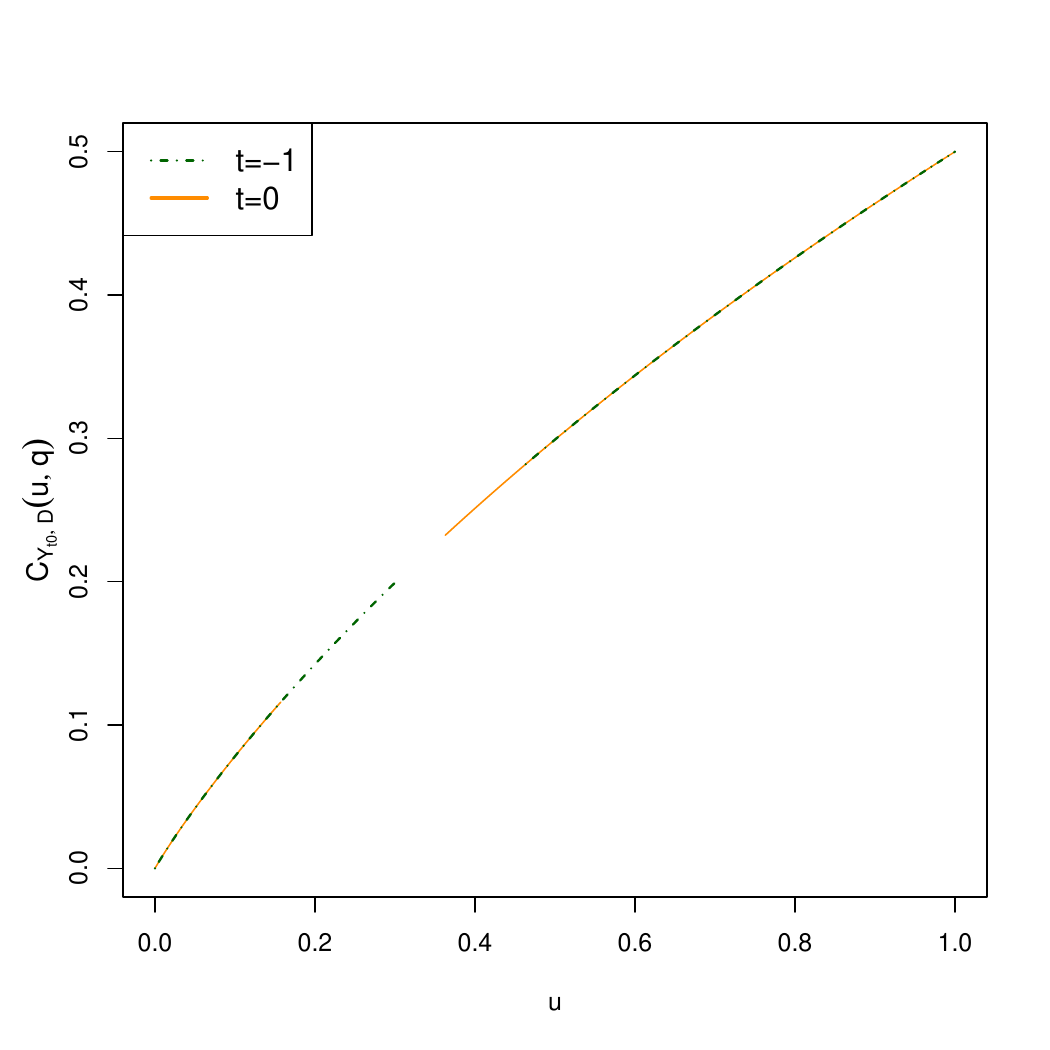}&\includegraphics[width=5cm]{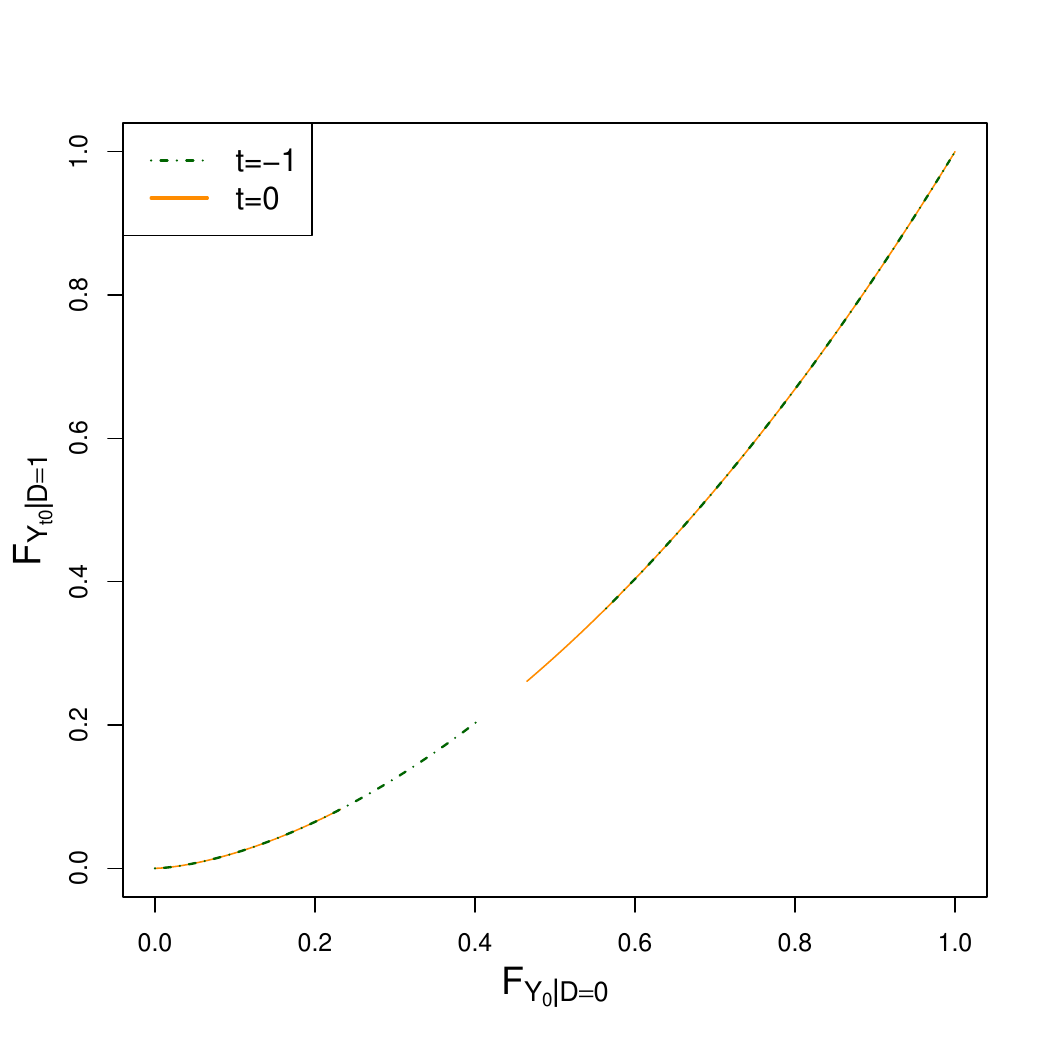}&\includegraphics[width=5cm]{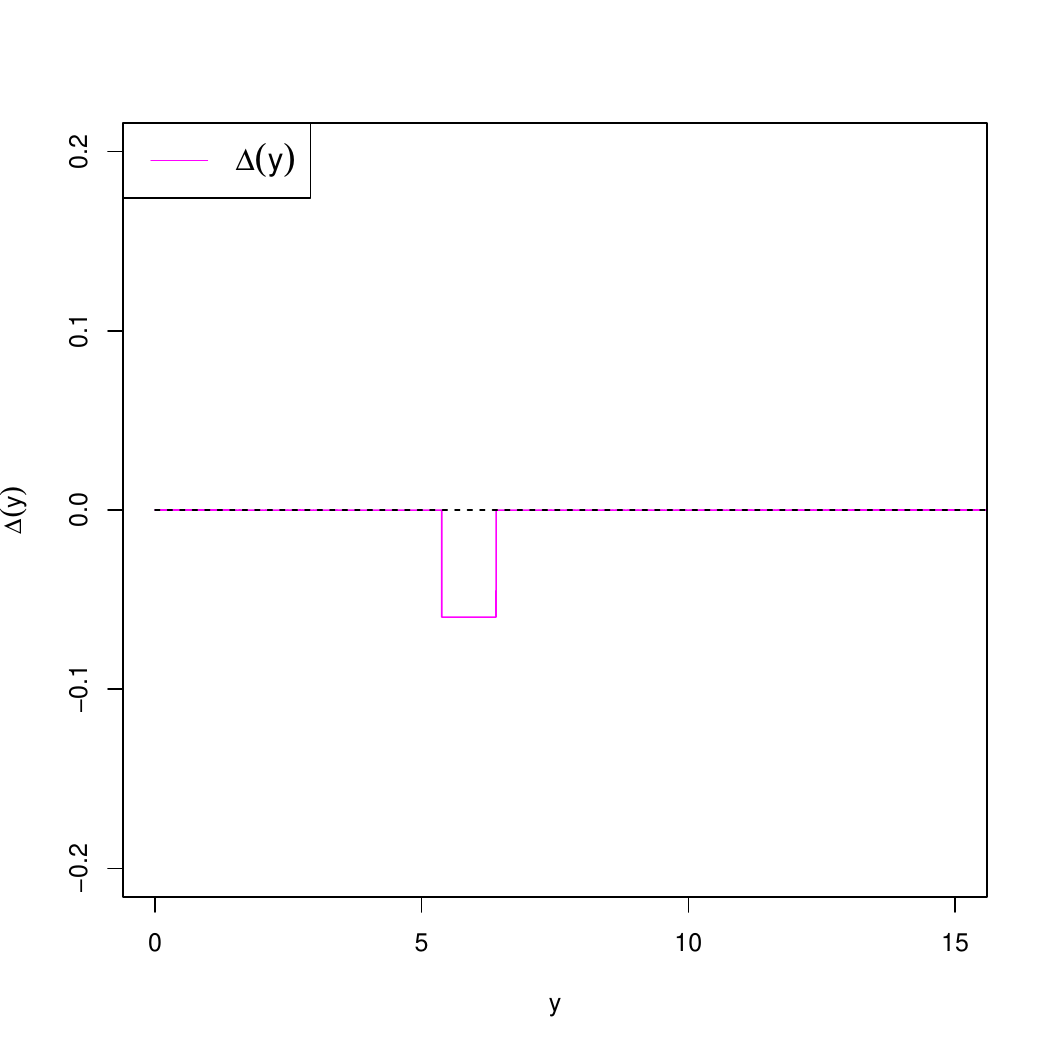}\\
    \multicolumn{3}{c}{\parbox{\textwidth}{\scriptsize{\emph{Notes}: $CF$ denotes the counterfactual distribution $F_{Y_{10}|D=1}$, $CS$-$LB$ and $CS$-$UB$ denote the CS lower and upper bound, respectively, on the counterfactual distribution.  To satisfy the copula stability assumption in periods $t\in\{-1,0,1\}$, we set $C_{Y_{-10,D}}=C_{Y_{00},D}=C_{Y_{10},D}$ to be the Clayton copula $C_{Y_0,D}(u,q)=(\max(u^{-\theta}+q^{-\theta}-1,0))^{-1/\theta}$ with $\theta=0.5$. The untreated potential outcome distributions for the treatment and control groups are given by the following for $t\in\{-1,0,1\}$:
$F_{Y_{t0}|D=0}(y)=\frac{1}{q}C_{Y_{t0},D}(F_{Y_{t}}(y),q)$,
$F_{Y_{t0}|D=1}(y)=\frac{1}{p}\left(F_{Y_{t0}}(y)-C_{Y_{t0},D}(F_{Y_{t0}}(y),q)\right)$.  The marginal distribution is given by  $F_{Y_{t0}}(y)=F_{Y_{t0}^*}(y)-b_{t0}(F_{Y_{t0}^*}(y)-F_{Y_{t0}^*}(\underline{w}_0))\mathbbm{1}\{y\in(\underline{w}_{0},c_{0})\}$. We set $c_0=7$, $\underline{w}_0=c_0-3$, $b_{-10}=0.5$, $b_{00}=0.75$, $b_{10}=0.5$, and $Y_{td}^*\sim \chi^2(k_{td})$ with $k_{-10}=8$, $k_{00}=9$ and $k_{10}=7$. }}}\\
    \end{tabular}}}\label{fig:mwexample_multiT0}
    \end{figure}
    \begin{figure}[htbp]\caption{CS bounds in the minimum-wage numerical example with CS holding for $t\in\{0,1\}$ only}
    \vspace{0.3cm}
    {\scriptsize{\begin{tabular}{ccc}
     (a) Using $t\in\{-1,1\}$&(b) Using $t\in\{0,1\}$&(c) Using $t\in\{-1,0,1\}$\\\\
    \includegraphics[width=5cm]{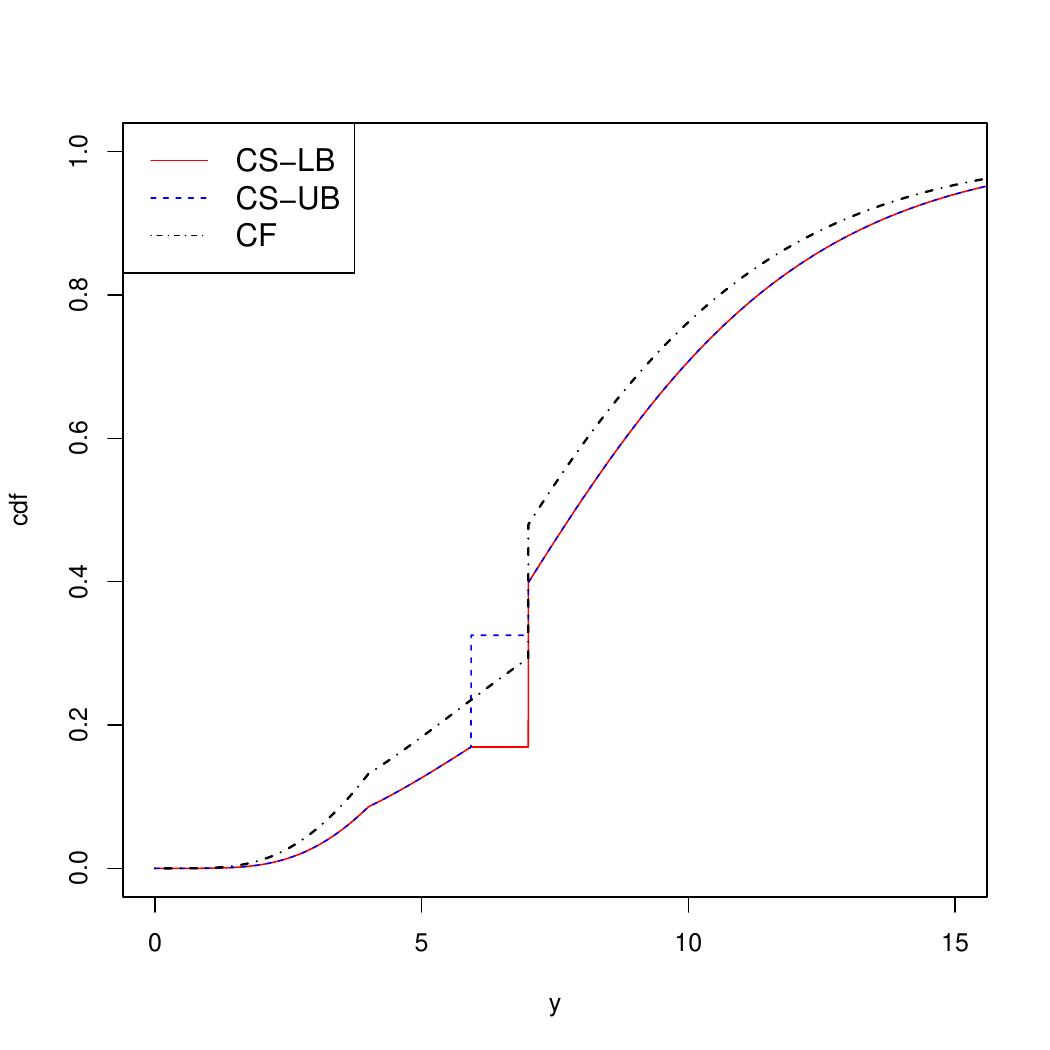}&  \includegraphics[width=5cm]{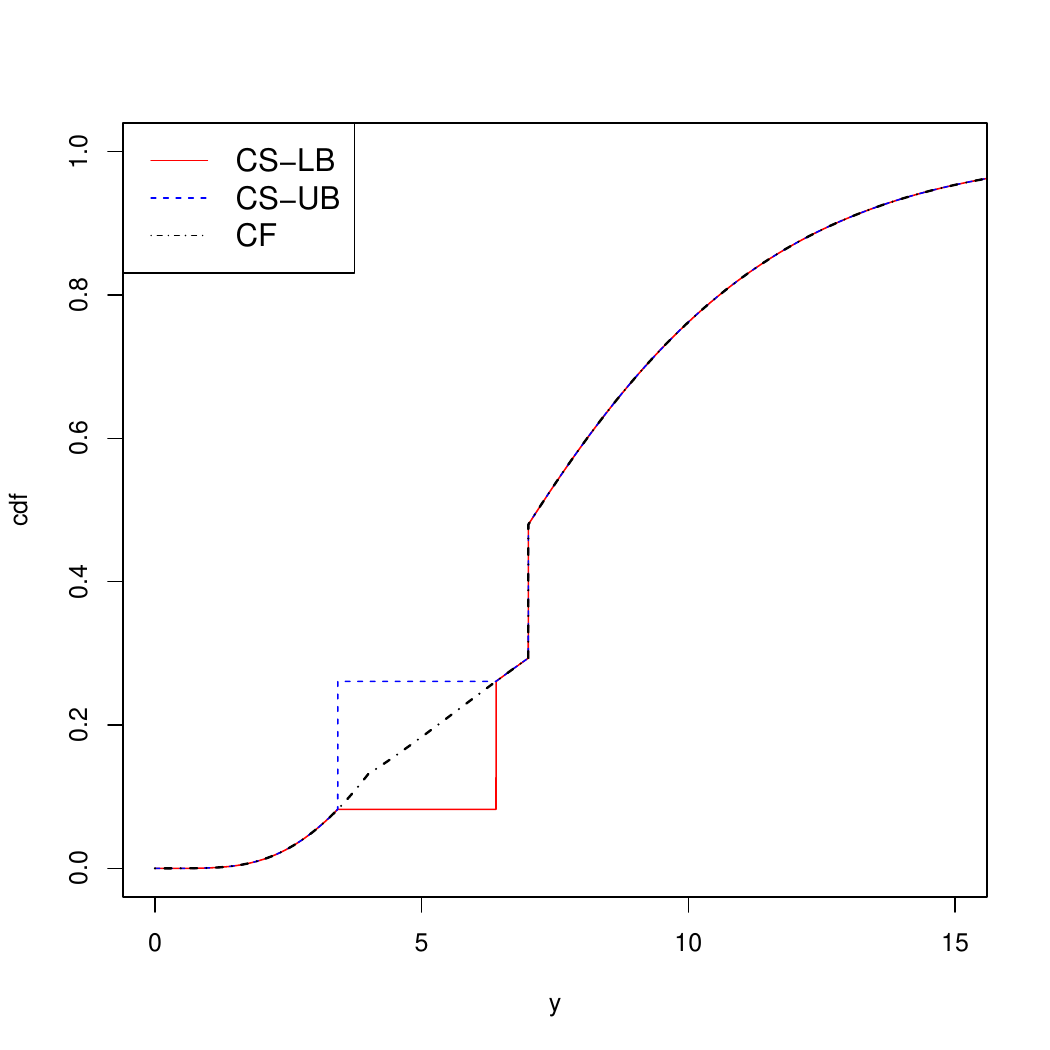}&    
    \includegraphics[width=5cm]{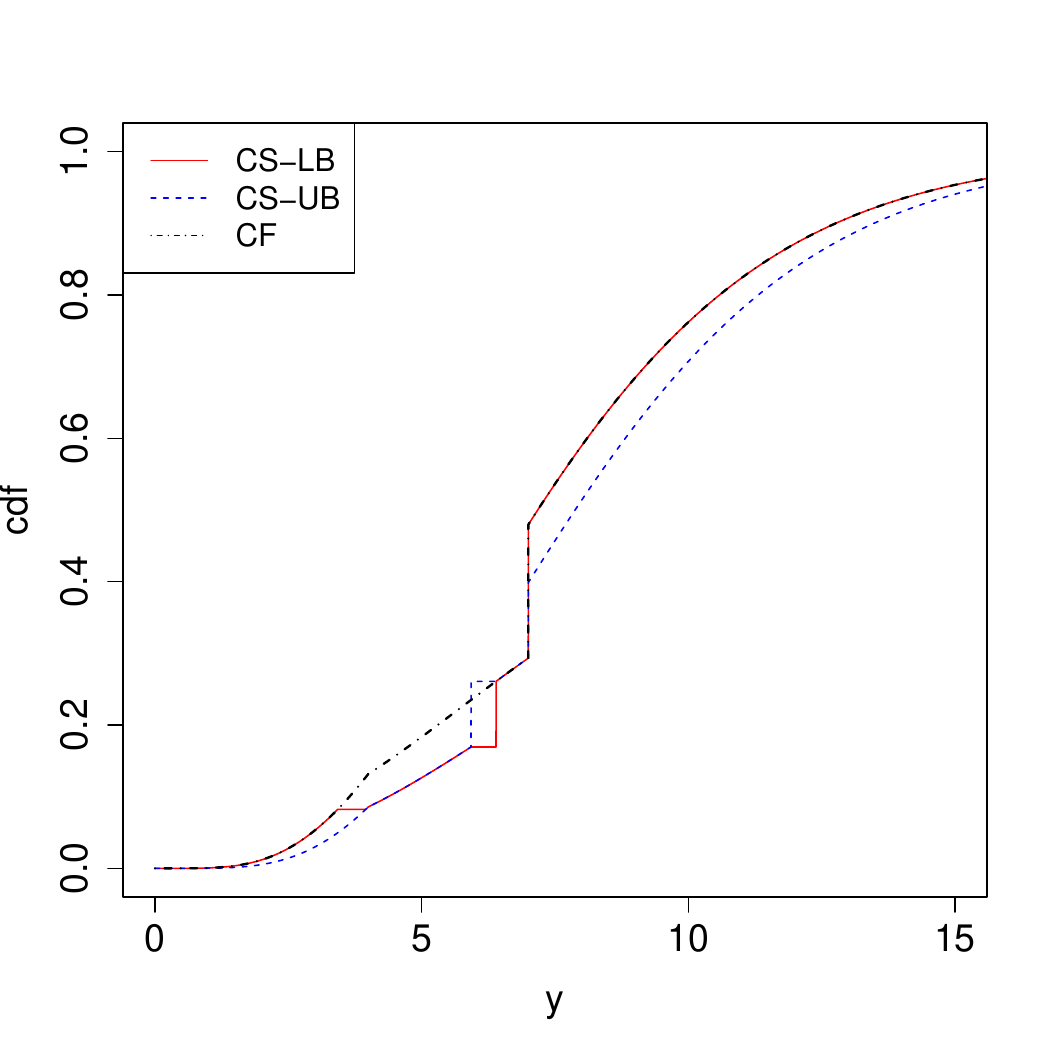}\\
    (d) $F_{Y_{t0}}(y)\mapsto F_{Y_{t0},D}(y,0)$&(e) $F_{Y_{t0}|D=0}(y)\mapsto F_{Y_{t0}|D=1}(y)$&(f) Model Testable Restriction ($\Delta$)\\
\includegraphics[width=5cm]{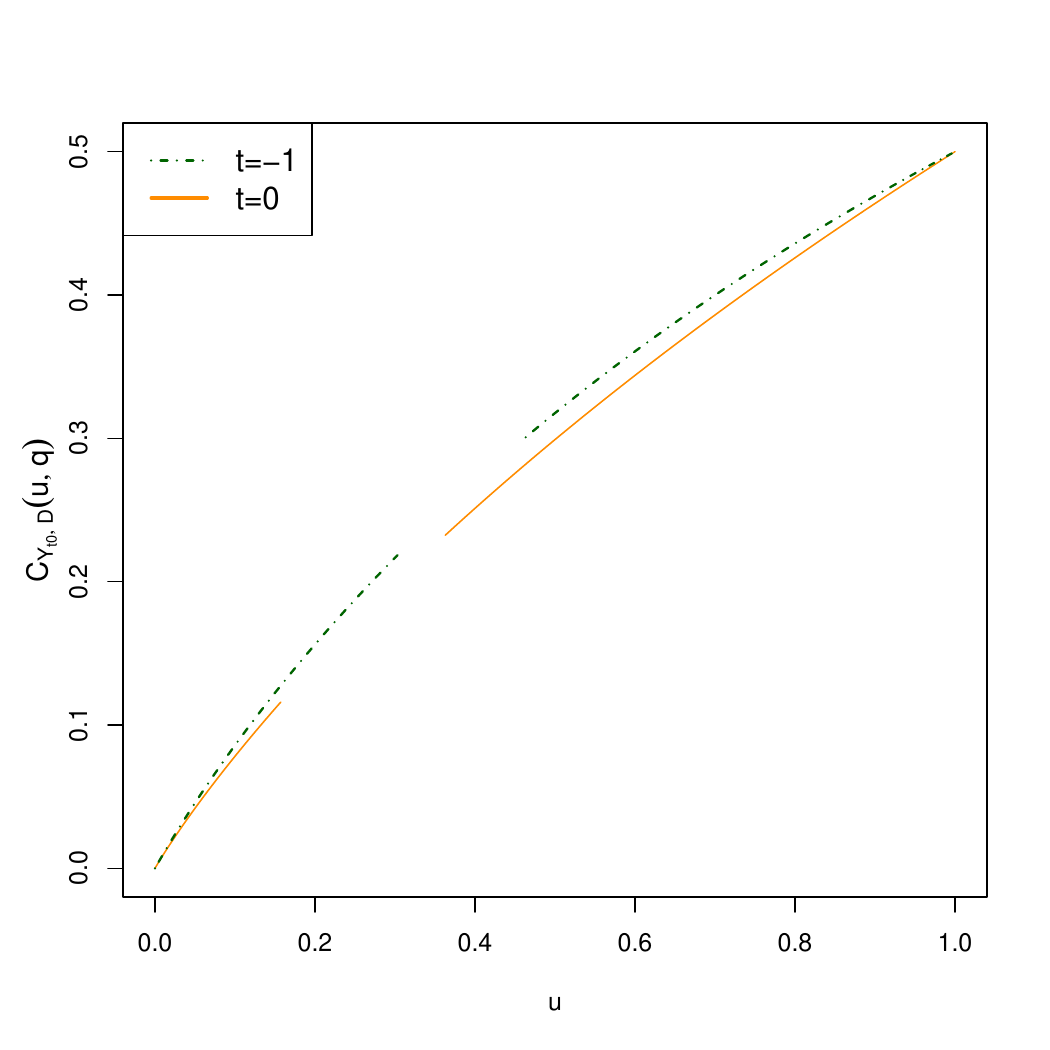}&\includegraphics[width=5cm]{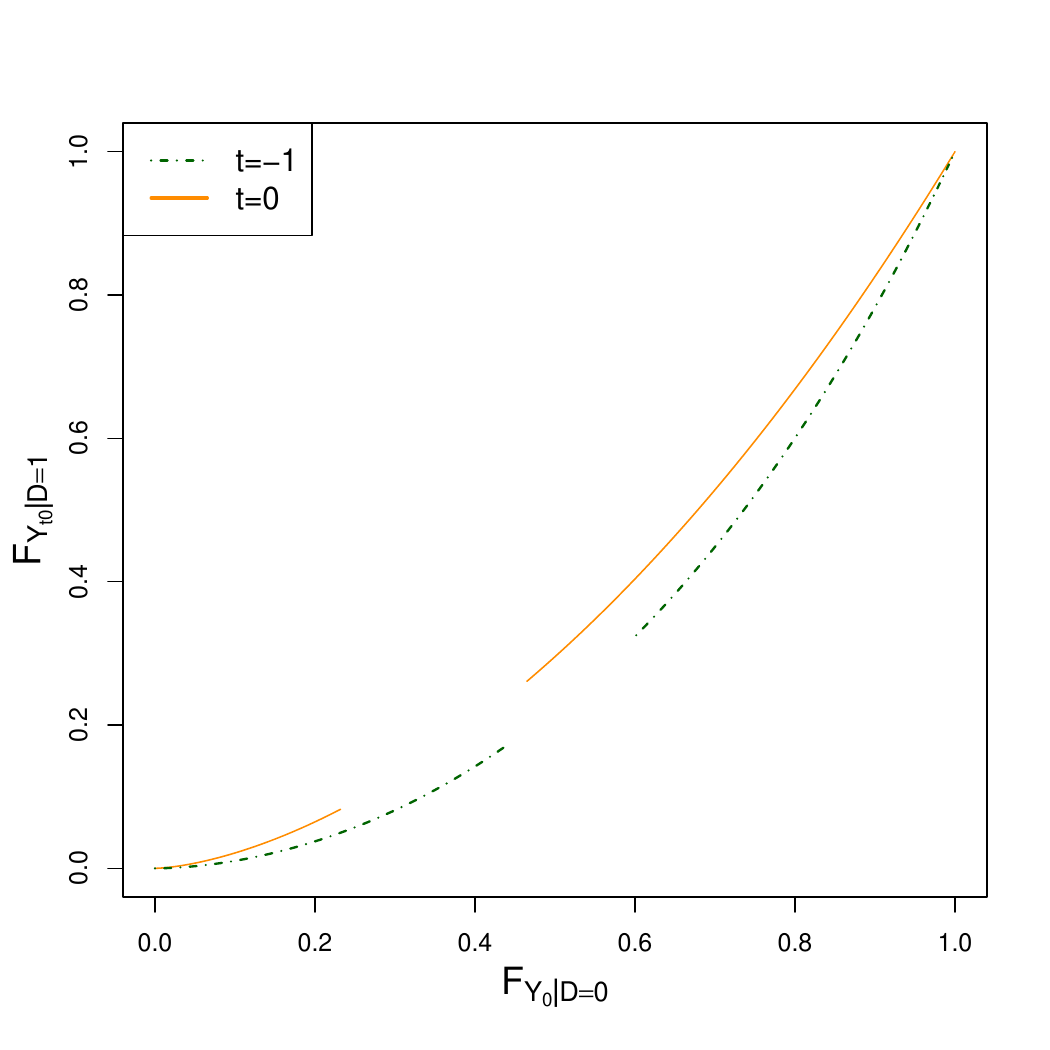}&\includegraphics[width=5cm]{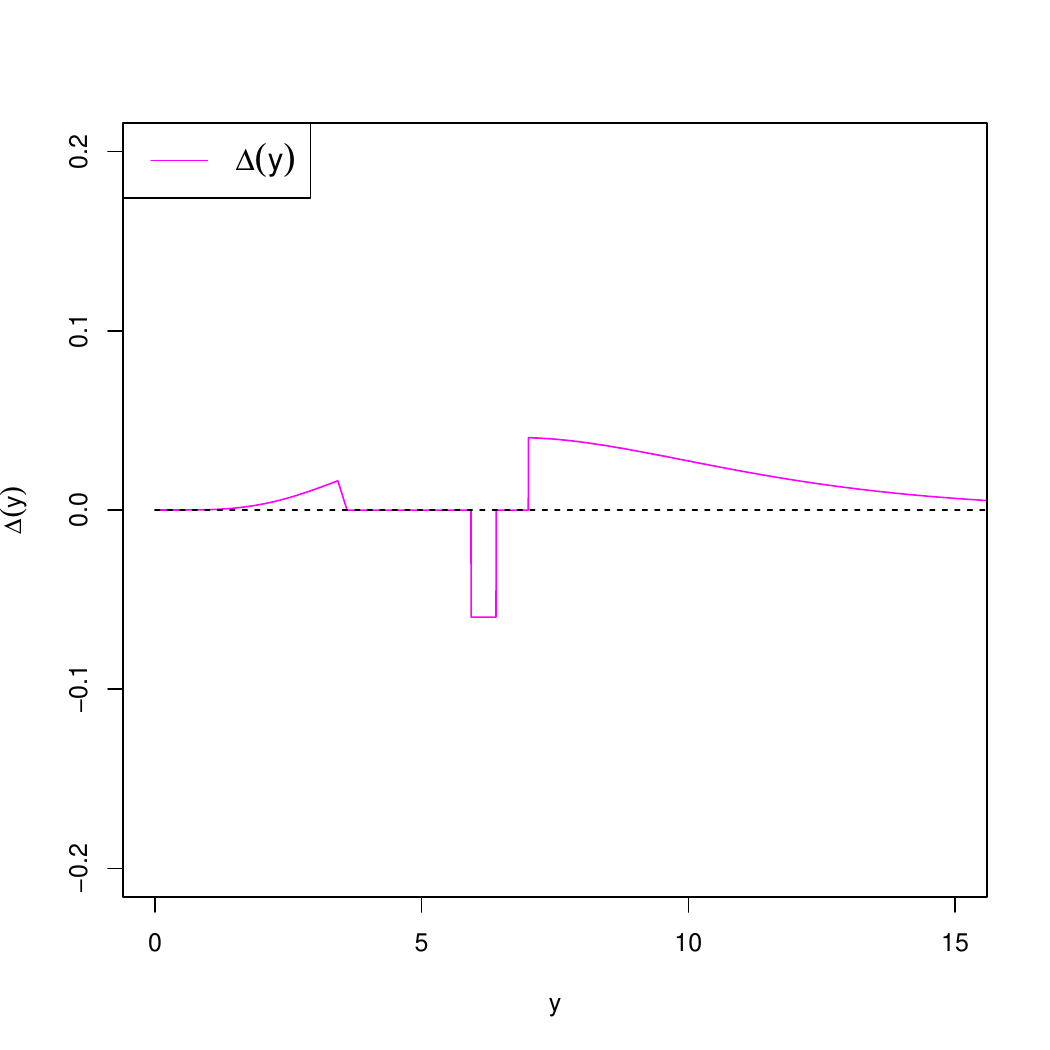}\\
\multicolumn{3}{c}{\parbox{\textwidth}{\scriptsize{\emph{Notes}: $CF$ denotes the counterfactual distribution $F_{Y_{10}|D=1}$, $CS$-$LB$ and $CS$-$UB$ denote the CS lower and upper bound, respectively, on the counterfactual distribution.  To satisfy the copula stability assumption in periods $t\in\{0,1\}$ only, we set $C_{Y_{00},D}=C_{Y_{10},D}$ to be the Clayton copula with $\theta=0.5$, whereas $C_{Y_{-10,D}}$ is the Clayton copula with $\theta=0.75$. The potential outcome distributions for the treatment and control groups are generated as described in Figure \ref{fig:mwexample_multiT0}.}}}\\
    \end{tabular}}}\\
    \medskip\label{fig:mwexample_multiT0_CSviolation}
    \end{figure}

\begin{remark}[Staggered adoption design]
    Suppose that we observe multiple post-treatment periods, $t=1,\dots,T_1$, where $D\in\{1,2,\dots,T_1,\infty\}$. For $t=1,\dots,T_1$, $D=t$ denotes the group that adopts the treatment in period $t$, and $D=\infty$ denotes the control group that is never-treated. Let $\mathbb{P}(D\leq \tau)=q_{\tau}$ for $\tau\in\{\infty,1,\dots, T_1-1\}$ and $Y_t^{\infty}$ denote the potential outcome in the control state. We can extend our identification approach to this setting under a suitable copula stability assumption, specifically assuming $C_{Y_{t}^{\infty},D}(\cdot,q_{\tau})=C_{Y_{(t-1)}^{\infty},D}(\cdot,q_{\tau})$ for $\tau\in\{\infty,1,\dots,T_1-1\}$ and $t\in\{-T_0+1,\dots,-1,0,1,\dots, T_1\}$.
\end{remark}

\subsubsection{Numerical Illustration}\label{sec:illustration_MP}
Here we illustrate the CS bounds with two pre-treatment periods as well as the testable restrictions in the context of a minimum-wage numerical example. Suppose that both treatment and control groups have a pre-existing minimum wage set at $c_{0}$ in the pre-treatment periods ($t=-1,0$). In the post-treatment period ($t=1$), the minimum wage increases for the treatment group to $c_{1}$. We consider two cases: (i) all model assumptions hold (Figure \ref{fig:mwexample_multiT0}), (ii) all assumptions except copula stability hold (Figure \ref{fig:mwexample_multiT0_CSviolation}).\footnote{In Appendix \ref{app:num_illustration_monviolation}, we demonstrate a third case, where copula stability holds, while the strict monotonicity of the horizontal copula is violated. This case demonstrates that we can detect violations of our model assumptions with only one pre-treatment period.} 

Figure \ref{fig:mwexample_multiT0} demonstrates that when copula stability holds for multiple pre-treatment periods, it can have significant gain in terms of identification as the multi-period CS bounds point-identifies the counterfactual distribution on a larger portion of its support in Panel (c) relative to Panels (a) and (b). Figure \ref{fig:mwexample_multiT0}(f) provides our model testable restriction, specifically $\Delta(y)\leq 0$, which holds in this case. Furthermore,  Panels (d) and (e) of Figure~\ref{fig:mwexample_multiT0} present $C_{Y_{t0},D}$ and $\Gamma_t$, respectively, for $t=-1,0$, which are equal on the intersection of their respective ranges.

Next, we demonstrate the case where copula stability only holds for $t\in\{0,1\}$, but not $t\in\{-1,1\}$. In Figure \ref{fig:mwexample_multiT0_CSviolation}, Panel (a) shows that using the pre-treatment period $t=-1$ only to construct the CS bounds yields bounds that do not include the counterfactual, whereas Panel (b) shows that the counterfactual is included in the CS bounds with pre-treatment period $t=0$ only. When considering the CS bounds using both pre-treatment periods in Figure \ref{fig:mwexample_multiT0_CSviolation}(c), we note that the CS lower bound is greater than the CS upper bound, and our model testable restriction is violated as indicated by Figure \ref{fig:mwexample_multiT0_CSviolation}(f). Relatedly, Figures \ref{fig:mwexample_multiT0_CSviolation}(d) and  \ref{fig:mwexample_multiT0_CSviolation}(e) demonstrate that the mappings $C_{Y_{t0},D}(\cdot,q)$ and $\Gamma_{t}(\cdot)$, respectively,  are not equal for $t=-1,0$, indicating a violation of copula stability.

\subsection{Connection to Changes-in-Changes}\label{sec:cic_connection}

In this section, we elaborate on the connection between our copula stability assumption and the CiC conditions in \citet{AtheyImbens2006}. We first show the equivalence between copula stability and the CiC conditions for continuous outcome distributions. Second, while the identification results in \citet{AtheyImbens2006} do not account for mixed outcomes, a researcher might still rely on their estimand. Here, we demonstrate that a na\"ive implementation of the CiC approach leads to a point/bound estimand that might not include the true counterfactual, whereas our CS bounds will.
Finally, for discrete outcomes, we demonstrate using an analytical example that copula stability can be compatible with multi-dimensional unobserved heterogeneity, whereas the CiC conditions require unobserved heterogeneity to be uni-dimensional.

\subsubsection{Continuous outcomes}
The following result demonstrates that the CiC conditions for continuous, strictly increasing outcome distributions are equivalent to our copula stability assumption. In Appendix \ref{Appen:equivalence_general}, we demonstrate how this result extends to all continuous outcomes. For other outcome distributions, this equivalence does not hold in general. 
\begin{claim}\label{claim:eq}
Assume the cdfs $F_{Y_{t0}}(.)$ for $t \in \{0,1\}$ are continuous and strictly increasing, then the following two statements are equivalent:
\begin{enumerate}
\item [(i)] $C_{Y_{00},D}(u,q)=C_{Y_{10},D}(u,q)$ for all $u \in [0,1]$.
\item [(ii)] There exist two strictly increasing functions $h_t(.), t \in \{0,1\}$ and two uniformly distributed random variables over $[0,1]$, $U_{00}$ and $U_{10}$, such that $Y_{t0}=h_t(U_{t0})$ and $U_{00}|D=d \sim U_{10}|D=d$ for $d \in \{0,1\}$. 
\end{enumerate}

\end{claim}
The proof of this claim is in Appendix \ref{proof:claim1}. The main intuition behind it is that for this class of distributions we can write $Y_{t0}=Q_{Y_{t0}}^{\mathbb{R},-}(U_{t0})$, where $U_{t0}=F_{Y_{t0}}(Y_{t0})\sim \mathcal{U}[0,1]$. As a result, the marginal distribution of $U_{t0}$ is stable across time by construction and the stability of the copula between $U_{t0}$ and $D$ is necessary and sufficient for the stability of $U_{t0}|D$, which is the conditional time invariance assumption in \citet{AtheyImbens2006}. Its equivalence to our copula stability assumption follows from the invariance of the copula under strictly monotonic transformations.

\subsubsection{Mixed outcomes} Here, we demonstrate that for mixed outcomes the CiC point/bound estimand may not cover the true counterfactual distribution in the context of the numerical minimum-wage example in Section \ref{sec:illustration_MP}.

The CiC bounds in the discrete case are defined for any $s \in \mathbb{Y}_{1|0}$ as follows for $t\in\{-1,0\}$,
\begin{eqnarray*}
F_{t,\text{CiC}}^{\text{LB}}(s) &=& F_{Y_t|D=1}\left(Q_{Y_t|D=0}^{\mathbb{Y}_{t|0},+}\left(F_{Y_1|D=0}(s)\right)\right), \\
F_{t,\text{CiC}}^{\text{UB}}(s) &=& F_{Y_t|D=1}\left(Q_{Y_t|D=0}^{\mathbb{Y}_{t|0},-}\left(F_{Y_1|D=0}(s)\right)\right).
\end{eqnarray*}

In the example illustrated in Figure~\ref{fig:mwexample_multiT0}, we have $\mathbb{Y}_{t|0} = \mathbb{R}^+$, and $Y_t|D=0$ has a strictly increasing cdf in $\mathbb{R}^+$. Then the following simplifications hold:
\begin{multline*}
F_{t,\text{CiC}}^{\text{UB}}(s) = F_{Y_t|D=1}\left(Q_{Y_t|D=0}^{\mathbb{Y}_{t|0},-}\left(F_{Y_1|D=0}(s)\right)\right) 
= F_{Y_t|D=1}\left(Q_{Y_t|D=0}^{\mathbb{R},-}\left(F_{Y_1|D=0}(s)\right)\right) 
= F_t^{\text{UB}}(s),
\end{multline*}
and
\begin{multline*}
F_{t,\text{CiC}}^{\text{LB}}(s) = F_{Y_t|D=1}\left(Q_{Y_t|D=0}^{\mathbb{Y}_{t|0},+}\left(F_{Y_1|D=0}(s)\right)\right) 
= F_{Y_t|D=1}\left(Q_{Y_t|D=0}^{\mathbb{R},+}\left(F_{Y_1|D=0}(s)\right)\right) \\
\geq \mathbb{P}\left(Y_t < Q_{Y_t|D=0}^{\mathbb{R},+}\left(F_{Y_1|D=0}(s)\right) \mid D=1\right) 
= F_t^{\text{LB}}(s),
\end{multline*}
where the inequality becomes strict at points of discontinuity.

More importantly, we can see that $F_{t,\text{CiC}}^{\text{LB}}(s)=F_{t,\text{CiC}}^{\text{UB}}(s)$, since $Q_{Y_t|D=0}^{\mathbb{Y}_{t|0},+}(u)=Q_{Y_t|D=0}^{\mathbb{Y}_{t|0},-}(u)$ for $u\in[0,1]$. However, this CiC point estimand is different from the true counterfactual of interest $F_{Y_{10}|D=1}$, as shown in Figure~\ref{fig:mwexample_CiC_T0}.

Therefore, in this case, our bounds  contain the CiC (point/bound) estimands  and the true counterfactual \[F_{t,\text{CiC}}^{\text{UB}}(s)\neq F_{Y_{10}|D=1}(s), \text{ where } \{F_{t,\text{CiC}}^{\text{UB}}(s),F_{Y_{10}|D=1}(s)\}\in [F_t^{\text{LB}}(s), F_t^{\text{UB}}(s)].\]

In sum, in this mixed-outcome example, if the researcher ignores the discontinuity and applies the CiC point estimand or applied the CiC bounds for the discrete case, their estimand will not cover the true counterfactual, as shown in Figure~\ref{fig:mwexample_CiC_T0}.

    \begin{figure}[H]\caption{CiC Point Estimand with CS holding for all three periods}
            \begin{tabular}{cc}
            Using $t\in\{-1,1\}$&Using $t\in\{0,1\}$\\
    \includegraphics[width=6cm]{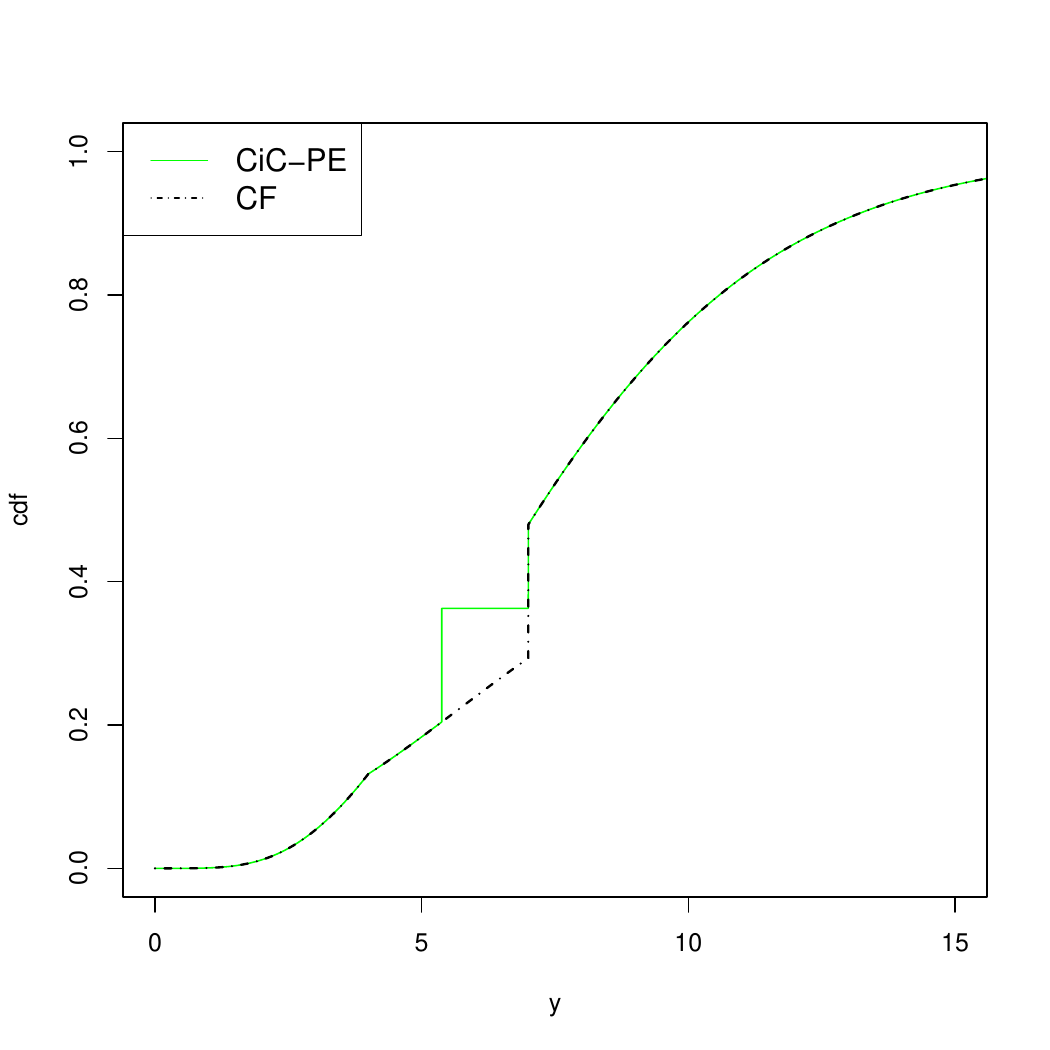}&\includegraphics[width=6cm]{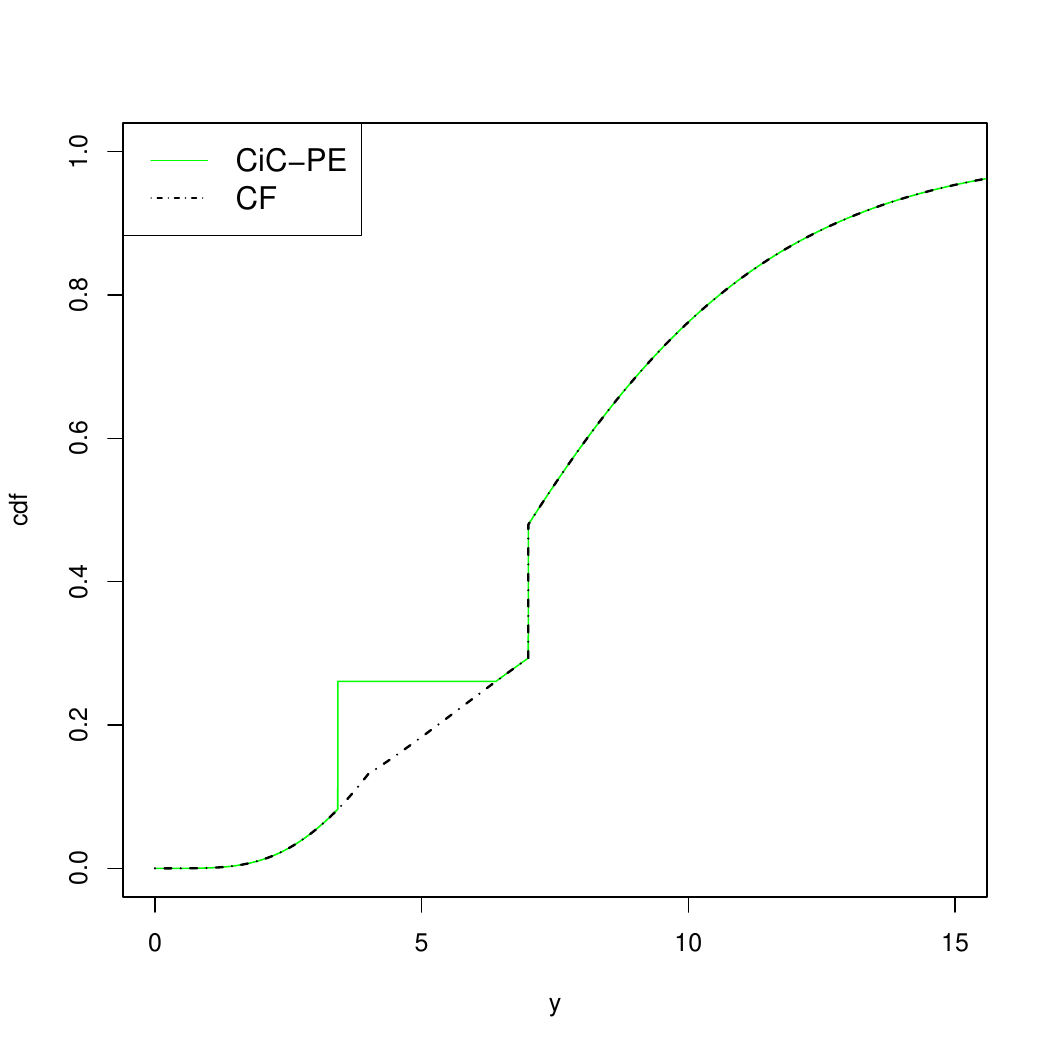}\\
\multicolumn{2}{c}{\parbox{12cm}{\scriptsize{\emph{Notes}: $CF$ denotes the counterfactual distribution $F_{Y_{10}|D=1}$, $CiC$-$PE$ denotes the CiC point-estimand. The copula and potential outcome distributions follow the specifications in Figure \ref{fig:mwexample_multiT0}.}}}\\
   \end{tabular}\label{fig:mwexample_CiC_T0}
   \end{figure}

\subsubsection{Discrete outcomes}
For the case of discrete outcomes, the following example illustrates that our identifying assumption can be compatible with multi-dimensional unobserved heterogeneity, whereas the CiC conditions require scalar unobserved heterogeneity.

\begin{example}[Binary outcome model with multidimensional unobserved heterogeneity]\label{ex:multidimensional}
    Consider the following model 
    \begin{eqnarray}
        Y_t &=&1-\mathbbm{1}\{\eta t D + U_t \leq c_t,\tilde{\eta} t D + \tilde{U}_t \leq \tilde{c}_t\},~~~ t=0,1,\nonumber\\
D&=&\mathbbm{1}\{V>q\},    \end{eqnarray}
    where $(Y_0,Y_1,D)$ is an observed random vector, $(\eta, \tilde{\eta}, U_0, U_1, \tilde{U}_0,\tilde{U}_1)$ is a latent random vector, and $(c_t, \tilde{c}_t)$ is a constant vector. For simplicity, we normalize $U_t$, $\tilde{U}_t$ and $V$ to be uniformly distributed on $[0,1]$. The untreated potential outcome $Y_{t0}$ is $$Y_{t0} =1-\mathbbm{1}\{U_t \leq c_t,\tilde{U}_t \leq \tilde{c}_t\}.$$ 
For instance, $D$ could be the student loan forgiveness program, $Y_t$ could be a college attendance decision, $U_t$ and $\tilde{U}_t$ could respectively be father's and mother's wealth in the absence of the program. This model assumes that an individual decides to attend college if at least one of the parents' wealth is above a (parent-specific) threshold, whether they were to receive the loan forgiveness program or not.  While the CiC approach does not allow multidimensional unobserved  heterogeneity, we show in Appendix \ref{proof:ex-multi-dimensional} that for the wide class of Archimedean copulas the stability of the dependence structure of the latent variables $(U_t,\tilde{U}_t,V)$ over time implies our copula stability assumption.
\end{example}

While our assumption accommodates a broader class of binary outcome models than the CiC model assumption, it does not necessarily yield tighter bounds. As illustrated in Figure \ref{fig:poisson} in the Online Appendix, both approaches produce the same bounds in the discrete-outcome case.

\section{Policy-relevant parameters: Social welfare treatment effect on the treated (SWTT)}\label{sec:swtt}
Building on our unifying, partial identification result for the counterfactual distribution, we provide a class of policy-relevant parameters that quantify the impact of policy on social welfare in the entire population, subpopulations in the lower tail of the distribution or over any interquantile range of the distribution. In general, when a policymaker decides to implement a new policy such as an increase in the legal minimum wage or legal minimum working time,  she expects the policy to have a specific social welfare impact. 
The social welfare function used by the policymaker is not necessarily known to the researcher, however. 
For instance, the policymaker may consider social welfare functions that put more weight on specific subpopulations, such as lower-income individuals, or considers only social welfare functions with specific properties like social welfare functions that respect the Pigou-Dalton principle of transfers\footnote{The Pigou-Dalton principle states that a transfer of income from a higher-ranked individual to a lower-ranked individual that does not change their ranks is always desirable.} or the rank-dependent social welfare functions introduced by
\citeauthor{Mehran1976} (\citeyear{Mehran1976}).\footnote{See \citeauthor{Aabergeetal2013} (\citeyear{Aabergeetal2013}) for a detailed discussion.}

 As we clarify below, the  widely used average treatment effect on the treated (ATT) corresponds to the case where the policymaker is inequality-neutral. If  the policymaker is averse to inequality, however, the ATT would not be an adequate causal parameter to measure the impact of the policy or judge its effectiveness.

 For this particular reason, we propose a class of parameters of interest that measure the causal effect of a particular policy in terms of a social welfare function,
 \begin{eqnarray*}
    SWTT_{\omega}&\equiv&SW_{\omega}(F_{Y_{11}|D=1})-SW_{\omega}(F_{Y_{10}|D=1}),\\
    &=&\int_{0}^{1}\omega(\tau)\left(Q^{\mathbb  R,-}_{Y_{11}|D=1}(\tau)-Q^{\mathbb  R,-}_{Y_{10}|D=1}(\tau)\right)d\tau, 
 \end{eqnarray*}
where $SW_{\omega}(F_X)=\int_{0}^{1}\omega(\tau)Q^{\mathbb  R,-}_{X}(\tau)$ denotes the social welfare function associated with a specific distribution $F_X$, and  $\omega(\tau) \in [0,1]$ is a weighting function. This social welfare function can be alternatively viewed as a weighted average of the outcomes of individuals $i$ where the weights depend on the rank of $X_i$, $SW_\omega=\int{X_i\omega(Rank(X_i))}di$ \citep{KitagawaTetenov2021}. Since the social welfare function essentially weights different quantiles of the distribution, the choice of the functional form of the weighting function relates to the inequality aversion of the policymaker and the extent thereof. We next consider several examples of weighting functions and discuss the properties of the social welfare functions they imply.

Before we proceed, it is important to emphasize that, while in many applications where measuring inequality is a concern, the outcome $Y$ is typically income or wages, our framework allows $Y$ to denote other outcomes as well as functions of different outcomes, such as consumption, income and/or human capital. Our $SWTT_{\omega}$ is also a generalization of the quantile treatment effect parameter discussed in 
\citet{Abadieetal2002}, \citet{Firpo2007}, and \citet{FrohlichMelly2008}.

\subsection{Generalized Gini social  welfare function}
The class of generalized Gini social welfare functions is the class of rank-dependent, equality-minded social welfare functions which satisfy the Pigou-Dalton principle of transfers and is given by 
\begin{eqnarray*} SW_{\Lambda}(F_X)=\int{\Lambda(F_X(x))}dx,\end{eqnarray*}
where $\Lambda(\cdot):[0,1] \mapsto  [0,1]$ 
is a convex, non-increasing, and non-negative  function with  boundary conditions $\Lambda(0)=1$ and $\Lambda(1)=0$.  This class admits the equivalent representation as a weighted sum of quantiles with weighting function $\omega(\tau)=\frac{\partial(1-\Lambda(\tau))}{\partial \tau}$,
\begin{eqnarray*}
SW_{\Lambda}(F_X)&=& SW_{\omega}(F_X)=\int{\omega(\tau)Q_X^{\mathbb{R},-}(\tau)}d\tau.
\end{eqnarray*}
As a result, the class of social welfare treatment effect parameters we introduce include this class as a special case. We proceed to present two important special cases of this class of social welfare functions, specifically the utilitarian and Gini social welfare functions.

\subsubsection*{Utilitarian welfare function}
When $\omega(\tau)=1$, we have $SW_{\omega}(F_X)=\int_{0}^{1}Q^{\mathbb  R,-}_{X}(\tau)d\tau$ $=\mathbb E[X].$ This corresponds to  the additive welfare function and in this case our proposed parameter boils down to the  ATT, i.e. $SWTT_{\omega}=ATT$. The ATT is therefore the appropriate parameter if the policymaker weights subpopulations at different quantiles of the distribution equally.

\subsubsection*{Gini social welfare function}\label{sec:giniswtt}
When $\omega(\tau)=2(1-\tau)$, we have  
$SW_{\omega}(F_X)=\int_{0}^{1}2(1-\tau)Q^{\mathbb  R,-}_{X}(\tau)d\tau$ $=\mathbb E[X]\left(1-I_{Gini}(F_X)\right),$
where $I_{Gini}(F_X)\equiv \frac{\int_{0}^{1}(2\tau-1)Q^{\mathbb  R,-}_{X}(\tau)d\tau}{\mathbb E[X]}$ is the widely used Gini inequality index, see \citet{Sen1974}. 
$SW_{\omega}(F_X)$ reflects the trade-off between the mean and (in)equality in the distribution $F_X$. The
product $\mathbb E[X]I_{Gini}(F_X)$ is a measure of the loss in social welfare due to inequality in the distribution $F_X$.
In that case, $SWTT_{\omega}$ captures the impact of the policy using the Gini social welfare function, see \citet{BlackorbyDonaldson1978} and \citet{Weymark1981}. In other words, if the policymaker implements the policy in order to reduce the  level of inequality measured by the Gini index, this parameter is the most adequate to judge the impact of this policy.

\subsection{Second-order dominance}
In many cases, when it is possible to do so, most inequality-averse policymakers  like to rank   distribution functions consistently
with second-degree dominance.
For instance, we say 
$F_{Y_{11}|D=1}$ second-order dominates $F_{Y_{10}|D=1}$
if and only if:
\begin{eqnarray*}
SWTT_{\omega}(u)&\equiv& SW_{\omega}(u,F_{Y_{11}|D=1})-SW_{\omega}(u,F_{Y_{10}|D=1})\\
&=&\int_{0}^{u}\left(Q^{\mathbb  R,-}_{Y_{11}|D=1}(\tau)-Q^{\mathbb  R,-}_{Y_{10}|D=1}(\tau)\right)d\tau\geq 0, 
 \end{eqnarray*}
 for all $u \in [0,1]$ and holds strictly for some $u$. In this special case, we have $\omega(\tau)=\mathbbm{1}\{\tau \leq u\}$. It is possible, however, that the observed and counterfactual distribution cannot be ranked using this criterion. Furthermore, the policy's objective may be to reduce inequality in a specific part of the distribution. We therefore consider the following quantile-specific Gini social welfare functions.

\subsection{Quantile-specific lower tail Gini social welfare
function}\label{sec:giniswtt_lower-tail}

In the Gini social welfare function discussed above, we assume that the policymaker is interested in the inequality of the whole population. Some policies may be concerned with reducing inequality up to specific quantiles of the distribution, such as minimum-wage policies \citep[e.g.][]{Dube2019,Cengizetal2019}. To quantify the impact of the policy on lower-tail quantiles, we extend the quantile-specific lower-tail Gini social welfare measures introduced in  \citeauthor{Aabergeetal2013} (\citeyear{Aabergeetal2013}) for continuous distributions to any type of distribution in order to accommodate the possibility of discontinuities resulting from censoring or bunching. To do so, we introduce the random variable
$X^u=Q_X^{\mathbb{R},-}(V)$, where $V\sim \mathcal{U}[0,u]$ for $u\in(0,1]$.\footnote{For $u\in Ran F_X$, $F_{X^u}(x)=\mathbb P(X\leq x|X\leq Q_X^{\mathbb{R},-}(u))$ for any $x\leq Q_X^{\mathbb{R},-}(u)$, thereby yielding the same truncated random variable introduced in \citeauthor{Aabergeetal2013}(\citeyear{Aabergeetal2013}). For $u\notin Ran F_X$, $X^u$ remains a well-defined random variable.}
 We relegate the derivations relevant to this section to Appendix \ref{sec:swtt_derivations}.

With this definition of $X^u$, we can show that the lower-tail Gini social welfare function can be decomposed into $\mathbb E[X^u]$ and the Gini coefficient associated with $F_{X^u}$ as follows
$$\int_{0}^{1}\frac{2}{u^2}(u-\tau)\mathbbm{1}\{\tau \leq u\}Q^{\mathbb  R,-}_{X}(\tau)d\tau=\mathbb E[X^u]\left(1-I_{Gini}(F_{X^u})\right),$$
where $I_{Gini}\left(F_{X^u}\right)\equiv\frac{\int_{0}^{1}(2\tau-u)\mathbbm{1}\{\tau\leq u\}Q^{\mathbb  R,-}_{X}(\tau)d\tau}{u^2\mathbb E[X^u]}$ is the lower-tail Gini coefficient at $u$ defined in \citet{Aabergeetal2013}. Therefore, 
$SWTT_{\omega}$ with  $\omega(\tau)=\frac{2}{u^2}(u-\tau)\mathbbm{1}\{\tau \leq u\}$ yields the following,
\begin{eqnarray*}
SWTT_{\omega}(u)&=&\int_{0}^{u}\frac{2}{u^2}(u-\tau)\left(Q^{\mathbb  R,-}_{Y_{11}|D=1}(\tau)-Q^{\mathbb  R,-}_{Y_{10}|D=1}(\tau)\right)d\tau,
 \end{eqnarray*}
and is interpreted as the Quantile-$u$ lower tail Gini social welfare treatment effect on the treated.

\subsection{Interquantile Gini social welfare function}\label{sec:giniswtt_interquantile} Since policies may target other parts of the distribution, such as the upper tail, we can generalize these quantile-specific social welfare treatment effect measures to any range of quantiles $[\underline{u},\overline{u}]$ a researcher may be interested in. Specifically, let $\underline{u}\in[0,1]$, $\overline{u}\in[0,1]$, $\underline{u}<\overline{u}$, $V\sim \mathcal{U}[\underline{u},\overline{u}]$, and $X^{\underline{u},\overline{u}}=Q_X^{\mathbb{R},-}(V)$. A derivation of $F_{X^{\underline{u},\overline{u}}}$ is relegated to Appendix \ref{sec:swtt_derivations}.
Now by letting  $\omega(\tau)=\frac{2}{(\overline{u}-\underline{u})^2}(\overline{u}-\tau)\mathbbm{1}\{\underline{u}<\tau\leq \overline{u}\}$, we obtain the Gini social welfare function specific to the quantile range $[\underline{u},\overline{u}]$,
$$SW_{\omega}(\underline{u},\overline{u})=\int_0^1\frac{2}{(\overline{u}-\underline{u})^2}(\overline{u}-\tau)\mathbbm{1}\{\underline{u}<\tau\leq \overline{u}\}Q_X^{\mathbb{R},-}(\tau)d\tau=\mathbb E[X^{\underline{u},\overline{u}}](1-I_{Gini}(F_{X^{\underline{u},\overline{u}}})),$$
where $\mathbb E[X^{\underline{u},\overline{u}}]\equiv\int_{\underline{u}}^{\overline{u}}Q_X^{\mathbb{R},-}(\tau)d\tau$ and $I_{Gini}\left(F_{X^{\underline{u},\overline{u}}}\right)\equiv \frac{\int_0^1(2\tau-\underline{u}-\overline{u})\mathbbm{1}\{\underline{u}<\tau\leq \overline{u}\}Q_X^{\mathbb{R},-}(\tau)d\tau}{(\overline{u}-\underline{u})^2\mathbb E[X^{\underline{u},\overline{u}}]}$.\footnote{This definition extends the upper tail Gini coefficient to any quantile range $[\underline{u},\overline{u}]$.} The interquantile Gini social welfare treatment effect on the treated over $[\underline{u},\overline{u}]$ is given by
 \begin{eqnarray*}
SWTT_{\omega}(\underline{u},\overline{u})&\equiv&SW_{\omega}(\underline{u},\overline{u},F_{Y_{11}|D=1})-SW_{\omega}(\underline{u},\overline{u},F_{Y_{10}|D=1})\\
&=&\int_{\underline{u}}^{\overline{u}}\frac{2}{(\overline{u}-\underline{u})^2}(\overline{u}-\tau)\left(Q^{\mathbb  R,-}_{Y_{11}|D=1}(\tau)-Q^{\mathbb  R,-}_{Y_{10}|D=1}(\tau)\right)d\tau.\end{eqnarray*}

\begin{remark}
It is important to note that when defining interquantile $SWTT_{\omega}(\underline{u},\overline{u})$ when $[\underline{u},\overline{u}]\subset [0,1]$, caution is required in interpreting these parameters, as we may not be comparing the same population unless certain assumptions hold. However, this concern is shared by most of the existing literature on recovering quantile treatment effects, including \citet{Abadieetal2002}, \citet{Firpo2007}, \citet{FrohlichMelly2008} and \citet{CallawayLi2019}, among many others. This issue disappears once we assume rank invariance—i.e., that there exists $U \sim \mathcal{U}[0,1]$ such that $Y_{1d|D=1} = Q_{Y_{1d|D=1}}^{\mathbb{R},-}(U), \quad \text{for } d \in \{0,1\}$.

\end{remark}

\section{Empirical Illustration}\label{sec:empirical}

In this section, we illustrate the CS bounds by revisiting the minimum wage study by \citet{Cengizetal2019}. This application demonstrates the usefulness of the class of policy-relevant parameters we introduce to examine the impact of the minimum wage increase. In particular, the lower-tail quantile social welfare treatment effect estimates allow us to zoom into the lower tail of the distribution, where we expect the minimum wage to have an impact. Overall, our CS bounds document proportionately larger impacts on the Gini social welfare in the lowest part of the distribution, where the minimum wage increase led to increase in the lower-tail mean and Gini social welfare. We also find that the distributional DiD exhibits violations of monotonicity in the lower tail of the distribution and is therefore not suitable for this application.

This empirical illustration highlights two practical advantages of our approach. First, our CS bounds relieve practitioners from having to take a stance on the support of the outcome of interest. Second, our multi-period CS bounds combine information from multiple pre-treatment periods to tighten the bounds on the parameters of interest and to simultaneously test the model assumptions.  

\subsection{Data and Implementation}
\citet{Cengizetal2019} examine 138 prominent state-level minimum wage increases between 1979 and 2016 using the individual-level NBER-merged Outgoing Rotation Group Earnings Data of the Current Population Survey. Their goal is to examine the impact of the policy on the wage distribution around the minimum wage, as illustrated in Figure \ref{fig:Ceng}. In order to make the empirical illustration of the multi-period CS bounds succinct, we focus on two pre-treatment periods, 2010 and 2011, and one post-treatment period, 2015, and examine the distributional impact of a nontrivial minimum wage increase of \$0.25 or more.\footnote{Note that starting 2009, the federal minimum has been \$7.25, so a minimum wage increase of \$0.25 or more constitutes an increase of more than 3\%. This definition of the treatment variable was also used in the empirical illustration in \citet{RothSantanna2021}.} For the purpose of this empirical illustration, we focus on the subgroup of states that had a pre-treatment minimum wage of \$8 or higher. We report the results for the remaining states in Appendix \ref{app:empirical_subgrouplessthan8}.

Table \ref{tab:Cengizetal_sumstats_multiT0} presents the summary statistics for hourly wage of both treatment and control groups in all three periods we consider. For both subgroups, the summary statistics show that the mean and standard deviation is different across treatment and control groups within the same year as well as within groups before and after the treatment.

In order to estimate the CS bounds on the counterfactual, we rely on Lemma \ref{lem:-X} to re-write the lower bound in a manner that admits straightforward numerical computation, specifically for $y\in\mathbb{Y}_{10|1}$ and for a given pre-treatment period $t$
\begin{eqnarray}
    F_{Y_{10}|D=1}^{LB,t}(y)&=&1-F_{-Y_{t}|D=1}(Q_{-Y_{t}|D=0}^{\mathbb{R},-}(1-F_{Y_1|D=0}(y)))\\
    F_{Y_{10}|D=1}^{UB,t}(y)&=&F_{Y_{t}|D=1}(Q_{Y_{t}|D=0}^{\mathbb{R},-}(F_{Y_1|D=0}(y)))
\end{eqnarray}
$F_{Y_{10}|D=1}^{LB,t}(y)$ and $F_{Y_{10}|D=1}^{UB,t}(y)$ are estimated by their sample analogues, $\widehat{F}_{Y_{10}|D=1}^{LB}(y)$ and $\widehat{F}_{Y_{10}|D=1}^{UB}(y)$, respectively, by replacing $F_X$ and $Q_X^{\mathbb{R},-}$ by their empirical counterparts, $\widehat{F}_X$ and $\widehat{Q}_X^{\mathbb{R},-}$, respectively. 

\begin{table}[H]\caption{Summary Statistics by Treatment and Control Groups}
\vspace{0.2cm}
\footnotesize{
    \begin{tabular}{lrrrrrrrrr}
    \toprule
    & \multicolumn{3}{c}{2010 (Pre-treatment)} & \multicolumn{3}{c}{2011 (Pre-treatment)}   & \multicolumn{3}{c}{2015 (Post-treatment)} \\
\cmidrule(lr){2-4}\cmidrule(lr){5-7}\cmidrule(lr){8-10}
(\$) & \multicolumn{1}{c}{Mean} & \multicolumn{1}{c}{S.D. } & \multicolumn{1}{c}{\# Obs} & \multicolumn{1}{c}{Mean} & \multicolumn{1}{c}{S.D. } & \multicolumn{1}{c}{\# Obs}& \multicolumn{1}{c}{Mean} & \multicolumn{1}{c}{S.D. } & \multicolumn{1}{c}{\# Obs} \\
 \midrule
    \multicolumn{10}{l}{States with Pre-Treatment Minimum Wage $\geq 8$}\\
    Control & 20.12 & 13.96 &            4,737  & 20.47 & 13.42 &            4,537  & 22.30 & 15.48 &            4,454  \\
    Treatment & 23.13 & 17.42 &          19,877  & 23.36 & 18.14 &          19,364  & 25.83 & 18.75 &          18,039  \\
    \bottomrule
    \end{tabular}}
\label{tab:Cengizetal_sumstats_multiT0}
\end{table}

\begin{figure} [htbp]\caption{Observed and Counterfactual Distributions: Subgroup 2 (Pre-MW$\geq \$ 8$), Bottom Quartile}
\vspace{0.25cm}{\footnotesize{
\begin{tabular}{cc}
(a) CS Bounds using 2010 pre-treatment period&(b) CS Bounds using 2011 pre-treatment period\\
\includegraphics[width=6cm]{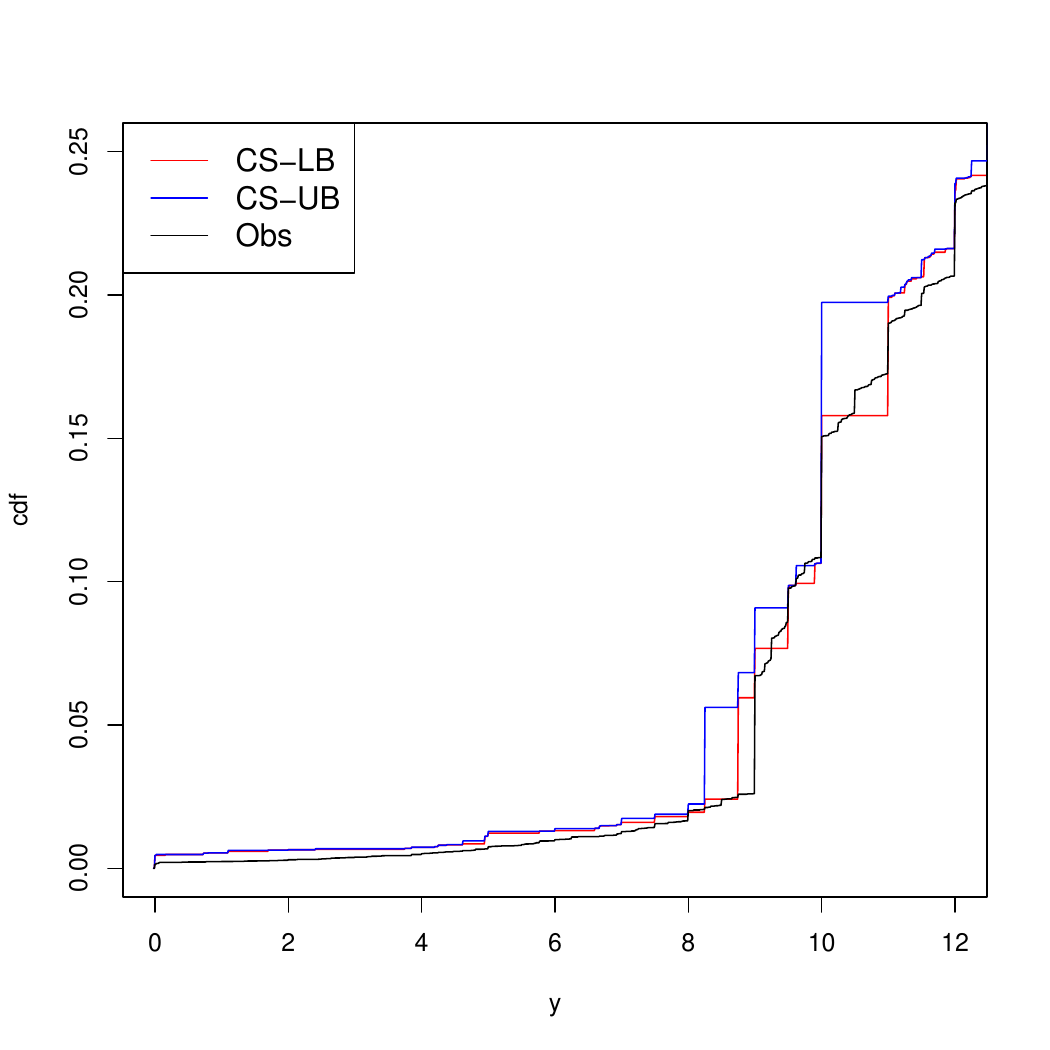}&\includegraphics[width=6cm]{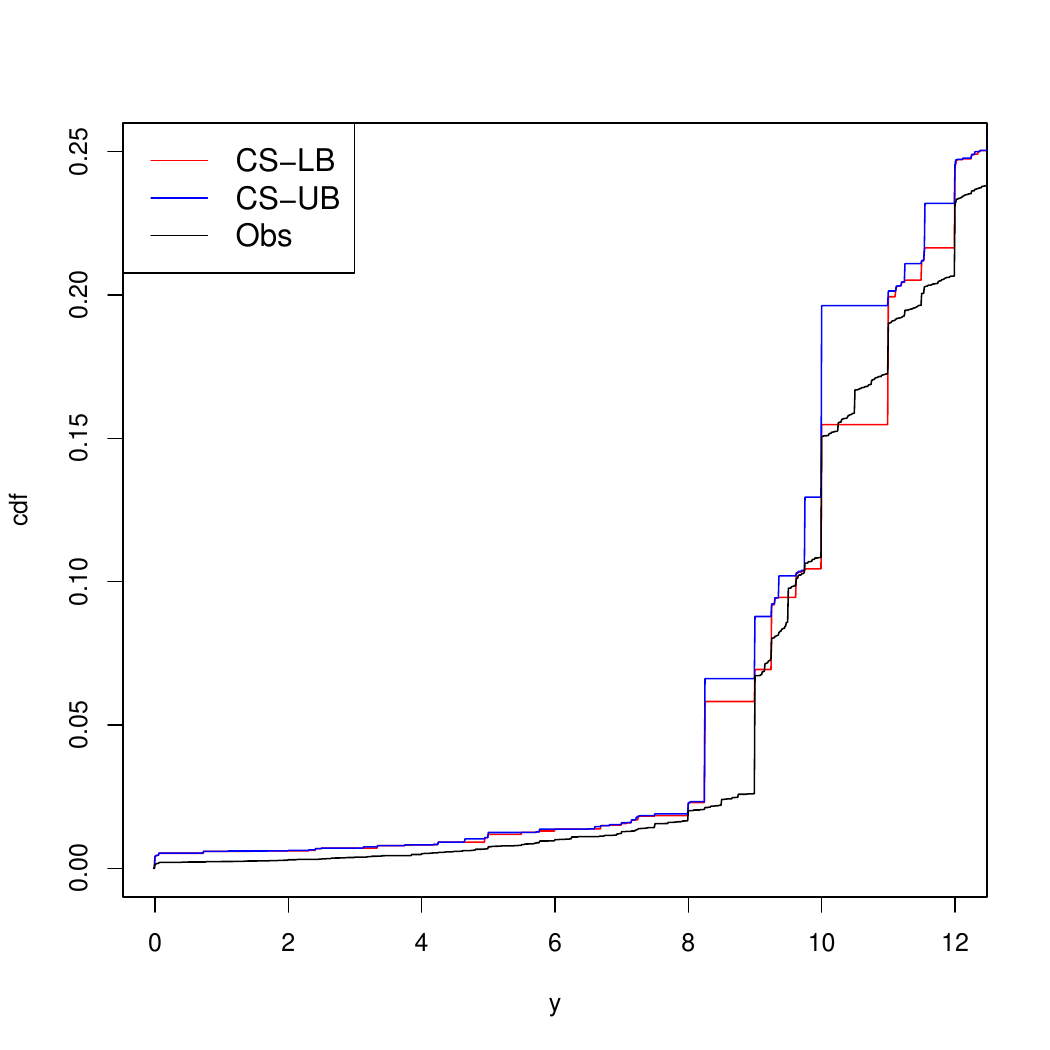}\\
(c) Dist-DiD using 2010 pre-treatment period&(d) Dist-DiD using 2011 pre-treatment period\\
\includegraphics[width=6cm]{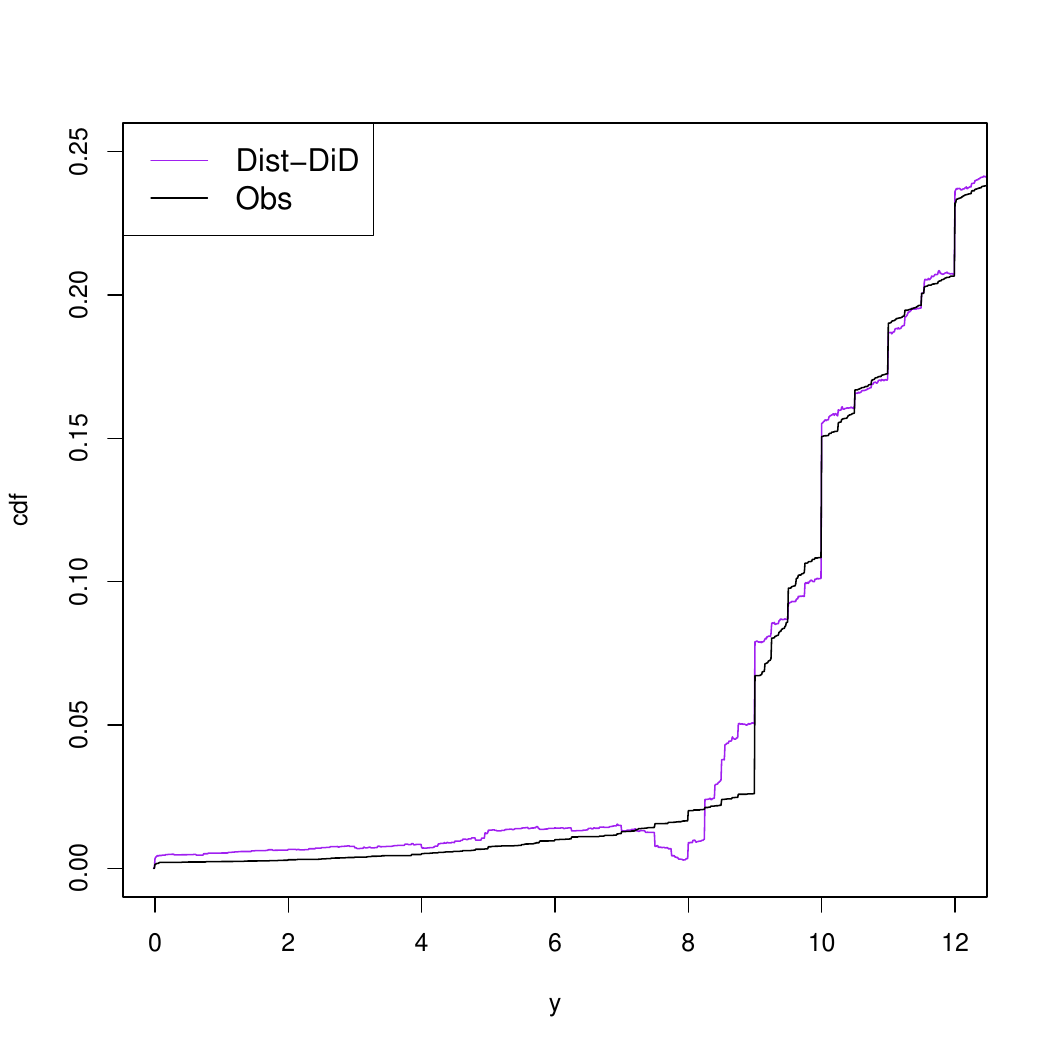}&\includegraphics[width=6cm]{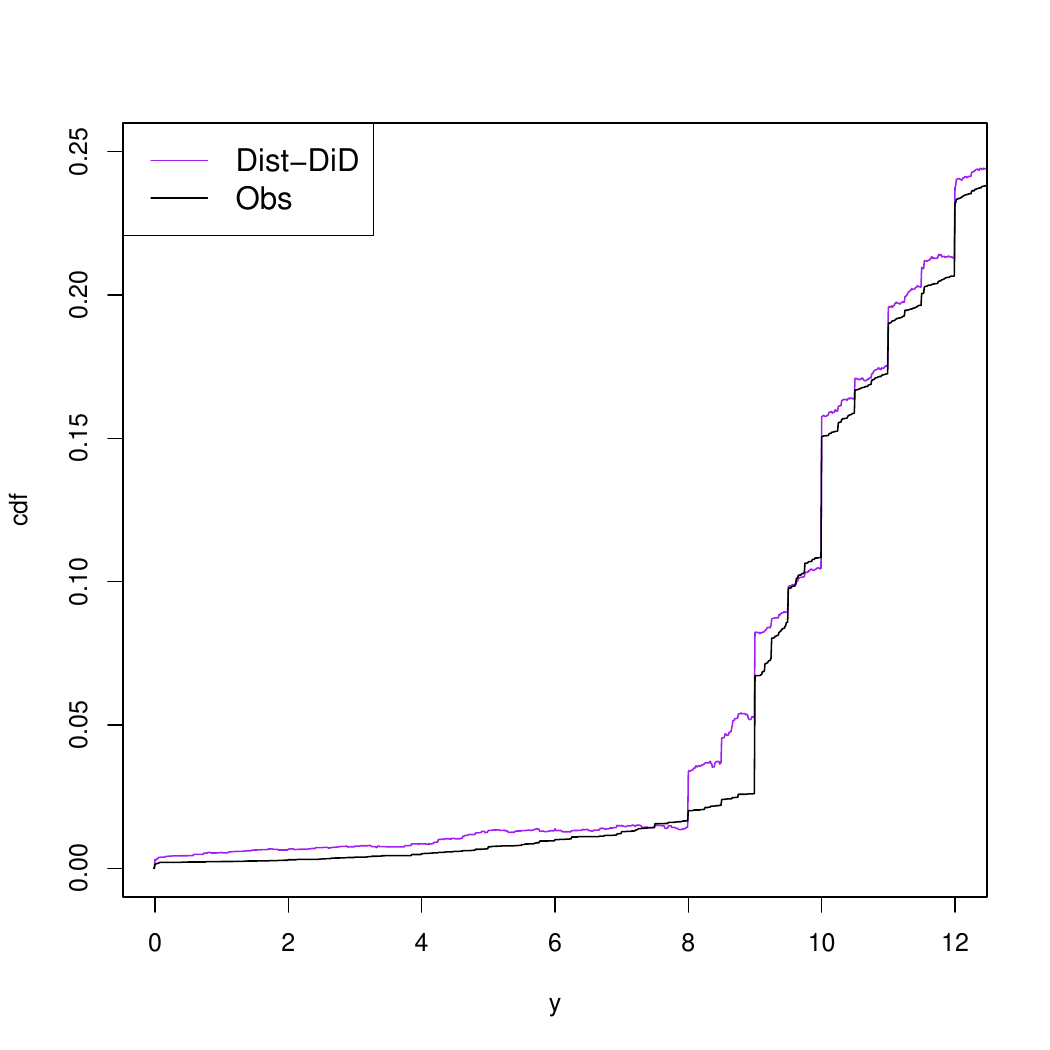}\\
(e) CiC using 2010 pre-treatment period&(f) CiC using 2011 pre-treatment period\\
\includegraphics[width=6cm]{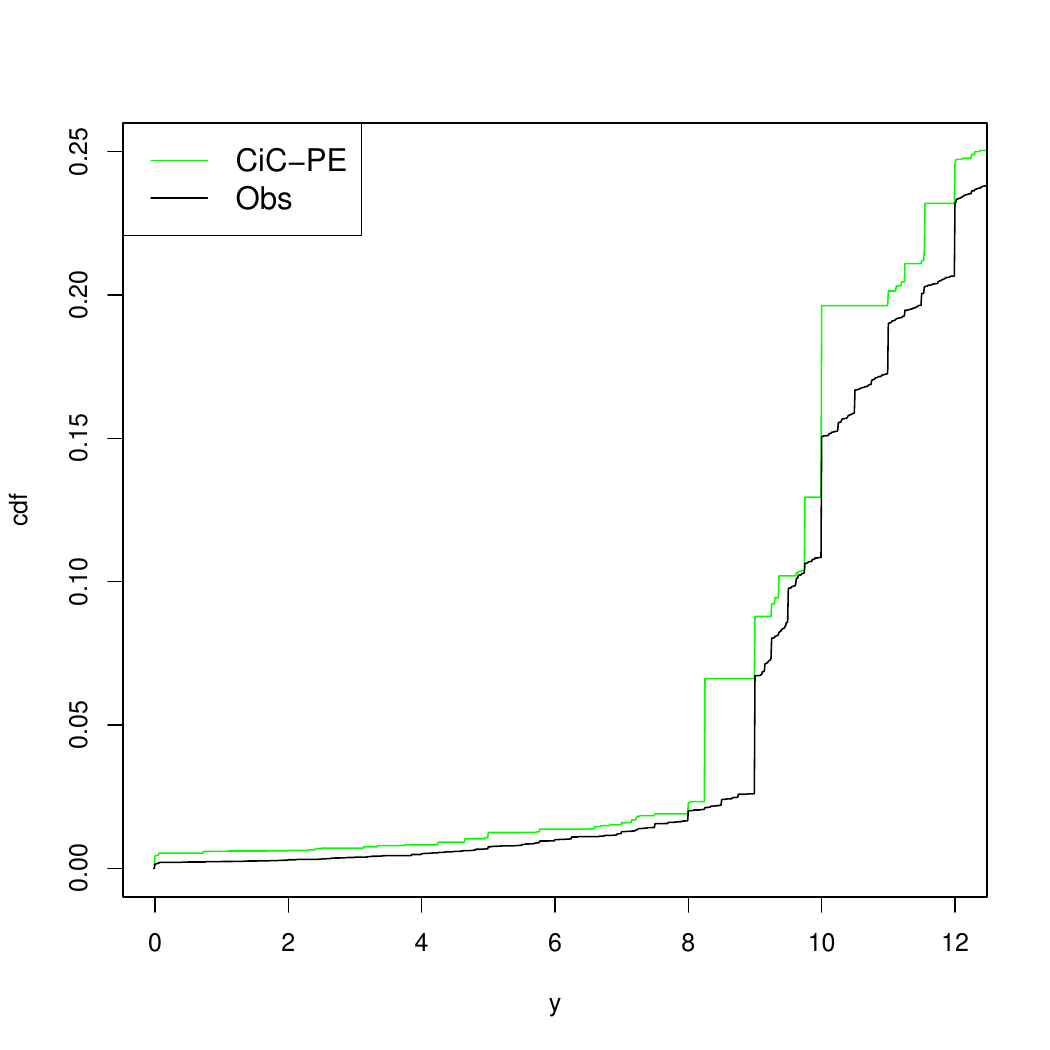}&\includegraphics[width=6cm]{Figures/FY10TCiC0_8_bottom25_06232025.pdf}\\
\multicolumn{2}{c}{\parbox{14cm}{\scriptsize{\emph{Notes}: $Obs$ denotes the observed factual $F_{Y_{11}|D=1}$, $CS$-$LB$ and $CS$-$UB$ denote the copula lower and upper bound estimates on the counterfactual distribution, respectively, $Dist$-$DiD$ depicts the distributional DiD estimator, and $CiC$-$PE$ denotes the CiC point estimator. For each point/bounds estimator, we provide estimates using each of the 2010 and 2011 pre-treatment periods. In this figure, we zoom into the lowest quartile of the distribution, see Figure \ref{fig:distribution_multiT0} for plots of the entire distribution. }}}
\end{tabular}}}
\label{fig:distribution_zoombottom_multiT0}
\end{figure}
\begin{figure}\caption{Horizontal subcopula plots scaled by $q$: $C_{Y_{t},D}(\cdot,q)/q$ on $Ran \widehat{F}_{Y_{t}}$ }\centering
    \includegraphics[width=10cm]{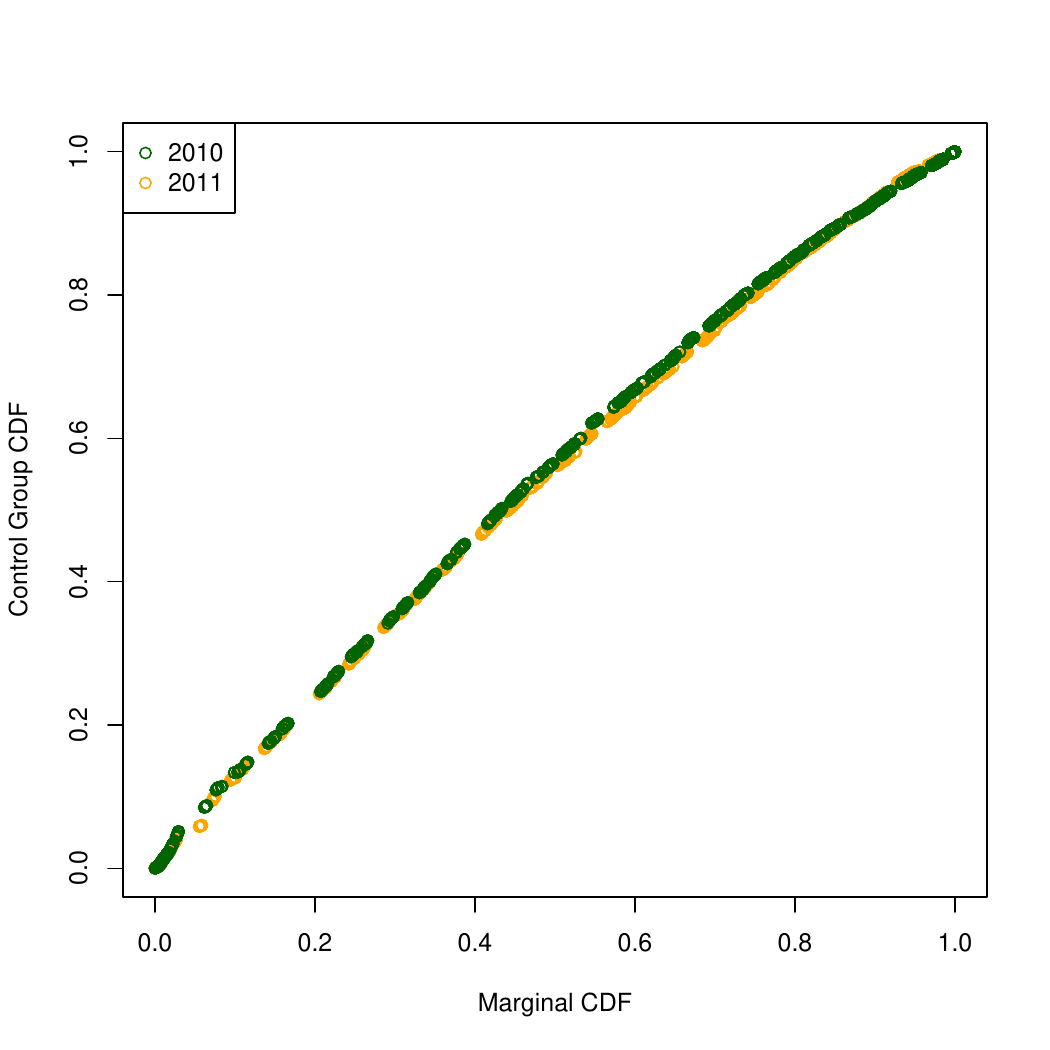}\\
    \parbox{12cm}{\scriptsize{\emph{Notes}: To plot the sample analogue of $C_{Y_{t},D}(\cdot,q)/q$, this figure plots the mapping $\widehat{F}_{Y_{t}}(y)\mapsto \widehat{F}_{Y_{t}|D=0}(y)$ for each $y$ in the empirical support of $Y_{t}$ for $t\in\{-1,0\}$, which refer to the 2010 and 2011 pre-treatment period, respectively.\\}}
    \label{fig:subcopula}
\end{figure}

The distributional DiD and CiC point estimators of $F_{Y_{10}|D=1}(y)$ are given by 
\begin{eqnarray}\widehat{F}_{Y_{10}|D=1}^{D\text{-}DiD,t}(y)&=&\widehat{F}_{Y_{t}|D=1}(y)+\widehat{F}_{Y_1|D=0}(y)-\widehat{F}_{Y_{t}|D=0}(y)\\
\widehat{F}_{Y_{10}|D=1}^{CiC,t}(y)&=&\widehat{F}_{Y_{t}|D=1}\left(\widehat{Q}_{Y_{t}|D=0}^{\mathbb{Y}_{t|0},-}\left(\widehat{F}_{Y_1|D=0}(y)\right)\right)
\end{eqnarray}

The CS bounds on the counterfactual as well as the observed factual distribution $\widehat{F}_{Y_1|D=1}$ can then be used to obtain the following sample analogues of the lower and upper bounds on the SWTT.\footnote{We compute the integral numerically using a grid with a step size of $0.01$.} For $t\in\{-1,0\}$, we obtain the following CS bounds estimator for the SWTT parameter
\begin{eqnarray}\widehat{SWTT}_\omega^{LB,t}&=&\int_\mathbb{R}{\left(\Lambda(\widehat{F}_{Y_1|D=1}(y)-\Lambda(\widehat{F}_{Y_{10}|D=1}^{LB,t}(y))\right)}dy\nonumber\\
\widehat{SWTT}_\omega^{UB,t}&=&\int_\mathbb{R}{\left(\Lambda(\widehat{F}_{Y_1|D=1}(y)-\Lambda(\widehat{F}_{Y_{10}|D=1}^{UB,t}(y))\right)}dy \label{eq:swtt-estimator}\end{eqnarray}
Similarly, we compute the multi-period CS bounds on the SWTT parameters.

To compute the SWTT parameters for the distributional DiD and CiC point estimators,
we use the following 
\begin{eqnarray}
\widehat{SWTT}_\omega^{D\text{-}DiD,t}&=&\int_\mathbb{R}{\left(\Lambda(\widehat{F}_{Y_1|D=1}(y)-\Lambda(\widehat{F}_{Y_{10}|D=1}^{D\text{-}DiD,t}(y))\right)}dy,\\
\widehat{SWTT}_\omega^{CiC,t}&=&\int_\mathbb{R}{\left(\Lambda(\widehat{F}_{Y_1|D=1}(y)-\Lambda(\widehat{F}_{Y_{10}|D=1}^{CiC,t}(y))\right)}dy.\label{eq:swtt-estimator_pe}\end{eqnarray}

\subsection{Bounds on the counterfactual distribution}\label{sec:empirical_counterfactual}
Figure \ref{fig:distribution_zoombottom_multiT0} presents the observed distribution of the treatment group in 2015, $\widehat{F}_{Y_1|D=1}$, as well as the CS bounds, distributional DiD and CiC point estimators of the counterfactual distribution using 2010 and 2011 as pre-treatment periods.  Since the minimum wage is likely to have an impact on the bottom of the distribution, we present those figures for the bottom quartile of the wage distribution where the minimum wage increase is likely to have an impact.\footnote{We relegate the figures of the entire distribution to Figure \ref{fig:distribution_multiT0} in the online appendix.}

First, we examine the CS bounds on the counterfactual distribution using each of the pre-treatment periods separately in Figure \ref{fig:distribution_zoombottom_multiT0}(a) and \ref{fig:distribution_zoombottom_multiT0}(b), respectively. Comparing the observed (factual) distribution with the CS bounds on the counterfactual using each of the pre-treatment periods, we note an obvious change in the censoring point as expected in the context of a minimum wage increase. For instance, in  Figure \ref{fig:distribution_zoombottom_multiT0}(b), the CS bounds on the counterfactual distribution exhibit a jump slightly above \$8, whereas the observed (factual) distribution exhibits a jump at about \$9. Furthermore, note that both upper and lower bounds satisfy the properties of a cdf. In addition, since the bounds do not cross, we do not have any detectable violation of the assumptions required for our identification approach. We also plot the sample analogue of the horizontal subcopula $C_{Y_{t0},D}(\cdot,q)$ for 2010 and 2011 to provide a visual check of our copula stability assumption in Figure \ref{fig:subcopula}. This plot is the counterpart of DiD pre-trends plots in our context. While this figure does not provide a formal test of the copula stability assumption, it demonstrates that the copulas governing the dependence between $Y_{t0}$ and $D$ for 2010 and 2011 are fairly similar.

Next, we examine the bottom quartile of the distributional DiD counterfactual estimates using 2010 and 2011 as pre-treatment period in Figure \ref{fig:distribution_zoombottom_multiT0}(c) and \ref{fig:distribution_zoombottom_multiT0}(d), respectively. At first glance, we note violations of the monotonicity property of cdfs in both counterfactual distributions, indicating a violation of the testable implication of the identifying assumption of distributional DiD \citep{RothSantanna2021}. The magnitude of the monotonocity violation is by far greater for the distributional DiD estimate using the 2010 pre-treatment period; the counterfactual estimate ``dips'' around the pre-treatment minimum wage of \$8, which is the part of the distribution particularly pertinent for the evaluation of the minimum wage increase.

Finally, we also present the CiC point estimator of the counterfactual using both pre-treatment periods in Figure \ref{fig:distribution_zoombottom_multiT0}(e) and \ref{fig:distribution_zoombottom_multiT0}(f), respectively. As demonstrated in Section \ref{sec:cic_connection}, the CiC point estimator coincides with the CS upper bound using the same pre-treatment period. This could translate to the CiC suffering from an upward bias in SWTT estimation as evident from comparing \eqref{eq:swtt-estimator} and \eqref{eq:swtt-estimator_pe}.

\subsection{Bounds on treatment effects}
Next, we quantify the impact of the minimum wage increase on the wage distribution using the ATT and the Gini SWTT both for the overall distribution as well as its lower tail. We report 95\% confidence intervals for all SWTT estimators using standard normal critical values and standard errors obtained using nonparametric bootstrap.\footnote{While the formal proof that these confidence intervals provide adequate coverage asymptotically is beyond the scope of the present paper, we have examined their performance in a simulation study mimicking our minimum wage setting which demonstrates that they provide adequate coverage in finite samples.}

\subsubsection{Overall social welfare treatment effects} Table \ref{tab:swtt_multiT0} presents 95\% confidence intervals on the ATT and Gini SWTT using the CS bounds, the distributional DiD and CiC point estimators.

\begin{table}[htbp]
  \centering
  \caption{Inference on SWTT using 2010 and/or 2011 as pre-treatment periods}
{\footnotesize    \begin{tabular}{lrrrrrrrrrrrl}
\multicolumn{7}{l}{Panel A. CS bounds}\\
          \toprule
          & \multicolumn{6}{c}{95\% CI}  \\
          \midrule
 Pre-period        & \multicolumn{2}{c}{2010}        & \multicolumn{2}{c}{2011} & \multicolumn{2}{c}{  2010 \& 2011}\\
          \cmidrule(lr){2-3}\cmidrule(lr){4-5}\cmidrule(lr){6-7}
          \cmidrule(lr){2-3}\cmidrule(lr){4-5}\cmidrule(lr){6-7}
  (\$)        &  \multicolumn{1}{c}{LB} & \multicolumn{1}{c}{UB} & \multicolumn{1}{c}{LB} & \multicolumn{1}{c}{UB} & \multicolumn{1}{c}{LB} & \multicolumn{1}{c}{UB} \\
          \midrule

    ATT   & -0.94 & 1.20  & -0.93 & 1.37  & -0.40 & 0.95 \\
    Gini SWTT & -0.43 & 0.85  & -0.13 & 1.11  & -0.06 & 0.79 \\
    \bottomrule 
    \end{tabular}}\\
    \vspace{0.5cm}
 {\footnotesize{   \begin{tabular}{lcrrcrrcrrcrr}

\multicolumn{9}{l}{Panel B. Distributional DiD and CiC point estimators}\\
     \toprule
 & \multicolumn{4}{c}{Dist DiD: 95\% CI}        & \multicolumn{4}{c}{CiC: 95\% CI}  \\
 \midrule
Pre-period    & \multicolumn{2}{c}{2010} 
   &\multicolumn{2}{c}{2011}
           & \multicolumn{2}{c}{2010} 
   &\multicolumn{2}{c}{2011}\\
             \cmidrule(lr){2-3}\cmidrule(lr){4-5}\cmidrule(lr){6-7}\cmidrule(lr){8-9}
  (\$)    &\multicolumn{1}{c}{LB} & \multicolumn{1}{c}{UB} &   \multicolumn{1}{c}{LB} & \multicolumn{1}{c}{UB} &    \multicolumn{1}{l}{LB} & \multicolumn{1}{l}{UB} &   \multicolumn{1}{l}{LB} & \multicolumn{1}{l}{UB}   \\
\midrule
ATT   & -0.16 & 1.23  & -0.03 & 1.37  & -0.49 & 1.20  & -0.41 & 1.37 \\
    Gini SWTT & -0.26 & 0.72  & -0.02 & 0.97  & -0.14 & 0.85  & 0.13  & 1.11 \\
    \midrule
    \multicolumn{9}{c}{\parbox{10cm}{\scriptsize{\emph{Notes}: The definitions of the SWTT bounds/point estimators are provided in \eqref{eq:swtt-estimator}--\eqref{eq:swtt-estimator_pe}. For the CS bounds, we report 95\% confidence intervals on the identified set. For the point estimators, we report 95\% confidence intervals on the SWTT parameter. All confidence intervals use standard normal critical values and nonparameteric bootstrap standard errors using 5,000 bootstrap replications.}}}\\
\end{tabular}}}
\label{tab:swtt_multiT0}
\end{table}

    \begin{table}[htbp]
  \centering
  \caption{Inference on Lower-tail SWTT using 2010 and/or 2011 as pre-treatment periods}
{\footnotesize    \begin{tabular}{lrrrrrr}
\multicolumn{7}{l}{Panel A. CS bounds}\\
          \toprule
          & \multicolumn{6}{c}{95\% CI}  \\
          \midrule
 Pre-period        & \multicolumn{2}{c}{2010}        & \multicolumn{2}{c}{2011} & \multicolumn{2}{c}{  2010\&2011}\\
          \cmidrule(lr){2-3}\cmidrule(lr){4-5}\cmidrule(lr){6-7}
          \cmidrule(lr){2-3}\cmidrule(lr){4-5}\cmidrule(lr){6-7}
  (\$)        &  \multicolumn{1}{c}{LB} & \multicolumn{1}{c}{UB} & \multicolumn{1}{c}{LB} & \multicolumn{1}{c}{UB} & \multicolumn{1}{c}{LB} & \multicolumn{1}{c}{UB} \\

\midrule
    $u=0.01$\\
      ATT(u)  & 0.43  & 2.80  & 0.52  & 3.13  & 0.73  & 2.74 \\
   Gini SWTT(u) & 0.52  & 2.41  & 0.61  & 2.56  & 0.81  & 2.42\\
   [0.25em]
    $u=0.025$\\
    ATT(u)& 0.12  & 1.95  & 0.23  & 2.11  & 0.35  & 1.85 \\
    Gini SWTT(u) & 0.31  & 2.36  & 0.39  & 2.56  & 0.57  & 2.29\\
 [0.25em]
    $u=0.05$\\
    ATT(u)  & 0.10  & 1.35  & 0.47  & 1.50  & 0.53  & 1.30 \\
    Gini SWTT(u) & 0.16  & 1.75  & 0.39  & 1.91  & 0.49  & 1.68\\
    [0.25em]
    $u=0.10$\\
    ATT(u)  & -0.03 & 0.90  & 0.23  & 1.01  & 0.31  & 0.85 \\
Gini SWTT(u)& 0.08  & 1.23  & 0.37  & 1.37  & 0.44  & 1.18\\
     [0.25em]
    $u=0.25$\\
    ATT(u)& -0.12 & 0.71  & 0.02  & 0.80  & 0.08  & 0.68 \\
 Gini SWTT(u) & -0.05 & 0.78  & 0.13  & 0.88  & 0.19  & 0.74\\
 \midrule
    \end{tabular}%
    }\\
 \vspace{0.5cm}
 {\footnotesize{  
\begin{tabular}{lcrrcrrcrrcrr}

\multicolumn{9}{l}{Panel B. Distributional DiD and CiC}\\
     \toprule
 & \multicolumn{4}{c}{Dist DiD: 95\% CI}        & \multicolumn{4}{c}{CiC: 95\% CI}  \\
 \midrule
Pre-period    & \multicolumn{2}{c}{2010} 
   &\multicolumn{2}{c}{2011}
           & \multicolumn{2}{c}{2010} 
   &\multicolumn{2}{c}{2011}\\
             \cmidrule(lr){2-3}\cmidrule(lr){4-5}\cmidrule(lr){6-7}\cmidrule(lr){8-9}
  (\$)    &\multicolumn{1}{c}{LB} & \multicolumn{1}{c}{UB} &   \multicolumn{1}{c}{LB} & \multicolumn{1}{c}{UB} &    \multicolumn{1}{l}{LB} & \multicolumn{1}{l}{UB} &   \multicolumn{1}{l}{LB} & \multicolumn{1}{l}{UB}   \\
\midrule
  $u=0.01$\\
      ATT(u) & 0.34  & 3.25  & 0.49  & 3.24  & 0.53  & 2.80  & 0.58  & 3.13 \\
      Gini SWTT(u) & 0.45  & 2.64  & 0.57  & 2.54  & 0.61  & 2.41  & 0.64  & 2.57\\
   [0.25em]
    $u=0.025$\\
    ATT(u)& -0.53 & 1.93  & -0.10 & 2.26  & 0.27  & 1.95  & 0.27  & 2.12 \\
       Gini SWTT(u) & -0.40 & 2.51  & 0.10  & 2.76  & 0.44  & 2.37  & 0.46  & 2.56\\
    $u=0.05$\\
    ATT(u)& -0.11 & 1.26  & 0.20  & 1.57  & 0.50  & 1.35  & 0.44  & 1.50 \\
    Gini SWTT(u) & -0.27 & 1.74  & 0.15  & 2.06  & 0.43  & 1.75  & 0.41  & 1.91\\
    $u=0.10$\\
    ATT(u) & -0.13 & 0.75  & 0.07  & 0.93  & 0.33  & 0.90  & 0.37  & 1.01 \\
   Gini SWTT(u) & -0.14 & 1.11  & 0.14  & 1.36  & 0.43  & 1.23  & 0.43  & 1.37\\
        [0.25em]
     $u=0.25$\\
    ATT(u) & -0.21 & 0.47  & -0.08 & 0.59  & 0.18  & 0.71  & 0.29  & 0.80 \\
    Gini SWTT(u)& -0.15 & 0.59  & 0.02  & 0.74  & 0.26  & 0.78  & 0.33  & 0.88\\
    \midrule
        \multicolumn{9}{c}{\parbox{10cm}{\scriptsize{\emph{Notes}: The definitions of the SWTT bounds/point estimators are provided in \eqref{eq:swtt-estimator}--\eqref{eq:swtt-estimator_pe}. For the CS bounds, we report 95\% confidence intervals on the identified set. For the point estimators, we report 95\% confidence intervals on the SWTT parameter. All confidence intervals use standard normal critical values and nonparameteric bootstrap standard errors using 5,000 bootstrap replications.}}}\\
    \end{tabular}%
    }}
  \label{tab:lowertail_swtt_multiT0}%
\end{table}%

When examining Table \ref{tab:swtt_multiT0}, we note that the 95\% confidence intervals on the CS bounds for the ATT and Gini SWTT include zero, whether we use 2010 and 2011 as pre-treatment periods separately or use them both in the multi-period CS bounds. This is consistent with the expectation that a minimum wage increase is unlikely to change the mean or inequality of the overall wage distribution. When we consider the 95\% confidence intervals using the distributional DiD and CiC point estimators, they suggest no improvement in terms of ATT and Gini SWTT, except using the CiC confidence interval that use the 2011 pre-treatment period. As pointed out in Section \ref{sec:empirical_counterfactual}, the CiC point estimator of the counterfactual coincides with the CS upper bound. As a result, the corresponding SWTT estimator may be upwardly biased.

\subsubsection{Lower-tail social welfare treatment effects} In the context of policies such as an increase in the legal minimum wage, the welfare of subpopulations at the lower tail of the wage distribution is an important policy target. Table \ref{tab:lowertail_swtt_multiT0} provides the lower-tail ATT and Gini social welfare treatment effects, $ATT(u)$ and $Gini~SWTT(u)$ for $u\in\{0.01,0.025,0.05,0.10,0.25\}$, respectively, introduced in Section \ref{sec:giniswtt_lower-tail}.

First, we consider the CS bounds using 2010 and 2011 as pre-treatment periods separately as well as the multi-period CS bounds that exploits both pre-treatment periods. Regardless of the pre-treatment year we use, for $u\in\{0.01,0.025,0.05\}$, the 95\% confidence intervals on the CS bounds demonstrate statistically significant improvement in terms of lower-tail mean and Gini social welfare. When we consider $u\in\{0.10,0.25\}$, we note that while the CS bounds using the 2011 pre-treatment period demonstrate statistically significant improvements in terms of lower-tail mean and Gini social welfare, the confidence intervals on the CS bounds using the 2010 pre-treatment period are not conclusive on the sign of this impact. Since the multiple-period CS bounds combine the information from both pre-treatment periods, they result in tighter confidence intervals than the CS bounds using 2010 or 2011 by itself for both the lower-tail ATT and Gini SWTT for all quantiles $u$ we consider. These tighter confidence intervals point to improvements both in terms of mean and Gini social welfare up to the lower quartile of the distribution ($u=0.25$). This demonstrates how exploiting the multiple pre-treatment periods can aid to provide tighter bounds that translate to shorter confidence intervals.

Next, we consider the distributional DiD and CiC estimators. The distributional DiD confidence intervals using the 2010 pre-treatment period do not suggest any significant improvement in terms of lower-tail mean and Gini social welfare, whereas the distributional DiD confidence intervals using the 2011 pre-treatment period suggest significant improvements in terms of both lower-tail mean and Gini social welfare for most of the quantiles we consider. When we examine the CiC point estimator, we note that the corresponding confidence intervals suggest significant improvements in terms of mean and Gini social welfare for all of the lower-tail quantiles we consider ($u=0.25$).

The confidence intervals on the lower-tail SWTT parameters demonstrate that the distributional DiD can yield contradictory results that then require an ad-hoc choice by the applied researcher regarding which period to use.\footnote{Since the distributional DiD point estimator of the counterfactual distribution using the 2010 pre-treatment period exhibits monotonicity violations, an applied researcher would likely discard those results and use the distributional DiD estimator using the 2011 pre-treatment period, for which the monotonicity violations are very minor. The selection of the pre-treatment period relies however on a pre-test, which raises the usual post-selection inference concerns. Pre-test bias issues in the context of difference-in-difference designs have been examined in \citet{Roth2022}.} The confidence intervals based on the CiC point estimator will coincide with the confidence interval on the CS upper bound and may therefore be upwardly biased.

Overall, our empirical application underscores the advantages of the CS bounds in terms of relieving the applied researcher from choosing the pre-treatment period as well as specifying the type of outcome distribution. It also demonstrates how to use the CS bounds on the counterfactual distribution to conduct inference on the SWTT parameters. Finally, The CS bounds on the counterfactual distribution can be used to bound other parameters, such as the parameters examined in \citet{Cengizetal2019}. We provide these estimates in Appendix \ref{app:empirical_deltas}.

\section{Conclusion}
With the goal of assessing the impact of regulatory policies on social welfare, this paper provides a unifying, partial identification result for the counterfactual distribution of the treatment group in difference-in-difference settings. Exploiting the stability of the dependence (copula) between group membership and the untreated potential outcome across time, our identification result has several advantages: (1) it applies to any outcome distribution, whether continuous, discrete or mixed, (2) it is invariant to monotonic transformations of the outcome, (3) it can allow for nonrandom selection into treatment without restricting the evolution of the marginal distribution of the potential outcomes across time. To quantify the impact of regulatory policies on social welfare, we introduce a broad class of treatment effect parameters. This class includes the ATT as well as the Gini social welfare treatment effect on the treated as a special case. We illustrate the empirical relevance of our results using a minimum wage application revisiting \citet{Cengizetal2019}. 
\begin{singlespace}
\bibliography{bibtex}
\bibliographystyle{econometrica}
\end{singlespace}

\renewcommand\thelemma{A.\arabic{lemma}} 
\setcounter{lemma}{0}
\renewcommand\theclaim{A.\arabic{claim}} 
\setcounter{claim}{0}
\renewcommand\theexample{A.\arabic{example}} 
\setcounter{example}{0}
\begin{appendix}

\section{Proofs of the main results} 
\label{Appen:Proof}
\subsection{An Additional  Result}

\begin{lemma}\label{lem:boundq}
Let $X$ be a random variable, we then have: 
\begin{enumerate}
\item \label{lem:boundq1} The following bounds are pointwise sharp,
\begin{eqnarray}
F_{X}\left(Q^{\mathbb R,+}_{X}\left(u\right)-\right) &\leq& u \leq F_{X}\left(Q^{\mathbb R,-}_{X}\left(u\right)\right), \text{ for all } u \in [0,1].
\end{eqnarray}

\item \label{lem:boundq2} Let $\mathbb{X}\subseteq \mathbb{Z}$, $\sup\left\{F_{X}\left(t\right): t\leq x \; \& \;  t \in \mathbb Z \cup\{-\infty \} \right\}=F_X(x).$

\end{enumerate} 
\end{lemma}
Before we proceed to provide a proof of the above lemma, we compare the bounds in Lemma \ref{lem:boundq}(1) with those used in \citet{AtheyImbens2006}, hereinafter AI2006, to bound the counterfactual distribution for discrete outcomes. These bounds are given by the following in our notation,
\begin{eqnarray}F_{X}(Q_X^{\mathbb{X},+}(u))\leq u\leq F_X(Q_X^{\mathbb{X},-}(u)).\label{eq:boundq_AI2006}
\end{eqnarray}
Now note that the upper bound employed in AI2006 only differs from the upper bound in Lemma \ref{lem:boundq}(1) in terms the use of $\mathbb{X}$ instead of $\mathbb{R}$.  These two quantiles only differ for $u=0$, since $\{x\in\mathbb{R}:F_X(x)\geq 0\}=\mathbb{R}$, whereas $\{x\in\mathbb{X}:F_X(x)\geq 0\}=\mathbb{X}$. As a result, $Q_X^{\mathbb{R},-}(0)=-\infty$ and $F_X(Q_X^{\mathbb{R},-}(0))=0$, whereas $Q_X^{\mathbb{X},-}(0)=\inf \mathbb{X}$ and $F_X(\inf\mathbb{X})\geq 0$. Therefore, our upper bound is lower than the one used in AI2006 for $u=0$.\footnote{Note that this is inconsequential for their identification result, since they provide bounds on the counterfactual distribution on its support, and set it to zero below the infimum of its support and to one above the supremum of its support.}

The lower bound in Lemma \ref{lem:boundq}(1) is starkly different from the lower bound in \eqref{eq:boundq_AI2006}. As we discuss in Section \ref{sec:cic_connection}, the lower bound in \eqref{eq:boundq_AI2006} equals the upper bound for several examples with mixed outcomes, due to censoring or bunching, because $Q_X^{\mathbb{X},+}(u)=Q_X^{\mathbb{X},-}(u)$ for $u\in[0,1]$ for some mixed outcome distributions. As a result, the lower bound is not valid in the mixed-outcome case in general. 
In those cases, the AI2006 bounds would not cover the counterfactual distribution. We demonstrate additional numerical examples in Appendix \ref{App:num_examples}. By contrast, our lower bound is valid and sharp for any outcome distribution. For discrete outcomes, our bounds collapse to theirs in numerical examples provided in Appendix \ref{App:num_examples}.

\begin{proof}(Lemma \ref{lem:boundq})\\
(1) 
$Q_X^{\mathbb{R},-}(u)\equiv\inf\{x\in\mathbb{R}:F_X(x)\geq u\}$. We know from the properties of a quantile function that $F_X(Q_X^{\mathbb{R},-}(u)) \geq u$. We now show that this inequality is sharp. Suppose that there exists $\tilde{x} \in \mathbb R: F_X(\tilde{x}) \geq u$ and $F_X(\tilde{x}) < F_X(Q_X^{\mathbb{R},-}(u))$. On the one hand, we have $F_X(\tilde{x}) < F_X(Q_X^{\mathbb{R},-}(u)) \Longrightarrow \tilde{x} < Q_X^{\mathbb{R},-}(u),$ since $F_X$ is nondecreasing. On the other hand, $F_X(\tilde{x}) \geq u \Longrightarrow \tilde{x} \in \{x\in \mathbb R: F_X(x) \geq u\}$. Therefore, $\tilde{x} \geq \inf\{x\in \mathbb R: F_X(x) \geq u\}=Q_X^{\mathbb{R},-}(u),$ which contradicts $\tilde{x} < Q_X^{\mathbb{R},-}(u)$.

We next show $F_X(Q_X^{\mathbb{R},+}(u)-)\leq u$. For a fixed $u \in [0,1]$, let us define $\Omega=\{y \in \mathbb R: F_X(y) \leq u\}$. We first show this implication: $z < Q_X^{\mathbb R,+}(u) \Longrightarrow F_X(z) \leq u$. By contradiction, suppose that (i) $z < Q_X^{\mathbb R,+}(u)$ and (ii) $F_X(z) > u$. Take $y \in \Omega$, then by (ii) we have $F_X(z) > u \geq F_X(y)$, which implies $F_X(z) > F_X(y)$, which in turn implies $y \leq z$ since $F_X$ is nondecreasing. Therefore, for all $y \in \Omega,$ we have $\ y \leq z$. It follows that $\sup \Omega \leq z$, i.e., $Q_X^{\mathbb R,+}(u) \leq z$. This leads to a contradiction since $z<Q_X^{\mathbb R,+}(u)$ by (i). Hence, we have shown that $z < Q_X^{\mathbb R,+}(u) \Longrightarrow F_X(z) \leq u$. Second, by definition, we have $F_X(Q_X^{\mathbb{R},+}(u)-)\equiv\sup_{z < Q_X^{\mathbb R,+}(u)} F_X(z) \leq \sup_{z < Q_X^{\mathbb R,+}(u)} u=u$, where the inequality holds from the previous implication.

Now we proceed to show that $F_X(Q_X^{\mathbb{R},+}(u)-)\leq u$ is sharp. First, let us show that there does not exist any $\tilde{x} \in \mathbb R$ such that (i) $F_X(\tilde{x}) \leq u$ and (ii) $F_X(Q_X^{\mathbb{R},+}(u)-) < F_X(\tilde{x}-)$. By contradiction, suppose there exists such an $\tilde{x}\in\mathbb R$. From (ii), $\sup_{z < Q^{\mathbb R,+}(u)} $ $F_X(z)$ $\equiv F_X(Q_X^{\mathbb{R},+}(u)-) < F_X(\tilde{x}-) \equiv \sup_{z < \tilde{x}} F_X(z)$, we deduce that $\{z < Q^{\mathbb R,+}(u)\} \subset \{z < \tilde{x}\}$. Therefore, $Q^{\mathbb R,+}(u) < \tilde{x}$. From (i), $F_X(\tilde{x}) \leq u$, we have $\tilde{x} \in \{x \in \mathbb R: F_X(x) \leq u\}$. Therefore, $\tilde{x} \leq \sup\{x \in \mathbb R: F_X(x) \leq u\}=Q^{\mathbb R,+}(u)$, which leads to a contradiction. It follows that there does not exist any $\tilde{x} \in \mathbb R$ such that $F_X(\tilde{x}) \leq u$ and $F_X(Q_X^{\mathbb{R},+}(u)-) < F_X(\tilde{x})$. \\

Second, let us show that there does not exist any $\tilde{x} \in \mathbb R$ such that $F_X(\tilde{x}-) \leq u$ and $F_X(Q_X^{\mathbb{R},+}(u)-) < F_X(\tilde{x}-)$. If $F_X(Q_X^{\mathbb{R},+}(u)-) < F_X(\tilde{x}-)$, then from the previous result, we must have $F_X(\tilde{x}) > u$. Hence, we have $F_X(\tilde{x}) > u \geq F_X(\tilde{x}-)$, which implies $\tilde{x}=Q_X^{\mathbb{R},+}(u)$, which in turn contradicts $F_X(Q_X^{\mathbb{R},+}(u)-) < F_X(\tilde{x}-)$.

\end{proof}

\subsection{Proof of Lemma \ref{cop:mon}}
By Sklar's Theorem \citep[Theorem 2.3.3]{Nelsen2006}, there is a unique subcopula  $C_{Y_{10}, D}$ determined  
on $Ran F_{Y_{10}} \times \{q\}$,  such that  the following hold:
\begin{eqnarray}
F_{Y_1,D}(y,0)&\equiv&\mathbb P(Y_1\leq y, D=0)=C_{Y_{10},D}\left(F_{Y_{10}}(y),q\right), \;\; y \in \overline{\mathbb R}.
\end{eqnarray}

Using Proposition 1(4) from \citet{Embrechts_al2013}, we have:
\begin{eqnarray}
C_{Y_{10},D}(u,q)&=&F_{Y_1,D}\left(Q^{\mathbb  R,-}_{Y_{10}}(u),0\right) \text{ for all } u \in \overline{\operatorname{Ran}}  F_{Y_{10}}. \label{eq:cop2}
\end{eqnarray}

The  latter equality  holds, because (i) for all $u \in \overline{\operatorname{Ran}}  F_{Y_{10}}$ there exists $y \in \overline{\mathbb R}$ such that  $y=Q^{\mathbb  R,-}_{Y_{10}}(u)$ and (ii) from Proposition 1(4) in \citet{Embrechts_al2013}  we have  $F_{Y_{10}}\left(Q^{\mathbb  R,-}_{Y_{10}}(u)\right)=u$ for all  $u \in \overline{\operatorname{Ran}}  F_{Y_{10}}$.
For $u, u' \in \overline{\operatorname{Ran}}  F_{Y_{10}}$ such that $u<u'$ we have  $Q^{\mathbb  R,-}_{Y_{10}}(u) < Q^{\mathbb  R,-}_{Y_{10}}(u') \Rightarrow  F_{Y_1,D}\left(Q^{\mathbb  R,-}_{Y_{10}}(u),0\right) <F_{Y_1,D}\left(Q^{\mathbb  R,-}_{Y_{10}}(u'),0\right) \iff C_{Y_{10},D}(u,q) <C_{Y_{10},D}(u',q)$.
The first strict inequality holds because by construction $Q^{\mathbb  R,-}_{Y_{10}}(u)$ is strictly increasing on $\overline{\operatorname{Ran}}  F_{Y_{10}}$. The second holds because $Q^{\mathbb  R,-}_{Y_{10}}(\cdot) \in \mathbb Y_{10} \subseteq  \mathbb Y_{10|0}$ since  $\mathbb Y_{10|1} \subseteq \mathbb Y_{10|0}$. 

\qed

\subsection{Proof of Theorem \ref{Main:theorem}}\label{proof:theorem1}

The proof follows in three steps. First, we derive the bounds (Section \ref{Bounds:deriv}), then we proceed to show sharpness (Section \ref{Bounds:sharpness}). Since the sharpness proof relies on two intermediate lemmata, the last step is then to prove these two lemmata (Section \ref{Bounds:intermediate_lemmata}).
\subsubsection{Derivation of the bounds}\label{Bounds:deriv}

Take a fixed $y \in \mathbb Y_{10|0}$, then the following holds for all  $\tilde y< Q^{\mathbb R,+}_{Y_0|D=0}\left(F_{Y_{1|D=0}}(y)\right)$ : 
\begin{eqnarray}
F_{Y_0|D=0}\left(\tilde y\right) &\leq& F_{Y_{1|D=0}}(y) \leq F_{Y_0|D=0}\left(Q^{\mathbb R,-}_{Y_0|D=0}\left(F_{Y_{1|D=0}}(y)\right)\right), \nonumber \\
F_{Y_0,D}\left(\tilde y,0\right) &\leq& F_{Y_{1,D}}(y,0) \leq F_{Y_0,D}\left(Q^{\mathbb R,-}_{Y_0|D=0}\left(F_{Y_{1|D=0}}(y)\right),0\right), \nonumber\\
C_{Y_0,D}\left(F_{Y_0}\left(\tilde y\right), q\right)&\leq& C_{Y_{10},D}\left(F_{Y_{10}}(y), q\right) \leq C_{Y_0,D}\left(F_{Y_0}\left(Q^{\mathbb R,-}_{Y_0|D=0}\left(F_{Y_{1|D=0}}(y)\right)\right), q\right), \nonumber\\
C_{Y_0,D}\left(F_{Y_0}\left(\tilde y\right), q\right)&\leq& C_{Y_{0},D}\left(F_{Y_{10}}(y), q\right) \leq C_{Y_0,D}\left(F_{Y_0}\left(Q^{\mathbb R,-}_{Y_0|D=0}\left(F_{Y_{1|D=0}}(y)\right)\right), q\right), \nonumber\\
 F_{Y_0}\left(\tilde y\right)&\leq& F_{Y_{10}}(y) \leq F_{Y_0}\left(Q^{\mathbb R,-}_{Y_0|D=0}\left(F_{Y_{1|D=0}}(y)\right)\right) \label{eq;it}
\end{eqnarray}
The first line of the inequality trivially holds
from Lemma \ref{lem:boundq}(\ref{lem:boundq1}) and the fact that $Y_0\leq \tilde{y}$ implies $Y_0 < Q^{\mathbb R,+}_{Y_0|D=0}\left(F_{Y_{1|D=0}}(y)\right)$.
The third line holds  by Sklar's Theorem \citep[Theorem 2.3.3.]{Nelsen2006}. The fourth line holds under Assumption \ref{stab}, and the last line holds under Assumption \ref{inc}. 
Notice that the last line requires $u \mapsto C_{Y_{10},D}(u,q)$ to be strictly increasing only on $\overline{\operatorname{Ran}}  F_{Y_{10}}\cup \overline{\operatorname{Ran}}  F_{Y_{00}} \subseteq [0,1]$.
Now, applying the monotonicity of the function $v-C_{Y_0,D}(v,q)$ on the inequality (\ref{eq;it}), for all  $\tilde y < Q^{\mathbb R,+}_{Y_0|D=0}\left(F_{Y_{1|D=0}}(y)\right)$ we have:
 \begin{eqnarray*}
&&F_{Y_0}\left(\tilde y\right)-C_{Y_0,D}\left(F_{Y_0}\left(\tilde y\right), q\right) \leq  F_{Y_{10}}(y)-C_{Y_0,D}\left(F_{Y_{10}}(y), q\right) \leq \label{ineq:it}
\\ &&F_{Y_0}\left(Q^{\mathbb R,-}_{Y_0|D=0}\left(F_{Y_{1|D=0}}(y)\right)\right)-C_{Y_0,D}\left(F_{Y_0}\left(Q^{\mathbb R,-}_{Y_0|D=0}\left(F_{Y_{1|D=0}}(y)\right)\right), q\right). \nonumber
\end{eqnarray*}
With a slight abuse of notation, we will use $F_{Y_{t0},D}(y,1)\equiv \mathbb{P}(Y_{t0}\leq y, D=1)$. Since $F_{Y_{t0}}(y)=F_{Y_{t0}, D}(y,1)+F_{Y_{t0}, D}(y,0)=F_{Y_{t0}, D}(y,1) + C_{Y_{t0},D}(F_{Y_{t0}}(y),q)$ for $t=0,1$, the latter equality implies the following:
\begin{eqnarray*}
F_{Y_0, D}\left(\tilde y, 1\right) &\leq&  F_{Y_{10}}(y)-C_{Y_0,D}\left(F_{Y_{10}}(y), q\right) \leq F_{Y_0, D}\left(Q^{\mathbb R,-}_{Y_0|D=0}\left(F_{Y_{1|D=0}}(y)\right), 1\right)\\
F_{Y_0, D}\left(\tilde y, 1\right) &\leq&  F_{Y_{10}}(y)-C_{Y_{10},D}\left(F_{Y_{10}}(y), q\right) \leq F_{Y_0, D}\left(Q^{\mathbb R,-}_{Y_0|D=0}\left(F_{Y_{1|D=0}}(y)\right), 1\right)\\
F_{Y_0, D}\left(\tilde y, 1\right)&\leq&  F_{Y_{10},D}(y, 1) \leq F_{Y_0, D}\left(Q^{\mathbb R,-}_{Y_0|D=0}\left(F_{Y_{1|D=0}}(y)\right), 1\right),\\
F_{Y_0, D}\left(\tilde y, 1\right) &\leq&  F_{Y_{10},D}(y, 1) \leq F_{Y_0, D}\left(Q^{\mathbb R,-}_{Y_0|D=0}\left(F_{Y_{1|D=0}}(y)\right), 1\right),\\
 F_{Y_0|D=1}\left(\tilde y\right) &\leq&  F_{Y_{10}|D=1}(y) \leq F_{Y_0|D=1}\left(Q^{\mathbb R,-}_{Y_0|D=0}\left(F_{Y_{1|D=0}}(y)\right)\right),
\end{eqnarray*}
where the second line holds under Assumption \ref{stab}.
So, to summarize, for any fixed $y \in \mathbb Y_{10|0}$,
we have:
$$F_{Y_0|D=1}\left(\tilde y\right) \leq  F_{Y_{10}|D=1}(y) \leq F_{Y_0|D=1}\left(Q^{\mathbb R,-}_{Y_0|D=0}\left(F_{Y_{1|D=0}}(y)\right)\right), \text{ for all } \tilde y < Q^{\mathbb R,+}_{Y_0|D=0}\left(F_{Y_{1|D=0}}(y)\right).$$
Taking the supremum over $\tilde{y}<Q_{Y_0|D=0}^{\mathbb{R},+}(F_{Y_1|D=0}(y))$ implies that:
$$ \sup_{\tilde y < Q^{\mathbb R,+}_{Y_0|D=0}\left(F_{Y_{1|D=0}}(y)\right)}F_{Y_0|D=1}\left(\tilde y\right) \leq  F_{Y_{10}|D=1}(y) \leq F_{Y_0|D=1}\left(Q^{\mathbb R,-}_{Y_0|D=0}\left(F_{Y_{1|D=0}}(y)\right)\right),$$
which is equivalent to:
$$\underbrace{F_{Y_0|D=1}\left(Q^{\mathbb R,+}_{Y_0|D=0}\left(F_{Y_{1|D=0}}(y)\right)-\right)}_{=F_{Y_0|D=1}\left(\left[Q^{\mathbb R,+}_{Y_0|D=0}\circ F_{Y_{1|D=0}}\right](y)-\right)\equiv F^{LB}(y) } \leq  F_{Y_{10}|D=1}(y) \leq \underbrace{F_{Y_0|D=1}\left(Q^{\mathbb R,-}_{Y_0|D=0}\left(F_{Y_{1|D=0}}(y)\right)\right)}_{=\left[F_{Y_0|D=1}\circ Q^{\mathbb R,-}_{Y_0|D=0}\circ F_{Y_1|D=0}\right] (y)\equiv F^{UB}(y) }.$$

We then finally have: 
\begin{eqnarray}\label{B1}
F^{LB}(y) \leq  F_{Y_{10}|D=1}(y) \leq F^{UB}(y) \text{ for all } y \in \mathbb Y_{10|0}.
\end{eqnarray}

While these above bounds are point-wise sharp for all $y \in \mathbb Y_{10|0},$ they may not be sharp for $y  \in \mathbb R \setminus \mathbb Y_{10|0}$.
And this is because  the upper bound may not be right-continuous in some cases, similarly for the  lower bound which may not be  right-continuous whenever $\{\tilde y \in \mathbb Y_{0|D=1} \cup \{-\infty\}:F_{Y_0|D=1}(\tilde y)\leq u\}$ is open for some $u\in Ran F_{Y_0|D=1}$.

To clarify this point, let us consider the simple case where $Y_{t0}$, $t \in \{0,1\}$ are all discrete random variables with $\mathbb Y_{10|0}=\{y_0,...,y_K\}$. In this case, $F^{LB}(.)$ is a well-defined cdf, while $F^{UB}(.)$ may not be a right-continuous function. Indeed, the function $u \mapsto Q^{\mathbb Y_{0|0},-}(u)$ is left-continuous and the discontinuities happen at $u\in Ran F_{Y_0|D=0}$.
Now, consider that there exists $u_k \in Ran F_{Y_0|D=0} \cap Ran F_{Y_{10}|D=0}$, thus $F^{UB}(.)$ could be left-continuous at
$y_k \in \mathbb Y_{10|0}$ such that $F_{Y_{10}|D=0}(y_k)=u_k$.
If it is left-continuous and not right-continuous in $y_k$, we have:
 $\{y \in \overline{\mathbb R}: F^{UB}(y)> F^{UB}(y_k)\}=(y_k,\infty]$. Let us consider $\epsilon>0$ such that $y_k +\epsilon < y_{k+1}$. In such a case, $F_{Y_{10}|D=1}(y_k+\epsilon)=F_{Y_{10}|D=1}(y_k)$, however, by applying naively the bounds to $y_k$ and $y_{k}+\epsilon$ we have:
\begin{eqnarray}
F^{LB}(y_k) &\leq&  F_{Y_{10}|D=1}(y_k) \leq F^{UB}(y_k), \text{ where } y_k \in \mathbb Y_{10|0}\label{eq1}\\
F^{LB}(y_k+\epsilon) &\leq&  F_{Y_{10}|D=1}(y_k+\epsilon) \leq F^{UB}(y_k+\epsilon), \text{ where } y_k +\epsilon \notin \mathbb Y_{10|0}\label{eq2}
\end{eqnarray}
which implies that the upper bound in (\ref{eq2}) is not sharp since $F^{UB}(y_k+\epsilon)> F^{UB}(y_k)$.
A valid tighter bound for $F^{LB}(y')$ for $y_{k}<y'<y_{k+1}$ is:
\begin{eqnarray*}
F^{LB}(y_k) &\leq&  F_{Y_{10}|D=1}(y') \leq F^{UB}(y_k),\;\; y_{k}\leq y'<y_{k+1}.
\end{eqnarray*}

Since extending the bounds in Eq. (\ref{B1}) to the case where $y \notin \mathbb Y_{10|0}$ provides non-sharp bounds,  we provide an alternative approach that internalizes the idea that our target function of interest must be right-continuous since it is a cdf.
Recall,
\begin{eqnarray}
F^{LB}(t) \leq  F_{Y_{10}|D=1}(t) \leq F^{UB}(t) \text{ for all } t \in \mathbb Y_{10|0}.
\end{eqnarray}
then for any fixed  $y \in \mathbb R$, we have:
\begin{eqnarray*}
&&\lim_{\tilde y \downarrow y}\sup\left\{F^{LB}(t): t\leq \tilde y \; \& \;  t \in \mathbb Y_{10|0} \cup\{-\infty \} \right\}\\&\leq& \lim_{\tilde y \downarrow y}\sup\left\{F_{Y_{10}|D=1}(t): t\leq \tilde y \; \& \;  t \in \mathbb Y_{10|0} \cup\{-\infty \} \right\} \\ &\leq &\lim_{\tilde y \downarrow y}\sup\left\{ F^{UB}(t): t\leq \tilde y \; \& \;  t \in \mathbb Y_{10|0} \cup\{-\infty \} \right\}, \;\; y \in \mathbb R.  
\end{eqnarray*}
Notice that because $\mathbb Y_{10|1} \subseteq \mathbb Y_{10|0}$, and $F_{Y_{10}|D=1}(\cdot)$ is a right-continuous function, we have the following equality by Lemma \ref{lem:boundq}(\ref{lem:boundq2}):
$$\lim_{\tilde y \downarrow y}\sup\left\{F_{Y_{10}|D=1}(t): t\leq \tilde y \; \& \;  t \in \mathbb Y_{10|0} \cup\{-\infty \} \right\}= F_{Y_{10}|D=1}(y) \text{ for all } y \in \mathbb R.$$ The last inequality therefore becomes:
\begin{multline}
\lim_{\tilde y \downarrow y}\sup\left\{F^{LB}(t): t\leq \tilde y \; \& \;  t \in \mathbb Y_{10|0} \cup\{-\infty \} \right\} \\ \leq  F_{Y_{10}|D=1}(y) \leq \lim_{\tilde y \downarrow y}\sup\left\{ F^{UB}(t): t\leq \tilde y \; \& \;  t \in \mathbb Y_{10|0} \cup\{-\infty \} \right\}, \;\; y \in \mathbb R.  
\end{multline}

\subsubsection{Sharpness of the bounds}\label{Bounds:sharpness}

In the previous subsection \ref{Bounds:deriv}, we showed that the bounds are valid. Now, we will show that both bounds are achievable. For the sake of brevity, we will focus only on the upper bound.
The main idea is to provide a DGP which is  only a function of the observable distributions but verifies the model assumptions and for which $\tilde{F}_{Y_{10}|D=1}(y)$ is equal to the  upper bound.

Consider that the unidentified counterfactual distribution is exactly the upper bound:
$$\tilde{F}_{Y_{10}\vert D=1}(y)\equiv \lim_{\tilde y \downarrow y}\sup\left\{F^{UB}(t): t\leq \tilde y \; \& \;  t \in \mathbb Y_{10|0} \cup\{-\infty \} \right\}.$$
For simplicity, we consider the case where $$\lim_{\tilde y \downarrow y}\sup\left\{F^{UB}(t): t\leq \tilde y \; \& \;  t \in \mathbb Y_{10|0} \cup\{-\infty \} \right\}=F^{UB}(y)\equiv F^{UB}_{Y_{10}\vert D=1}(y).$$
We need to define a joint distribution on $(Y_{00},Y_{10},Y_{11}, D)$ such that it is compatible with the data $(Y_0,Y_1,D)$, and Assumptions \ref{stab} and \ref{inc} hold. For any vector $X$, denote $F_{X,D}(x,d)=\mathbb P(X\leq x, D=d)$. Let $F_{Y_{00},Y_{10},Y_{11},D}(y_0,y_{10},y_{11},d)$ be a candidate joint distribution. 
We define
\begin{eqnarray*}
    \tilde{F}_{Y_{00},Y_{10},Y_{11},D}(y_0,y_{10},y_{11},0)&\equiv&F_{Y_{0},Y_{1},D}(y_0,y_{10},0)*F^{UB}_{Y_{10}\vert D=1}(y),\\
    \tilde{F}_{Y_{00},Y_{10},Y_{11},D}(y_0,y_{10},y_{11},1)&\equiv&F_{Y_{0},Y_{1},D}(y_0,y_{11},1)*F^{UB}_{Y_{10}\vert D=1}(y).
\end{eqnarray*}
We construct the proposed distribution using the following rule.
For $\tilde{F}_{Y_{00},Y_{10},Y_{11},D}(y_0,y_{10},y_{11},d)$ to be compatible with the data $(Y_0,Y_1,D)$, we must have
\begin{eqnarray*}
    \tilde{F}_{Y_{00},Y_{10},Y_{11},D}(y_0,y_{10},y_{11},0)&=&F_{Y_{0},Y_{1},D}(y_0,y_{10},0)*\tilde{F}_{Y_{11}\vert Y_{00}\leq y_0,Y_{10}\leq y_{10},D=0}(y_{11}),\\
    \tilde{F}_{Y_{00},Y_{10},Y_{11},D}(y_0,y_{10},y_{11},1)&=&F_{Y_{0},Y_{1},D}(y_0,y_{11},1)*\tilde{F}_{Y_{10}\vert Y_{00}\leq y_0,Y_{11}\leq y_{11},D=1}(y_{10}).
\end{eqnarray*}
The distributions $\tilde{F}_{Y_{11}\vert Y_{00}\leq y_0,Y_{10}\leq y_{10},D=0}(y_{11})$ and $\tilde{F}_{Y_{10}\vert Y_{00}\leq y_0,Y_{11}\leq y_{11},D=1}(y_{10})$ are counterfactual. We set both of them equal to  $\tilde{F}_{Y_{10}\vert Y_{00}\leq \infty,Y_{11}\leq \infty,D=1}(y_{10})=F^{UB}_{Y_{10}\vert D=1}(y)$, which is the counterfactual distribution that we consider above. 

We now show that $\tilde{F}_{Y_{10}|D=1}(y)$ is a cdf. It is easy to see that $\tilde{F}_{Y_{10}|D=1}(y)$ is nondecreasing since for $y \leq y'$ we have 
\begin{eqnarray*}
&&\left\{F^{UB}(t): t\leq y \; \& \;  t \in \mathbb Y_{10|0} \cup\{-\infty \} \right\}  \subseteq \left\{F^{UB}(t): t\leq y' \; \& \;  t \in \mathbb Y_{10|0} \cup\{-\infty \} \right\}.
\end{eqnarray*}
The limits of the function $\tilde{F}_{Y_{10}|D=1}(y)$ at $-\infty$ and $\infty$ are 0 and 1, respectively. By construction, the function $\tilde{F}_{Y_{10}|D=1}(y)$ is a right-continuous function. 

We have
\begin{eqnarray*}
F_{Y_{00},Y_{10},D}(y_0,y_{10},0) &=& q C_{Y_{0},Y_{1}\vert D=0}\left(\frac{1}{q} C_{Y_0,D}(F_{Y_0}(y_0),q),\frac{1}{q} C_{Y_{10},D}(F_{Y_{10}}(y_{10}),q)\right),\\
&=&C_{Y_0,Y_1,D}\left(F_{Y_0}(y_0),F_{Y_{10}}(y_{10}),q\right).
\end{eqnarray*}
We now need to construct copulas $\tilde{C}_{Y_0,D}(u,q)$, $\tilde{C}_{Y_{10},D}(u,q)$, $\tilde{C}_{Y_{0},Y_{1} \vert D=0}(u_0,u_1)$, and $\tilde{C}_{Y_{0},Y_{10},D}(u_0,u_1,q)$ such that the following holds:
\begin{eqnarray*}
F_{Y_{00},Y_{10},D}(y_0,y_{10},0) &=& q \tilde{C}_{Y_{0},Y_{1}\vert D=0}\left(\frac{1}{q} \tilde{C}_{Y_0,D}(F_{Y_0}(y_0),q),\frac{1}{q} \tilde{C}_{Y_{10},D}(\tilde{F}_{Y_{10}}(y_{10}),q)\right),\\
&=&\tilde{C}_{Y_0,Y_1,D}\left(F_{Y_0}(y_0),\tilde{F}_{Y_{10}}(y_{10}),q\right),
\end{eqnarray*}
where $\tilde{F}_{Y_{10}}(y_{10}) = p F^{UB}_{Y_{10}|D=1}(y_{10}) + q F_{Y_1 \vert D=0}(y_{10})\equiv F^{UB}_{Y_{10}}(y_{10})$. 

Since $\overline{Ran}F_{Y_0}$ is closed, we define 
{\tiny{\begin{eqnarray*}
\tilde{C}_{Y_0,D}(u,q)&=& \left\{ \begin{array}{lcl}
      F_{Y_0,D}(Q_{Y_0}^{\mathbb{R},-}(u),0)  
      \text{  if } u \in \overline{\operatorname{Ran}}  F_{Y_0} \cap \overline{\operatorname{Ran}}  \tilde{F}_{Y_{10}} \\ \\
      F_{Y_0,D}(Q_{Y_0}^{\mathbb{R},-}(u),0) \text{  if } u \in \overline{\operatorname{Ran}}  F_{Y_0} \cap (\overline{\operatorname{Ran}}  \tilde{F}_{Y_{10}})^c \\ \\
      F_{Y_1,D}(\tilde{Q}_{Y_{10}}^{\mathbb{R},-}(u),0) \text{  if } u \in (\overline{\operatorname{Ran}}  F_{Y_0})^c \cap \overline{\operatorname{Ran}}  \tilde{F}_{Y_{10}} \\ \\
       F_{Y_0,D}\left(Q_{Y_0}^{\mathbb{R},-}(\underline{u}(u)),0\right) + \left(F_{Y_0,D}\left(Q_{Y_0}^{\mathbb{R},-}(\overline{u}(u)),0\right)-F_{Y_0,D}\left(Q_{Y_0}^{\mathbb{R},-}(\underline{u}(u)),0\right)\right)\frac{u-\underline{u}(u)}{\overline{u}(u)-\underline{u}(u)}
        \\ \\
         \qquad \qquad \qquad \qquad \qquad \text{  if } u \in (\overline{\operatorname{Ran}}  F_{Y_0})^c \cap (\overline{\operatorname{Ran}}  \tilde{F}_{Y_{10}})^c,  \underline{u}(u) \in \overline{\operatorname{Ran}} F_{Y_0}, \text{ and }  \overline{u}(u) \in \overline{\operatorname{Ran}} F_{Y_0}
        \\ \\
         F_{Y_1,D}\left(\tilde{Q}_{Y_{10}}^{\mathbb{R},-}(\underline{u}(u)),0\right) + \left( F_{Y_1,D}\left(\tilde{Q}_{Y_{10}}^{\mathbb{R},-}(\overline{u}(u)),0\right)- F_{Y_1,D}\left(\tilde{Q}_{Y_{10}}^{\mathbb{R},-}(\underline{u}(u)),0\right)\right)\frac{u-\underline{u}(u)}{\overline{u}(u)-\underline{u}(u)}
        \\ \\
         \qquad \qquad \qquad \qquad \qquad \text{  if } u \in (\overline{\operatorname{Ran}}  F_{Y_0})^c \cap (\overline{\operatorname{Ran}}  \tilde{F}_{Y_{10}})^c,  \underline{u}(u) \in \overline{\operatorname{Ran}} \tilde{F}_{Y_{10}}, \text{ and }  \overline{u}(u) \in \overline{\operatorname{Ran}} \tilde{F}_{Y_{10}}
        \\ \\
        F_{Y_1,D}\left(\tilde{Q}_{Y_{10}}^{\mathbb{R},-}(\underline{u}(u)),0\right) + \left(F_{Y_0,D}\left(Q_{Y_0}^{\mathbb{R},-}(\overline{u}(u)),0\right)-F_{Y_1,D}\left(\tilde{Q}_{Y_{10}}^{\mathbb{R},-}(\underline{u}(u)),0\right)\right)\frac{u-\underline{u}(u)}{\overline{u}(u)-\underline{u}(u)}
        \\ \\
         \qquad \qquad \qquad \qquad \qquad \text{  if } u \in (\overline{\operatorname{Ran}}  F_{Y_0})^c \cap (\overline{\operatorname{Ran}}  \tilde{F}_{Y_{10}})^c,  \underline{u}(u) \in \overline{\operatorname{Ran}} \tilde{F}_{Y_{10}}, \text{ and }  \overline{u}(u) \in \overline{\operatorname{Ran}} F_{Y_0}
        \\ \\ 
       F_{Y_0,D}\left(Q_{Y_0}^{\mathbb{R},-}(\underline{u}(u)),0\right) + \left(F_{Y_1,D}\left(\tilde{Q}_{Y_{10}}^{\mathbb{R},-}(\overline{u}(u)),0\right)-F_{Y_0,D}\left(Q_{Y_0}^{\mathbb{R},-}(\underline{u}(u)),0\right)\right)\frac{u-\underline{u}(u)}{\overline{u}(u)-\underline{u}(u)}
        \\ \\
         \qquad \qquad \qquad \qquad \qquad \text{  if } u \in (\overline{\operatorname{Ran}}  F_{Y_0})^c \cap (\overline{\operatorname{Ran}}  \tilde{F}_{Y_{10}})^c,  \underline{u}(u) \in \overline{\operatorname{Ran}} F_{Y_0}, \text{ and } \overline{u}(u) \in \overline{\operatorname{Ran}} \tilde{F}_{Y_{10}}
     \end{array}\right.
     \\
   \tilde{C}_{Y_{10},D}(u,q)&=& \tilde{C}_{Y_0,D}(u,q),  
\end{eqnarray*}}}

\hspace{-0.4cm}where for any $u \in [0,1]$, $\underline{u}(u)\equiv \sup\{q\in \overline{\operatorname{Ran}} F_{Y_0} \cup \overline{\operatorname{Ran}} \tilde{F}_{Y_{10}}: q \leq u\}$, $\overline{u}(u)\equiv \inf\{q\in \overline{\operatorname{Ran}} F_{Y_0} \cup \overline{\operatorname{Ran}} \tilde{F}_{Y_{10}}: q \geq u\}$, and $   \tilde{Q}^{\mathbb{R},-}_{Y_{10}}(u)\equiv\inf\{y\in\mathbb{R}:\tilde{F}_{Y_{10}}(y)\geq u\}$.

{\tiny{\begin{eqnarray*}
\tilde{C}_{Y_0,Y_1 \vert D=0}(u_0,u_1)= \left\{ \begin{array}{lcl}
      F_{Y_0,Y_1 \vert D=0}(Q_{Y_0\vert D=0}^{\mathbb{R},-}(u_0),Q_{Y_1\vert D=0}^{\mathbb{R},-}(u_1)) \quad\text{  if } (u_0,u_1) \in \overline{\operatorname{Ran}}  F_{Y_0 \vert D=0}\times \overline{\operatorname{Ran}}  F_{Y_1 \vert D=0}\\ \\
       F_{Y_0,Y_1 \vert D=0}\bigg(Q_{Y_0\vert D=0}^{\mathbb{R},-}(\underline{u}(u_0)),Q_{Y_1\vert D=0}^{\mathbb{R},-}(\underline{u}(u_1))\bigg)+\\\bigg[F_{Y_0,Y_1 \vert D=0}\bigg(Q_{Y_0\vert D=0}^{\mathbb{R},-}(\overline{u}_0(u_0)),Q_{Y_1\vert D=0}^{\mathbb{R},-}(\overline{u}_1(u_1))\bigg)-F_{Y_0,Y_1 \vert D=0}\bigg(Q_{Y_0\vert D=0}^{\mathbb{R},-}(\underline{u}_0(u_0)),Q_{Y_1\vert D=0}^{\mathbb{R},-}(\underline{u}_1(u_1))\bigg)\bigg]\\
       \\
       * \frac{(u_0-\underline{u}_0(u_0)) (u_1-\underline{u}_1(u_1))}{(\overline{u}_0(u_0)-\underline{u}_0(u_0))(\overline{u}_1(u_1)-\underline{u}_1(u_1))} \text{  if } (u_0,u_1) \notin \overline{\operatorname{Ran}}  F_{Y_0 \vert D=0}\times \overline{\operatorname{Ran}}  F_{Y_1 \vert D=0}
     \end{array}\right.
  \end{eqnarray*}}}
  \hspace{-0.3cm}where for $t\in \{0,1\}$ and for any $(u_0,u_1) \in [0,1]^2$, $\underline{u_t}(u)\equiv \sup\{q\in \overline{\operatorname{Ran}} F_{Y_t\vert D=0}: q \leq u\}$, while $\overline{u_t}(u)\equiv \inf\{q\in \overline{\operatorname{Ran}} F_{Y_t \vert D=0}: q \geq u\}$.

We then define for $(u_0,u_1) \in [0,1]^2$
\begin{eqnarray*}
\tilde{C}_{Y_{00},Y_{10},D}(u_0,u_1,q) &=& q \tilde{C}_{Y_{0},Y_{1}\vert D=0}\left(\frac{1}{q} \tilde{C}_{Y_0,D}(u_0,q),\frac{1}{q} \tilde{C}_{Y_0,D}(u_1,q)\right).
\end{eqnarray*}
We can verify that $\tilde{C}_{Y_{00},Y_{10},D}(u_0,u_1,q)$ is a well-defined copula. We start by showing that $\tilde{C}_{Y_0,D}(u_0,q)$ is a well-defined subcopula. 
To do so, we need to introduce two  intermediate lemmata:

\begin{lemma}\label{lem:sharp3}
For any $u \in RanF_{Y_{10}}^{UB}$ and $v\in RanF_{Y_0}$ such that $u < v$, we have $\tilde{C}_{Y_0,D}(u,q) < \tilde{C}_{Y_0,D}(v,q)$.
\end{lemma}

\begin{lemma}\label{lem:sharp5}
Suppose $F_{Y_0}(y-) \in RanF_{Y_0}$ for all $y$. For any $u \in RanF_{Y_{10}}^{UB}$ and $v\in RanF_{Y_0}$ such that $v < u$, we have $\tilde{C}_{Y_0,D}(v,q) < \tilde{C}_{Y_0,D}(u,q)$.
\end{lemma}

First, we have $\tilde{C}_{Y_0,D}(1,q)=F_{Y_0,D}(Q_{Y_0}^{\mathbb{R},-}(1),0)=q$. Now let us show that for all $(u,v)\in[0,1]^2$ such that $u < v$, we have $\tilde{C}_{Y_0,D}(u,q)< \tilde{C}_{Y_0,D}(v,q)$. From the definition of $\tilde{C}_{Y_0,D}(u,q)$ and Lemma \ref{cop:mon}, it follows that, when $u$ and $v$ belong to the same range, this monotonicity condition holds. We are going to prove it when $u$ and $v$ belong to different ranges. On the one hand, if $u \in Ran\tilde{F}_{Y_{10}}$ and $v \in RanF_{Y_0}$, then from Lemma \ref{lem:sharp3}, we have $\tilde{C}_{Y_0,D}(u,q) < \tilde{C}_{Y_0,D}(v,q)$. On the other hand, if $v \in Ran\tilde{F}_{Y_{10}}$ and $u \in RanF_{Y_0}$, then from Lemma \ref{lem:sharp5}, we have $\tilde{C}_{Y_0,D}(u,q) < \tilde{C}_{Y_0,D}(v,q)$.    
Since $\tilde{C}_{Y_0,Y_1 \vert D=0}(u_0,u_1)$ is an extended copula of the identified part of the copula of $(Y_0,Y_1) \vert D=0$ through the Sklar theorem, it is a well-defined copula. Any extended copula of this form should work for the proof, as we do not impose any additional restrictions on the true copula of $(Y_0,Y_1) \vert D=0$.   

We also need to check that $\tilde{C}_{Y_{00},Y_{10},D}\left(F_{Y_0}(y_0),\tilde{F}_{Y_{10}}(y_{10}),q\right)=F_{Y_0,Y_1,D}(y_0,y_{10},0).$
This latter equality holds by construction of $\tilde{C}_{Y_{00},Y_{10},D}(u_0,u_1,q)$.

When we let $u_0$ go to 1, we obtain
\begin{eqnarray*}
\tilde{C}_{Y_{10},D}(u_1,q)=\tilde{C}_{Y_{00},Y_{10},D}(1,u_1,q) &=& q \tilde{C}_{Y_{0},Y_{1}\vert D=0}\left(\frac{1}{q} \tilde{C}_{Y_0,D}(1,q),\frac{1}{q} \tilde{C}_{Y_0,D}(u_1,q)\right)\nonumber\\
&=&q \frac{1}{q} \tilde{C}_{Y_0,D}(u_1,q)=\tilde{C}_{Y_0,D}(u_1,q).
\end{eqnarray*}
Similarly,
\begin{eqnarray*}
\tilde{C}_{Y_{00},D}(u_0,q)=\tilde{C}_{Y_{00},Y_{10},D}(u_0,1,q) &=& q \tilde{C}_{Y_{0},Y_{1}\vert D=0}\left(\frac{1}{q} \tilde{C}_{Y_0,D}(u_0,q),\frac{1}{q} \tilde{C}_{Y_0,D}(1,q)\right)\nonumber\\
&=&q \frac{1}{q} \tilde{C}_{Y_0,D}(u_0,q)=\tilde{C}_{Y_0,D}(u_0,q).
\end{eqnarray*}
And by construction, we have $\tilde{C}_{Y_{10},D}(u,q)=\tilde{C}_{Y_0,D}(u,q)$ for all $u\in 
[0,1]$ (Assumption \ref{stab} holds). 
Furthermore, we have shown above that $\tilde{C}_{Y_0,D}(u,q)$ is strictly increasing in $u$ (Assumption \ref{inc} holds). 

By construction, the proposed joint distribution $\tilde{F}_{Y_{00},Y_{10},Y_{11},D}(y_0,y_1,y_2,d)$ is compatible with the data and the proposed copulas $\tilde{C}_{Y_0,D}(u,q)$, and $\tilde{C}_{Y_{10},D}(u,q)$ satisfy Assumptions \ref{stab} and \ref{inc}. 

The proof is similar for the lower bound on $F_{Y_{10\vert D=1}}(y)$ and any distribution in the identified set of $F_{Y_{10}\vert D=1}(y_{10})$.

To complete the proof, it remains to show the two intermediate lemmata. 

\subsubsection{Proofs of Intermediate Lemmata}\label{Bounds:intermediate_lemmata}

\subsubsection*{Proof of Lemma \ref{lem:sharp3}}

First, we start by the following claims:

\begin{claim}\label{claim:sharp2}
For any $u_y=F_{Y_{10}}^{UB}(y) \in RanF_{Y_{10}}^{UB}$, the smallest $v \in RanF_{Y_{0}}$ such that $u_y \leq v$ is $v_y=F_{Y_0}(h(y))$.
\end{claim}
\noindent \emph{Proof}. We have 
\begin{eqnarray*}
u_y = q F_{Y_{10}\vert D=0}(y) + p F^{UB}_{Y_{10}\vert D=1}(y) = q F_{Y_{10}\vert D=0}(y) + p F_{Y_0\vert D=1}(h(y)).
\end{eqnarray*}
Since $F_{Y_0\vert D=1}(h(y)) \in RanF_{Y_0 \vert D=1},$ to obtain the smallest element $v \in RanF_{Y_{0}}$, we need to find the smallest element $s$ on $RanF_{Y_0\vert D=0}$ such that $F_{Y_{1}\vert D=0}(y) \leq s$. From Lemma \ref{lem:boundq}.(\ref{lem:boundq1}), $s=F_{Y_0\vert D=0}(Q^{\mathbb R,-}_{F_{Y_0\vert D=0}}(F_{Y_{1}\vert D=0}(y)))$. This completes the proof of Claim \ref{claim:sharp2}.\qed

\begin{claim}\label{claim:sharp1}
For any $u \in RanF_{Y_{10}}^{UB}$, there exists $v \in RanF_{Y_{0}}$ such that $u \leq v \Longrightarrow \tilde{C}_{Y_0,D}(u,q)\leq \tilde{C}_{Y_0,D}(v,q).$
\end{claim}
\noindent\emph{Proof}. $u_y\equiv F_{Y_{10}}^{UB}(y) = q F_{Y_{10}\vert D=0}(y) + p F^{UB}_{Y_{10}\vert D=1}(y),$ where $F^{UB}_{Y_{10}\vert D=1}(y)=F_{Y_0\vert D=1}(h(y))$ with $h(y)=Q_{Y_{0\vert 0}}^{\mathbb R,-}(F_{Y_1\vert D=0}(y))$. Then, 
$u_y \leq q F_{Y_{0}\vert D=0}(h(y)) + p F_{Y_0\vert D=1}(h(y))$, since $F_{Y_{10}\vert D=0}(y) \leq F_{Y_{0}\vert D=0}(h(y))$ by construction. So, $u_y \leq F_{Y_0}(h(y))\equiv v_y \in RanF_{Y_0}$. Now, the following hold:
\begin{eqnarray*}
u_y \leq v_y &\Longrightarrow& q F_{Y_{10}\vert D=0}(y) + p F_{Y_0\vert D=1}(h(y)) \leq q F_{Y_{0}\vert D=0}(h(y)) + p F_{Y_0\vert D=1}(h(y)),\\
&\Longrightarrow& q F_{Y_{10}\vert D=0}(y) \leq q F_{Y_{0}\vert D=0}(h(y)),\\
&\Longrightarrow& F_{Y_{10},D}(y,0) \leq F_{Y_{0},D}(h(y),0),\end{eqnarray*}
\begin{eqnarray*}
&\Longrightarrow& \tilde{C}_{Y_0,D}(F_{Y_{10}}^{UB}(y),q)\leq \tilde{C}_{Y_0,D}(F_{Y_0}(h(y)),q),\\
&\Longrightarrow& \tilde{C}_{Y_0,D}(u_y,q)\leq \tilde{C}_{Y_0,D}(v_y,q).
\end{eqnarray*}
This completes the proof of Claim \ref{claim:sharp1}.\qed

Now we proceed to complete the proof of the lemma. Take $u \in RanF_{Y_{10}}^{UB}$ and $v\in RanF_{Y_0}$ such that $u < v$. Since $u \in RanF_{Y_{10}}^{UB}$, there exits $y$ such that $u_y=F_{Y_{10}}^{UB}(y)$. Then, from Claim \ref{claim:sharp2}, there exists $v_y= F_{Y_0}(h(y))$ such that $u_y \leq v_y$. From Claim \ref{claim:sharp2}, we have $v_y \leq v$. If $v=v_y$, then we have $F_{Y_{10}}^{UB}(y) < F_{Y_0}(h(y))$, which implies successively  
\begin{eqnarray*}
q F_{Y_{10}\vert D=0}(y) + p F_{Y_0\vert D=1}(h(y)) &<& q F_{Y_{0}\vert D=0}(h(y)) + p F_{Y_0\vert D=1}(h(y)),\\
q F_{Y_{10}\vert D=0}(y) &<& q F_{Y_{0}\vert D=0}(h(y)),\\
F_{Y_{10},D}(y,0) &<& F_{Y_{0},D}(h(y),0),\\
\tilde{C}_{Y_0,D}(F_{Y_{10}}^{UB}(y),q) &<& \tilde{C}_{Y_0,D}(F_{Y_0}(h(y)),q),\\
\tilde{C}_{Y_0,D}(u,q) &<& \tilde{C}_{Y_0,D}(v,q).
\end{eqnarray*}
If $v_y < v$, then from Claim \ref{claim:sharp1} we have $\tilde{C}_{Y_0,D}(u,q)\leq \tilde{C}_{Y_0,D}(v_y,q)$. And since $\tilde{C}_{Y_0,D}(v,q)$ is strictly increasing on $RanF_{Y_0}$ from Lemma \ref{cop:mon}, we have $\tilde{C}_{Y_0,D}(v_y,q) < \tilde{C}_{Y_0,D}(v,q)$. Therefore, $\tilde{C}_{Y_0,D}(u,q) < \tilde{C}_{Y_0,D}(v,q)$.
\qed

\subsubsection*{ Proof of Lemma \ref{lem:sharp5}}

We first start by stating and proving the following claim:

\begin{claim}\label{claim:sharp4}
Suppose $F_{Y_0}(y-) \in RanF_{Y_0}$ for all $y$. For any $u_y=F_{Y_{10}}^{UB}(y) \in RanF_{Y_{10}}^{UB}$, there exist $w_y \in [0,1]$ and $v_y\in RanF_{Y_0}$ such that $v_y \leq w_y \leq u_y$ and $\tilde{C}_{Y_0,D}(v_y,q) \leq \tilde{C}_{Y_0,D}(w_y,q) \leq \tilde{C}_{Y_0,D}(u_y,q)$.
\end{claim}
\noindent\emph{Proof}. We have 
    \begin{eqnarray*}
    u_y&=&F_{Y_{10}}^{UB}(y)= q F_{Y_{10}\vert D=0}(y) + p F^{UB}_{Y_{10}\vert D=1}(y),\\
    &\geq& q F_{Y_{10}\vert D=0}(y) + p F^{LB}_{Y_{10}\vert D=1}(y)= q F_{Y_{10}\vert D=0}(y) + p F_{Y_0\vert D=1}(\underline{h}(y)-)\equiv w_y,\\
    &\geq& q F_{Y_{0}\vert D=0}(\underline{h}(y)-) + p F_{Y_0\vert D=1}(\underline{h}(y)-)\equiv v_y,
    \end{eqnarray*}
    where $\underline{h}(y)= Q_{Y_0\vert 0}^{\mathbb R,+}(F_{Y_1\vert D=0}(y))$, and the second inequality holds from Lemma \ref{lem:boundq}. 
We discuss two cases.

    Case 1: $F^{UB}_{Y_{10}\vert D=1}(y)=F^{LB}_{Y_{10}\vert D=1}(y)$

    In this case, $u_y=w_y$, we have
\begin{eqnarray*}
u_y \geq v_y &\Longrightarrow& q F_{Y_{10}\vert D=0}(y) + p F_{Y_0\vert D=1}(\underline{h}(y)-) \geq q F_{Y_{0}\vert D=0}(\underline{h}(y)-) + p F_{Y_0\vert D=1}(\underline{h}(y)-),\\
&\Longrightarrow& q F_{Y_{10}\vert D=0}(y) \leq q F_{Y_{0}\vert D=0}(\underline{h}(y)-),\\
&\Longrightarrow& F_{Y_{10},D}(y,0) \leq F_{Y_{0},D}(\underline{h}(y)-,0),\\
&\Longrightarrow& \tilde{C}_{Y_0,D}(F_{Y_{10}}^{UB}(y),q)\leq \tilde{C}_{Y_0,D}(F_{Y_0}(\underline{h}(y)-),q),\\
&\Longrightarrow& \tilde{C}_{Y_0,D}(u_y,q)\leq \tilde{C}_{Y_0,D}(v_y,q).
\end{eqnarray*}

Case 2: $F^{LB}_{Y_{10}\vert D=1}(y)< F^{UB}_{Y_{10}\vert D=1}(y)$

In this case, $w_y \notin \overline{\operatorname{Ran}} F_{Y_{10}}^{UB}$. From Lemma \ref{lem:boundq}, $v_y$ is the highest element of $\overline{\operatorname{Ran}} F_{Y_0}$ such that $w_y \geq v_y$. 
First, suppose $w_y \notin \overline{\operatorname{Ran}} F_{Y_0}.$ Then $w_y \in (\overline{\operatorname{Ran}} F_{Y_0})^c \cap (\overline{\operatorname{Ran}} F^{UB}_{Y_{10}})^c$. Let $\overline{u}(w_y)\equiv \inf\{q\in \overline{\operatorname{Ran}} F_{Y_0} \cup \overline{\operatorname{Ran}} F^{UB}_{Y_{10}}: q \geq w_y\}$. We have $v_y \leq w_y < \overline{u}(w_y) \leq u_y$, and either $\overline{u}(w_y) \in \overline{\operatorname{Ran}} F_{Y_0}$ or $\overline{u}(w_y) \in \overline{\operatorname{Ran}} F^{UB}_{Y_{10}}$. 

If $\overline{u}(w_y) \in \overline{\operatorname{Ran}} F_{Y_0}$, then
\begin{eqnarray*}&&
\tilde{C}_{Y_0,D}(w_y,q)\\&=& 
       F_{Y_0,D}\left(Q_{Y_0}^{\mathbb{R},-}(v_y),0\right)+\bigg[F_{Y_0,D}\left(Q_{Y_0}^{\mathbb{R},-}(\overline{u}(w_y)),0\right)-F_{Y_0,D}\left(Q_{Y_0}^{\mathbb{R},-}(v_y),0\right)\bigg] * \frac{w_y-v_y}{\overline{u}(w_y)-v_y}.
\end{eqnarray*}
Since $0 \leq \frac{w_y-v_y}{\overline{u}(w_y)-v_y} \leq 1$ and $\bigg[F_{Y_0,D}\left(Q_{Y_0}^{\mathbb{R},-}(\overline{u}(w_y)),0\right)-F_{Y_0,D}\left(Q_{Y_0}^{\mathbb{R},-}(v_y),0\right)\bigg] \geq 0$ from Lemma \ref{cop:mon}, the following holds: 
\begin{eqnarray*}
&&\tilde{C}_{Y_0,D}(v_y,q) \equiv F_{Y_0,D}\left(Q_{Y_0}^{\mathbb{R},-}(v_y),0\right) 
 \leq \tilde{C}_{Y_0,D}(w_y,q) \leq  F_{Y_0,D}\left(Q_{Y_0}^{\mathbb{R},-}(\overline{u}(w_y)),0\right)\\
 && \qquad \qquad \qquad \qquad \qquad \leq  F_{Y_0,D}\left(Q_{Y_0}^{\mathbb{R},-}(u_y),0\right) \equiv \tilde{C}_{Y_0,D}(u_y,q),
\end{eqnarray*}
where the last inequality holds because $Q_{Y_{0}}^{\mathbb{R},-}(u)$ is monotone in $u$. Hence, 
$$\tilde{C}_{Y_0,D}(v_y,q) \leq \tilde{C}_{Y_0,D}(w_y,q) \leq \tilde{C}_{Y_0,D}(u_y,q).$$

If $\overline{u}(w_y) \in \overline{\operatorname{Ran}} F^{UB}_{Y_{10}}$, then $\overline{u}(w_y) \in \overline{\operatorname{Ran}} F^{UB}_{Y_{10}}=u_y$, and 
\begin{eqnarray*}
&&\tilde{C}_{Y_0,D}(w_y,q)\\
&=& 
       F_{Y_0,D}\left(Q_{Y_0}^{\mathbb{R},-}(v_y),0\right)+\bigg[F_{Y_{1},D}\left(\tilde{Q}_{Y_{10}}^{\mathbb{R},-}(u_y),0\right)-F_{Y_0,D}\left(Q_{Y_0}^{\mathbb{R},-}(v_y),0\right)\bigg] * \frac{w_y-v_y}{\overline{u}(w_y)-v_y},\end{eqnarray*}
\begin{eqnarray*}       &=& 
       \tilde{C}_{Y_0,D}(v_y,q)+\bigg[F_{Y_1,D}(y,0)-F_{Y_0,D}\left(\underline{h}(y)-,0\right)\bigg] * \frac{w_y-v_y}{\overline{u}(w_y)-v_y}
\end{eqnarray*}
Since $0 \leq \frac{w_y-v_y}{\overline{u}(w_y)-v_y} \leq 1$ and $\bigg[F_{Y_1,D}(y,0)-F_{Y_0,D}\left(\underline{h}(y)-,0\right)\bigg] \geq 0$ from Lemma \ref{lem:boundq}, the following holds: 
\begin{eqnarray*}
\tilde{C}_{Y_0,D}(v_y,q) \leq \tilde{C}_{Y_0,D}(w_y,q) \leq F_{Y_1,D}\left(y,0\right) \equiv \tilde{C}_{Y_0,D}(u_y,q).
\end{eqnarray*}

Second, suppose $w_y \in \overline{\operatorname{Ran}} F_{Y_0}.$ Then, from Lemma \ref{lem:boundq}, we must have $w_y=v_y$, which implies $F_{Y_{1}, D}(y,0)=F_{Y_{0},D}(\underline{h}(y)-,0)$, which in turn implies $\tilde{C}_{Y_0,D}(u_y,q)=\tilde{C}_{Y_0,D}(v_y,q)=\tilde{C}_{Y_0,D}(w_y,q)$.

\noindent This completes the proof of Claim \ref{claim:sharp4}. \qed

Now we proceed to complete the proof of the lemma. Take $u \in RanF_{Y_{10}}^{UB}$ and $v\in RanF_{Y_0}$ such that $v < u$. Since $u \in RanF_{Y_{10}}^{UB}$, there exits $y$ such that $u_y=F_{Y_{10}}^{UB}(y)$. From Claim \ref{claim:sharp4}, there exists $w_y \in [0,1]$ and $v_y \in \overline{\operatorname{Ran}} F_{Y_0}$ such that $v\leq v_y < w_y \leq u$.

 Case 1: $F^{UB}_{Y_{10}\vert D=1}(y)=F^{LB}_{Y_{10}\vert D=1}(y)$

    In this case, $u_y=w_y$, we have
\begin{eqnarray*}
u_y > v_y &\Longrightarrow& q F_{Y_{10}\vert D=0}(y) + p F_{Y_0\vert D=1}(\underline{h}(y)-) \geq q F_{Y_{0}\vert D=0}(\underline{h}(y)-) + p F_{Y_0\vert D=1}(\underline{h}(y)-),\\
&\Longrightarrow& q F_{Y_{10}\vert D=0}(y) > q F_{Y_{0}\vert D=0}(\underline{h}(y)-),\\
&\Longrightarrow& F_{Y_{10},D}(y,0) > F_{Y_{0},D}(\underline{h}(y)-,0),\\
&\Longrightarrow& \tilde{C}_{Y_0,D}(F_{Y_{10}}^{UB}(y),q) > \tilde{C}_{Y_0,D}(F_{Y_0}(\underline{h}(y)-),q),\\
&\Longrightarrow& \tilde{C}_{Y_0,D}(u_y,q) > \tilde{C}_{Y_0,D}(v_y,q) \geq \tilde{C}_{Y_0,D}(v,q),\ \text{ since } v_y, v \in \overline{\operatorname{Ran}} F_{Y_0},\\
&\Longrightarrow& \tilde{C}_{Y_0,D}(u,q) > \tilde{C}_{Y_0,D}(v,q).
\end{eqnarray*}

Case 2: $F^{LB}_{Y_{10}\vert D=1}(y)< F^{UB}_{Y_{10}\vert D=1}(y)$

The proof here is very similar to Case 2 in Claim \ref{claim:sharp4}, except the strict inequality $0 < \frac{w_y-v_y}{\overline{u}(w_y)-v_y} < 1$. This strict inequality implies
$$\tilde{C}_{Y_0,D}(v,q) \leq \tilde{C}_{Y_0,D}(v_y,q) < \tilde{C}_{Y_0,D}(w_y,q) \leq \tilde{C}_{Y_0,D}(u_y,q).$$
Hence, $\tilde{C}_{Y_0,D}(v,q) < \tilde{C}_{Y_0,D}(u,q)$.

Now we have completed the proof of the two intermediate lemmata and thereby the proof of Theorem \ref{Main:theorem}.

\qed

\subsection{Proof of Corollary \ref{cor:Iden}}\label{apx:corident}
\begin{proof}
    In the continuous cdfs case, we have
    \begin{eqnarray*}
    F^{LB}(t)&=&F_{Y_0|D=1}\left(Q^{\mathbb R,+}_{Y_0|D=0}\left(F_{Y_{1|D=0}}(t)\right)-\right)\; \\\; 
    &=&\mathbb P\left(Y_0 < Q^{\mathbb R,+}_{Y_0|D=0}\left(F_{Y_{1|D=0}}(t)\right)\vert D=1\right) \; \text{ by definition} \\\; 
    &=&\mathbb P\left(Y_0 \leq Q^{\mathbb R,+}_{Y_0|D=0}\left(F_{Y_{1|D=0}}(t)\right)\vert D=1\right) \; \text{ under continuity}\\\; 
    &=&F_{Y_0|D=1}\left(Q^{\mathbb R,+}_{Y_0|D=0}\left(F_{Y_{1|D=0}}(t)\right)\right)\; \\\;
    F^{UB}(t)&=&F_{Y_0|D=1}\left(Q^{\mathbb R,-}_{Y_0|D=0}\left(F_{Y_{1|D=0}}(t)\right)\right).
\end{eqnarray*}
We know that $Q^{\mathbb R,-}_{Y_0|D=0}(u) \leq Q^{\mathbb R,+}_{Y_0|D=0}(u)$ since $Q^{\mathbb R,+}_{Y_0|D=0}(u)=\sup\{y \in \mathbb R: F_{Y_0\vert D=0}(y)=u\}$ and $Q^{\mathbb R,-}_{Y_0|D=0}(u)=\inf\{y \in \mathbb R: F_{Y_0\vert D=0}(y)=u\}$ in the continuous cdf case.
Since $F_{Y_0\vert D=1}$ is nondecreasing, $F_{Y_0|D=1}\left(Q^{\mathbb R,-}_{Y_0|D=0}\left(F_{Y_{1|D=0}}(t)\right)\right) \leq F_{Y_0|D=1}\left(Q^{\mathbb R,+}_{Y_0|D=0}\left(F_{Y_{1|D=0}}(t)\right)\right)$, that is,
$F^{UB}(t)\leq F^{LB}(t)$. We know that under our model assumptions $F^{LB}(t)\leq F^{UB}(t)$, therefore it follows that $F^{LB}(t)=F^{UB}(t)$. 

\end{proof}
\subsection{Proof of Theorem \ref{GMain:theorem}}
The proof of this theorem follows by similar arguments to the proof of Theorem \ref{Main:theorem} and is therefore provided in Section \ref{app:proof_theorem2} of the online appendix.

\subsection{Proof of Claim \ref{claim:eq}}\label{proof:claim1}
(i) $\Longrightarrow$ (ii).

Since the cdf $F_{Y_{t0}}$ is continuous and strictly increasing, we have 
\begin{eqnarray*}
Y_{t0}&=&Q^{\mathbb R,-}_{Y_{t0}}\left(F_{Y_{t0}}(Y_{t0})\right),\\
&=& Q^{\mathbb R,-}_{Y_{t0}}\left(U_{t0})\right),\ \text{ where } U_{t0}\equiv F_{Y_{t0}}(Y_{t0})\sim \mathcal U_{[0,1]},\\
&=& h_t(U_{t0}), \ \text{ where } h_t(u)\equiv Q^{\mathbb R,-}_{Y_{t0}}(u).
\end{eqnarray*}

By definition, $h_t$ is continuous and strictly increasing as is the quantile function $Q^{\mathbb R,-}_{Y_{t0}}$. Then, the following equalities hold: 
\begin{eqnarray*}
C_{Y_{t0},D}&=&C_{h_t(U_{t0}),D}=C_{U_{t0},D},
\end{eqnarray*}
where the second equality holds from the invariance principle in \citeauthor{Embrechts_al2013} (\citeyear{Embrechts_al2013}, Proposition 4(2)). Therefore, 
\begin{eqnarray*}
C_{Y_{00},D}(u,q) = C_{Y_{10},D}(u,q) &\Longrightarrow& C_{U_{00},D}(u,q) = C_{U_{10},D}(u,q),\\
&\Longrightarrow& C_{U_{00},D}(F_{U_{00}}(u), F_D(0)) =C_{U_{10},D}(F_{U_{10}}(u), F_D(0)),\\
&\Longrightarrow& F_{U_{00},D}(u,0) = F_{U_{10},D}(u,0),\\
&\Longrightarrow& u-F_{U_{00},D}(u,0) = u-F_{U_{10},D}(u,0),\\
&\Longrightarrow& F_{U_{00}}(u)-F_{U_{00},D}(u,0) = F_{U_{10}}(u)-F_{U_{10},D}(u,0),\\
&\Longrightarrow& \mathbb P(U_{00} \leq u, D=1) = \mathbb P(U_{10} \leq u, D=1),
\end{eqnarray*}
where the second implication follows from $U_{t0} \sim \mathcal U_{[0,1]}$ and $F_D(0)=q$, the third holds from Sklar's theorem, and the fifth follows from $U_{t0} \sim \mathcal U_{[0,1]}$. Hence, we have:
\begin{eqnarray*}
C_{Y_{00},D}(u,q) = C_{Y_{10},D}(u,q)
&\Longrightarrow& \mathbb P(U_{00} \leq u, D=d) = \mathbb P(U_{10} \leq u, D=d) \text{ for } d\in\{0,1\},\\ 
&\Longrightarrow& \mathbb P(U_{00} \leq u, D=d)/\mathbb P(D=d) = \mathbb P(U_{10} \leq u, D=d)/\mathbb P(D=d),\\
&\Longrightarrow& F_{U_{00}\vert D}(u\vert d) \equiv \mathbb P(U_{00} \leq u \vert D=d) = \mathbb P(U_{10} \leq u \vert D=d)\equiv F_{U_{10}\vert D}(u\vert d),\\
&\Longrightarrow& U_{00}\vert D=d \sim U_{10}\vert D=d.
\end{eqnarray*}

(ii) $\Longrightarrow$ (i). Suppose there exist two strictly increasing functions $h_t(.), t \in \{0,1\}$ and two uniformly distributed random variables over $[0,1]$ $U_{00}$ and $U_{10}$ such that $Y_{t0}=h_t(U_{t0})$ and $U_{00}|D=d \sim U_{10}|D=d$. Then, we have 
\begin{eqnarray*}
F_{U_{00}\vert D}(u\vert d) =F_{U_{10}\vert D}(u\vert d)
&\Longrightarrow& F_{U_{00}\vert D}(u\vert d)\mathbb P(D=d) = F_{U_{10}\vert D}(u\vert d)\mathbb P(D=d),\\ 
&\Longrightarrow& \mathbb P(U_{00} \leq u, D=d) = \mathbb P(U_{10} \leq u, D=d),\\
&\Longrightarrow& F_{U_{00},D}(u, 0) =F_{U_{10}, D}(u, 0) \text{ for } d=0,\\
&\Longrightarrow& C_{U_{00},D}(F_{U_{00}}(u), F_D(0)) =C_{U_{10},D}(F_{U_{10}}(u), F_D(0)),\\
&\Longrightarrow& C_{U_{00},D}(u, q) =C_{U_{10},D}(u, q),\\
&\Longrightarrow& C_{h_0(U_{00}),D}(u, q)=C_{U_{00},D}(u, q) =C_{U_{10},D}(u, q)=C_{h_1(U_{00}),D}(u, q),\\
&\Longrightarrow& C_{Y_{00},D}(u, q)=C_{Y_{10},D}(u, q),
\end{eqnarray*}
where the fourth implication holds from Sklar's theorem, the fifth follows from $U_{t0} \sim \mathcal U_{[0,1]}$, the sixth follows by the invariance principle in \citeauthor{Embrechts_al2013} (\citeyear{Embrechts_al2013}, Proposition 4.(2)), and the last holds from $Y_{t0}=h_t(U_{t0})$.   
\qed

\newpage

\renewcommand\thefigure{A.\arabic{figure}} 
\setcounter{figure}{0}

\renewcommand\thetable{A.\arabic{table}} 
\setcounter{table}{0}
\setcounter{page}{0} 
\renewcommand\theclaim{D.\arabic{claim}} 
\setcounter{claim}{0}
\pagenumbering{gobble}
\begin{center}
\vspace{2cm}
    \huge{{Online Appendix}}

\huge{Evaluating the Impact of Regulatory Policies on Social Welfare\\ in Diff-in-Diff Settings}\\
\bigskip
    \Large{Dalia Ghanem \quad D\'esir\'e K\'edagni\quad Ismael Mourifi\'e}
\end{center}
\vspace{1cm}
\startcontents[sections]
\printcontents[sections]{l}{1}{\setcounter{tocdepth}{1}}
\setcounter{page}{0}
\newpage
\pagenumbering{arabic}

\section{Supplementary results for Section \ref{Sec: Main-results}}

\subsection{Parallel trends as covariance stability}\label{Appen:covstability}
\begin{lemma}\label{lem:covstability} Suppose $\mathbb P(D=1)\in(0,1)$.
$$\mathbb E[Y_{10}-Y_{00}\vert D=1]=\mathbb E[Y_{10}-Y_{00}\vert D=0]~~\Longleftrightarrow~~ Cov(Y_{00},D)=Cov(Y_{10,D}).$$
\end{lemma}
\begin{proof} The result follows by first multiplying $\mathbb E[Y_{10}-Y_{00}\vert D=1]-\mathbb E[Y_{10}-Y_{00}\vert D=0]$ by $\mathbb P(D=1)\mathbb P(D=0)$ and then simplifying the resulting expression as follows, 
\begin{eqnarray*}&&\mathbb P(D=1)\mathbb P(D=0)\mathbb E[Y_{10}-Y_{00}\vert D=1]-\mathbb P(D=1)\mathbb P(D=0)\mathbb E[Y_{10}-Y_{00}\vert D=0]\\&=& \mathbb P(D=0)\mathbb E[(Y_{10}-Y_{00})D]-\mathbb P(D=1)\mathbb E[(Y_{10}-Y_{00})(1-D)]\\
&=&
\mathbb E\left[(Y_{10}-Y_{00})(1-\mathbb P(D=1))D-(Y_{10}-Y_{00})(1-D)\mathbb P(D=1)\right]\\
&=&\mathbb E\left[(Y_{10}-Y_{00})(D-\mathbb P(D=1))\right]=
\mathbb E\left[(Y_{10}-Y_{00})(D-\mathbb E[D])\right]\\
&=&Cov(Y_{10}-Y_{00},D).
\end{eqnarray*}
The $\Longrightarrow$ ($\Longleftarrow$) direction follows from noting that it would imply the left-hand (right-hand) side of the equality is zero.
\end{proof}

\subsection{Dependence stability vs parallel trends in Example \ref{ex:1}}\label{sec:example}
Consider the DGP in Example \ref{ex:1}. We have $Q^{\mathbb R,-}_{Y_{0}}(u)=\Phi^{-1}(u)\sigma_0$, and $Q^{\mathbb R,-}_{Y_{10}}(u)=\Phi^{-1}(u)\sigma_1$, where $\Phi^{-1}(u)$ denotes the quantile of the standard normal distribution. We also have:
\begin{eqnarray*}
F_{Y_0,D}(y,0)\equiv\mathbb P(Y_0 \leq y, D\leq 0) &=& \Phi_{\Sigma_{U_0 \eta}}\left(\frac{y}{\sigma_0},0;\rho_0\right),\\
F_{Y_{10},D}(y,0)\equiv\mathbb P(Y_{10} \leq y, D\leq 0) &=& \Phi_{\Sigma_{U_1 \eta}}\left(\frac{y}{\sigma_1},0;\rho_1\right),
\end{eqnarray*}
where $\Phi_{\Sigma}(.,.; \rho)$ is the joint cdf of a bivariate normal random variable with variance-covariance matrix $\Sigma$ and coefficient of correlation $\rho$. 

From \citet[Corollary 2.3.7]{Nelsen2006}, we have for $u\in[0,1]$,
\begin{eqnarray*}
C_{Y_0,D}(u,q)=F_{Y_0,D}(Q^{\mathbb R,-}_{Y_{0}}(u),Q^{\mathbb R,-}_{D}(q))= \Phi_2\left(\Phi^{-1}(u),0;\rho_0\right),\\
C_{Y_{10},D}(u,q)=F_{Y_{10},D}(Q^{\mathbb R,-}_{Y_{10}}(u),Q^{\mathbb R,-}_{D}(q))= \Phi_2\left(\Phi^{-1}(u),0;\rho_1\right),
\end{eqnarray*}
where $\Phi_2(.,.; \rho)$ is the joint cdf of a standard bivariate normal random variable with parameter $\rho$.
Since the function $\Phi_2(.,.; \rho)$ is strictly increasing in $\rho$,\footnote{See \cite{Sibuya1959} and \cite{Sungur1990}.} we conclude that $C_{Y_0,D}(u,q)=C_{Y_{10},D}(u,q)$ if and only if $\rho_0=\rho_1$.

In Example \ref{ex:1}, parallel trends in distribution implies $\sigma_1=\sigma_0$ and $\rho_1=\rho_0$, i.e., $U_0$ and $U_1$ have the same distribution $N(0,\sigma^2_1)$, and copula stability (Assumption \ref{stab}) holds. Indeed, parallel trends in distribution states:
\begin{eqnarray*}
F_{Y_{10}\vert D=1}(y)-F_{Y_{0}\vert D=1}(y)&=& F_{Y_{10}\vert D=0}(y)-F_{Y_{0}\vert D=0}(y),
\end{eqnarray*}
which implies
\begin{eqnarray*}
\frac{F_{Y_{10},D}(y,1)-F_{Y_{0},D}(y,1)}{\mathbb P(D=1)}&=& \frac{F_{Y_{10},D}(y,0)-F_{Y_{0},D}(y,0)}{\mathbb P(D=0)},\\
\frac{F_{Y_{10},D}(y,1)-F_{Y_{0},D}(y,1)}{0.5}&=& \frac{F_{Y_{10},D}(y,0)-F_{Y_{0},D}(y,0)}{0.5},\\
F_{Y_{10}}(y)-F_{Y_{10},D}(y,0)-F_{Y_{0}}(y)+F_{Y_{0},D}(y,0)&=& F_{Y_{10},D}(y,0)-F_{Y_{0},D}(y,0),\end{eqnarray*}
\begin{eqnarray*}
F_{Y_{10}}(y)-F_{Y_{0}}(y)&=& 2(F_{Y_{10},D}(y,0)-F_{Y_{0},D}(y,0)),
\end{eqnarray*}
that is,
$\Phi(\frac{y}{\sigma_1})-\Phi(\frac{y}{\sigma_0})= 2(\Phi_{\Sigma_{U_1 \eta}}(\frac{y}{\sigma_1},0;\rho_1)-\Phi_{\Sigma_{U_0 \eta}}(\frac{y}{\sigma_0},0;\rho_0))$ for all $y$. 
For $y=0$, this equality implies $\Phi_{\Sigma_{U_1 \eta}}(0,0;\rho_1)-\Phi_{\Sigma_{U_0 \eta}}(0,0;\rho_0)=0$, that is, $\frac{1}{4}+\frac{\arcsin(\rho_1)}{2\pi}=\frac{1}{4}+\frac{\arcsin(\rho_0)}{2\pi}$, which implies $\rho_1=\rho_0$ because the function $\arcsin$ is continuous and strictly increasing.

Parallel trends in distribution implies the standard parallel trends, which according to Lemma \ref{lem:covstability} is equivalent to covariance stability $Cov(Y_{00},D)=Cov(Y_{10,D})$, that is, $\rho_0 \sigma_0=\rho_0 \sigma_1$. Since $\rho_1=\rho_0$, we have $\sigma_0=\sigma_1$ because $\rho_t\neq 0$ by assumption. 

\subsection{A variant of Example \ref{ex:1} with non-normal marginals}\label{Appen:example1_nonnormal}
In this section, we present a variant on Example \ref{ex:1} with exponential, instead of Gaussian, marginals. We make two observations on the following example: (i) the parallel trends assumption no longer has a simple interpretation as in Example \ref{ex:1}, (ii) the copula stability restriction is identical to Example \ref{ex:1} despite the difference in the marginal distribution.

\begin{example}\label{example1:1}
Consider the following data generating process (DGP) in which the treatment is received when its gain (treatment effect) is bigger than or equal to a threshold, say 0 for simplicity. This is a simple Roy model where selection into treatment is on the gain. 
\begin{eqnarray}\label{eq:example1}
\left\{ \begin{array}{lcl}
     Y_0 &=& U_0\\ \\
     Y_{1} &=& \eta D+ U_1  \\ \\
     D &=& \mathbbm{1}\{\eta\geq 0\}
     \end{array} \right.
\end{eqnarray}
where $U_t \sim \exp{(\theta_t)}$, $C_{U_t,\eta}(u,v)=\Phi_2\left(\Phi^{-1}(u),\Phi^{-1}(v);\rho_t\right)$, $\rho_t \neq 0$. 

In this case, we have the following:
\begin{enumerate}
\item [(a)] {\bf Copula  stability:} $\rho_0=\rho_1$ 

$C_{U_0,\eta}=C_{U_1,\eta} \Leftrightarrow \rho_0=\rho_1$
since $\Phi_2\left(.,.;\rho\right)$ is strictly increasing in $\rho$. 
\smallskip
\item [(b)] {\bf Parallel trends:} $\int( C_{U_0,D}(1-e^{-\theta_0 u},q)-(1-e^{-\theta_0 u})q )du=\int( C_{U_1,D}(1-e^{-\theta_1 u},q)-(1-e^{-\theta_1 u})q) du$ $$\Leftrightarrow$$ 
$\int (\Phi_2\left(\Phi^{-1}(1-e^{-\theta_0 u}),\Phi^{-1}(q);\rho_0\right)-(1-e^{-\theta_0 u})q) du=\int( \Phi_2\left(\Phi^{-1}(1-e^{-\theta_1 u}),\Phi^{-1}(q);\rho_1\right)-(1-e^{-\theta_1 u})q )du$
where
\begin{eqnarray*}
   C_{U_t,D}(1-e^{-\theta_t u},q)&=&\mathbb P(U_t \leq u, D=0),\\
   &=& \mathbb P(U_t \leq u, \eta \leq 0),\\
   &=&C_{U_t,\eta}(F_{U_t}(u),F_{\eta}(0)),\\
   &=& \Phi_2\left(\Phi^{-1}(1-e^{-\theta_t u}),\Phi^{-1}(q);\rho_t\right).
\end{eqnarray*}
\item [(c)] {\bf Distributional DiD:} $\rho_0=\rho_1$ and $\theta_0=\theta_1$ 

Since $\rho_t \neq 0,$ $D \not\independent U_t$. Therefore, from \cite{RothSantanna2021}, distributional PT holds iff stationarity holds, i.e., $\mathbb P(U_0\leq u \vert D=d)=\mathbb P(U_1\leq u \vert D=d)$ for all $u$ and $d$, which implies $\mathbb P(U_0\leq u,D=d)=\mathbb P(U_1\leq u,D=d)$ for all $u$ and $d$, which in turn implies $\mathbb P(U_0\leq u)=\mathbb P(U_1\leq u)$, which finally implies $\theta_0=\theta_1$. Now, using the equality $\mathbb P(U_0\leq u,D=0)=\mathbb P(U_1\leq u,D=0)$, we have 
$\Phi_2\left(\Phi^{-1}(1-e^{-\theta_0 u}),\Phi^{-1}(q);\rho_0\right)=\Phi_2\left(\Phi^{-1}(1-e^{-\theta_0 u}),\Phi^{-1}(q);\rho_1\right)$, which implies $\rho_0=\rho_1$ since $\Phi_2\left(.,.;\rho\right)$ is strictly increasing in $\rho$. 
\end{enumerate}

\end{example}

\subsection{Proof of Theorem \ref{GMain:theorem}}\label{app:proof_theorem2}

Take a fixed $y \in \mathbb Y_{10|0}$, then  for any $t \in \{-T_0,\dots,0\}$, the following holds for all  $\tilde y< Q^{\mathbb R,+}_{Y_t|D=0}\left(F_{Y_{1|D=0}}(y)\right)$: 
\begin{eqnarray}
F_{Y_t|D=0}\left(\tilde y\right) &\leq& F_{Y_{1|D=0}}(y) \leq F_{Y_t|D=0}\left(Q^{\mathbb R,-}_{Y_t|D=0}\left(F_{Y_{1|D=0}}(y)\right)\right), \nonumber \\
F_{Y_t,D}\left(\tilde y,0\right) &\leq& F_{Y_{1,D}}(y,0) \leq F_{Y_t,D}\left(Q^{\mathbb R,-}_{Y_t|D=0}\left(F_{Y_{1|D=0}}(y)\right),0\right), \nonumber\\
C_{Y_t,D}\left(F_{Y_t}\left(\tilde y\right), q\right)&\leq& C_{Y_{10},D}\left(F_{Y_{10}}(y), q\right) \leq C_{Y_t,D}\left(F_{Y_t}\left(Q^{\mathbb R,-}_{Y_t|D=0}\left(F_{Y_{1|D=0}}(y)\right)\right), q\right), \nonumber\\
C_{Y_t,D}\left(F_{Y_t}\left(\tilde y\right), q\right)&\leq& C_{Y_{t},D}\left(F_{Y_{10}}(y), q\right) \leq C_{Y_t,D}\left(F_{Y_t}\left(Q^{\mathbb R,-}_{Y_t|D=0}\left(F_{Y_{1|D=0}}(y)\right)\right), q\right), \nonumber\\
 F_{Y_t}\left(\tilde y\right)&\leq& F_{Y_{10}}(y) \leq F_{Y_t}\left(Q^{\mathbb R,-}_{Y_t|D=0}\left(F_{Y_{1|D=0}}(y)\right)\right) \label{eq;it}
\end{eqnarray}
The first line of the inequality trivially holds
from Lemma \ref{lem:boundq}(\ref{lem:boundq1}) and the fact that $Y_0 \leq \tilde{y}$ implies $Y_0 < Q^{\mathbb R,+}_{Y_0|D=0}\left(F_{Y_{1|D=0}}(y)\right)$.
The third line holds  by Sklar's Theorem \citep[Theorem 2.3.3.]{Nelsen2006}. The fourth line holds under Assumption \ref{Gstab}, and the last line holds under Assumption \ref{inc}. 
Notice that the last line requires $u \mapsto C_{Y_{10},D}(u,q)$ to be strictly increasing only on $\overline{\operatorname{Ran}}  F_{Y_{10}}\cup \overline{\operatorname{Ran}}  F_{Y_{t0}} \subseteq [0,1]$.
Now, applying the monotonicity of the function $v-C_{Y_t,D}(v,q)$ on the inequality (\ref{eq;it}), for all  $\tilde y < Q^{\mathbb R,+}_{Y_t|D=0}\left(F_{Y_{1|D=0}}(y)\right)$ we have:
 \begin{eqnarray*}
&&F_{Y_t}\left(\tilde y\right)-C_{Y_t,D}\left(F_{Y_t}\left(\tilde y\right), q\right) \leq  F_{Y_{10}}(y)-C_{Y_t,D}\left(F_{Y_{10}}(y), q\right) \leq \label{ineq:it}
\\ &&F_{Y_t}\left(Q^{\mathbb R,-}_{Y_t|D=0}\left(F_{Y_{1|D=0}}(y)\right)\right)-C_{Y_t,D}\left(F_{Y_t}\left(Q^{\mathbb R,-}_{Y_t|D=0}\left(F_{Y_{1|D=0}}(y)\right)\right), q\right). \nonumber
\end{eqnarray*}
In addition, since $F_{Y_{t0}}(y)=F_{Y_{t0}, D}(y,1)+F_{Y_{t0}, D}(y,0)=F_{Y_{t0}, D}(y,1) + C_{Y_{t0},D}(F_{Y_{t0}}(y),q)$ for $t=-T_0,\cdots,1$, the latter equality implies the following:
\begin{eqnarray*}
F_{Y_t, D}\left(\tilde y, 1\right) &\leq&  F_{Y_{10}}(y)-C_{Y_t,D}\left(F_{Y_{10}}(y), q\right) \leq F_{Y_t, D}\left(Q^{\mathbb R,-}_{Y_t|D=0}\left(F_{Y_{1|D=0}}(y)\right), 1\right)\\
F_{Y_t, D}\left(\tilde y, 1\right) &\leq&  F_{Y_{10}}(y)-C_{Y_{10},D}\left(F_{Y_{10}}(y), q\right) \leq F_{Y_t, D}\left(Q^{\mathbb R,-}_{Y_t|D=0}\left(F_{Y_{1|D=0}}(y)\right), 1\right)\\
F_{Y_t, D}\left(\tilde y, 1\right)&\leq&  F_{Y_{10},D}(y, 1) \leq F_{Y_t, D}\left(Q^{\mathbb R,-}_{Y_t|D=0}\left(F_{Y_{1|D=0}}(y)\right), 1\right),\\
F_{Y_t, D}\left(\tilde y, 1\right) &\leq&  F_{Y_{10},D}(y, 1) \leq F_{Y_t, D}\left(Q^{\mathbb R,-}_{Y_t|D=0}\left(F_{Y_{1|D=0}}(y)\right), 1\right),\\
 F_{Y_t|D=1}\left(\tilde y\right) &\leq&  F_{Y_{10}|D=1}(y) \leq F_{Y_t|D=1}\left(Q^{\mathbb R,-}_{Y_t|D=0}\left(F_{Y_{1|D=0}}(y)\right)\right),
\end{eqnarray*}
where the second line holds under Assumption \ref{Gstab}.
So, to summarize, for any fixed $y \in \mathbb Y_{10|0}$, and for any $t \in \{-T_0,\dots,0\}$
we have:
$$F_{Y_t|D=1}\left(\tilde y\right) \leq  F_{Y_{10}|D=1}(y) \leq F_{Y_t|D=1}\left(Q^{\mathbb R,-}_{Y_t|D=0}\left(F_{Y_{1|D=0}}(y)\right)\right), \text{ for all } \tilde y < Q^{\mathbb R,+}_{Y_t|D=0}\left(F_{Y_{1|D=0}}(y)\right).$$
Taking the supremum over $\tilde{y}<Q_{Y_t|D=0}^{\mathbb{R},+}(F_{Y_1|D=0}(y))$ implies that:
$$ \sup_{\tilde y < Q^{\mathbb R,+}_{Y_t|D=0}\left(F_{Y_{1|D=0}}(y)\right)}F_{Y_t|D=1}\left(\tilde y\right) \leq  F_{Y_{10}|D=1}(y) \leq F_{Y_t|D=1}\left(Q^{\mathbb R,-}_{Y_t|D=0}\left(F_{Y_{1|D=0}}(y)\right)\right),$$
which is equivalent to:
$$\underbrace{F_{Y_t|D=1}\left(Q^{\mathbb R,+}_{Y_t|D=0}\left(F_{Y_{1|D=0}}(y)\right)-\right)}_{=F_{Y_t|D=1}\left(\left[Q^{\mathbb R,+}_{Y_t|D=0}\circ F_{Y_{1|D=0}}\right](y)-\right)\equiv F_t^{LB}(y) } \leq  F_{Y_{10}|D=1}(y) \leq \underbrace{F_{Y_t|D=1}\left(Q^{\mathbb R,-}_{Y_t|D=0}\left(F_{Y_{1|D=0}}(y)\right)\right)}_{=\left[F_{Y_t|D=1}\circ Q^{\mathbb R,-}_{Y_t|D=0}\circ F_{Y_1|D=0}\right] (y)\equiv F_t^{UB}(y) }.$$

Thus we have: for any $t \in \{-T_0,\dots,0\}$
\begin{eqnarray}\label{B1mp}
F_t^{LB}(y) \leq  F_{Y_{10}|D=1}(y) \leq F_t^{UB}(y), \text{ for all } y \in \mathbb Y_{10|0}.
\end{eqnarray}
Then we finally have the following bounds:

\begin{eqnarray}\label{B2mp}
\max_{t \in \{-T_0,\dots,0\}}F_t^{LB}(y) \leq  F_{Y_{10}|D=1}(y) \leq \min_{t \in \{-T_0,\dots,0\}}F_t^{UB}(y), \text{ for all } y \in \mathbb Y_{10|0}.
\end{eqnarray}

Notice that the above bounds naturally extend to the case where  $y  \in \mathbb R \setminus \mathbb Y_{10|0}$, however for $y  \in \mathbb R \setminus \mathbb Y_{10|0}$ the bounds may no longer be (point-wise) sharp. And this is because  the upper bound may not be right-continuous in some cases, similarly for the  lower bound which may not be  right-continuous whenever $\{\tilde y \in \mathbb Y_{t|D=1} \cup \{-\infty\}:F_{Y_t|D=1}(\tilde y)\leq u\}$ is open for some $u\in Ran F_{Y_t|D=1}$.

To clarify this point, let us consider the simple case where $Y_{t0}$, for all $t$, are all discrete random variables with $\mathbb Y_{10|0}=\{y_0,...,y_K\}$. In this case, $F_t^{LB}(.)$ is a well-defined cdf, while $F_t^{UB}(.)$ may not be a right-continuous function. Indeed, the function $u \mapsto Q^{\mathbb Y_{t|0},-}(u)$ is left-continuous and the discontinuities happen at $u\in Ran F_{Y_t|D=0}$.
Now, consider that there exists $u_k \in Ran F_{Y_t|D=0} \cap Ran F_{Y_{10}|D=0}$, thus $F^{UB}(.)$ could be left-continuous at
$y_k \in \mathbb Y_{10|0}$ such that $F_{Y_{10}|D=0}(y_k)=u_k$.
If it is left-continuous and not right-continuous in $y_k$, we have:
 $\{y \in \overline{\mathbb R}: F_t^{UB}(y)> F_t^{UB}(y_k)\}=(y_k,\infty]$. Let us consider $\epsilon>0$ such that $y_k +\epsilon < y_{k+1}$. In such a case, $F_{Y_{10}|D=1}(y_k+\epsilon)=F_{Y_{10}|D=1}(y_k)$, however, by applying naively the bounds to $y_k$ and $y_{k}+\epsilon$ we have:
\begin{eqnarray}
F_t^{LB}(y_k) &\leq&  F_{Y_{10}|D=1}(y_k) \leq F_t^{UB}(y_k), \text{ where } y_k \in \mathbb Y_{10|0}\label{eq1mp}\\
F_t^{LB}(y_k+\epsilon) &\leq&  F_{Y_{10}|D=1}(y_k+\epsilon) \leq F_t^{UB}(y_k+\epsilon), \text{ where } y_k +\epsilon \notin \mathbb Y_{10|0}\label{eq2mp}
\end{eqnarray}
which implies that the upper bound in (\ref{eq2mp}) is not sharp since $F_t^{UB}(y_k+\epsilon)> F_t^{UB}(y_k)$.
A valid tighter bound for $F_t^{LB}(y')$ for $y_{k}<y'<y_{k+1}$ is:
\begin{eqnarray*}
F_t^{LB}(y_k) &\leq&  F_{Y_{10}|D=1}(y') \leq F_t^{UB}(y_k),\;\; y_{k}\leq y'<y_{k+1}.
\end{eqnarray*}

Since extending the bounds in Eq. (\ref{B2mp}) to the case where $y \notin \mathbb Y_{10|0}$ provides non-sharp bounds,  we provide an alternative approach that internalizes the idea that our targeting function of interest must be right-continuous since it is a cdf.
Recall,
\begin{eqnarray}
\max_{t \in \{-T_0,\dots,0\}}F_t^{LB}(s) \leq  F_{Y_{10}|D=1}(s) \leq \min_{t \in \{-T_0,\dots,0\}}F_t^{UB}(s), \text{ for all } s \in \mathbb Y_{10|0}.
\end{eqnarray}
then for any fixed  $y \in \mathbb R$, we have:
\begin{eqnarray*}
&&\lim_{\tilde y \downarrow y}\sup\left\{\max_{t \in \{-T_0,\dots,0\}}F_t^{LB}(s): s\leq \tilde y \; \& \;  s \in \mathbb Y_{10|0} \cup\{-\infty \} \right\}\\&\leq& \lim_{\tilde y \downarrow y}\sup\left\{F_{Y_{10}|D=1}(s): s\leq \tilde y \; \& \;  s \in \mathbb Y_{10|0} \cup\{-\infty \} \right\} \leq \\ &&\lim_{\tilde y \downarrow y}\sup\left\{ \min_{t \in \{-T_0,\dots,0\}}F_t^{UB}(s): s\leq \tilde y \; \& \;  s \in \mathbb Y_{10|0} \cup\{-\infty \} \right\}, \;\; y \in \mathbb R.  
\end{eqnarray*}
Notice that because $\mathbb Y_{10|1} \subseteq \mathbb Y_{10|0}$, and $F_{Y_{10}|D=1}(\cdot)$ is a right-continuous function, we have the following equality by Lemma \ref{lem:boundq}(\ref{lem:boundq2}):
$$\lim_{\tilde y \downarrow y}\sup\left\{F_{Y_{10}|D=1}(s): s\leq \tilde y \; \& \;  s \in \mathbb Y_{10|0} \cup\{-\infty \} \right\}= F_{Y_{10}|D=1}(y) \text{ for all } y \in \mathbb R;$$ therefore the last inequality becomes:
\begin{multline*}
\lim_{\tilde y \downarrow y}\sup\left\{ \max_{t \in \{-T_0,\dots,0\}}F_t^{LB}(s): s\leq \tilde y \; \& \;  s \in \mathbb Y_{10|0} \cup\{-\infty \} \right\} \\ \leq  F_{Y_{10}|D=1}(y) \leq \lim_{\tilde y \downarrow y}\sup\left\{ \min_{t \in \{-T_0,\dots,0\}}F_t^{UB}(s): s\leq \tilde y \; \& \;  s \in \mathbb Y_{10|0} \cup\{-\infty \} \right\}, \;\; y \in \mathbb R.  
\end{multline*}
\qed
\subsection{Proof of Example \ref{ex:multidimensional}}\label{proof:ex-multi-dimensional}

We have: $C_{U_0,\tilde{U}_0,V}(u,\tilde{u},q)=C_{U_1,\tilde{U}_1,V}(u,\tilde{u},q)$ for all $(u,\tilde{u},q)\in[0,1]^3$ implies successively 
\begin{eqnarray*}
   && C_{U_0,\tilde{U}_0,V}(u,1,q)=C_{U_1,\tilde{U}_1,V}(u,1,q),\\
   && C_0(C_{U_0,\tilde{U}_0}(u,1),q)=C_1(C_{U_1,\tilde{U}_1}(u,1),q),\\
   && C_0(u,q) = C_1(u,q),\\
   && C_{Y_{00},D}(u,q) = C_{Y_{10},D}(u,q).
\end{eqnarray*}
We need to check that the Sklar theorem holds on the range in this model. We have
\begin{eqnarray*}
    \mathbb P(Y_{t0}=0,D=0) &=& \mathbb P(U_t \leq c_t, \tilde{U}_t \leq \tilde{c}_t, V\leq q),\\
    &=& C_{U_t,\tilde{U}_t,V}(c_t,\tilde{c}_t,q),\\
    &=& C_t(C_{U_t,\tilde{U}_t}(c_t,\tilde{c}_t),q),\\
    &=& C_{Y_{t0},D}(C_{U_t,\tilde{U}_t}(c_t,\tilde{c}_t),q),\\
    &=& C_{Y_{t0},D}(\mathbb P(Y_{t0}=0),q)\ \text{as }\ \mathbb P(Y_{t0}=0)=C_{U_t,\tilde{U}_t}(c_t,\tilde{c}_t),\\
    &=& C_{Y_{t0},D}(\mathbb P(Y_{t0}=0),\mathbb P(D=0))\ \text{as }\ \mathbb P(D=0)=q.    
\end{eqnarray*}
\qed
\subsection{Auxiliary lemma}
\begin{lemma}\label{lem:-X} For a random variable $X$ with cdf denoted by $F_X(x)$ for $x\in\mathbb{R}$,
    \begin{enumerate}[(i)]\item $F_X(x-)=1-F_{-X}(-x)$ for $x\in\mathbb{R}$.\label{FX-}
    \item $Q_{X}^{\mathbb{R},+}(q)=-Q_{-X}^{\mathbb{R},-}(1-q)$ for $q\in[0,1]$.\label{quantile}
\end{enumerate}
    \end{lemma}
\begin{proof} 
(i) is straightforward from the following.
\begin{eqnarray}
    F_X(x-)&=&\mathbb{P}(X<x)=1-\mathbb{P}(X\geq x)=1-\mathbb{P}(-X\leq -x)=1-F_{-X}(-x).\label{eq:FXleft}\end{eqnarray}
To show (ii), the following equality is convenient
    \begin{eqnarray}
    F_X(x)&=&1-\mathbb{P}(X>x)=1-\mathbb{P}(-X< -x)\label{eq:Fx}
\end{eqnarray}
Using the above equality, we can write $Q_{X}^{\mathbb{R},+}(q)$ as follows. First note that:
$$\{x\in\mathbb{R}:F_X(x)\leq q\}=\{x\in\mathbb{R}: 1-\mathbb{P}(-X< -x)\leq q\}=\{x\in\mathbb{R}:\mathbb{P}(-X<-x)\geq 1-q\}$$
By the above and the left-continuity of $\mathbb{P}(-X<-x)$, it follows that 
\begin{eqnarray}
    Q_{X}^{\mathbb{R},+}(q)&=&\sup\{x\in\mathbb{R}:F_X(x)\leq q\}\nonumber=\sup\{x\in\mathbb{R}:\mathbb{P}(-X<-x)\geq 1-q\}\nonumber\\
    &=&-\inf\{-x\in\mathbb{R}:\mathbb{P}(-X<-x)\geq 1-q\}=-\inf\{-x\in\mathbb{R}:\mathbb{P}(-X\leq-x)\geq 1-q\}\nonumber\\
    &=&-Q_{-X}^{\mathbb{R},-}(1-q)
\end{eqnarray}
where the first equality follows by definition. The second equality follows from \eqref{eq:FXleft}. The penultimate equality follows by the left-continuity of $\mathbb{P}(-X<-x)$ and $\mathbb{P}(-X\leq -x)$ being its right-continuous counterpart.

\end{proof} 

\section{Supplementary results for Section \ref{sec:swtt}}\label{sec:swtt_derivations}
Here, we provide the distributions of $X^u$ and $X^{\underline{u},\overline{u}}$ which are used to define the quantile-specific social welfare functions in Section \ref{sec:swtt}.

Let $X^u=Q_X^{\mathbb{R},-}(V)$, where $V\sim \mathcal{U}[0,u]$. Note that by definition, $F_{X^u}(x)=1$ for $x\geq Q_X^{\mathbb{R},-}(u)$. As for $x<Q_X^{\mathbb{R},-}(u)$, by Proposition 1(5) in \citet{Embrechts_al2013}, it follows that
\begin{align}F_{X^u}(x)=\mathbb P(Q_X^{\mathbb{R},-}(V)\leq x)=\mathbb P\left(V\leq F_X(x)\right)=\frac{F_X(x)}{u}\end{align}
As a result,
\begin{eqnarray}F_{X^u}(x)&=\left\{\begin{array}{cc}\frac{F_X(x)}{u}& \text{for }x<Q_X^{\mathbb{R},-}(u),\\1&\text{for }x\geq Q_X^{\mathbb{R},-}(u).\end{array}\right.\end{eqnarray}
For $u\in Ran F_X$, $F_{X^u}(x)=\frac{F_X(x)}{u}=\frac{F_X(x)}{F_X(Q_X^{\mathbb{R},-}(u))}=\mathbb P(X\leq x|X\leq Q_X^{\mathbb{R},-}(u))$ for any $x\leq Q_X^{\mathbb{R},-}(u)$, thereby yielding the same truncated random variable introduced in \citeauthor{Aabergeetal2013}(\citeyear{Aabergeetal2013}). For $u\notin Ran F_X$, $X^u$ remains a well-defined random variable.

Now consider $X^{\underline{u},\overline{u}}=Q_X^{\mathbb{R},-}(V)$, where $V\sim \mathcal{U}[\underline{u},\overline{u}]$. By similar arguments to the case of $X^u$, it follows that
\begin{eqnarray}F_{X^{\underline{u},\overline{u}}}(x)&=\left\{\begin{array}{cl}0&\text{for }x<Q_X^{\mathbb{R},-}(\underline{u}),\\
\frac{F_X(x)-\underline{u}}{\overline{u}-\underline{u}}& \text{for }Q_X^{\mathbb{R},-}(\underline{u})\leq x <Q_X^{\mathbb{R},-}(\overline{u}),\\
1&\text{for }x\geq Q_X^{\mathbb{R},-}(\overline{u}).\end{array}\right.\end{eqnarray}

\section{Equivalence between copula stability and CiC assumptions for continuous outcomes: General result}\label{Appen:equivalence_general}

In this section, we generalize the equivalence result in Claim \ref{claim:eq} to any continuous outcome. To do so, we rely on two lemmas. The first lemma characterizes the implication of the copula stability of $(U_t,D)$ for the copula $(Y_{t0},D)$, and vice versa, under a representation condition, specifically $(Y_{t0},D)\overset{d}{=}(Q_{Y_{t0}}^{\mathbb R,-}(U_t),D)$ for $U_t\sim\mathcal{U}[0,1]$. This lemma is of independent interest, as it demonstrates why copula stability and the CiC conditions are not equivalent outside of the continuous outcome case.
\begin{lemma}\label{lemma:representation_cs}  For $t=0,1$, consider $(Y_{t0},D)$ such that $Y_{t0}\sim F_{Y_{t0}}$ and $D$ is a binary variable with $\mathbb P(D=0)=q\in(0,1)$.  Suppose that there exist $U_t\sim \mathcal{U}[0,1]$ for $t=0,1$ such that $(Y_{t0},D)\overset{d}{=}(Q_{Y_{t0}}^{\mathbb R,-}(U_t),D)$ for $t=0,1$. 
\begin{enumerate}[(i)]
\item If $C_{Y_{00},D}(u,q)=C_{Y_{10},D}(u,q)$ for all $u\in[0,1]$, then for
 $v\in \overline{\operatorname{Ran}} (F_{Y_{00}})\cap \overline{\operatorname{Ran}} (F_{Y_{10}})$
\begin{align}C_{U_0,D}(v,q)=C_{U_1,D}(v,q).\end{align}
\label{claim:general1}
\item If $C_{U_0,D}(v,q)=C_{U_1,D}(v,q)$ for all $v\in[0,1]$, then for $u\in \overline{\operatorname{Ran}} (F_{Y_{00}})\cap \overline{\operatorname{Ran}} (F_{Y_{10}})$,
\begin{align}C_{Y_{00},D}(u,q)=C_{Y_{10},D}(u,q).\end{align} \label{claim:general2}
\end{enumerate}
\end{lemma}
\begin{proof}
(i) For $y\in\mathbb{R}$
\begin{align}&C_{Y_{t0},D}(F_{Y_t}(y),q)=F_{Y_{t0},D}(y,0)=\mathbb P(Y_{t0}\leq y, D=0)=\mathbb P(Q_{Y_{t0}}^{\mathbb R,-}(U_t)\leq y,D=0)\nonumber\\
=&\mathbb P(U_t\leq F_{Y_{t0}}(y),D=0)=C_{U_t,D}(F_{Y_{t0}}(y),q),
\label{eq:cop_invariance}
\end{align}
where the first two equalities follow by definition. The third equality follows by the assumption that $(Y_t,D)\overset{d}{=}(Q_{Y_{t0}}^{\mathbb{R},-}(U_t),D)$. The penultimate equality holds by Proposition 1(5) in \citet{Embrechts_al2013} and the right-continuity of $F_{Y_{t0}}$, which ensure that $Q_{Y_{t0}}^-(u)\leq y \Leftrightarrow u\leq F_{Y_{t0}}(y)$. As a result, for $t=0,1$, $C_{Y_{t0},D}(v,q)=C_{U_t,D}(v,q)$ for $v\in \overline{\operatorname{Ran}} (F_{Y_{t0}})$.

As a result, the dependence stability condition in Lemma \ref{lemma:representation_cs}(i), $C_{Y_{00},D}(u,q)=C_{Y_{10},D}(u,q)$ for all $u\in[0,1]$, implies the following for $v\in \overline{\operatorname{Ran}} (F_{Y_{00}})\cap \overline{\operatorname{Ran}} (F_{Y_{10}})$
\begin{align}C_{U_0,D}(v,q)=C_{Y_{00},D}(v,q)=C_{Y_{10},D}(v,q)=C_{U_1,q}(v,q),\end{align}
where the first and last equalities follow from \eqref{eq:cop_invariance}, whereas the second follows by the dependence stability assumption on $C_{Y_{t0},D}$ imposed in Lemma \ref{lemma:representation_cs}(i).

(ii) For $y\in\mathbb{R}$,
\begin{align}&F_{U_t,D}(F_{Y_{t0}}(y),0)=C_{U_t,D}(F_{Y_{t0}}(y),q)=\mathbb P(U_t\leq F_{Y_{t0}}(y),D=0)=\mathbb P(Q_{Y_{t0}}^{\mathbb R,-}(U_t)\leq y, D=0)\nonumber\\
=&\mathbb P(Y_{t0}\leq y, D=0)=C_{Y_{t0},D}(F_{Y_{t0}}(y),q),\label{eq:cop_invariance_U}\end{align}
where the first two equalities follow by definition, the third follows from  Proposition 1(5) in \citet{Embrechts_al2013} since $F_{Y_{t0}}$ is increasing and right-continuous.  The last two equalities follow by definition.

As a result, the dependence stablity condition imposed in Lemma \ref{lemma:representation_cs}(ii) implies the following for $u\in \overline{\operatorname{Ran}} (F_{Y_{00}})\cap \overline{\operatorname{Ran}} (F_{Y_{10}})$,
\begin{align}C_{Y_{00},D}(u,q)=C_{U_0,D}(u,q)=C_{U_1,D}(u,q)=C_{Y_{10},D}(u,q),\end{align}
where the first and last equalities follow from \eqref{eq:cop_invariance_U}, whereas the second follows by the dependence stability assumption on $C_{U_t,D}$ imposed in Lemma \ref{lemma:representation_cs}(ii).
\end{proof}
The following lemma is well-established in the literature. We provide a proof for completeness, as we cannot find a reference for it.
\begin{lemma}\label{lem:as_representation}
    $Y_{t0}=Q_{Y_{t0}}^{\mathbb R,-}(F_{Y_{t0}}(Y_{t0}))$ a.s.
\end{lemma}
\begin{proof}
From Proposition 2(2) in \citet{Embrechts_al2013}, $Y_{t0}$ has the same distribution as $\tilde{Y}_{t0}\equiv Q_{Y_{t0}}^{\mathbb R,-}(U)$ where $U \sim \mathcal{U}_{[0,1]}.$ We can write $$Q_{Y_{t0}}^{\mathbb R,-}(F_{Y_{t0}}(\tilde{Y}_{t0}))=Q_{Y_{t0}}^{\mathbb R,-}(F_{Y_{t0}}(Q_{Y_{t0}}^{\mathbb R,-}(U))).$$
From the definition of $Q_{Y_{t0}}^{\mathbb R,-}$, we have $F_{Y_{t0}}(Q_{Y_{t0}}^{\mathbb R,-}(U)) \geq U$. Therefore, since the quantile function is nondecreasing, $Q_{Y_{t0}}^{\mathbb R,-}(F_{Y_{t0}}(Q_{Y_{t0}}^{\mathbb R,-}(U))) \geq Q_{Y_{t0}}^{\mathbb R,-}(U),$ which implies $Q_{Y_{t0}}^{\mathbb R,-}(F_{Y_{t0}}(\tilde{Y}_{t0})) \geq \tilde{Y}_{t0}.$

On the other hand, since $F_{\tilde{Y}_{t0}}(y) \geq F_{\tilde{Y}_{t0}}(y)$, from the definition of the quantile function, it follows that $Q_{\tilde{Y}_{t0}}^{\mathbb R,-}(F_{\tilde{Y}_{t0}}(y)) \leq y.$ This latter inequality implies $Q_{\tilde{Y}_{t0}}^{\mathbb R,-}(F_{\tilde{Y}_{t0}}(\tilde{Y}_{t0})) \leq \tilde{Y}_{t0}.$ Finally, since $F_{\tilde{Y}_{t0}}=F_{Y_{t0}}$, we have $Q_{Y_{t0}}^{\mathbb R,-}(F_{Y_{t0}}(\tilde{Y}_{t0})) \leq \tilde{Y}_{t0}.$ As a result, we have $Q_{Y_{t0}}^{\mathbb R,-}(F_{Y_{t0}}(\tilde{Y}_{t0})) = \tilde{Y}_{t0}.$

Now, define $S=\left\{y \in \mathbb R: Q_{Y_{t0}}^{\mathbb R,-}(F_{Y_{t0}}(y)) = y \right\}$. We have 
$$\mathbb P(Y_{t0} \in S)=\mathbb P(\tilde{Y}_{t0} \in S)=\mathbb P\left(Q_{Y_{t0}}^{\mathbb R,-}(F_{Y_{t0}}(\tilde{Y}_{t0})) = \tilde{Y}_{t0}\right)=1,$$
where the first equality follows from $F_{\tilde{Y}_{t0}}=F_{Y_{t0}}$ and the second follows from the definition of $S$.

Hence, 
$$1=\mathbb P(Y_{t0} \in S)=\mathbb P\left(Q_{Y_{t0}}^{\mathbb R,-}(F_{Y_{t0}}(Y_{t0}))=Y_{t0}\right).$$

\end{proof}

Finally, we proceed to demonstrate the equivalence between conditional time invariance and copula stability for continuous outcomes.
\begin{claim}
Assume $F_{Y_{t0}}$ is a continuous outcome distribution. Then, (i) and (ii) from Claim \ref{claim:eq} are equivalent almost surely.
\end{claim}
\begin{proof}
$\Longrightarrow$
First, we note that for any potential outcome $Y_{t0}$, we have $Y_{t0}=Q_{Y_{t0}}^{\mathbb R,-}(F_{Y_{t0}}(Y_{t0}))$ almost surely by Lemma \ref{lem:as_representation}. Suppose now that the potential outcome $Y_{t0}$ is continuous. Then $U_{t0}\equiv F_{Y_{t0}}(Y_{t0}) \sim \mathcal U_{[0,1]},$ and $\overline{\operatorname{Ran}} F_{Y_{t0}}=[0,1]$, which implies $\overline{\operatorname{Ran}} F_{Y_{00}} \cap \overline{\operatorname{Ran}} F_{Y_{10}} = [0,1].$ Hence, $(Y_{t0},D)=(Q_{Y_{t0}}^{\mathbb R,-}(U_t),D)$ a.s., which implies $(Y_{t0},D)\overset{d}{=}(Q_{Y_{t0}}^{\mathbb R,-}(U_t),D)$. The conditions of Lemma \ref{lemma:representation_cs} hold.

Therefore, from Lemma \ref{lemma:representation_cs} (\ref{claim:general1}), if $C_{Y_{00},D}(u,q)=C_{Y_{10},D}(u,q)$ for all $u\in[0,1]$, then for $C_{U_{00},D}(u,q)=C_{U_{10},D}(u,q)$ for all $u\in \overline{\operatorname{Ran}} F_{Y_{00}} \cap \overline{\operatorname{Ran}} F_{Y_{10}} = [0,1]$. Note that $C_{U_{00},D}(u,q)=C_{U_{10},D}(u,q)$ is equivalent to $\mathbb P(U_{00} \leq u, V \leq q)=\mathbb P(U_{10} \leq u,  V \leq q),$ i.e., $\mathbb P(U_{00} \leq u, D=0)=\mathbb P(U_{10} \leq u,  D=0)$. Since $\mathbb P(U_{00}\leq u)=\mathbb P(U_{10}\leq u)$, we have $\mathbb P(U_{00}\leq u)-\mathbb P(U_{00} \leq u, D=0)=\mathbb P(U_{10}\leq u)-\mathbb P(U_{10} \leq u, D=0),$ i.e., $\mathbb P(U_{00} \leq u, D=1)=\mathbb P(U_{10} \leq u,  D=1)$. As a result, if $C_{Y_{00},D}(u,q)=C_{Y_{10},D}(u,q)$ for all $u\in[0,1]$, then $\mathbb P(U_{00} \leq u, D=d)=\mathbb P(U_{10} \leq u,  D=d)$ for all $u$ and $d$, i.e., $U_{00} \mid D=d \sim U_{10} \mid D=d.$

$\Longrightarrow$ From Lemma \ref{lemma:representation_cs} (\ref{claim:general2}), if $C_{U_{00},D}(u,q)=C_{U_{10},D}(u,q)$ for all $u\in[0,1]$, then for all $u\in \overline{\operatorname{Ran}} F_{Y_{00}} \cap \overline{\operatorname{Ran}} F_{Y_{10}} = [0,1],$ $C_{Y_{00},D}(u,q)=C_{Y_{10},D}(u,q)$. Since we have shown above that $C_{Y_{00},D}(u,q)=C_{Y_{10},D}(u,q)$ and $U_{00} \mid D=d \sim U_{10} \mid D=d$ are equivalent (given that $U_{00} \sim U_{10}$), we conclude that if $U_{00} \mid D=d \sim U_{10} \mid D=d$, then $C_{Y_{00},D}(u,q)=C_{Y_{10},D}(u,q)$ for all $u \in [0,1]$. 

\end{proof}
\section{CS bounds for binary outcomes}
\subsection{Bounds for binary outcomes with multiple pre-treatment periods}
Suppose that the outcome of interest is binary, i.e., $Y_t\in\{0,1\}$ for all periods $t$. We have  
\begin{eqnarray*}
    Q^{\mathbb Y,+}_{Y}(u)&=& -\infty \mathbbm{1}\{u < \mathbb P(Y=0)\}+1-\mathbbm{1}\{\mathbb P(Y=0) \leq u <1\} \; \\\; 
    Q^{\mathbb R,+}_{Y}(u)&=& \mathbbm{1}\{\mathbb P(Y=0) \leq u <1\} + \infty \mathbbm{1}\{u=1\}\; \\\; 
    Q^{\mathbb R,-}_{Y}(u)&=& 1- \mathbbm{1}\{\mathbb P(Y=0) \geq u\}=\mathbbm{1}\{\mathbb P(Y=0) < u\}.
\end{eqnarray*}
The following corollary of Theorem \ref{GMain:theorem} holds.
\begin{corollary}\label{Cor:binary}
  Suppose that $\mathbb Y_{t0|1} \subseteq \mathbb Y_{t0|0}$ for $t \in \{-T_0,\ldots,0\}$.
 If  Assumptions \ref{inc} and \ref{Gstab} hold, then  the bounds  for the unobserved counterfactual $\mathbb P(Y_{10}=0|D=1)$ are: 

\begin{eqnarray*}
\max_{t \in \{-T_0,\dots,0\}} F^{LB}_t(0) \leq \mathbb P(Y_{10}=0|D=1) \leq \min_{t \in \{-T_0,\dots,0\}} F^{UB}_t(0),
\end{eqnarray*}
where 
\begin{eqnarray*}
    F_t^{LB}(0)&=& F_{Y_t|D=1}\bigg(-\infty \mathbbm{1}\{\mathbb P(Y_1=0|D=0) < \mathbb P(Y_t=0|D=0)\}\\
    && \qquad +1-\mathbbm{1}\{\mathbb P(Y_t=0|D=0) \leq \mathbb P(Y_1=0|D=0) <1\}\bigg)\\
    &=&
    \mathbb P(Y_t=0|D=1)\mathbbm{1}\{\mathbb P(Y_t=0|D=0) \leq \mathbb P(Y_1=0|D=0) <1\}\\
    && + \mathbbm{1}\{\mathbb P(Y_1=0|D=0)=1\}\; \\\; 
    F_t^{UB}(0)&=& F_{Y_t|D=1}\left(1-\mathbbm{1}\{\mathbb P(Y_t=0|D=0) \geq \mathbb P(Y_1=0|D=0)\}\right),\\
    &=& \mathbb P(Y_t=0|D=1)\mathbbm{1}\{\mathbb P(Y_t=0|D=0) \geq \mathbb P(Y_1=0|D=0)\}\\
    && + \mathbbm{1}\{\mathbb P(Y_t=0|D=0) < \mathbb P(Y_1=0|D=0)\}.
\end{eqnarray*}
\end{corollary}

The following example demonstrates a case where (distributional) DiD would yield a negative counterfactual probability $\mathbb P(Y_{10}=1|D=1)$, which does not equal to the true counterfactual probability, whereas the multi-period CS bounds would contain $\mathbb P(Y_{10}=1|D=1)$.
\begin{example}\label{ex:multi_DGP}
    Consider the following anti-double hurdle model 
    \begin{eqnarray*}
        Y_t =1-\mathbbm{1}\{-0.5\mathbbm{1}\{t=1\} D + U_t \leq 0.6^{|t|+1},-0.4\mathbbm{1}\{t=1\} D + \tilde{U}_t \leq 0.7^{|t|+1}\},\ \ \  t\in\{-1,0,1\}.
    \end{eqnarray*}
Suppose $D=\mathbbm{1}\{V>0.5\}$, $U_t, \tilde{U}_t, V \sim \mathcal U_{[0,1]}$, and $C_{U_t,\tilde{U}_t,V}(u,\tilde{u},v)=C_t\left(C_{U_t,\tilde{U}_t}(u,\tilde{u}),v\right)$ where $C_t(u,v)=\left(u^{-1/2}+v^{-1/2}-1\right)^{-2}$ and $C_{U_t,\tilde{U}_t}(u,\title{u})=u \tilde{u}$. Define $C_{Y_{t0},D}(u,q)\equiv C_t(u,q)$.
The probability that $Y_t=0$ in the control and treatment group is given by:
\begin{eqnarray*}
    \mathbb P(Y_t=0|D=0) &=& \mathbb P(U_t \leq 0.6^{|t|+1}, \tilde{U}_t \leq 0.7^{|t|+1}| V\leq 0.5),\\
    &=& 2\left(0.42^{-(|t|+1)/2}+0.5^{-1/2}-1\right)^{-2}\\
    \mathbb P(Y_t=0|D=1) &=& \mathbb P(-0.5\mathbbm{1}\{t=1\} + U_t \leq 0.6^{|t|+1},-0.4\mathbbm{1}\{t=1\} + \tilde{U}_t \leq 0.7^{|t|+1} | V>0.5),\\
    &=&\frac{1}{\mathbb P(V>0.5)}\mathbb P(-0.5\mathbbm{1}\{t=1\} + U_t \leq 0.6^{|t|+1},-0.4\mathbbm{1}\{t=1\} + \tilde{U}_t \leq 0.7^{|t|+1}, V>0.5),\\
    &=& 2\bigg[(0.6^{|t|+1}+0.5\mathbbm{1}\{t=1\})(0.7^{|t|+1}+0.4\mathbbm{1}\{t=1\})\\
    && - \left([(0.6^{|t|+1}+0.5\mathbbm{1}\{t=1\})(0.7^{|t|+1}+0.4\mathbbm{1}\{t=1\})]^{-1/2}+0.5^{-1/2}-1\right)^{-2}\bigg],
\end{eqnarray*}
whereas the true counterfactual probability $\mathbb P(Y10=1|D=1)$ is $0.9715$, and the true ATT is $-0.4079$.
\end{example}

\noindent The (distributional) DiD estimand would yield a negative counterfactual probability $$\mathbb P(Y_{10}=0\mid D=1)=\mathbb P(Y_{1}=0\mid D=0)+\mathbb P(Y_{0}=0\mid D=1)-\mathbb P(Y_{0}=0\mid D=0)=-0.0845<0$$
and the corresponding DiD estimand $\theta_{DiD}=-0.5209$, which is not equal to the true counterfactual probability.

Our identifying assumptions (horizontal copula stability assumption + its strict monotonicity in the first argument) hold, our CS bounds are valid and yield $\mathbb P(Y_{10}=1\mid D=1)^{LB}=0.6821$, $\mathbb P(Y_{10}=1\mid D=1)^{UB}=1$, and $ATT \in [-0.4364, -0.1185]$. 

Note that since the untreated potential outcome depends on two-dimensional unobservables $(U_t, \tilde{U}_t)$, an informed researcher would not use the CiC approach.
\subsubsection{Point-identification in the binary outcome case}
We provide a condition under which we can achieve point-identification when the outcome variable is binary. 
\begin{corollary}\label{cor:pointid_binary}
Suppose that $\mathbb Y_{t0|1} \subseteq \mathbb Y_{t0|0}$ for $t \in \{-T_0,\ldots,0\}$, and Assumptions \ref{inc} and \ref{Gstab} hold. Suppose there exists $t_0 \in \{-T_0,\ldots,0\}$ such that $\mathbb P(Y_{t_0}=0|D=0)=\mathbb P(Y_1=0|D=0)$ (this can be checked). Then, $\mathbb P(Y_{10}=0|D=1)=\mathbb P(Y_{t_0}=0|D=1)$.
\end{corollary}

\begin{proof}
The condition in the corollary implies $F^{LB}_{t_0}(0)=F_{t_0}^{UB}(0)=\mathbb P(Y_{t_0}=0|D=1)$ (from the definition of $F^{LB}_{t_0}(0)$ and $F_{t_0}^{UB}(0)$). Therefore, 
$$\min_{t \in \{-T_0,\dots,0\}} F^{UB}_t(0) \leq F^{UB}_{t_0}(0)=F^{LB}_{t_0}(0) \leq \max_{t \in \{-T_0,\dots,0\}} F^{LB}_t(0).$$
From Corollary \ref{result:fals-test}, we must have under our identifying assumptions
$$\max_{t \in \{-T_0,\dots,0\}} F^{LB}_t(0) \leq \min_{t \in \{-T_0,\dots,0\}} F^{UB}_t(0).$$
Therefore, the following equality holds. 
$$\max_{t \in \{-T_0,\dots,0\}} F^{LB}_t(0) = \min_{t \in \{-T_0,\dots,0\}} F^{UB}_t(0).$$
Hence,
$$\max_{t \in \{-T_0,\dots,0\}} F^{LB}_t(0) = \min_{t \in \{-T_0,\dots,0\}} F^{UB}_t(0)=F^{UB}_{t_0}(0)=F^{LB}_{t_0}(0)=\mathbb P(Y_{t_0}=0|D=1).$$

\end{proof}

Notice that the condition in Corollary \ref{cor:pointid_binary} is satisfied in Example \ref{ex:multi_DGP} for $-T_0=-1$. Hence, point identification is achieved with two pre-treatments periods $\{-1,0\}$.

\subsection{Comparison with \cite{wooldridge2023simple}}
In the binary outcome setting, our (horizontal) copula stability assumption states: $C_{Y_{00},D}(u,q)=C_{Y_{10},D}(u,q)$ for $u \in [0,1]$, where $q\equiv \mathbb P(D=0)$. The \cite{wooldridge2023simple} parallel trends assumption states that there is a known, strictly increasing, continuously differentiable function $G(.)$ such that 
\begin{eqnarray}\label{WPT}
    G^{-1}(\mathbb E[Y_{10} \vert D=1])-G^{-1}(\mathbb E[Y_{00} \vert D=1])=G^{-1}(\mathbb E[Y_{10} \vert D=0])-G^{-1}(\mathbb E[Y_{00} \vert D=0]). 
\end{eqnarray} 
In general, the two assumptions are not nested. 
To illustrate his assumption, \cite{wooldridge2023simple} considers the following specification
\begin{eqnarray*}
    Y_{t0} &=& \mathbbm{1}\{Y_{t0}^* >0 \},\ \ t=0,1,\\
    Y_{t0}^* &=& \alpha + \beta D +\gamma\cdot t + U_t,\ a.s. \ \ t=0,1,   
\end{eqnarray*}
where $U_0$, $U_1$ are continuous and independent of $D$, and $\alpha$ and $\beta$ are fixed; $U_0$, $U_1$ are identically distributed with a known strictly increasing cdf $F$.

In this example, $\mathbb E[Y_{t0}\vert D=d]=1-F[-(\alpha+\beta d+\gamma\cdot t)]\equiv G(\alpha+\beta d+\gamma\cdot t)$. 
While the usual linear PT holds for $Y_{t0}^*$, it generally fails for $Y_{t0}$. But, the \cite{wooldridge2023simple} PT condition in Equation \eqref{WPT} holds.

Since $U_0$ and $U_1$ are continuous and independent of $D$ with strictly increasing cdfs, then $C_{Y_{00},D}(u,q)=u\cdot q=C_{Y_{10},D}(u,q)$ for some $u$. The model is therefore consistent with (horizontal) copula stability assumption. However, our assumption does not require that the copula  $C_{Y_{00},D}(u,q)$ be known. Furthermore, our assumption does not require that $U_0$ and $U_1$ have the same marginal distribution, as we allow $U_0$ and $U_1$ to follow different distributions. Finally, if $U_0$ and $U_1$ have the same marginal distribution $F$, but the researcher does not know $F$, \citeauthor{wooldridge2023simple}'s (\citeyear{wooldridge2023simple}) approach would not identify the counterfactual quantity $\mathbb E[Y_{t0}\vert D=1]$.

\section{Sufficient condition for copula stability in Example \ref{ex:firm_wages}}\label{proof:mwexample}
Suppose $Y_{10}=(1+R_f)Y_{00},$ where $R_f > 0,$ $R_f \perp (Y_{00},D)$, and $RanF_{Y_{10}}=RanF_{Y_{00}}=[0,1]$. We have 
\begin{eqnarray*}
C_{Y_{10},D}(u,q)&=& \int C_{Y_{10},D \vert R_f=r}(u,q) dF_{R_f}(r),\\
   C_{Y_{10},D \vert R_f=r}(u,q) &=& \mathbb P(Y_{10} \leq Q_{Y_{10}\vert R_f=r}^{\mathbb R,-}(u),D \leq Q_{D \vert R_f=r}^{\mathbb R,-}(q) \vert R_f=r).
\end{eqnarray*}

\begin{eqnarray*}
   \mathbb P(Y_{10} \leq Q_{Y_{10} \vert R_f=r}^{\mathbb R,-}(u),D \leq Q_{D \vert R_f=r}^{\mathbb R,-}(q) \vert R_f=r)
    &=& \mathbb P((1+r)Y_{00} \leq Q_{Y_{10} \vert R_f=r}^{\mathbb R,-}(u),D \leq Q_{D}^{\mathbb R,-}(q) \vert R_f=r),\\
    &=& \mathbb P\left(Y_{00} \leq \frac{Q_{Y_{10}\vert R_f=r}^{\mathbb R,-}(u)}{1+r},D=0 \vert R_f=r\right),
\end{eqnarray*}
where the first equality because $R_f \perp D$, and the second holds because $Q_{D}^{\mathbb R,-}(q)=0$. 

\begin{eqnarray*}
   Q_{Y_{10}\vert R=r}^{\mathbb R,-}(u) &=& \inf\left\{y \in \mathbb R: F_{Y_{10}\vert R_f=r}(y) \geq u\right\},\\
   &=& \inf\left\{y \in \mathbb R: \mathbb P(Y_{10} \leq y \vert R_f=r) \geq u\right\},\\
   &=& \inf\left\{y \in \mathbb R: \mathbb P((1+r)Y_{00} \leq y \vert R_f=r) \geq u\right\},\\
   &=& \inf\left\{y \in \mathbb R: \mathbb P(Y_{00} \leq \frac{y}{1+r} \vert R_f=r) \geq u\right\},\\
   &=& \inf\left\{y \in \mathbb R: F_{Y_{00}\vert R_f=r}\left(\frac{y}{1+r}\right) \geq u\right\},\\
   &=& \inf\left\{(1+r)y' \in \mathbb R: F_{Y_{00}\vert R_f=r}\left(\frac{(1+r)y'}{1+r}\right) \geq u\right\},\\
   &=& (1+r)\inf\left\{y' \in \mathbb R: F_{Y_{00}\vert R_f=r}\left(y'\right) \geq u\right\},\\
   &=& (1+r)Q_{Y_{00}}^{\mathbb R,-}(u), 
\end{eqnarray*}
where the last equality holds because $R_f \perp Y_{00}$.

Hence,
\begin{eqnarray*}
   \mathbb P(Y_{10} \leq Q_{Y_{10} \vert R_f=r}^{\mathbb R,-}(u),D \leq Q_{D \vert R_f=r}^{\mathbb R,-}(q) \vert R_f=r)
    &=& \mathbb P(Y_{00} \leq Q_{Y_{00}}^{\mathbb R,-}(u),D=0 \vert R_f=r),\\
    &=& \mathbb P(Y_{00} \leq Q_{Y_{00}}^{\mathbb R,-}(u),D=0),\\
    &=& C_{Y_{00},D}(u,q),
\end{eqnarray*}
where the second equality holds because $R_f \perp (Y_{00},D)$.  Therefore,
\begin{eqnarray*}
    C_{Y_{10},D}(u,q) &=& \int C_{Y_{00},D}(u,q) dF_{R_f}(r),\\
    &=& C_{Y_{00},D}(u,q) \int  dF_{R_f}(r)=C_{Y_{00},D}(u,q),
\end{eqnarray*}
where the last equality holds because $\int  dF_{R_f}(r)=1$.

\section{Supplementary numerical illustration}
\subsection{Numerical illustration of identification result}\label{app:num_illustration_identification}
To provide a graphical illustration of the identification result, it is helpful to consider a numerical example motivated by our minimum wage setting. Suppose that both treatment and control groups have a pre-existing minimum wage set at $c_{0}$ in the pre-treatment period ($t=0$). In the post-treatment period ($t=1$), the minimum wage increases for the treatment group to $c_{1}$. 

Following the conceptual framework presented in Figure \ref{fig:Ceng}, we expect to find bunching at the relevant minimum wage. Due to the presence of such a threshold policy in both periods, bunching at the relevant threshold is prevalent in all observed distributions as demonstrated by Figure \ref{fig:mwexample_observed}(a)--(d). The counterfactual distribution, $F_{Y_{10}|D=1}$, presented in Figure \ref{fig:mwexample_observed}(d) also exhibits a discontinuity at the pre-treatment policy threshold $c_0$. 

While the distributions in Figure \ref{fig:mwexample_observed}(a)-(d) satisfy copula stability (Assumption \ref{stab}), a visual inspection of $F_{Y_{t0}|D=d}$ for $t=0,1$ and $d=0,1$ demonstrates our point that this assumption is compatible with time and group heterogeneity in the distribution of the potential outcomes. What it requires, however, is the time-invariance of the horizontal copula $C_{Y_{t0},D}(\cdot,q)$ and subsequently the rank mapping between the control and treatment group's distribution, $\Gamma(\cdot)$. Figure \ref{fig:mwexample_observed}(e) and \ref{fig:mwexample_observed}(f) plot $C_{Y_{00},D}(\cdot,q)$ and $\Delta(\cdot)$ analytically, respectively. Figure \ref{fig:mwexample_observed}(g) and \ref{fig:mwexample_observed}(h) plot the mappings $F_{Y_{00}}(y)\mapsto F_{Y_{00}|D=0}(y)$ and $F_{Y_{00}|D=0}(y)\mapsto F_{Y_{00}|D=1}(y)$, which point-identify $C_{Y_{00},D}(\cdot,q)$ on $Ran F_{Y_{00}}$ and $\Gamma(\cdot)$ on $RanF_{Y_{00}|D=0}$, respectively. 

To bound the counterfactual distribution $F_{Y_{10}|D=1}$, we transport the dependence structure from the pre-treatment period to the post-treatment period. We can point-identify $F_{Y_{10}|D=1}(y)$ for $y\in\mathbb{Y}_{10|0}$ with $F_{Y_{10}|D=0}(y)\in RanF_{Y_{00}|D=0}\cap RanF_{Y_{10}|D=0}$. Outside of this intersection, we have to extend the (horizontal) copula. Since there are multiple extensions possible, we can only partially identify $F_{Y_{10}|D=1}$ as demonstrated in Figure \ref{fig:mwexample_observed}(i).

\begin{figure}[htbp]\caption{Numerical Minimum-Wage Example: Observed and counterfactual distributions}
\vspace{0.2cm}
    \begin{tabular}{cccc}\\
\scriptsize{(a) $F_{Y_{00}|D=0}$} & \scriptsize{(b) $F_{Y_{10}|D=0}$}&\scriptsize{(c) $F_{Y_{00}|D=1}$} &\scriptsize{(d) $F_{Y_{11}|D=1}$\& $F_{Y_{10}|D=1}$}\\
\includegraphics[width=3.5cm]{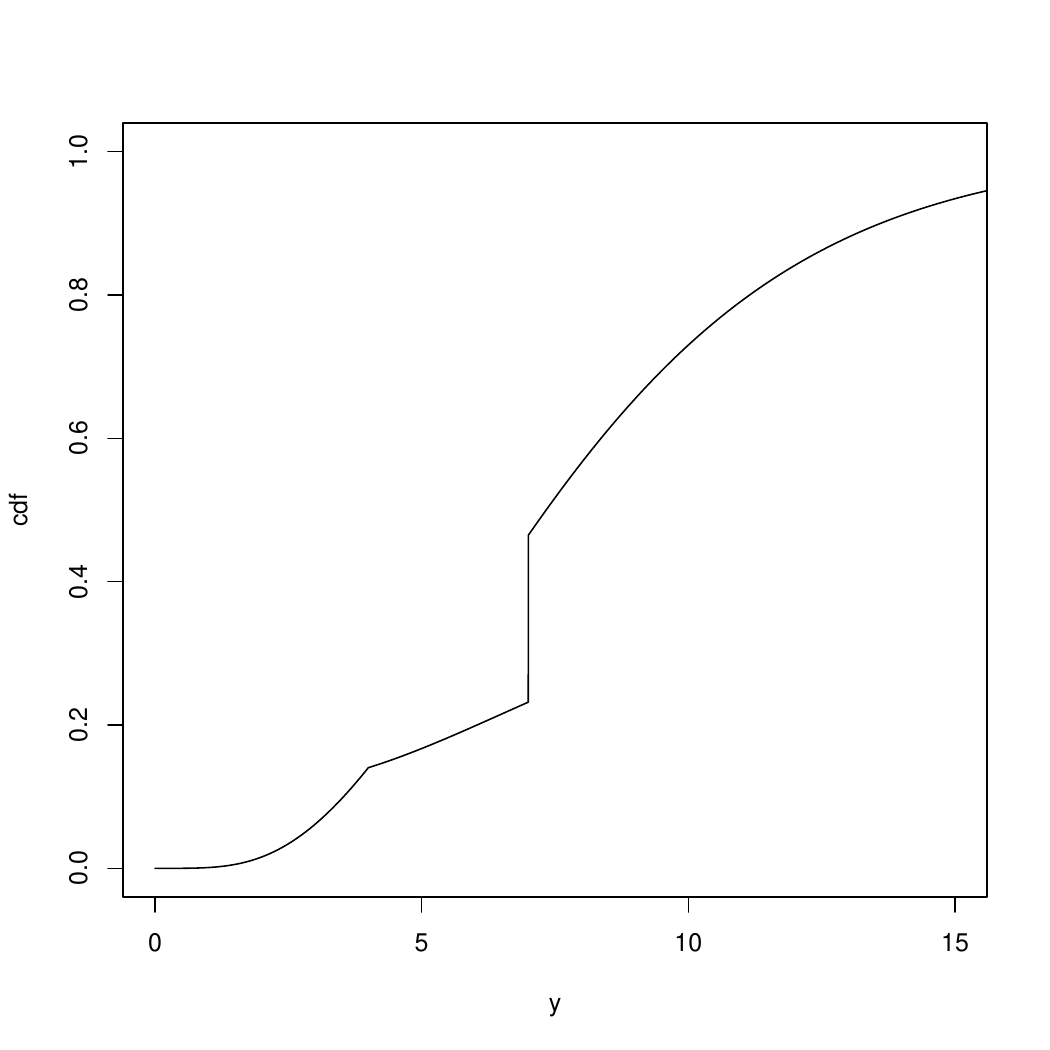}&\includegraphics[width=3.5cm]{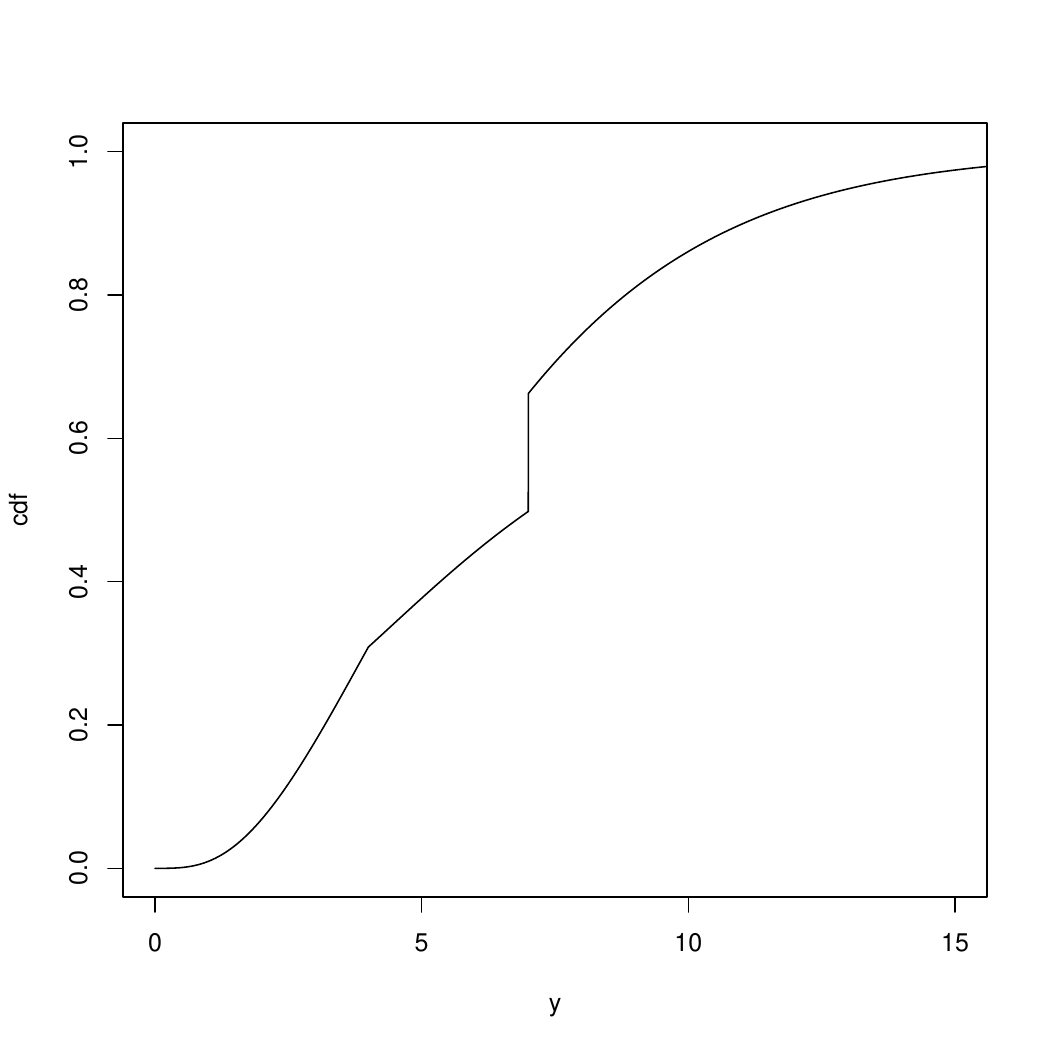}&
\includegraphics[width=3.5cm]{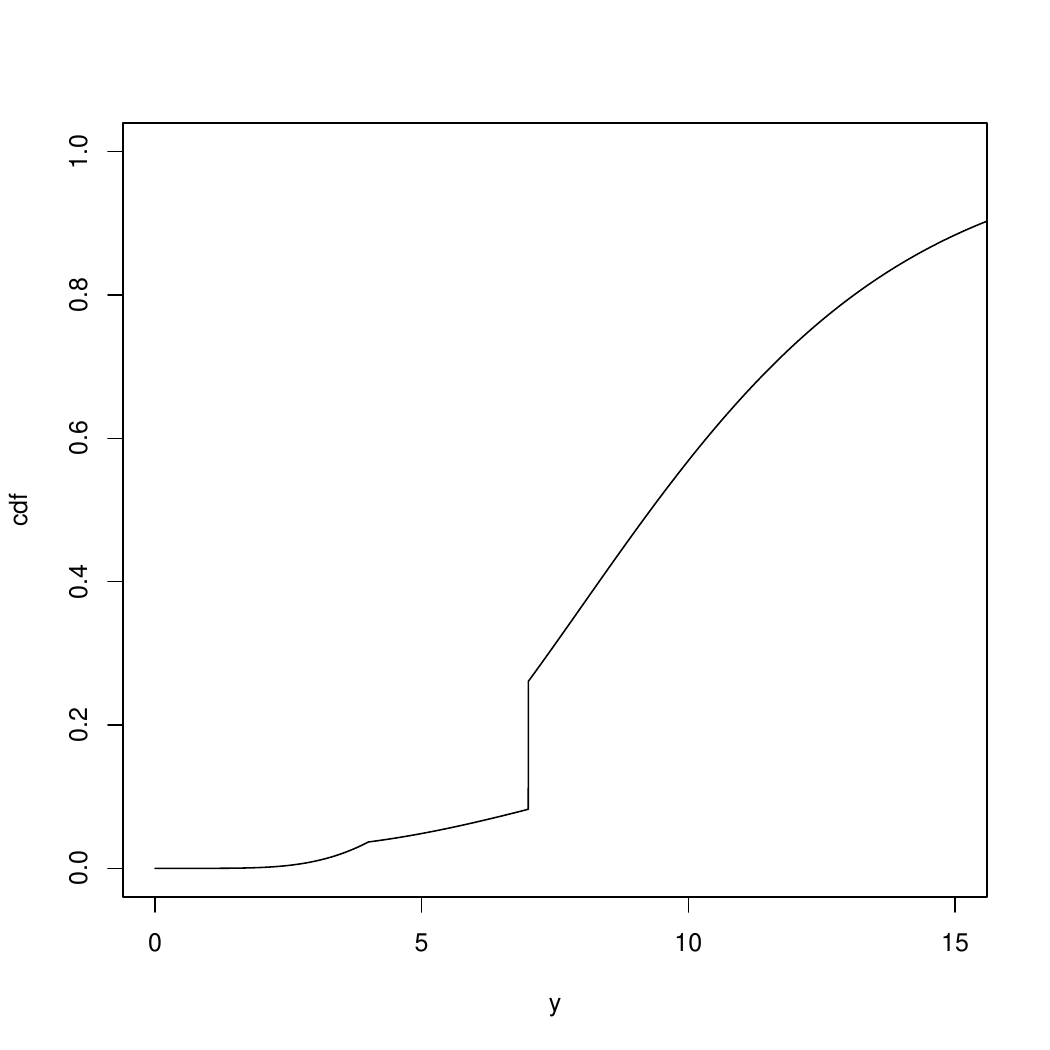}&
\includegraphics[width=3.5cm]{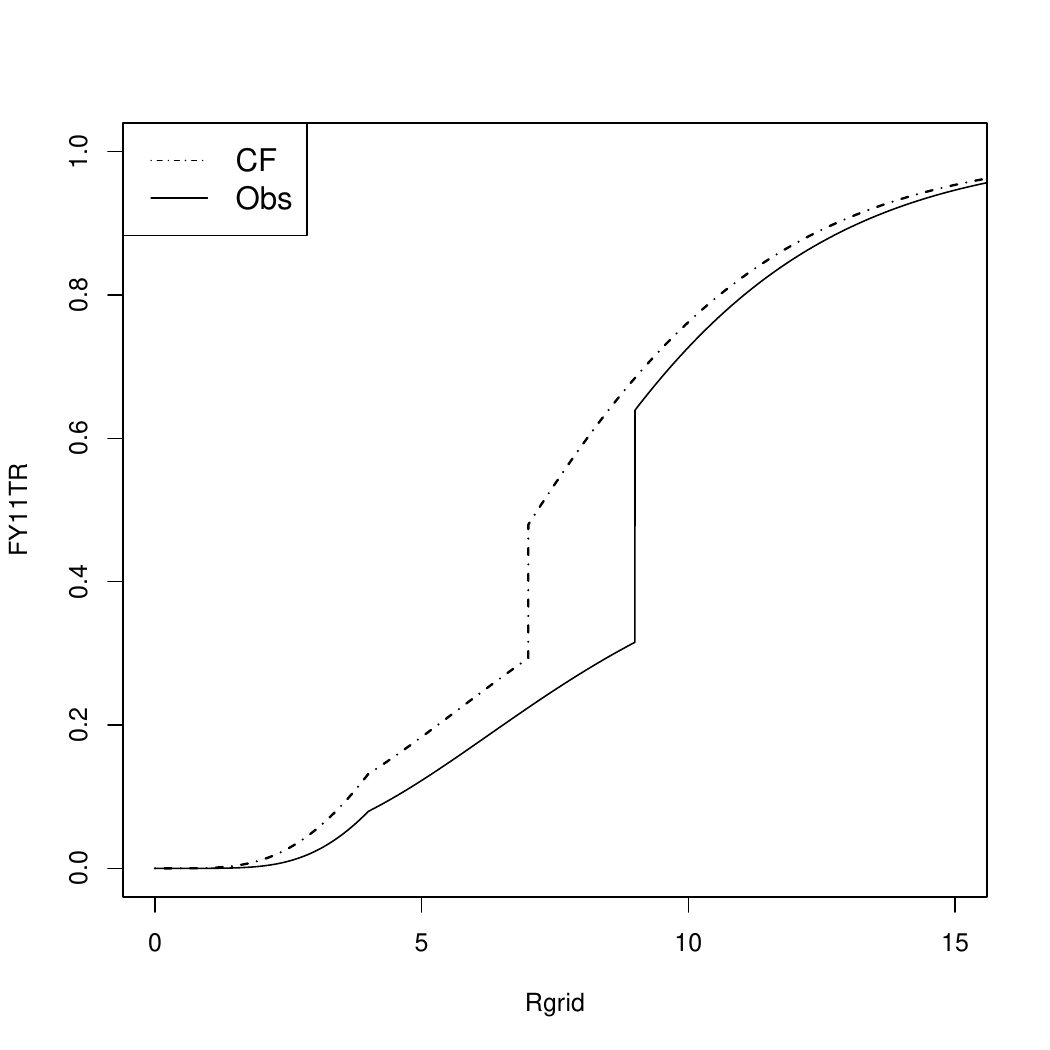}\\
\\
\\
\scriptsize (e) $C_{Y_{00},D}(\cdot,q)$ (analytical)&\scriptsize (f) $\Gamma(\cdot)$ (analytical) &\scriptsize (g) $F_{Y_{00}}(y)\mapsto F_{Y_{00}|D=0}(y)$ &\scriptsize (h) $F_{Y_{00}|D=0}(y)\mapsto F_{Y_{00}|D=1}(y)$ \\
\includegraphics[width=3.5cm]{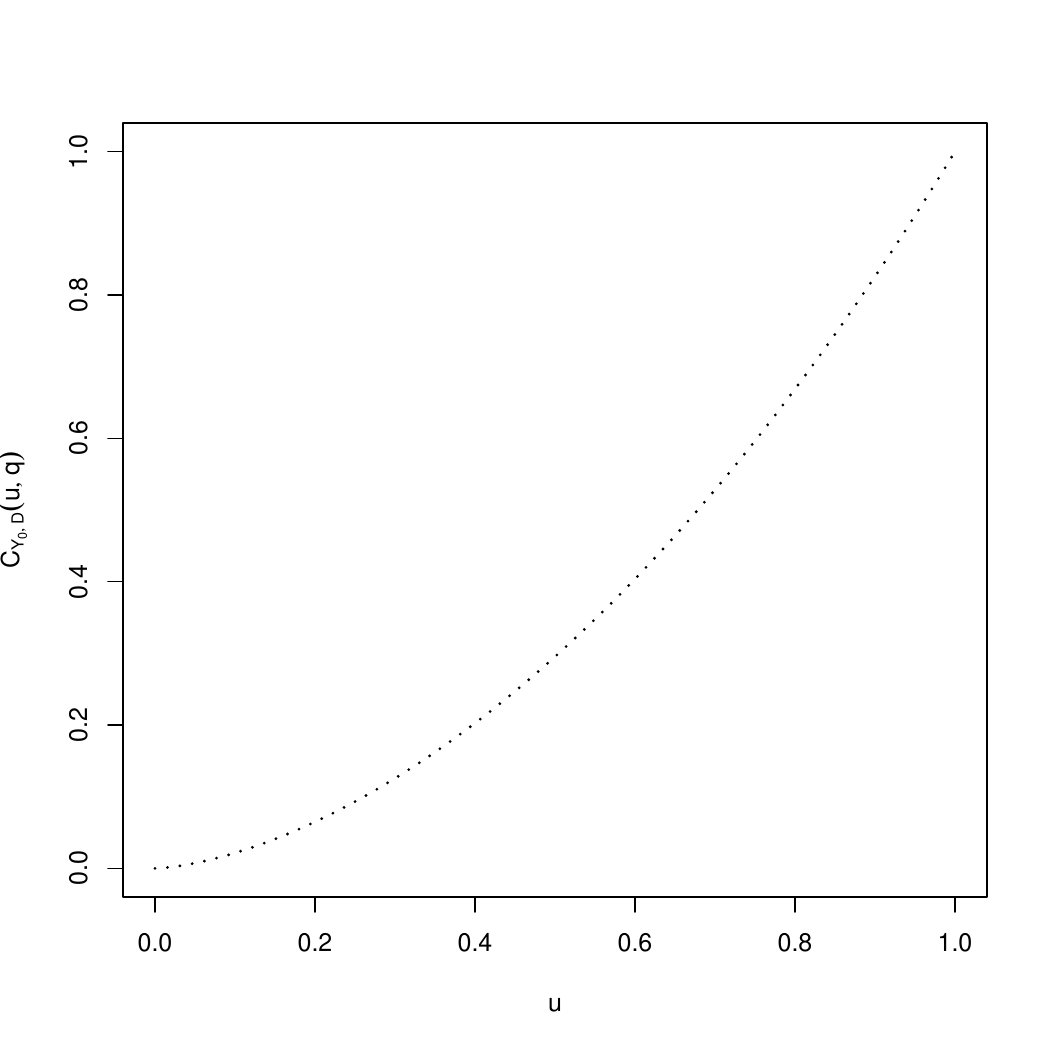}&\includegraphics[width=3.5cm]{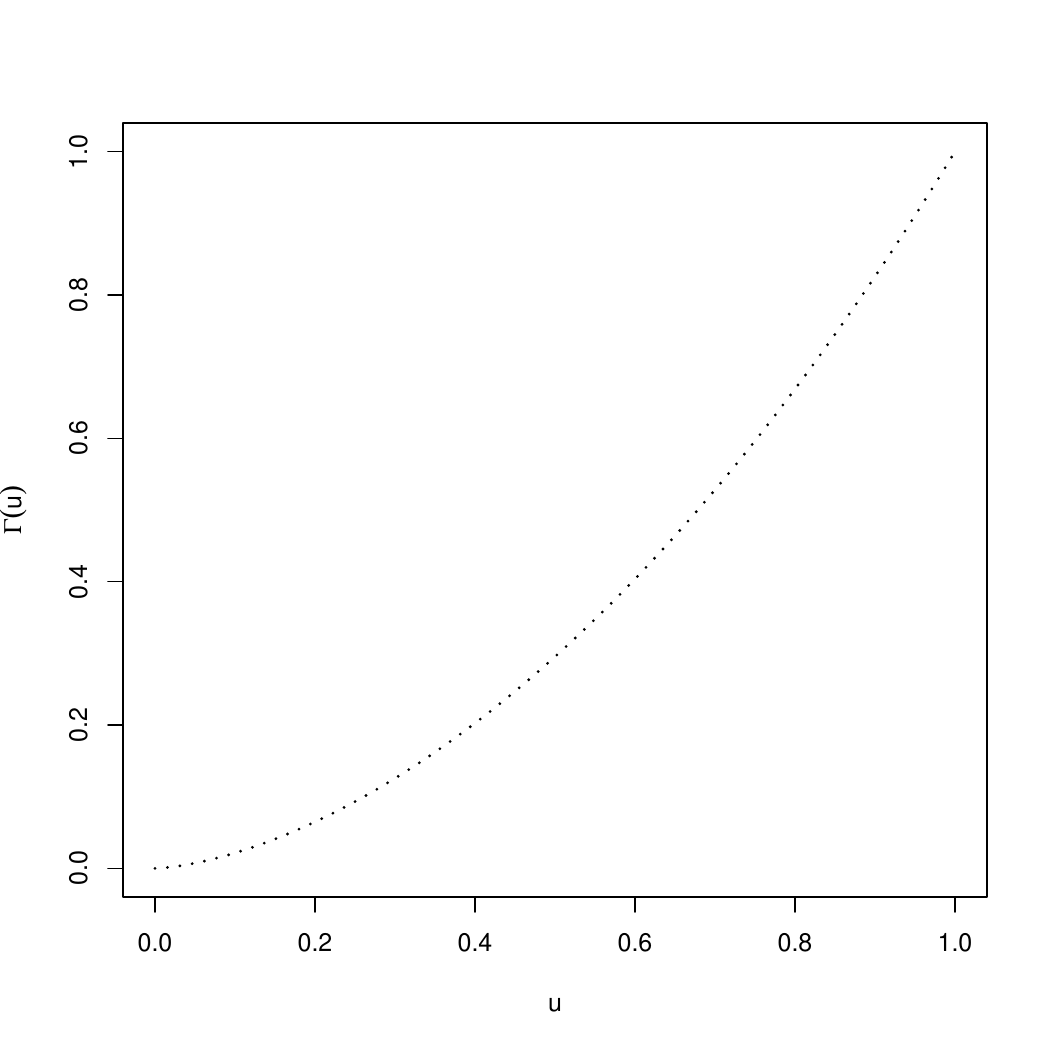}&\includegraphics[width=3.5cm]{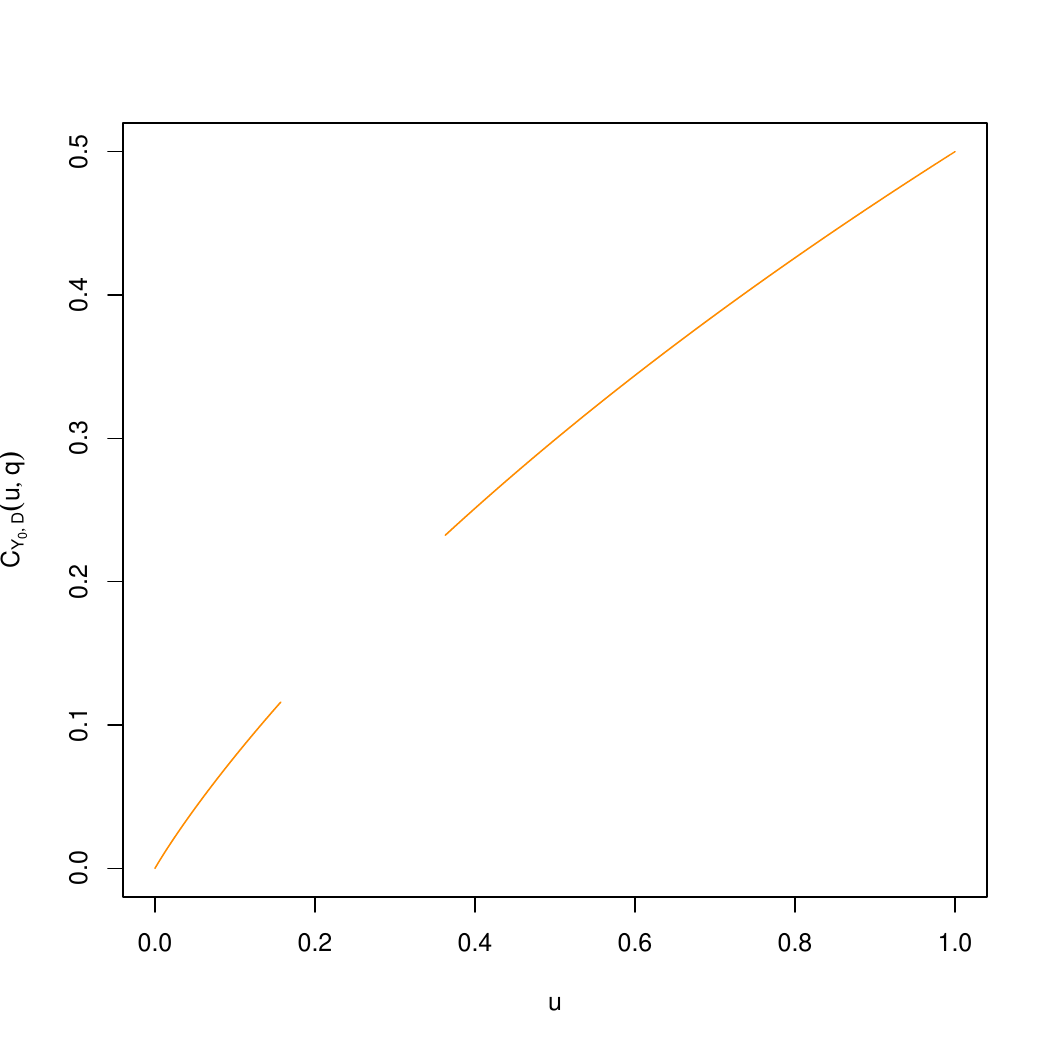}&\includegraphics[width=3.5cm]{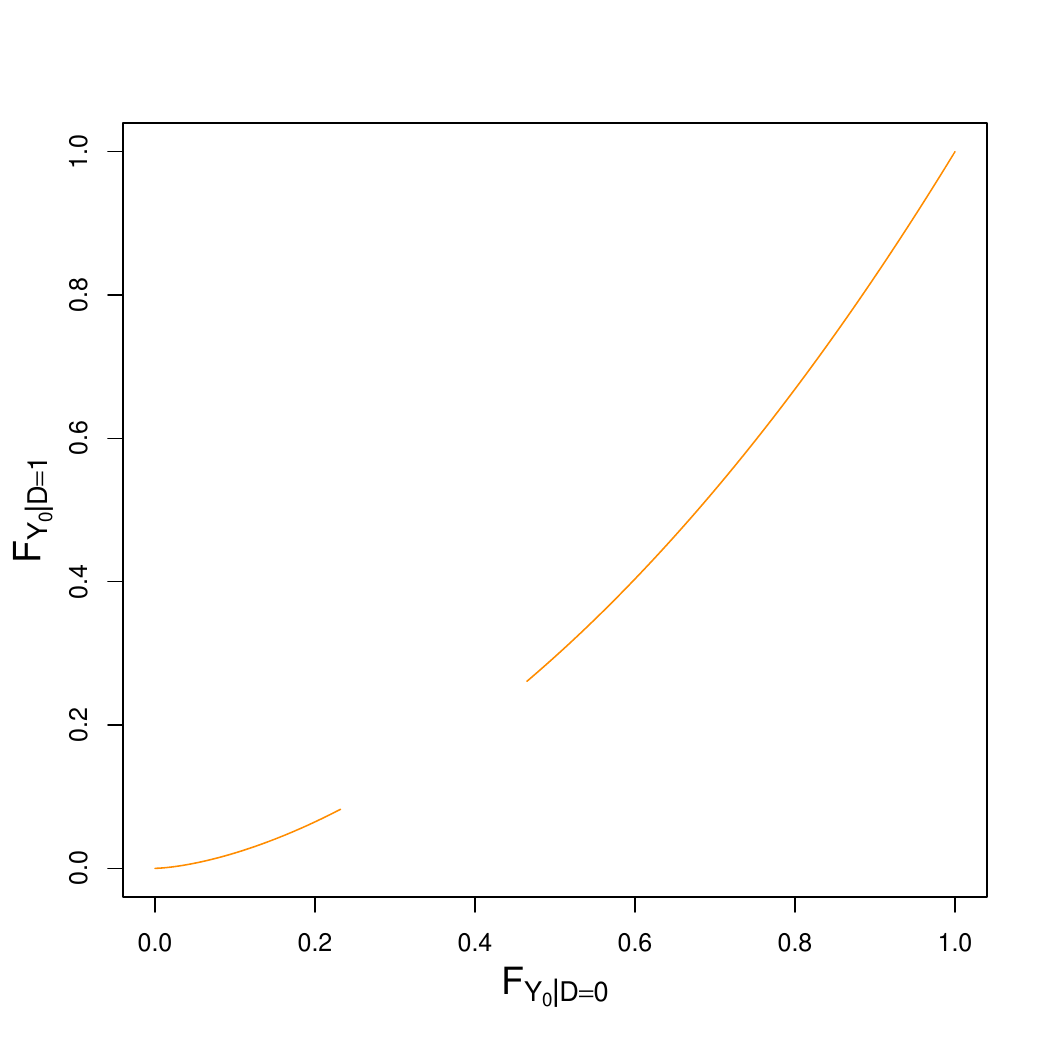}\\
\\
&\multicolumn{2}{c}{\scriptsize (i) CS Bounds on $F_{Y_{10}|D=1}$}&\\
&\multicolumn{2}{c}{\includegraphics[width=6cm]{Figures/FTCSbounds0_mwexample_0.5_2_7_7_0.75_0.5_9_7.pdf}}&\\
\multicolumn{4}{c}{\parbox{0.95\textwidth}{\scriptsize{
\emph{Notes}: Dotted black curves demonstrate curves that depend on unobservables. In Panel (d), $CF$ denotes the counterfactual distribution $F_{Y_{10}|D=1}$. The copula and potential outcome distributions are specified as in Figure \ref{fig:mwexample_multiT0} for $t=0,1$. }}}\\
\end{tabular}\\
\medskip
\label{fig:mwexample_observed}
\end{figure}

\subsection{Numerical illustration of violation of Assumption \ref{cop:mon} in the multiple pre-treatment case}\label{app:num_illustration_monviolation} Here we consider a case where copula stability holds for both pre-treatment periods, but the strict monotonicity of the copula is violated. The violation of the strict monotonicity of the copula leads to support violations in this case.  Unlike the previous two cases, the CS lower and upper bounds obtained from one pre-treatment period cross as in Panels (a) and (b) in Figure \ref{fig:mwexample_multiT0_monviolation}, demonstrating that it is possible to detect a violation of Assumption \ref{cop:mon} when relying on a single pre-treatment period. The testable restriction of Assumptions \ref{cop:mon} and \ref{Gstab} as well as the support condition demonstrates violations in Figure \ref{fig:mwexample_multiT0_monviolation}(f). Note that Figures \ref{fig:mwexample_multiT0_monviolation}(d) and \ref{fig:mwexample_multiT0_monviolation}(e) show that the mappings $C_{Y_{t0},D}(\cdot,q)$ ($\Gamma_t(\cdot)$) are equal for $t=-1,0$ (at the respective intersection of their ranges), a consequence of copula stability. They indicate, however, that the copula is not strictly monotonic (Figure \ref{fig:mwexample_multiT0_monviolation}(d)).

    \begin{figure}[htbp]\caption{CS bounds in the minimum-wage numerical example with Assumption 2 violated but CS holding for all periods}
    \vspace{0.3cm}
    {\scriptsize{\begin{tabular}{ccc}
   (a) Using $t\in\{-1,1\}$&(b) Using $t\in\{0,1\}$&(c) Using $t\in\{-1,0,1\}$\\\\
    \includegraphics[width=5cm]{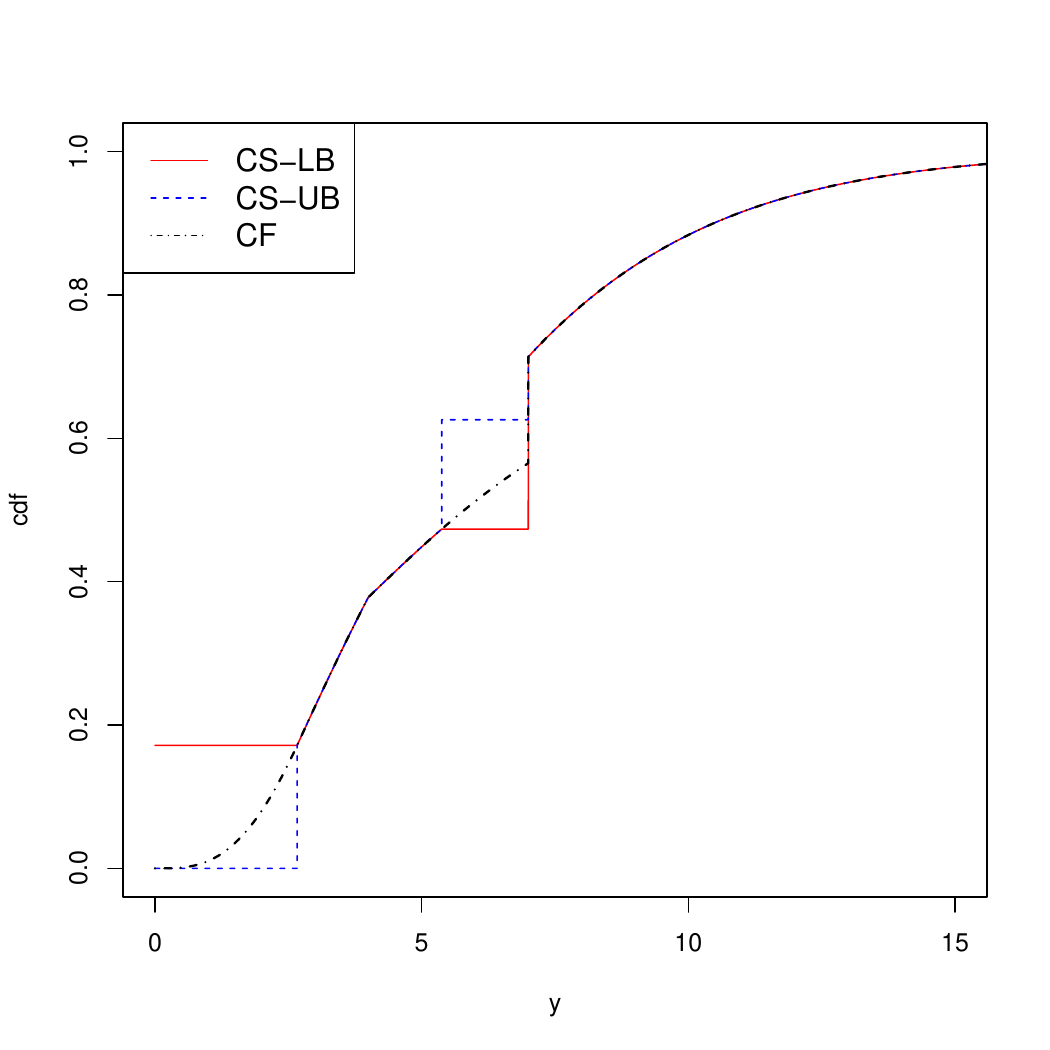}&  \includegraphics[width=5cm]{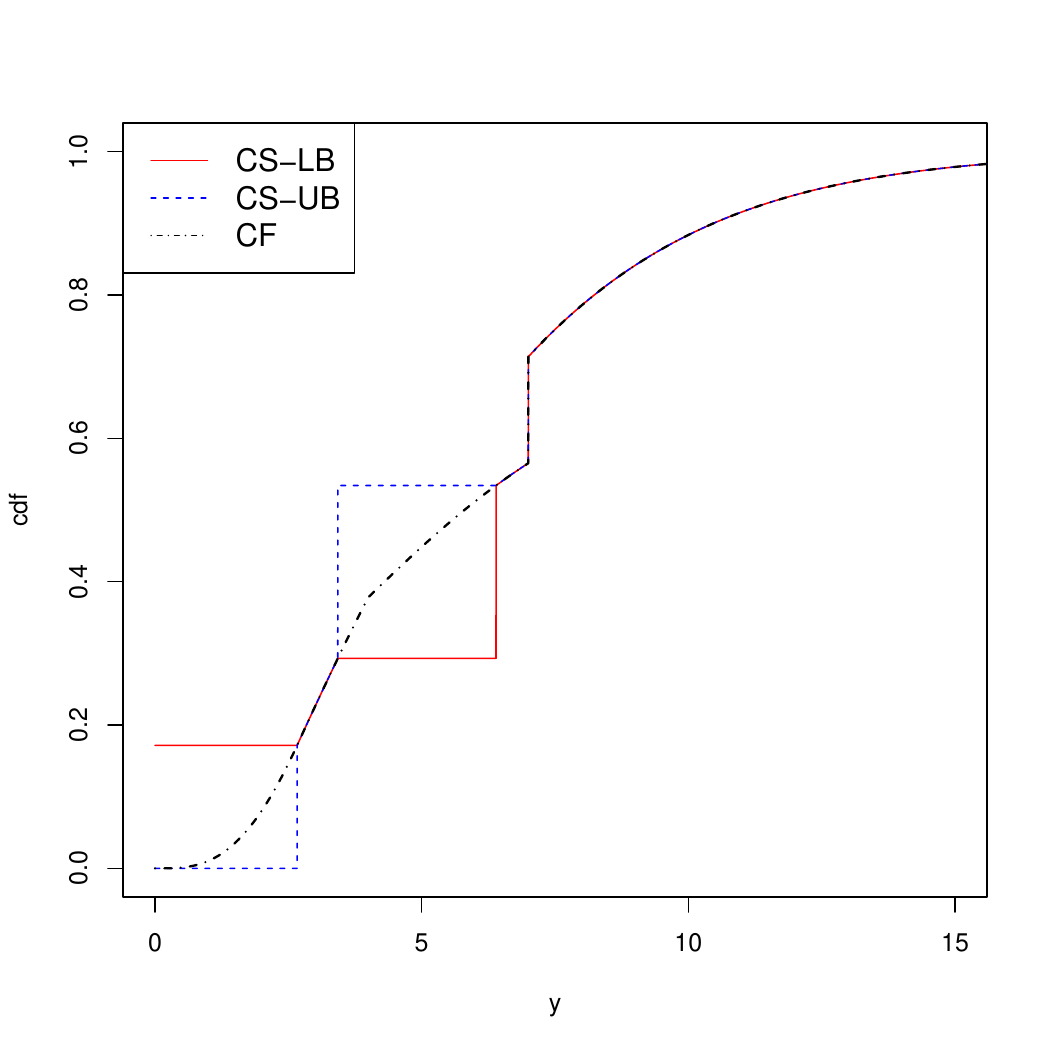}&    
    \includegraphics[width=5cm]{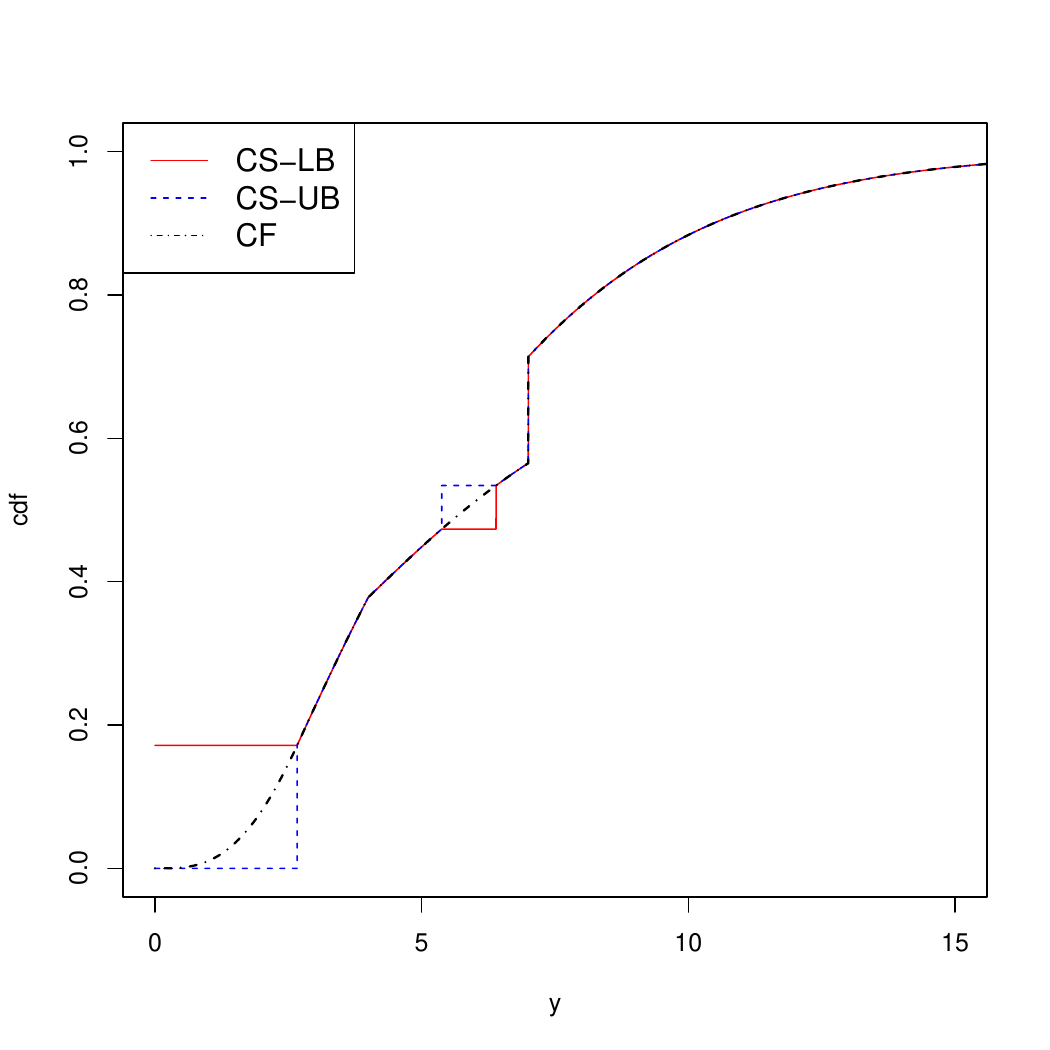}\\
    (d) $F_{Y_{t0}}(y)\mapsto F_{Y_{t0},D}(y,0)$&(e) $F_{Y_{t0}|D=0}(y)\mapsto F_{Y_{t0}|D=1}(y)$&(f) Model Testable Restriction ($\Delta$)\\
\includegraphics[width=5cm]{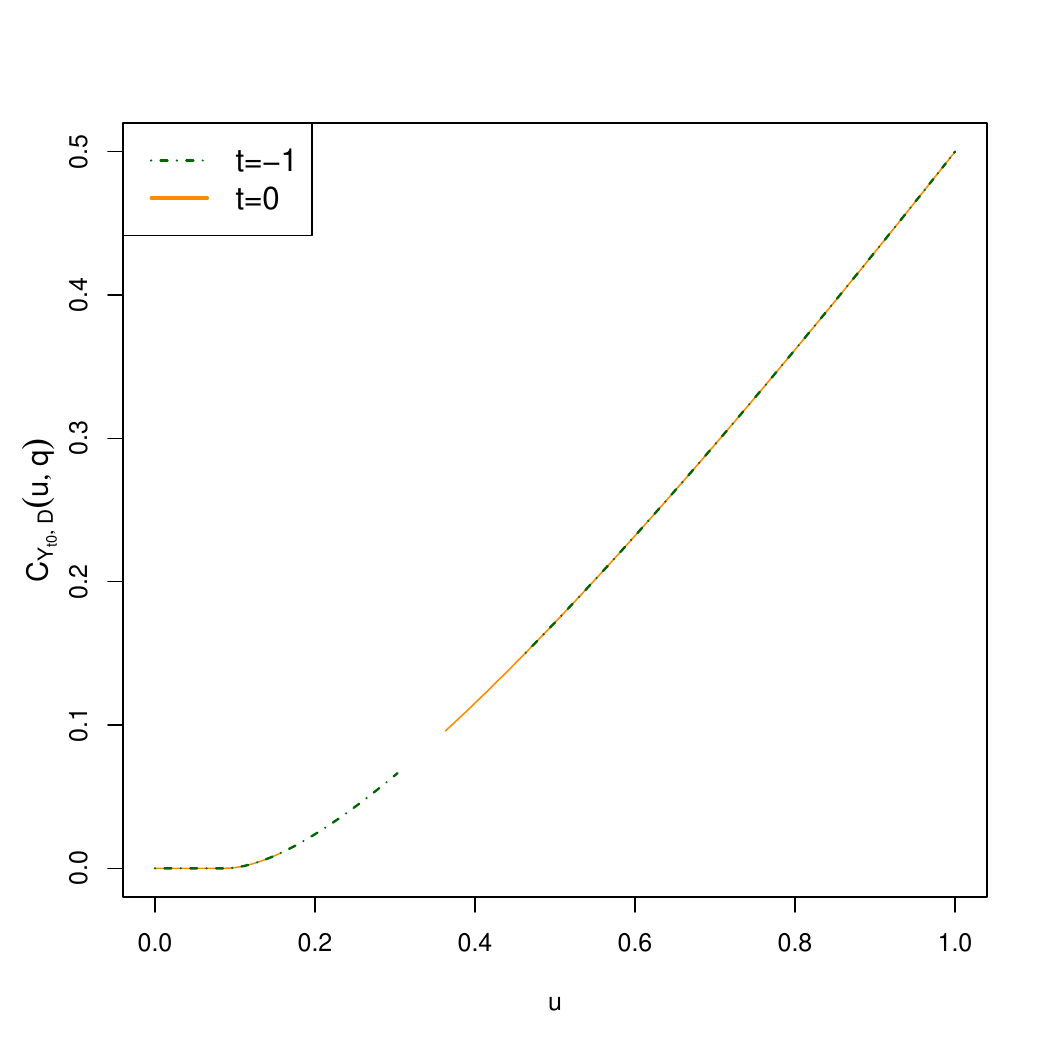}&\includegraphics[width=5cm]{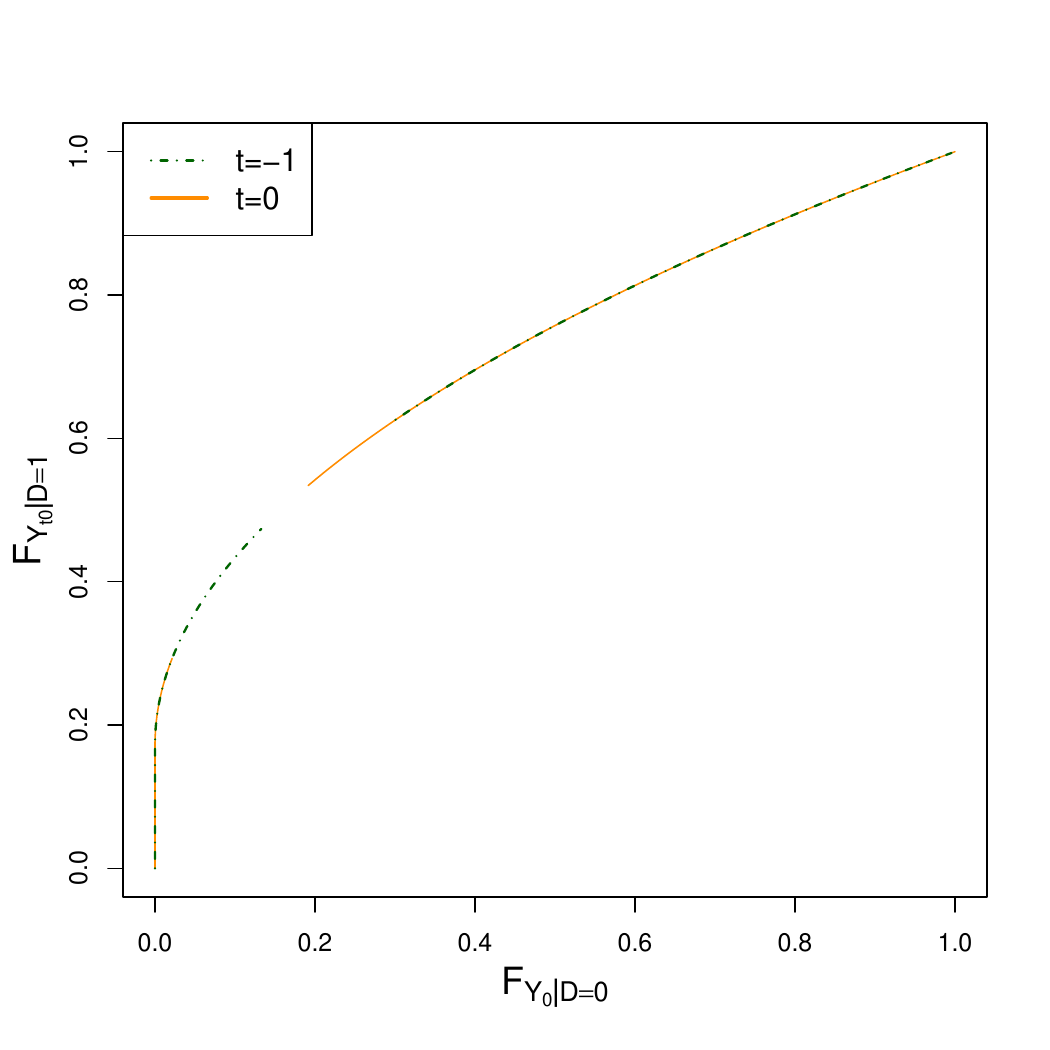}&\includegraphics[width=5cm]{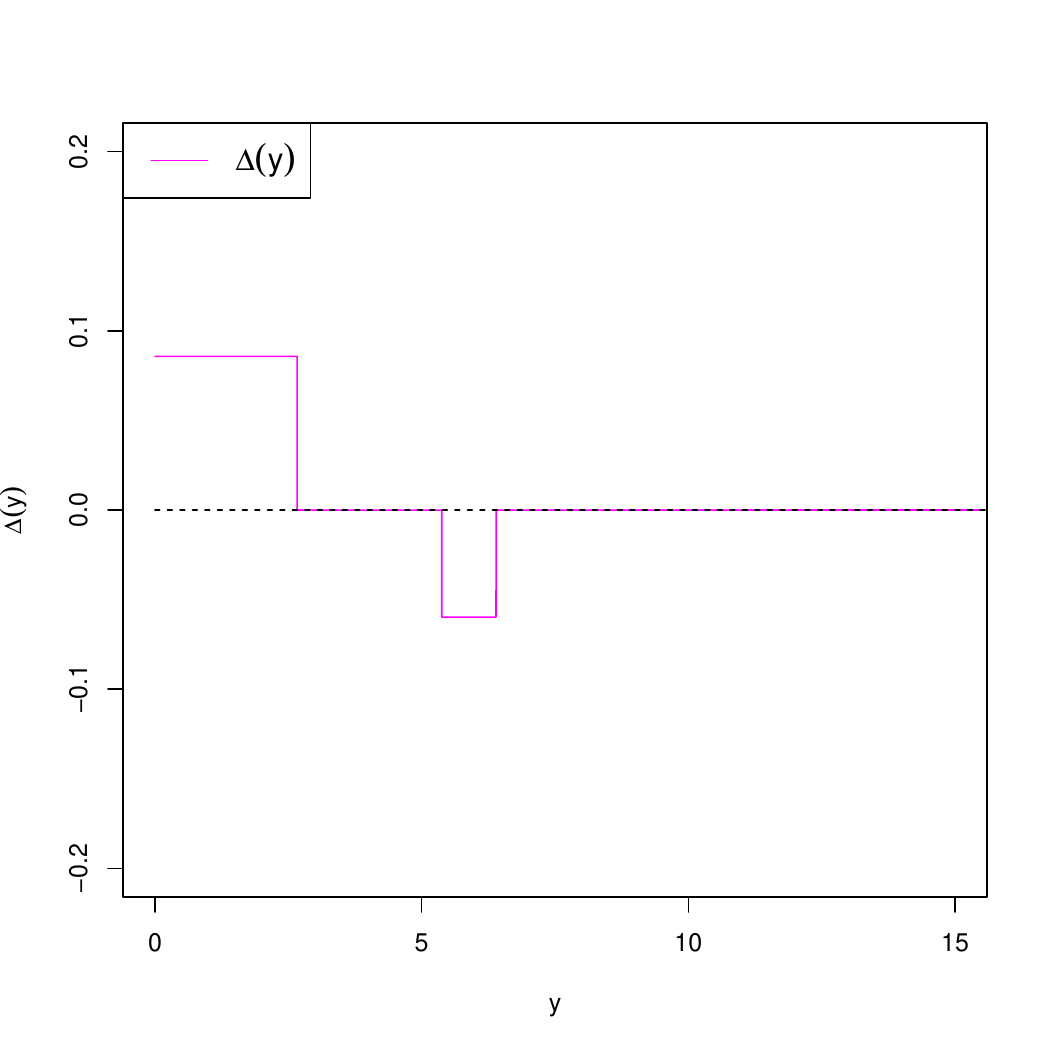}\\
\multicolumn{3}{c}{\parbox{\textwidth}{\scriptsize{\emph{Notes}: To satisfy the copula stability assumption in periods $t\in\{-1,0,1\}$ while violating the strict monotonicity of the copula, we set $C_{Y_{-1,0},D}=C_{Y_{00},D}=C_{Y_{10},D}$ to be the Clayton copula with $\theta=-0.5$. The potential outcome distributions for the treatment and control groups are generated as described in Figure \ref{fig:mwexample_multiT0}. }}}\\

    \end{tabular}}}\\
    \medskip
    \label{fig:mwexample_multiT0_monviolation}
    \end{figure}
\section{Supplementary empirical analysis}
\subsection{Supplementary Figures for Section \ref{sec:empirical}}
We include the CS bounds, distributional DiD and CiC estimates of the counterfactual distributional in Figure \ref{fig:distribution_multiT0}.
\begin{figure} [htbp]\caption{Empirical Application: Observed and Counterfactual Estimates (States with Pre-MW $\geq $ \$8)}
\vspace{0.25cm}{\footnotesize{
\begin{tabular}{cc}
(a) CS Bounds using 2010 pre-treatment period&(b) CS Bounds using 2011 pre-treatment period\\
\includegraphics[width=7cm]{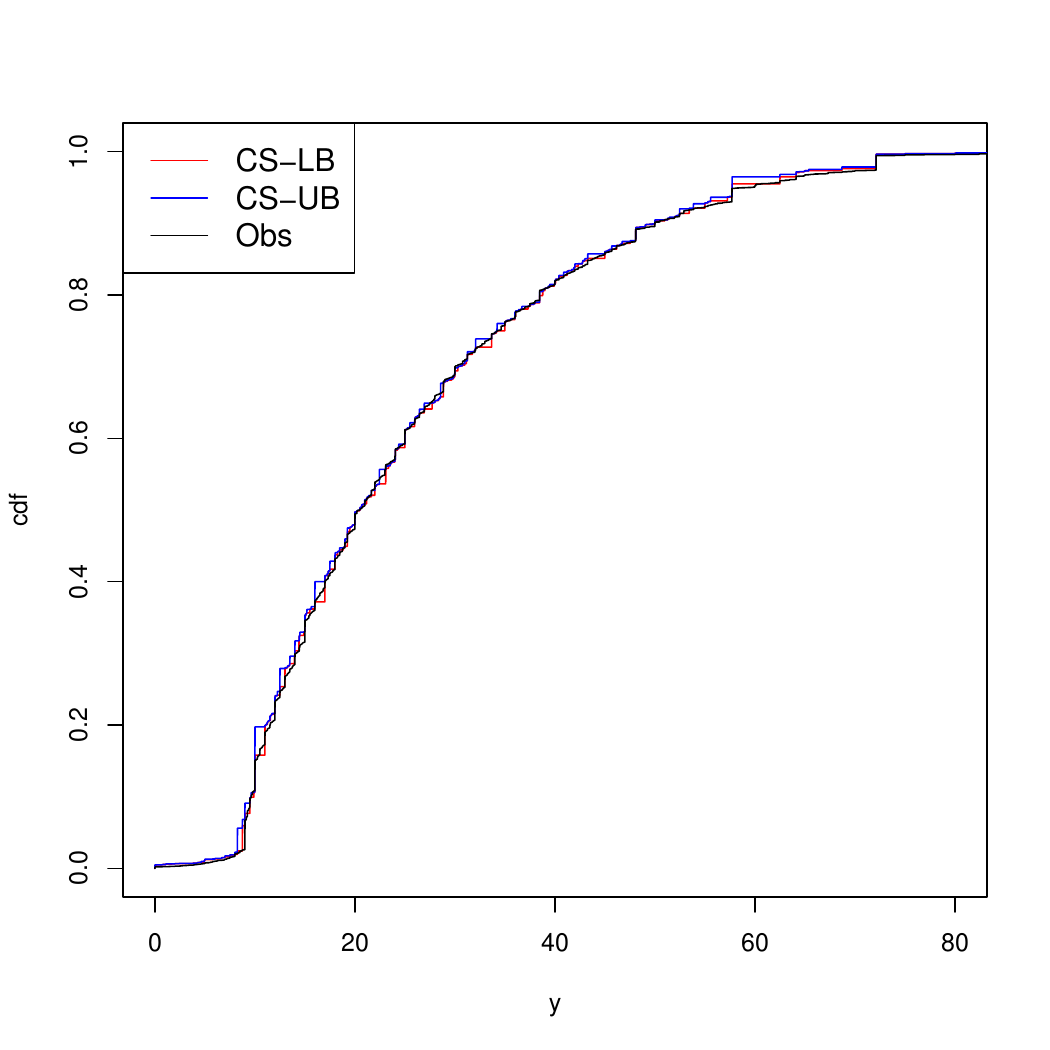}&\includegraphics[width=7cm]{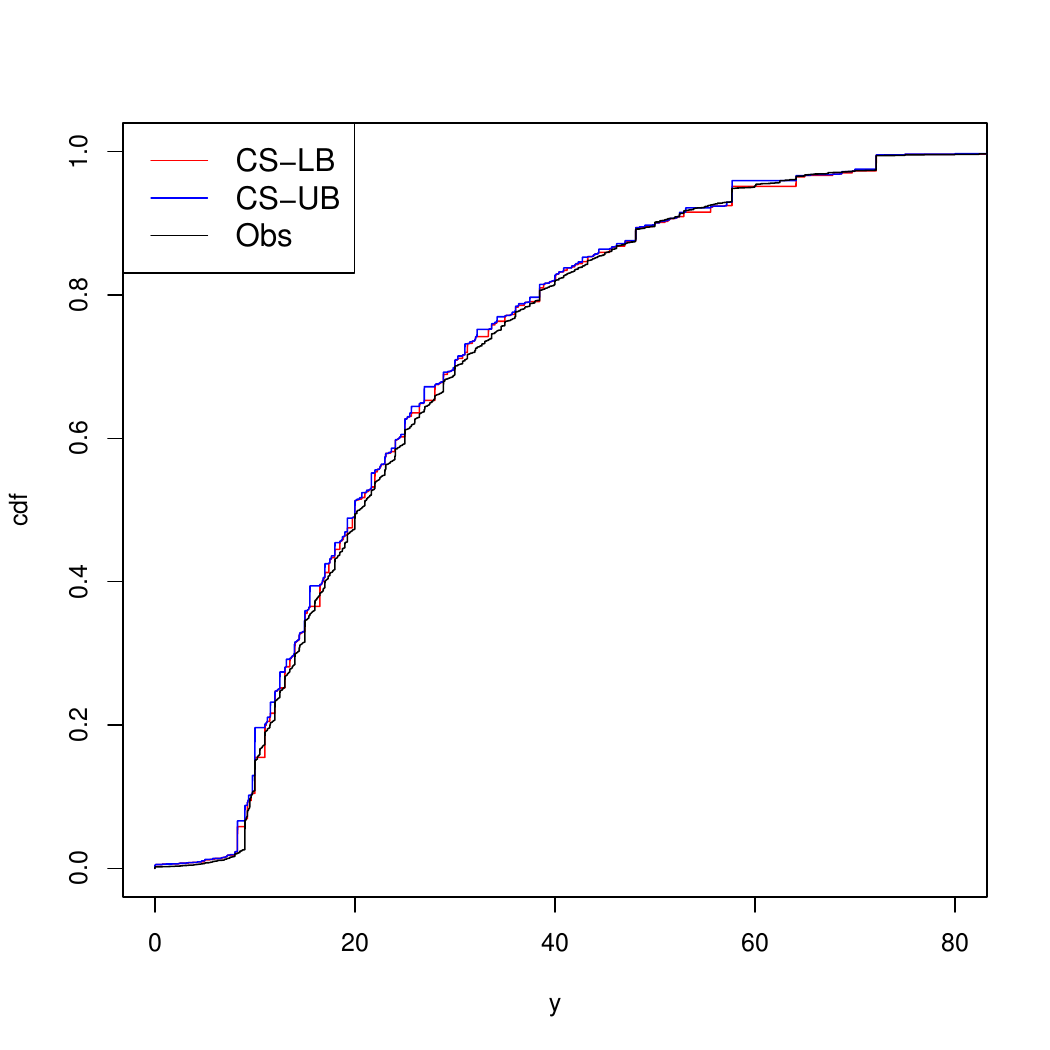}\\
(c) Dist-DiD using 2010 pre-treatment period&(d) Dist-DiD using 2011 pre-treatment period\\
\includegraphics[width=7cm]{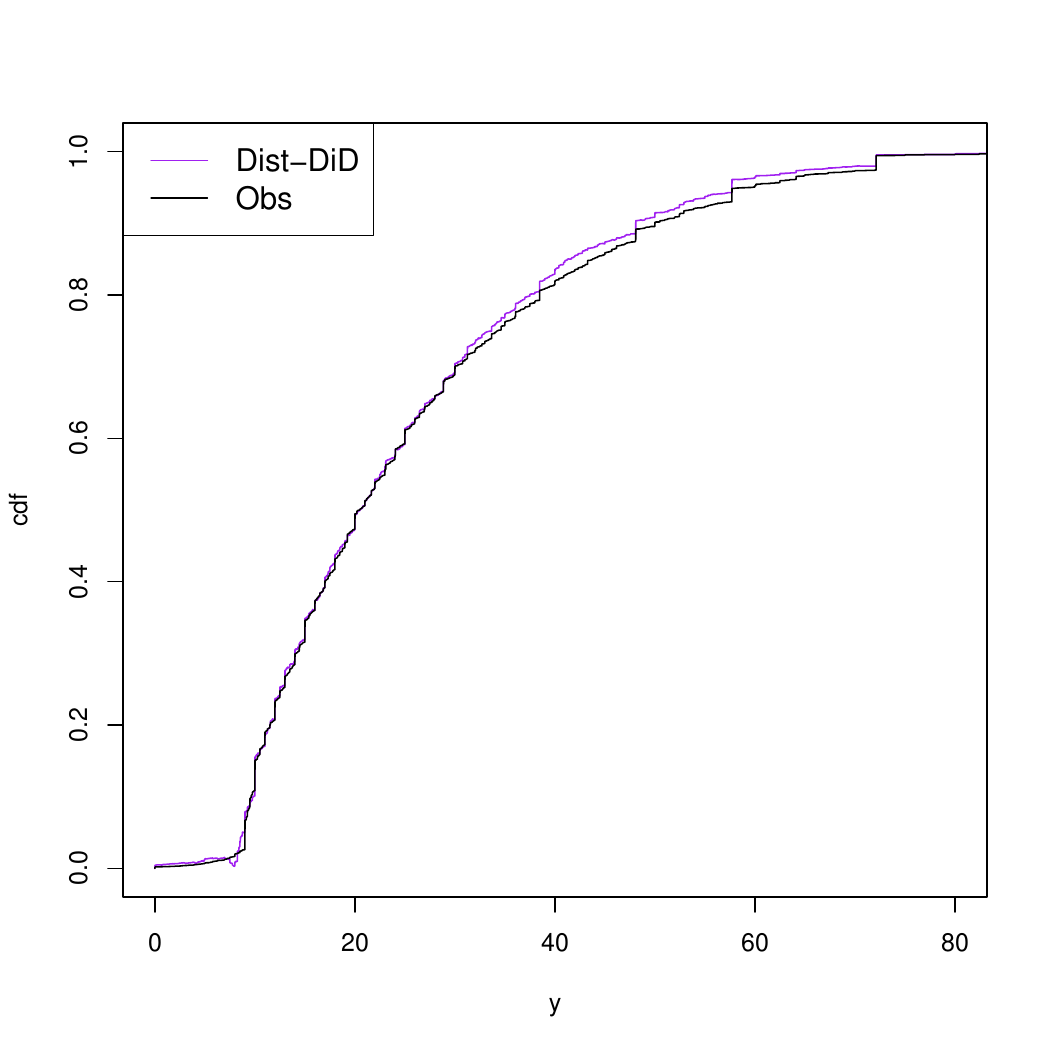}&\includegraphics[width=7cm]{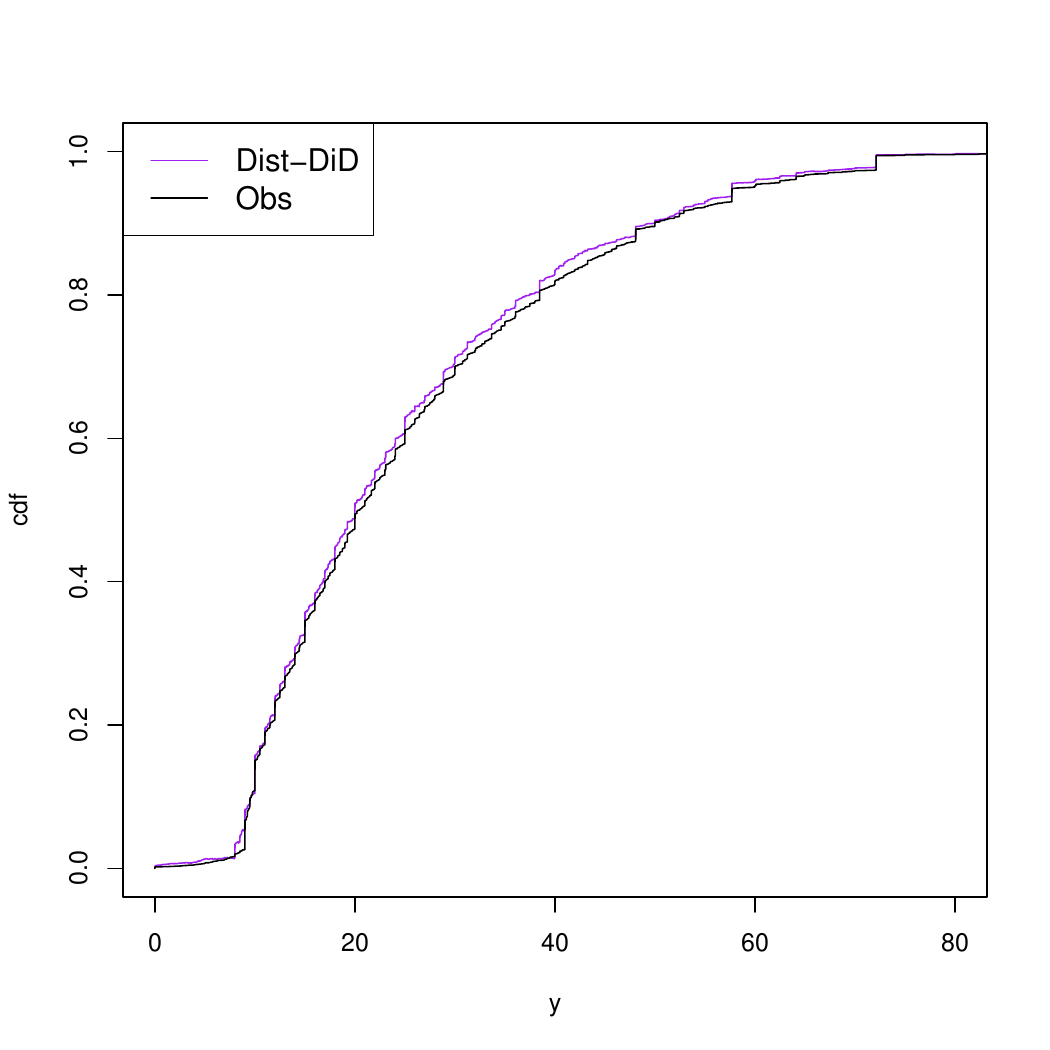}\\
(e) CiC using 2010 pre-treatment period&(f) CiC using 2011 pre-treatment period\\
\includegraphics[width=7cm]{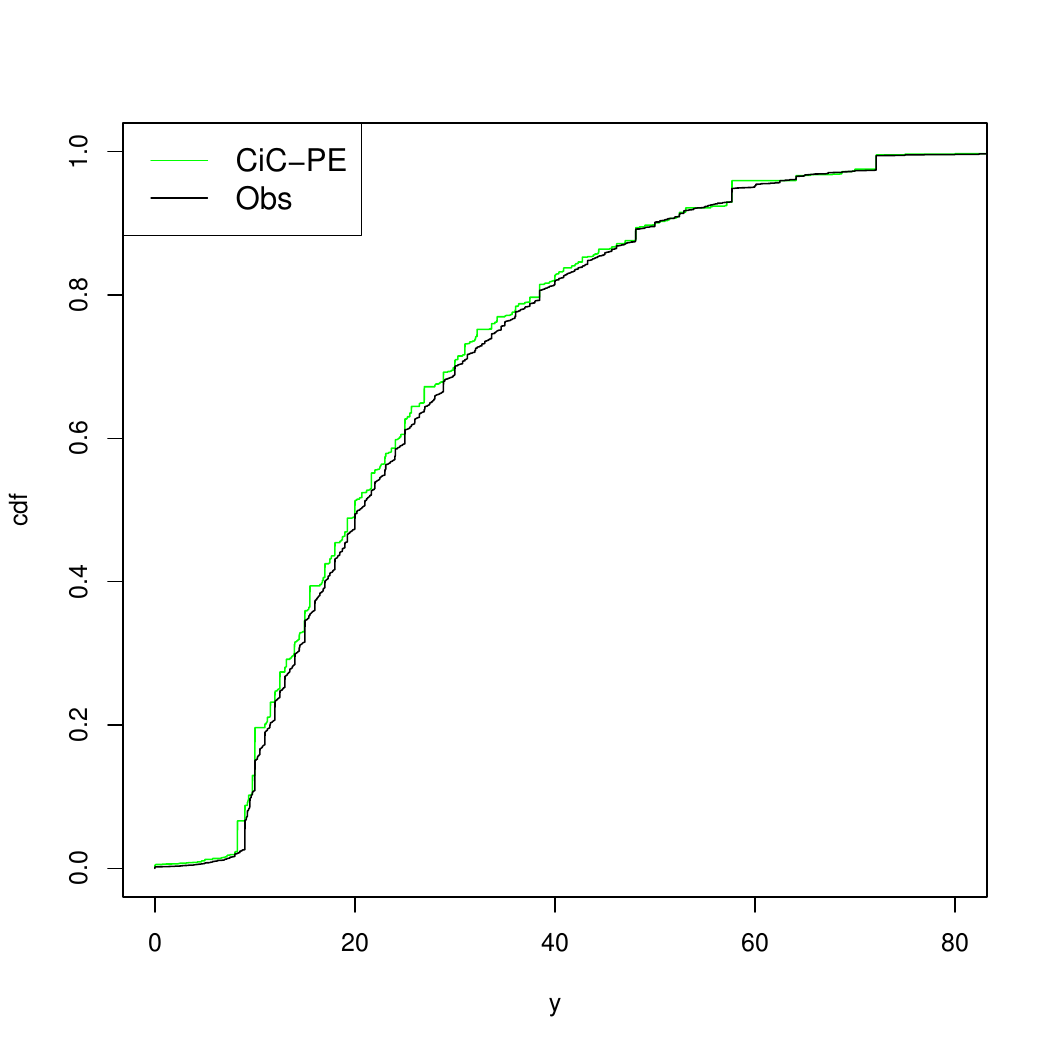}&\includegraphics[width=7cm]{Figures/FY10TCiC0_8_06232025.pdf}\\
\end{tabular}}}
\label{fig:distribution_multiT0}
\end{figure}
\newpage
\subsection{Parameters from \citet{Cengizetal2019}}\label{app:empirical_deltas} Finally, we compute the objects of interest in \citet{Cengizetal2019}, $\Delta b$ and $\Delta a$  depicted in Figure \ref{fig:Ceng}, which quantify the change in employment rates around the new minimum wage, as well as their sum $\Delta e$, which measures the overall impact on employment. Note that these quantities can be obtained from the cdf of the observed and counterfactual distribution as follows, 
\begin{eqnarray}\Delta b&=&F_{Y_1|D=1}(MW)-F_{Y_1|D=1}(0)-(F_{Y_{10}|D=1}(MW)-F_{Y_{10}|D=1}(0)),\label{eq:deltab}\\
\Delta a&=&F_{Y_1|D=1}(\overline{W})-F_{Y_1|D=1}(MW)-(F_{Y_{10}|D=1}(\overline{W})-F_{Y_{10}|D=1}(MW)),\label{eq:deltaa}\\
\Delta e&=&\Delta a+\Delta b=F_{Y_1|D=1}(\overline{W})-F_{Y_1|D=1}(0)-(F_{Y_{10}|D=1}(\overline{W})-F_{Y_{10}|D=1}(0)),\nonumber\end{eqnarray}
where $MW$ denotes the new minimum wage, and $\overline{W}$ is a user-specified quantity that should be the wage level beyond which the increase in the minimum wage should not have an impact on employment. The first quantity $\Delta b$ measures the impact of the minimum wage increase on the proportion of wage-earners with a wage below the new minimum wage, $MW$, whereas $\Delta a$ measures the impact of the minimum wage increase on the proportion of wage earners with hourly wages between $MW$ and $\overline{W}$. Finally, $\Delta e$, which equals the sum of $\Delta a$ and $\Delta b$ by definition, quantifies the impact on the proportion of employment around the minimum wage (below $\overline{W}$).

\begin{table}[htbp]
\caption{Parameters from \citet{Cengizetal2019}}
 {\footnotesize   \begin{tabular}{lrrrrrrrrrr}
    \toprule
  \parbox{1.25cm}{Pre-treatment period} &\multicolumn{4}{c}{2010}&\multicolumn{4}{c}{2011}&\multicolumn{2}{c}{ 2010 \& 2011}\\
  \cmidrule(lr){2-5}\cmidrule(lr){6-9}\cmidrule(lr){10-11}
   & \multicolumn{2}{c}{CS Bounds} & \multicolumn{1}{c}{DistDiD} & \multicolumn{1}{c}{CiC} & \multicolumn{2}{c}{CS Bounds}  & \multicolumn{1}{l}{DistDiD} & \multicolumn{1}{l}{CiC} & \multicolumn{2}{c}{CS Bounds}\\
   \cmidrule(lr){2-3}\cmidrule(lr){6-7}\cmidrule(lr){10-11}
    &\multicolumn{1}{c}{LB}&\multicolumn{1}{c}{UB}&&&\multicolumn{1}{c}{LB}&\multicolumn{1}{c}{UB}&&&\multicolumn{1}{c}{LB}&\multicolumn{1}{c}{UB}\\
    \midrule
$\Delta b$  &  -2.9\% & 0.3\% & -1.2\% & -2.9\% & -3.9\% & -3.1\% & -2.0\% & -3.9\% & -2.9\% & -3.1\% \\
$\Delta a$  &  -0.9\% & 2.3\% & 1.7\% & 2.3\% & 2.3\% & 3.3\% & 1.6\% & 3.1\% & 2.5\% & 2.3\% \\
$\Delta e$&    -0.6\% & -0.6\% & 0.5\% & -0.6\% & -0.8\% & -0.6\% & -0.4\% & -0.8\% & -0.6\% & -0.6\% \\
\midrule
\multicolumn{11}{c}{\parbox{0.8\textwidth}{\footnotesize \emph{Notes}: We compute the estimates of $\Delta b$ and $\Delta a$ using the sample analogues of Eq. \eqref{eq:deltab} and \eqref{eq:deltaa}, respectively, with $MW=8.5$ and $\bar{W}=11$. }}\\
\end{tabular}%
}\label{tab:deltas_multiT0}
\end{table}

Table \ref{tab:deltas_multiT0} presents the estimates of $\Delta b$, $\Delta a$ and $\Delta e$ for all estimators we consider using 2010 and 2011 as pre-treatment periods.\footnote{The CS bounds on $\Delta b$ are given by the following,
\begin{eqnarray}&&F_{Y_1|D=1}(MW)-F_{Y_1|D=1}(0)-(F_{Y_{10}|D=1}^{UB}(MW)-F_{Y_{10}|D=1}^{LB}(0))\nonumber\\
&\leq& \Delta b\leq F_{Y_1|D=1}(MW)-F_{Y_1|D=1}(0)-\max\left \{(F_{Y_{10}|D=1}^{LB}(MW)-F_{Y_{10}|D=1}^{UB}(0)),0\right\}.\label{eq:bounds_deltas}\end{eqnarray}

CS bounds on the $\Delta a$ and $\Delta e$ are obtained in the same manner.} The CS bounds on $\Delta b$ and $\Delta a$ using the 2010 pre-treatment period are wide and inconclusive regarding the sign of this parameter. The CS bounds on $\Delta e$ however collapse to a point and equal -0.6\% suggesting that the minimum wage may have slightly reduced employment. Using the 2011 pre-treatment period, the CS bounds estimates on $\Delta b$ suggest a reduction in employment below the minimum wage, whereas the CS bounds estimates on $\Delta a$ suggest an increase in employment just above the minimum wage. The CS bounds on $\Delta e$ using the 2011 pre-treatment period also suggest a slight reduction in employment. The multi-period CS bounds lead to a similar conclusion.

When we examine the distributional DiD estimates, we find that the distributional DiD point estimates using each pre-treatment period suggest similar results for $\Delta b$ and $\Delta a$, though slightly different magnitudes, but differ in the sign of $\Delta e$. Using the 2010 pre-treatment period, the distributional DiD estimate suggest that a slight increase in employment $\Delta e$, whereas it suggest a slight decrease in employment if one uses the 2011 pre-treatment period. Given the monotonicity violation around \$8 exhibited in the distributional DiD counterfactual estimate using the 2010 pre-treatment period, one would conclude that the resulting $\Delta e$ estimate is unreliable. 

Finally, since $\Delta b$, $\Delta a$ and $\Delta e$ consist of differences between the observed and counterfactual outcome distribution,  the CiC point estimates will thus fall inside the CS bounds on $\Delta b$, $\Delta a$ and $\Delta e$, respectively.
This is because the CiC point estimates of these objects will use the CS upper bound as the counterfactual estimate.\footnote{This is straightforward from examining Eq.\ \eqref{eq:bounds_deltas}.}

\subsection{Empirical analysis for subsample with pre-MW $<$\$8}\label{app:empirical_subgrouplessthan8}
Here we provide the empirical analysis for the subsample with pre-treatment minimum wage of less than \$8. Since the federal minimum wage was \$7.25 in 2009, this subsample consists of states with pre-treatment minimum wage at or slightly above the federal minimum wage. Table \ref{tab:Cengizetal_sumstats_multiT0_subsample} provides the summary statistics, Figure \ref{fig:distribution_zoombottom_multiT0_subsample} and \ref{fig:distribution_multiT0_subsample} provides the estimates of the counterfactual distribution for the bottom quartile as well as the entire distribution, respectively. Tables \ref{tab:swtt_multiT0_subsample} and \ref{tab:lowertail_swtt_multiT0_subsample}  provide confidence intervals on the SWTT parameters we consider.
\begin{table}[H]\caption{Summary Statistics by Treatment and Control Groups:  States with Pre-MW$<$\$8}
\vspace{0.2cm}
\footnotesize{
    \begin{tabular}{lrrrrrrrrr}
    \toprule
    & \multicolumn{3}{c}{2010 (Pre-treatment)} & \multicolumn{3}{c}{2011 (Pre-treatment)}   & \multicolumn{3}{c}{2015 (Post-treatment)} \\
\cmidrule(lr){2-4}\cmidrule(lr){5-7}\cmidrule(lr){8-10}
(\$) & \multicolumn{1}{c}{Mean} & \multicolumn{1}{c}{S.D. } & \multicolumn{1}{c}{\# Obs} & \multicolumn{1}{c}{Mean} & \multicolumn{1}{c}{S.D. } & \multicolumn{1}{c}{\# Obs}& \multicolumn{1}{c}{Mean} & \multicolumn{1}{c}{S.D. } & \multicolumn{1}{c}{\# Obs} \\
 \midrule
    \multicolumn{10}{l}{States with Pre-Treatment Minimum Wage $<8$}\\Control & 18.43 & 12.78 &          44,574  & 18.78 & 13.48 &              43,864  & 20.41 & 15.90 &              42,322  \\
    Treatment & 20.20 & 14.26 &          38,261  & 20.64 & 14.47 &              37,127  & 22.12 & 18.21 &              32,489  \\
    \bottomrule
    \end{tabular}}
\label{tab:Cengizetal_sumstats_multiT0_subsample}
\end{table}

\begin{table}[H]
  \centering
  \caption{Inference on SWTT using 2010 and/or 2011 as pre-treatment periods: States with Pre-MW$<$\$8}
{\footnotesize    \begin{tabular}{lrrrrrrrrrrrl}
\multicolumn{7}{l}{Panel A. CS bounds}\\
          \toprule
          & \multicolumn{6}{c}{95\% CI}  \\
          \midrule
 Pre-period        & \multicolumn{2}{c}{2010}        & \multicolumn{2}{c}{2011} & \multicolumn{2}{c}{  2010 \& 2011}\\
          \cmidrule(lr){2-3}\cmidrule(lr){4-5}\cmidrule(lr){6-7}
          \cmidrule(lr){2-3}\cmidrule(lr){4-5}\cmidrule(lr){6-7}
  (\$)        &  \multicolumn{1}{c}{LB} & \multicolumn{1}{c}{UB} & \multicolumn{1}{c}{LB} & \multicolumn{1}{c}{UB} & \multicolumn{1}{c}{LB} & \multicolumn{1}{c}{UB} \\
          \midrule

    ATT   & -0.94 & 0.18  & -0.92 & 0.16  & -0.68 & -0.02 \\
    Gini SWTT & -0.38 & 0.25  & -0.46 & 0.15  & -0.30 & 0.08 \\
    \bottomrule 
    \end{tabular}}\\
    \vspace{0.5cm}
 {\footnotesize{   \begin{tabular}{lcrrcrrcrrcrr}

\multicolumn{9}{l}{Panel B. Distributional DiD and CiC point estimators}\\
     \toprule
 & \multicolumn{4}{c}{Dist DiD: 95\% CI}        & \multicolumn{4}{c}{CiC: 95\% CI}  \\
 \midrule
Pre-period    & \multicolumn{2}{c}{2010} 
   &\multicolumn{2}{c}{2011}
           & \multicolumn{2}{c}{2010} 
   &\multicolumn{2}{c}{2011}\\
             \cmidrule(lr){2-3}\cmidrule(lr){4-5}\cmidrule(lr){6-7}\cmidrule(lr){8-9}
  (\$)    &\multicolumn{1}{c}{LB} & \multicolumn{1}{c}{UB} &   \multicolumn{1}{c}{LB} & \multicolumn{1}{c}{UB} &    \multicolumn{1}{l}{LB} & \multicolumn{1}{l}{UB} &   \multicolumn{1}{l}{LB} & \multicolumn{1}{l}{UB}   \\
\midrule
ATT   & -0.40 & 0.23  & -0.50 & 0.15  & -0.55 & 0.17  & -0.60 & 0.15 \\
    Gini SWTT & -0.26 & 0.11  & -0.35 & 0.02  & -0.12 & 0.27  & -0.20 & 0.16 \\
  
    \midrule
    \multicolumn{9}{c}{\parbox{10cm}{\scriptsize{\emph{Notes}: The definitions of the SWTT bounds/point estimators are provided in \eqref{eq:swtt-estimator}--\eqref{eq:swtt-estimator_pe}. For the CS bounds, we report 95\% confidence intervals on the identified set. For the point estimators, we report 95\% confidence intervals on the SWTT parameter. All confidence intervals use standard normal critical values and nonparameteric bootstrap standard errors using 500 bootstrap replications.}}}\\
\end{tabular}}}
\label{tab:swtt_multiT0_subsample}
\end{table}

\begin{figure} [htbp]\caption{Observed and Counterfactual Distributions: States with Pre-MW$< \$ 8$, Bottom Quartile}
\vspace{0.25cm}{\footnotesize{
\begin{tabular}{cc}
(a) CS Bounds using 2010 pre-treatment period&(b) CS Bounds using 2011 pre-treatment period\\
\includegraphics[width=6cm]{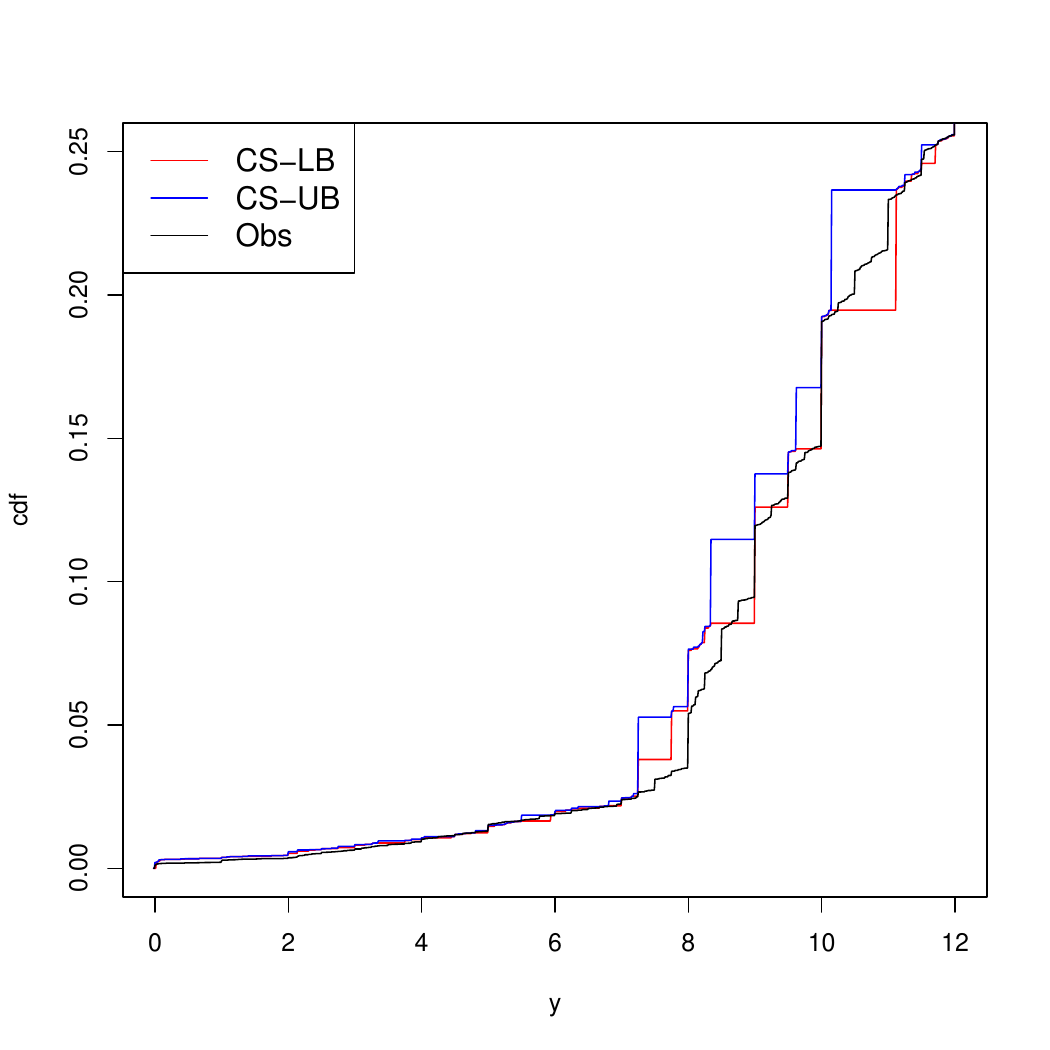}&\includegraphics[width=6cm]{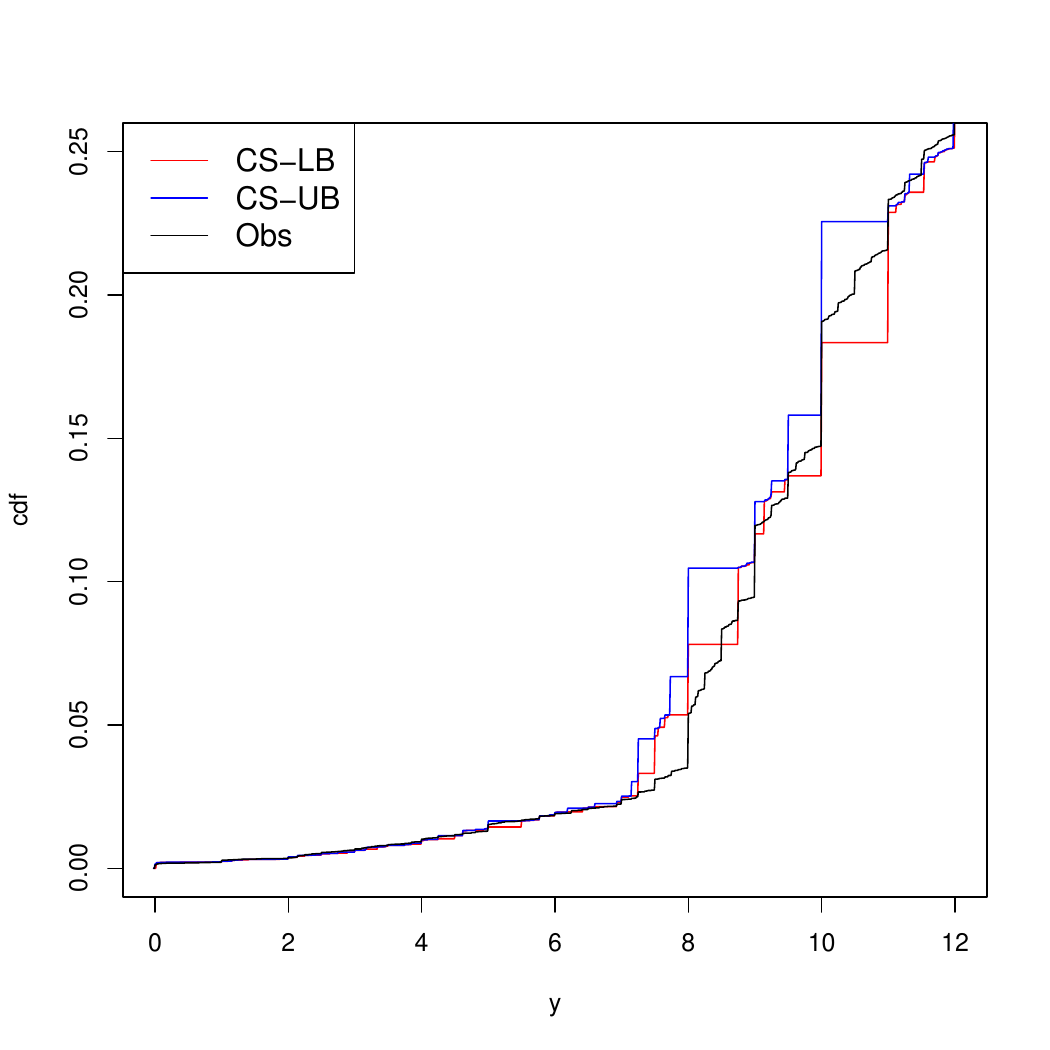}\\
(c) Dist-DiD using 2010 pre-treatment period&(d) Dist-DiD using 2011 pre-treatment period\\
\includegraphics[width=6cm]{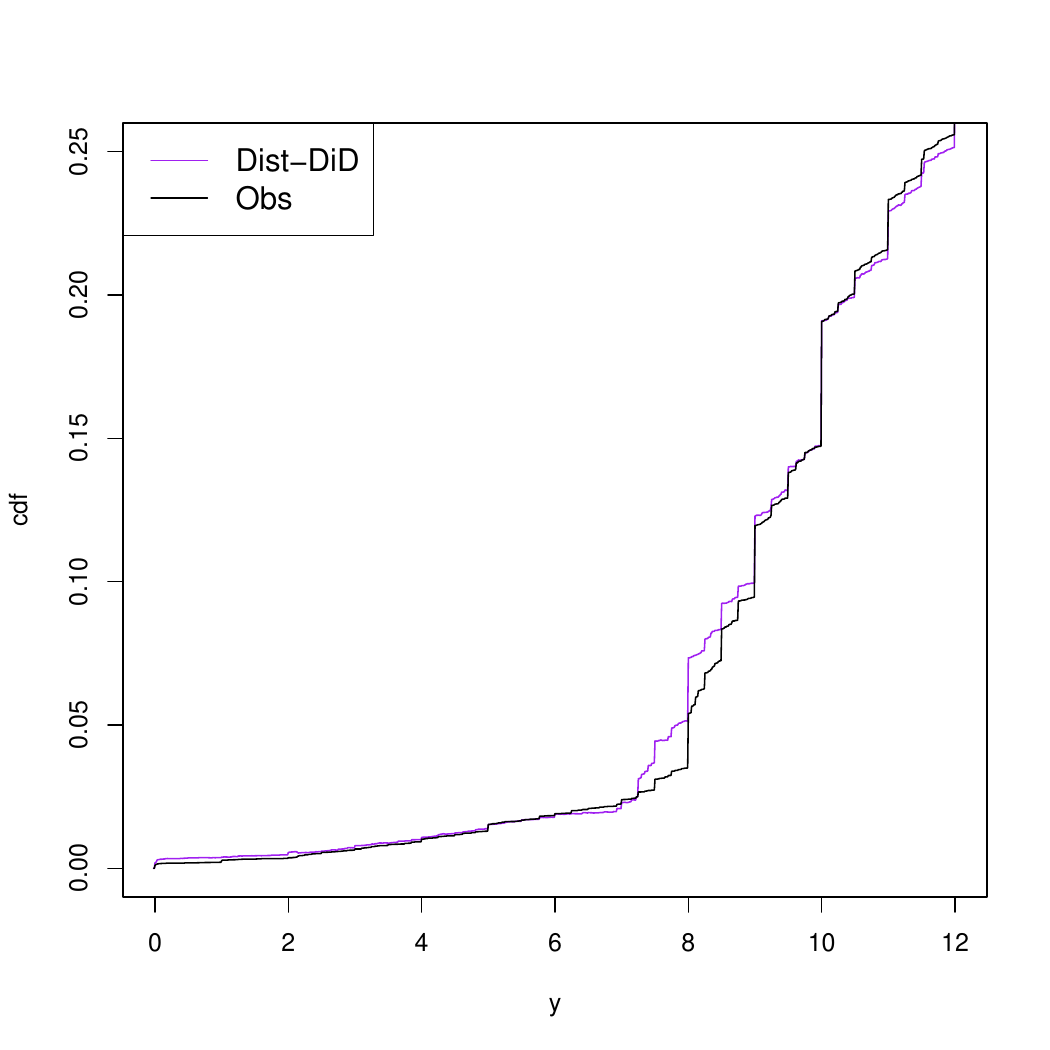}&\includegraphics[width=6cm]{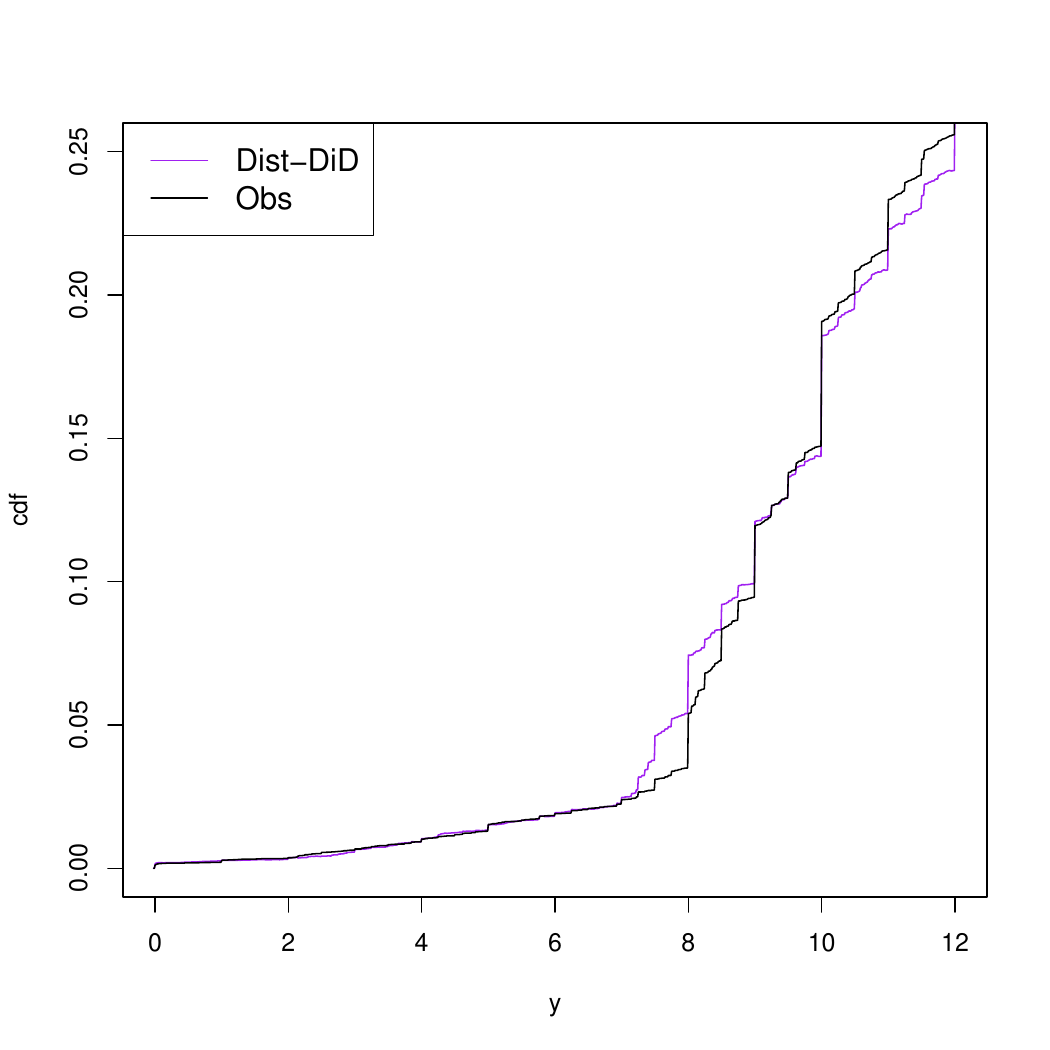}\\
(e) CiC using 2010 pre-treatment period&(f) CiC using 2011 pre-treatment period\\
\includegraphics[width=6cm]{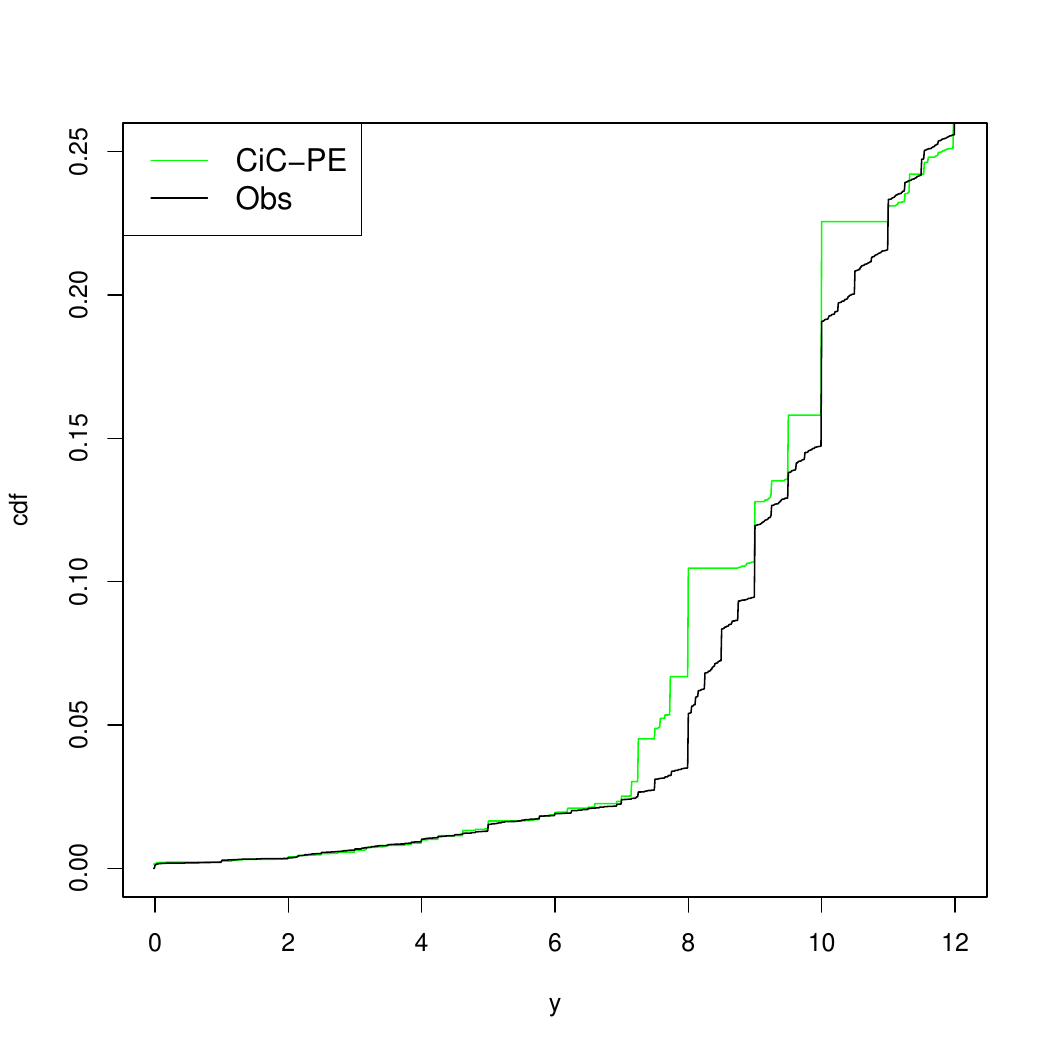}&\includegraphics[width=6cm]{Figures/FY10TCiC0_7.25_bottom25_06232025.pdf}\\
\multicolumn{2}{c}{\parbox{14cm}{\scriptsize{\emph{Notes}: $Obs$ denotes the observed factual $F_{Y_{11}|D=1}$, $CS$-$LB$ and $CS$-$UB$ denote the copula lower and upper bound estimates on the counterfactual distribution, respectively, $Dist$-$DiD$ depicts the distributional DiD estimator, and $CiC$-$PE$ denotes the CiC point estimator. For each point/bounds estimator, we provide estimates using each of the 2010 and 2011 pre-treatment periods. In this figure, we zoom into the lowest quartile of the distribution, see Figure \ref{fig:distribution_multiT0} for plots of the entire distribution. }}}
\end{tabular}}}
\label{fig:distribution_zoombottom_multiT0_subsample}
\end{figure}


\begin{figure} [H]\caption{Observed and Counterfactual Distributions: States with Pre-MW$< \$ 8$, Entire Distribution}
\vspace{0.25cm}{\footnotesize{
\begin{tabular}{cc}
(a) CS Bounds using 2010 pre-treatment period&(b) CS Bounds using 2011 pre-treatment period\\
\includegraphics[width=6cm]{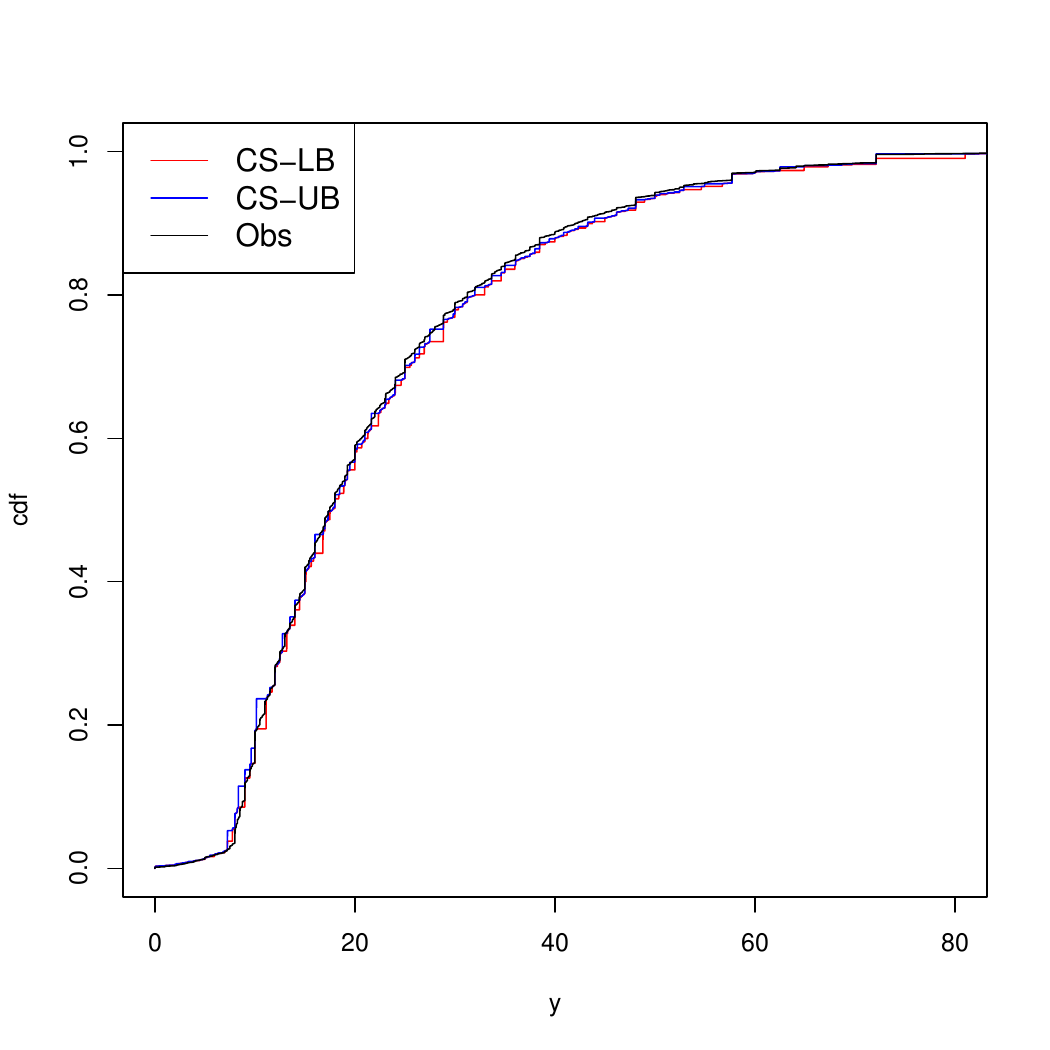}&\includegraphics[width=6cm]{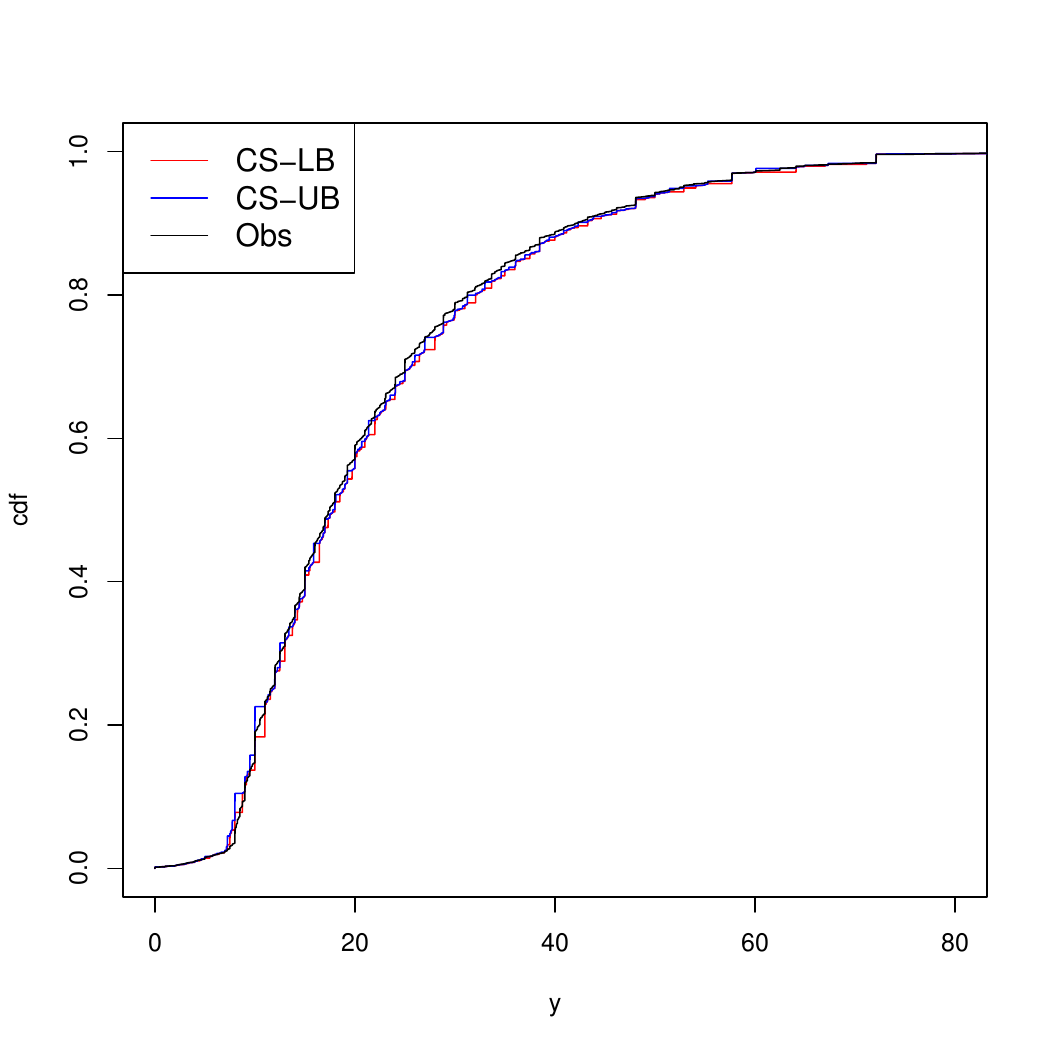}\\
(c) Dist-DiD using 2010 pre-treatment period&(d) Dist-DiD using 2011 pre-treatment period\\
\includegraphics[width=6cm]{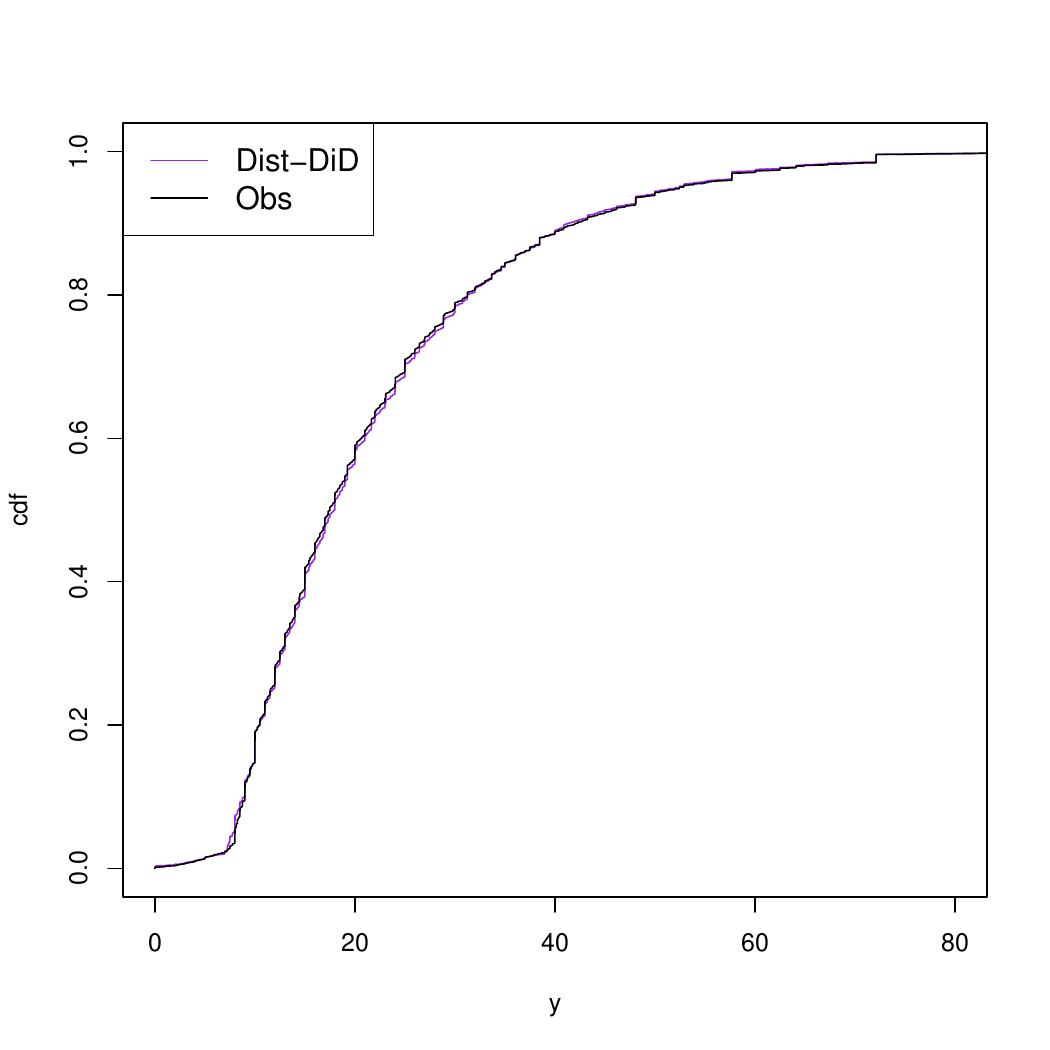}&\includegraphics[width=6cm]{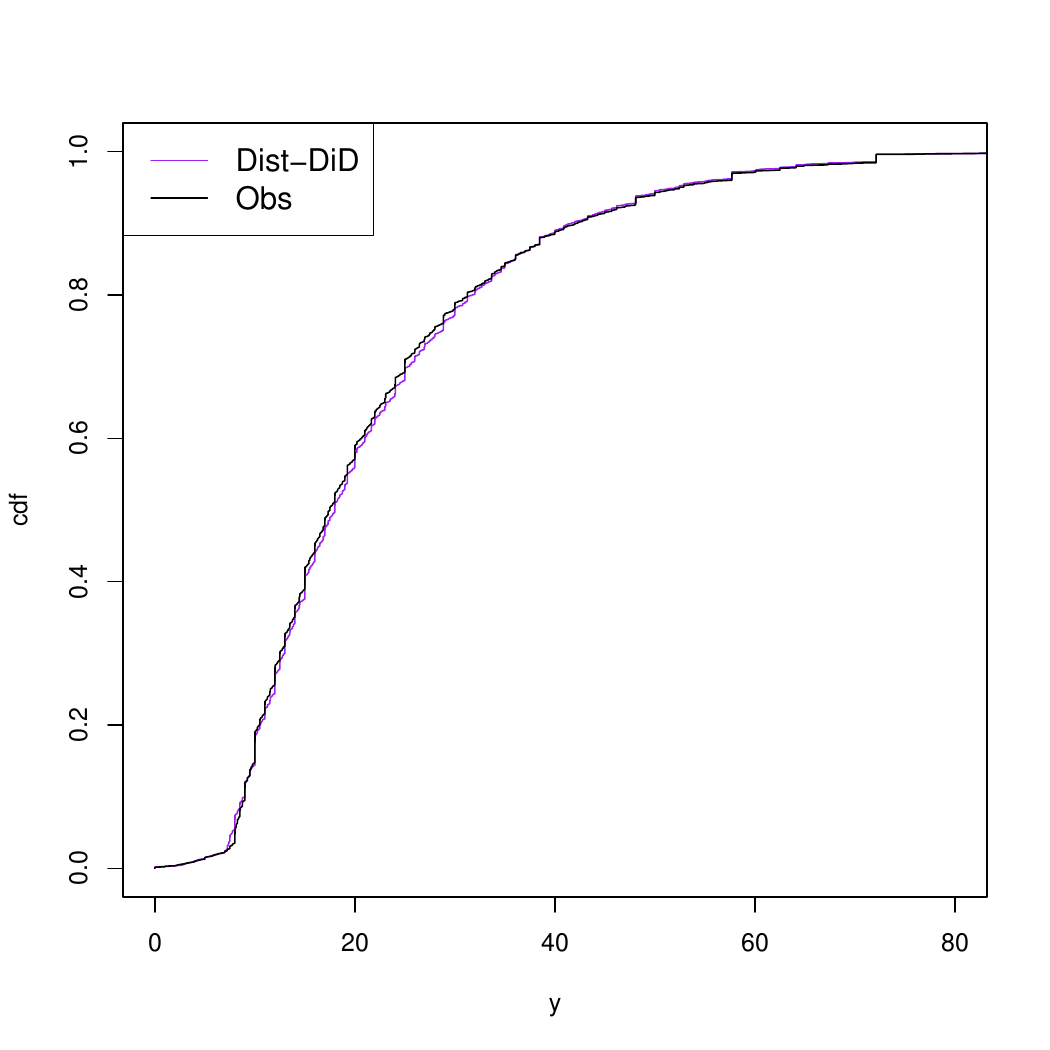}\\
(e) CiC using 2010 pre-treatment period&(f) CiC using 2011 pre-treatment period\\
\includegraphics[width=6cm]{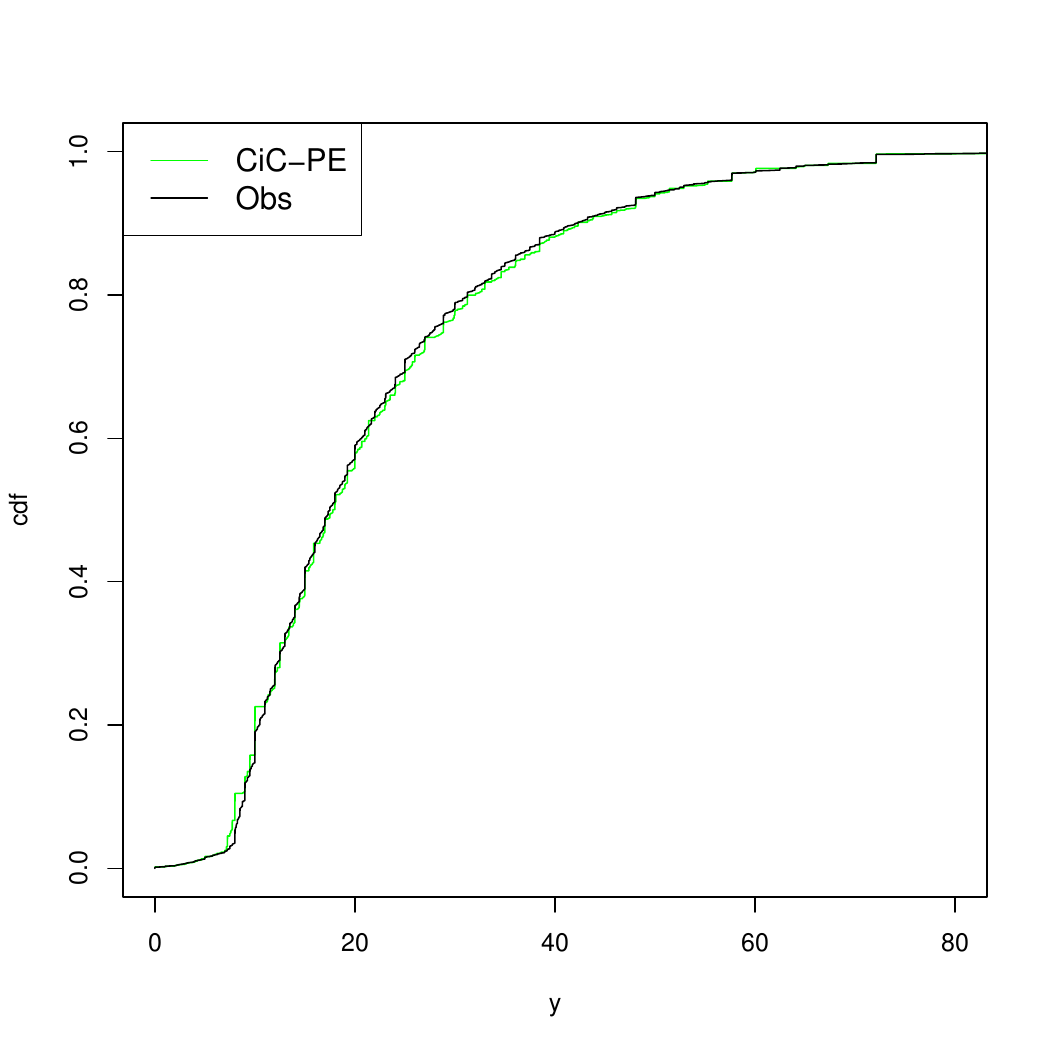}&\includegraphics[width=6cm]{Figures/FY10TCiC0_7.25_06232025.pdf}\\
\multicolumn{2}{c}{\parbox{14cm}{\scriptsize{\emph{Notes}: $Obs$ denotes the observed factual $F_{Y_{11}|D=1}$, $CS$-$LB$ and $CS$-$UB$ denote the copula lower and upper bound estimates on the counterfactual distribution, respectively, $Dist$-$DiD$ depicts the distributional DiD estimator, and $CiC$-$PE$ denotes the CiC point estimator. For each point/bounds estimator, we provide estimates using each of the 2010 and 2011 pre-treatment periods. }}}
\end{tabular}}}
\label{fig:distribution_multiT0_subsample}
\end{figure}

\begin{table}[H]
  \centering
  \caption{Inference on Lower-tail SWTT using 2010 and/or 2011 as pre-treatment periods: States with Pre-MW$<$\$8}
{\footnotesize    \begin{tabular}{lrrrrrr}
\multicolumn{7}{l}{Panel A. CS bounds}\\
          \toprule
          & \multicolumn{6}{c}{95\% CI}  \\
          \midrule
 Pre-period        & \multicolumn{2}{c}{2010}        & \multicolumn{2}{c}{2011} & \multicolumn{2}{c}{  2010\&2011}\\
          \cmidrule(lr){2-3}\cmidrule(lr){4-5}\cmidrule(lr){6-7}
          \cmidrule(lr){2-3}\cmidrule(lr){4-5}\cmidrule(lr){6-7}
  (\$)        &  \multicolumn{1}{c}{LB} & \multicolumn{1}{c}{UB} & \multicolumn{1}{c}{LB} & \multicolumn{1}{c}{UB} & \multicolumn{1}{c}{LB} & \multicolumn{1}{c}{UB} \\

\midrule
    $u=0.01$\\
    ATT(u) & -0.04 & 0.99  & -0.54 & 0.37  & -0.01 & 0.40 \\
    Gini SWTT(u) & 0.01  & 0.92  & -0.47 & 0.37  & 0.04  & 0.40\\
       [0.25em]
    $u=0.025$\\
    ATT(u) & -0.29 & 0.74  & -0.49 & 0.45  & -0.21 & 0.38 \\
    Gini SWTT(u) & -0.20 & 0.81  & -0.52 & 0.41  & -0.13 & 0.37\\
       [0.25em]
    $u=0.05$\\
    ATT(u) & -0.02 & 0.68  & -0.08 & 0.52  & 0.09  & 0.47 \\
    Gini SWTT(u) & -0.13 & 0.70  & -0.30 & 0.46  & -0.05 & 0.40\\
        [0.25em]
    $u=0.10$\\
    ATT(u) & 0.03  & 0.57  & 0.00  & 0.58  & 0.13  & 0.45 \\
    Gini SWTT(u)& 0.00  & 0.62  & -0.08 & 0.50  & 0.09  & 0.42\\
        [0.25em]
    $u=0.25$\\
    ATT(u) & -0.10 & 0.52  & -0.15 & 0.43  & -0.01 & 0.35 \\
    Gini SWTT(u) & -0.02 & 0.53  & -0.07 & 0.47  & 0.07  & 0.37\\
   \midrule

    \end{tabular}%
    }\\
 \vspace{0.5cm}
 {\footnotesize{  
\begin{tabular}{lcrrcrrcrrcrr}

\multicolumn{9}{l}{Panel B. Distributional DiD and CiC}\\
     \toprule
 & \multicolumn{4}{c}{Dist DiD: 95\% CI}        & \multicolumn{4}{c}{CiC: 95\% CI}  \\
 \midrule
Pre-period    & \multicolumn{2}{c}{2010} 
   &\multicolumn{2}{c}{2011}
           & \multicolumn{2}{c}{2010} 
   &\multicolumn{2}{c}{2011}\\
             \cmidrule(lr){2-3}\cmidrule(lr){4-5}\cmidrule(lr){6-7}\cmidrule(lr){8-9}
  (\$)    &\multicolumn{1}{c}{LB} & \multicolumn{1}{c}{UB} &   \multicolumn{1}{c}{LB} & \multicolumn{1}{c}{UB} &    \multicolumn{1}{l}{LB} & \multicolumn{1}{l}{UB} &   \multicolumn{1}{l}{LB} & \multicolumn{1}{l}{UB}   \\
\midrule
    $u=0.01$\\
    ATT(u) & -0.03 & 0.93  & -0.59 & 0.39  & 0.06  & 0.98  & -0.53 & 0.42 \\
    Gini SWTT(u) & 0.08  & 0.89  & -0.56 & 0.40  & 0.09  & 0.92  & -0.44 & 0.40\\
       [0.25em]
        $u=0.025$\\
    ATT(u) & -0.35 & 0.64  & -0.48 & 0.46  & -0.20 & 0.74  & -0.43 & 0.50 \\
    Gini SWTT(u) & -0.18 & 0.78  & -0.51 & 0.43  & -0.10 & 0.81  & -0.47 & 0.45\\
       [0.25em]
        $u=0.05$\\
    ATT(u) & -0.04 & 0.54  & -0.08 & 0.47  & 0.17  & 0.68  & 0.02  & 0.55 \\
    Gini SWTT(u) & -0.16 & 0.61  & -0.30 & 0.44  & -0.02 & 0.70  & -0.23 & 0.50\\
       [0.25em]
          $u=0.10$\\
    ATT(u) & 0.04  & 0.42  & 0.03  & 0.38  & 0.23  & 0.57  & 0.26  & 0.59 \\
    Gini SWTT(u) & -0.02 & 0.48  & -0.07 & 0.41  & 0.14  & 0.61  & 0.05  & 0.53\\
       [0.25em]
              $u=0.25$\\
    ATT(u) & -0.06 & 0.22  & -0.12 & 0.15  & 0.23  & 0.51  & 0.20  & 0.43 \\
    Gini SWTT(u) & -0.01 & 0.31  & -0.04 & 0.26  & 0.23  & 0.52  & 0.20  & 0.47 \\
    
    \midrule
        \multicolumn{9}{c}{\parbox{10cm}{\scriptsize{\emph{Notes}: The definitions of the SWTT bounds/point estimators are provided in \eqref{eq:swtt-estimator}--\eqref{eq:swtt-estimator_pe}. For the CS bounds, we report 95\% confidence intervals on the identified set. For the point estimators, we report 95\% confidence intervals on the SWTT parameter. All confidence intervals use standard normal critical values and nonparameteric bootstrap standard errors using 500 bootstrap replications.}}}\\
    \end{tabular}%
    }}
  \label{tab:lowertail_swtt_multiT0_subsample}%
\end{table}%

\begin{table}[H]
\caption{Parameters from \citet{Cengizetal2019}}
 {\footnotesize   \begin{tabular}{lrrrrrrrrrr}
    \toprule
  \parbox{1.25cm}{Pre-treatment period} &\multicolumn{4}{c}{2010}&\multicolumn{4}{c}{2011}&\multicolumn{2}{c}{ 2010 \& 2011}\\
  \cmidrule(lr){2-5}\cmidrule(lr){6-9}\cmidrule(lr){10-11}
   & \multicolumn{2}{c}{CS Bounds} & \multicolumn{1}{c}{DistDiD} & \multicolumn{1}{c}{CiC} & \multicolumn{2}{c}{CS Bounds}  & \multicolumn{1}{l}{DistDiD} & \multicolumn{1}{l}{CiC} & \multicolumn{2}{c}{CS Bounds}\\
   \cmidrule(lr){2-3}\cmidrule(lr){6-7}\cmidrule(lr){10-11}
    &\multicolumn{1}{c}{LB}&\multicolumn{1}{c}{UB}&&&\multicolumn{1}{c}{LB}&\multicolumn{1}{c}{UB}&&&\multicolumn{1}{c}{LB}&\multicolumn{1}{c}{UB}\\
    \midrule
$\Delta b$   & -3.2\% & -0.1\% & -0.8\% & -3.0\% & -2.2\% & 0.6\% & -0.8\% & -2.1\% & -2.2\% & -0.2\% \\
$\Delta a$  & -0.1\% & 7.0\% & 1.3\% & 2.8\% & -0.3\% & 2.6\% & 1.9\% & 2.3\% & 0.4\% & 2.6\% \\
$\Delta e$& -0.4\% & 4.0\% & 0.5\% & -0.2\% & 0.1\% & 0.5\% & 1.1\% & 0.3\% & 0.1\% & 0.5\% \\
\midrule
\multicolumn{11}{c}{\parbox{0.8\textwidth}{\footnotesize \emph{Notes}: We compute the estimates of $\Delta b$ and $\Delta a$ using the sample analogues of Eq. \eqref{eq:deltab} and \eqref{eq:deltaa}, respectively, with $MW=8.5$ and $\bar{W}=11$. }}\\
\end{tabular}%
}\label{tab:deltas_multiT0_subsample}
\end{table}

\section{Additional numerical examples}\label{App:num_examples}
In this section, we illustrate the wide applicability of the CS identification approach using several numerical examples of outcomes with discrete and mixed distributions. We consider four different marginal distributions presented in Table \ref{tab:num_examples}, including the Poisson distribution (Example I), left- and right-censoring (Examples II-III) and a bunching example (Example IV). While Example I falls under the AI2006 identification results, the remaining examples are not covered by their approach. 

\begin{table}[H]\caption{Examples of Marginal Distributions of $Y_{t0}$}
            \centering
{\footnotesize{\begin{tabular}{ll}
    \toprule
I. Poisson&~~$F_{Y_{t0}}(y)=\Pi_t(y)$, where $\Pi_t(\cdot)$ is the Poisson cdf with mean $\lambda_t$.\\
[0.5em]
II. Left-censoring&$\left\{\text{\parbox{12cm}{$F_{Y_{t0}}(y)=\left\{\begin{array}{ll}0&\text{if }y < c_t\\
\Lambda_t(y)&\text{if }y\geq c_t\end{array}\right.,$ \\ 
where $\Lambda_t(\cdot)$ is the $\chi^2$ cdf with $k_t$ degrees of freedom.}}\right.$\\
[1.5em]
III. Right-censoring&$\left\{\text{\parbox{12cm}{$F_{Y_{t0}}(y)=\left\{\begin{array}{ll}\Lambda_t(y)&\text{if }y < c_t\\
1&\text{if }y\geq c_t\end{array}\right.$,\\ where $\Lambda_t(\cdot)$ is the $\chi^2$ cdf with $k_t$ degrees of freedom.}}\right.$\\
[1.5em]
IV. Bunching&$\left\{\text{\parbox{12cm}{$F_{Y_{t0}}(y)=\left\{\begin{array}{ll}\Phi_t(y)&\text{if }y\not\in [c_t,w_t)\\\Phi_t(c_t)+b_t(\Phi_t(w_t)-\Phi_t(c_t))&\text{if }y=c_t\\\Phi_t(c_t)+b_t(\Phi_t(w_t)-\Phi_t(c_t))+(1-b_t)(\Phi_t(y)-\Phi_t(c_t)&\text{if }y\in(c_t,w_t)\end{array}\right.$\\ where $\Phi_t(.)$ is the standard normal cdf with mean $\mu_t$ and standard deviation $\sigma_t$.}}\right.$\\
[0.75cm]
\bottomrule
\end{tabular}}}
\label{tab:num_examples}
    \end{table}
    
Given marginal distributions of $Y_{00}$ and $Y_{10}$, we can generate conditional potential outcome distributions that satisfy the copula stability condition by the following, for $t=0,1$,
\begin{align}F_{Y_{t0}|D=0}(y)&=\frac{1}{q}C_{Y_0,D}(F_{Y_{t0}}(y),q),\label{eq:cs_FYt0C}\\
F_{Y_{t0}|D=1}(y)&=\frac{1}{p}\left(F_{Y_{t0}}(y)-C_{Y_0,D}(F_{Y_{t0}}(y),q)\right).\label{eq:cs_FYt0T}\end{align}
We set $C_{Y_0,D}(u,q)=(max(u^{-\theta}+q^{-\theta}-1,0))^{-1/\theta}$. In the following examples, we let $\theta=1$ to fulfill the strict monotonicity condition imposed on the horizontal copula for $u\in[0,1]$. Note that all parameters of the marginal distributions we consider are allowed to vary across time in an arbitrary manner.  

Figures \ref{fig:poisson}-\ref{fig:bunching} present the numerical examples. Each figure presents a plot of each of the observed distribution used in the evaluation of the CS bounds ($F_{Y_{0}|D=0}$, $F_{Y_{1}|D=0}$ and $F_{Y_{0}|D=1}$) in Panels A-C. Panel D of each figure presents the counterfactual distribution for the treatment group ($F_{Y_{10}|D=1}$) together with the CS bounds labeled as $CF$ and $LB$/$UB$, respectively. 

Figure \ref{fig:poisson} illustrates our bounds for the Poisson example with $\lambda_0=1$ and $\lambda_1=3$. Since the CiC bounds proposed in AI2006 can be applied, we compute them and compare them to the CS bounds proposed here. In this numerical example, both bounding approaches coincide as illustrated in Panel D of Figure \ref{fig:poisson}. 

Next, we examine mixed outcome distributions that fall outside the scope of the AI2006 identification results. Figures \ref{fig:left_censoring_1}-\ref{fig:left_censoring_2} provide two different parametrizations of the left-censoring example (Example II). In the first case (Figure \ref{fig:left_censoring_1}), $Ran F_{Y_{10}|D=1}\subset Ran F_{Y_0|D=1}$ and, as a result, the counterfactual distribution is point-identified. In the second case (Figure \ref{fig:left_censoring_2}), $Ran F_{Y_{10}|D=1}\not\subseteq Ran F_{Y_0|D=1}$, and we therefore only attain partial identification of the counterfactual distribution. Figure \ref{fig:right_censoring} illustrates the CS bounds for a right-censoring example (Example III), where the censoring cutoff as well as the degrees of freedom of the $\chi^2$ distribution vary across time. Finally, we consider a bunching example (Example IV), where the bunching cutoff ($c_t$), the width of the bunching window ($w_t-c_t$) and the bunching probability ($b_t$) are time-varying. One notable feature of the bunching example is that the potential outcome distributions are strictly increasing, but discontinuous. Panel D of Figure \ref{fig:bunching} shows that the CS bounds in this bunching example cover the counterfactual distribution. Overall, for these mixed outcome distributions, our numerical analysis illustrates that point-identification of the counterfactual distribution is possible on the intersection of the range of $F_{Y_{10|D=1}}$ and $F_{Y_0|D=1}$, whereas only set-identification is possible outside this intersection. 

Finally, it is important to discuss how the AI2006 CiC bounds would perform in the context of the mixed-outcome examples we consider. In several of these examples, the two quantiles used in the upper and lower bound in the AI2006 CiC bounds equal each other, specifically $Q_{Y_{0}|D=0}^{\mathbb{Y}_{0|0},+}(u)=Q_{Y_{0}|D=0}^{\mathbb{Y}_{0|0},-}(u)$ for $u\in(0,1)$ (e.g. Examples III and IV). It follows that the AI2006 CiC lower bound would equal its upper bound, and the CiC bounds would not include the counterfactual distribution. As AI2006 point out, the bound on quantiles that they exploit in their partial identification result for discrete outcomes is not valid for outcomes with mixed distributions.

\begin{singlespace}
\begin{figure}[H]\caption{Numerical Example I: Poisson with $\lambda_0=1$, $\lambda_1=3$}\label{fig:poisson}\centering
\footnotesize{\begin{tabular}{cccc}
\\
Panel A. $F_{Y_{00}|D=0}$ &Panel B. $F_{Y_{10}|D=0}$&
Panel C. $F_{Y_{00}|D=1}$ & Panel D. CS Bounds on $F_{Y_{10}|D=1}$\\
\includegraphics[width=4cm]{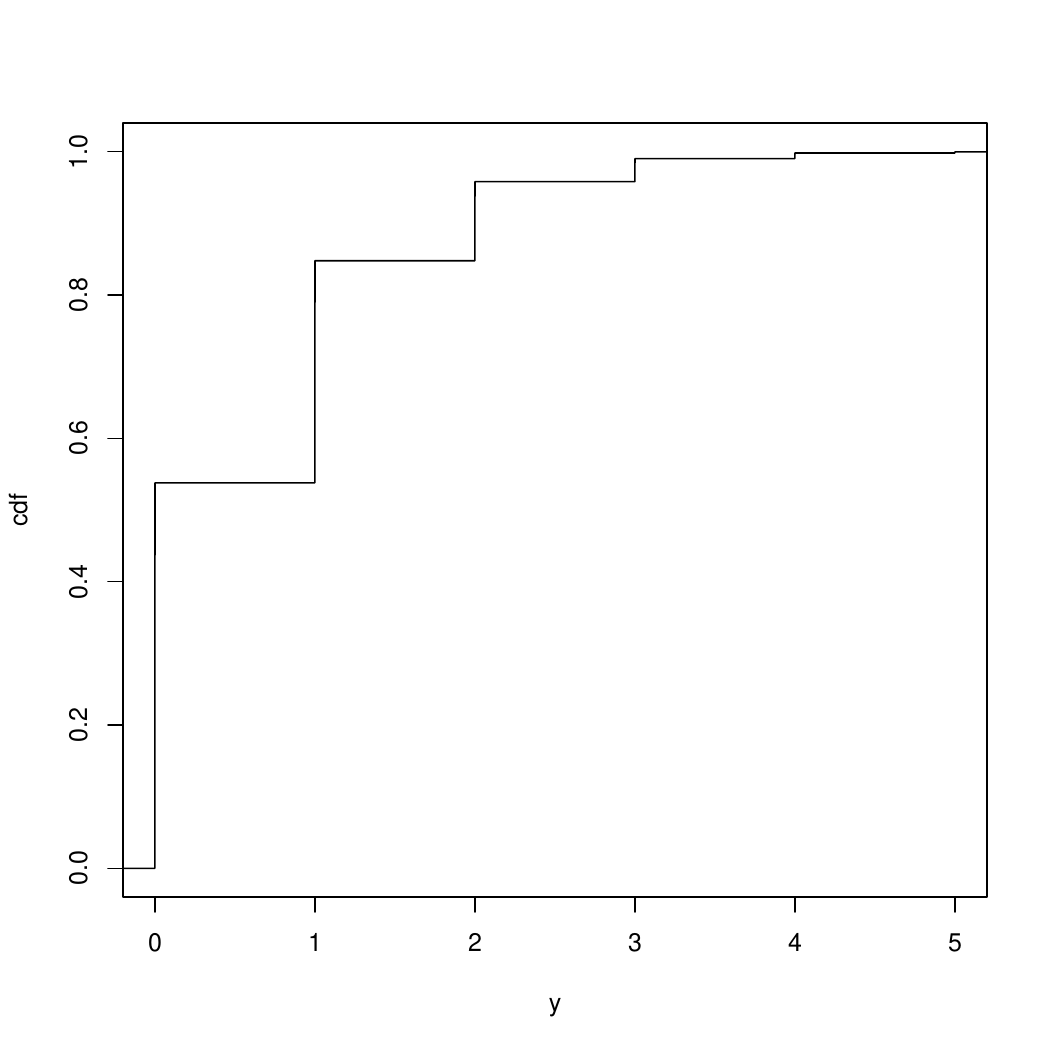}&\includegraphics[width=4cm]{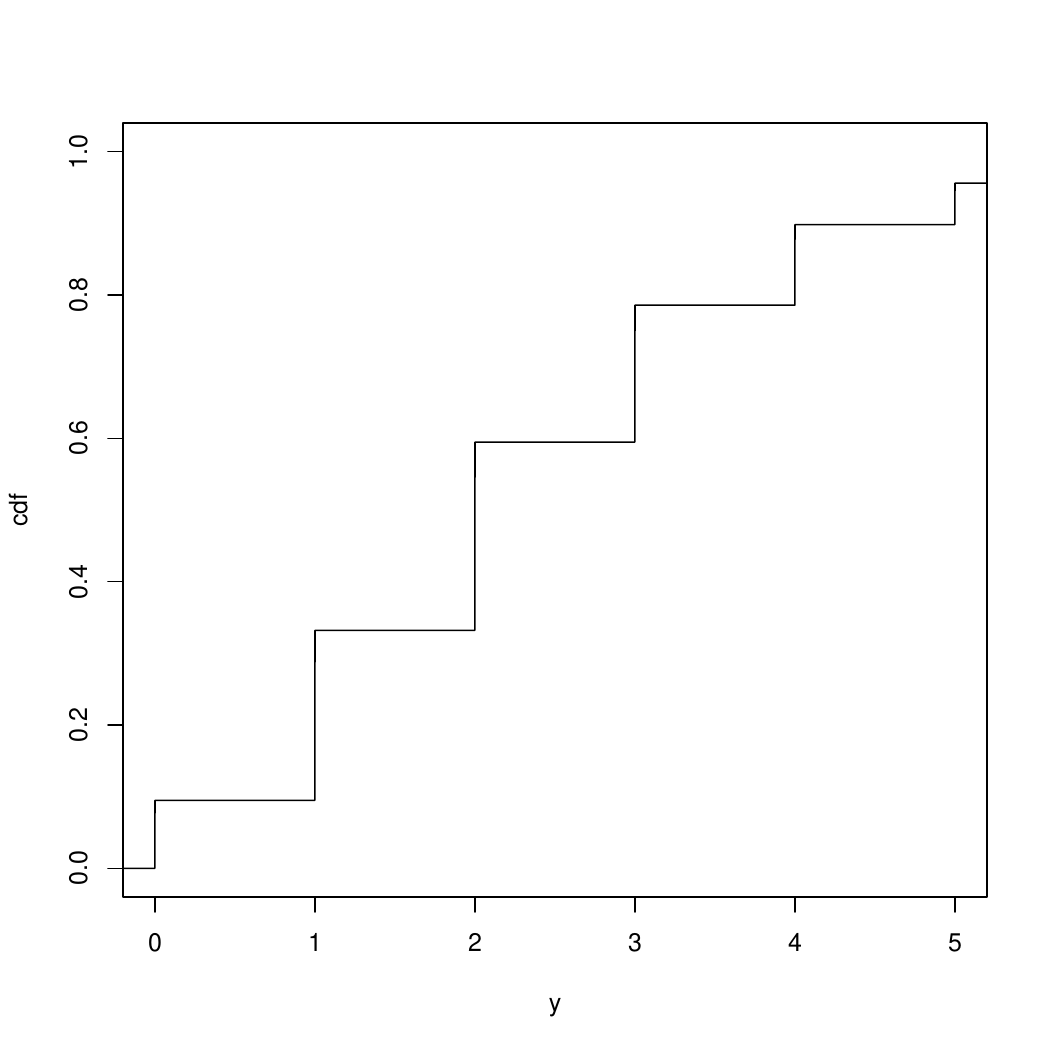}&
\includegraphics[width=4cm]{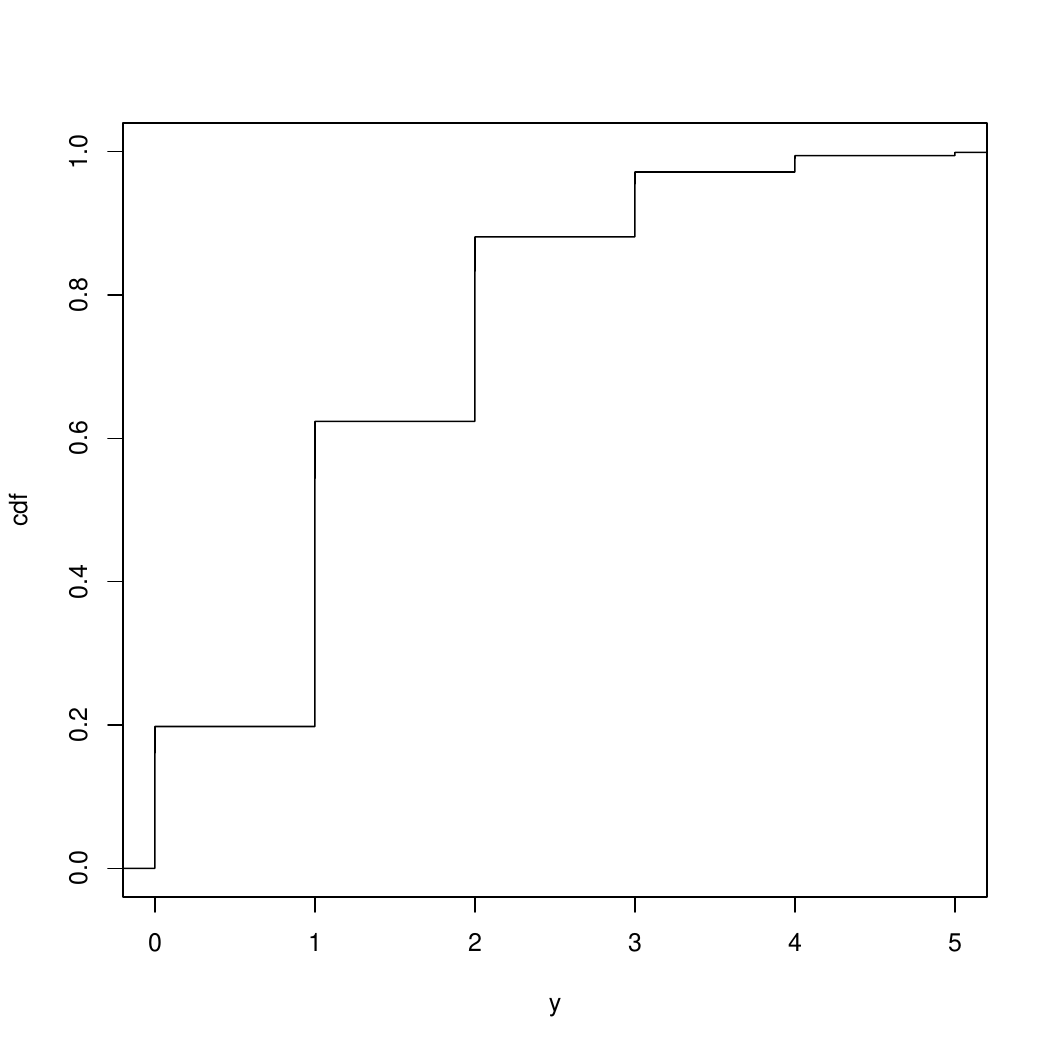}&\includegraphics[width=4cm]{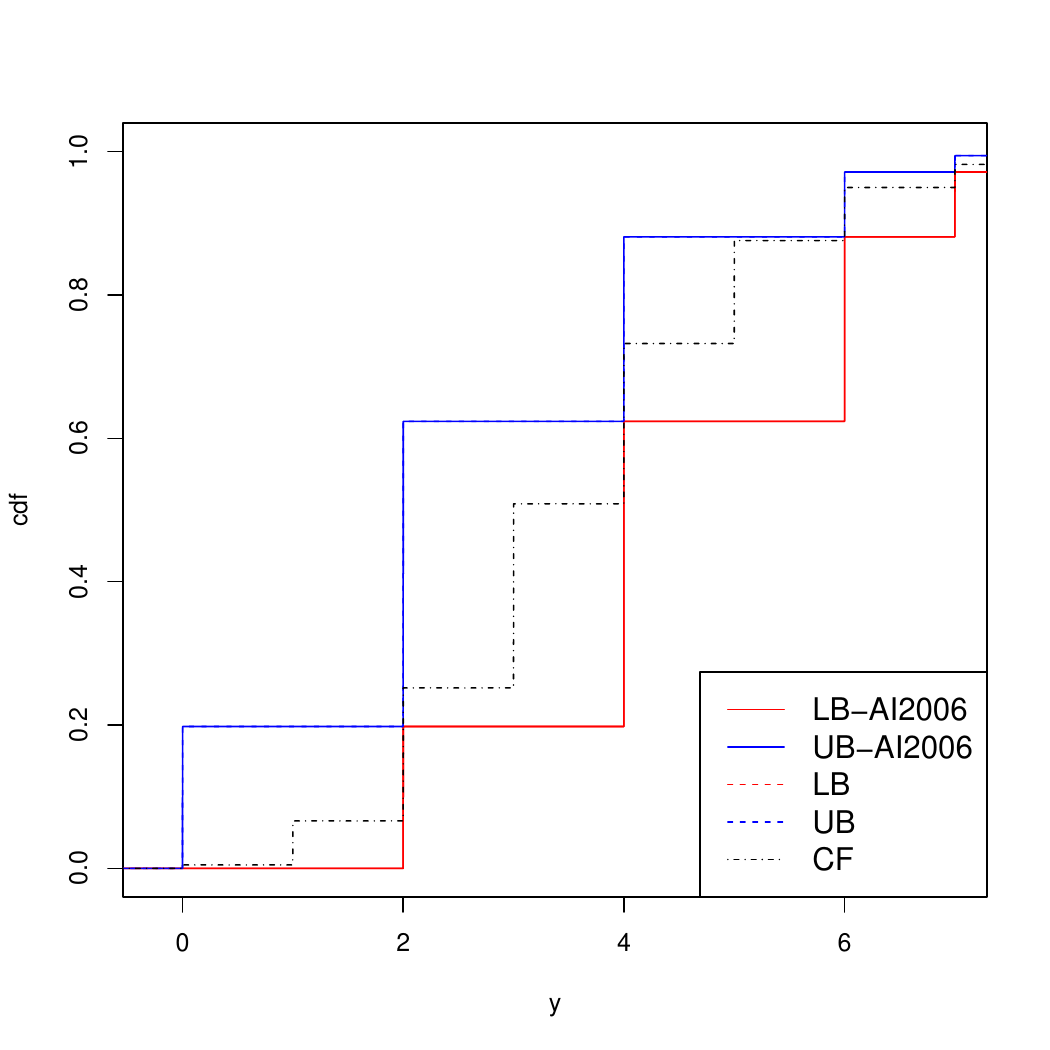}\\
\multicolumn{4}{c}{\parbox{\textwidth}{\scriptsize{\emph{Notes}: In Panel D, $CF$ denotes the counterfactual distribution for the treatment group ($F_{Y_{10}|D=1}$), $LB$-$AI2006$ ($UB$-$AI2006$) denotes the CiC lower (upper) bound from AI2006, and $LB$ ($UB$) denote the CS lower (upper) bound proposed here. }}}
\end{tabular}}
\end{figure}

\begin{figure}[htbp]\caption{Numerical Example II: Left-censoring, $c_0=c_1=5$,  $k_0=5$, $k_1=3$}\label{fig:left_censoring_1}
\footnotesize{\begin{tabular}{cccc}
\\
Panel A. $F_{Y_{00}|D=0}$ &Panel B. $F_{Y_{10}|D=0}$&Panel C. $F_{Y_{00}|D=1}$ & Panel D. CS Bounds on $F_{Y_{10}|D=1}$\\

\includegraphics[width=4cm]{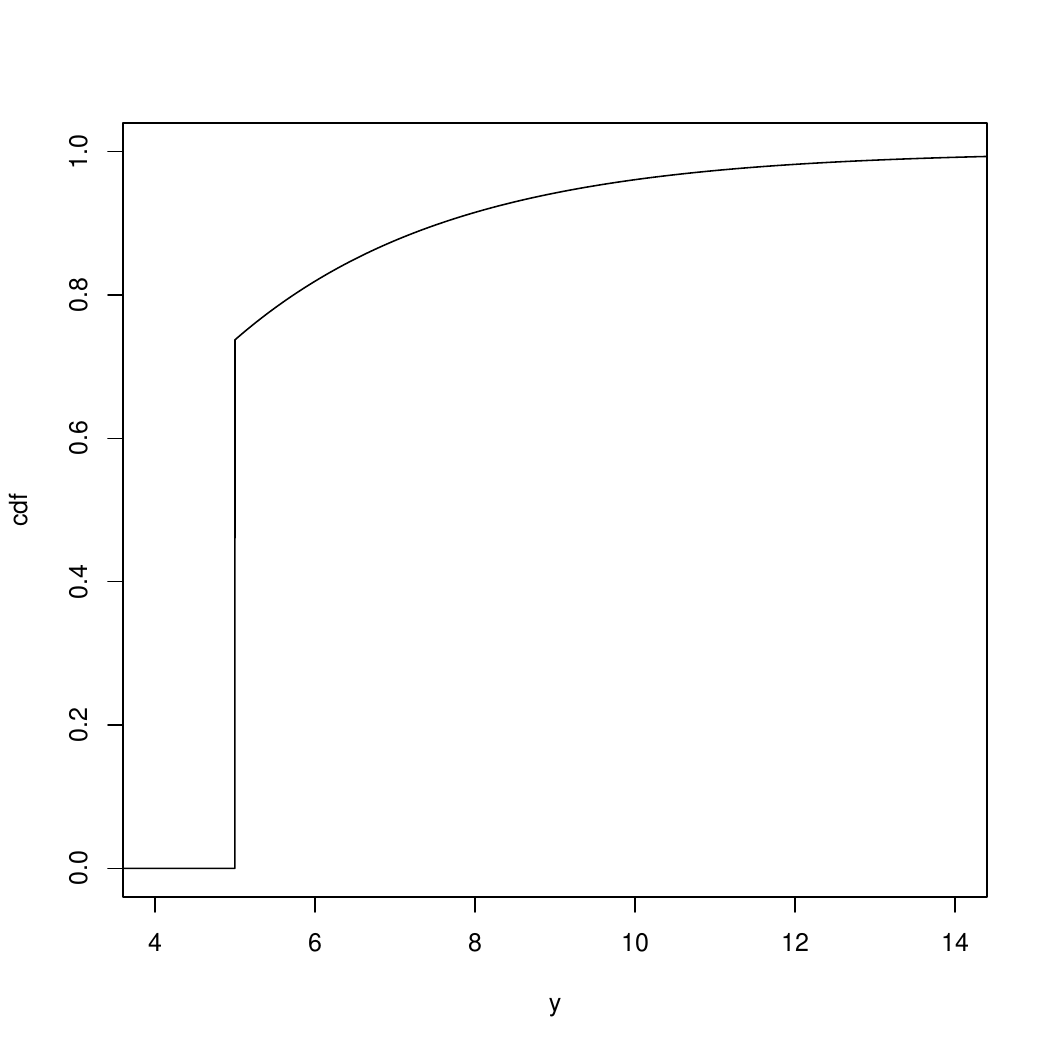}&\includegraphics[width=4cm]{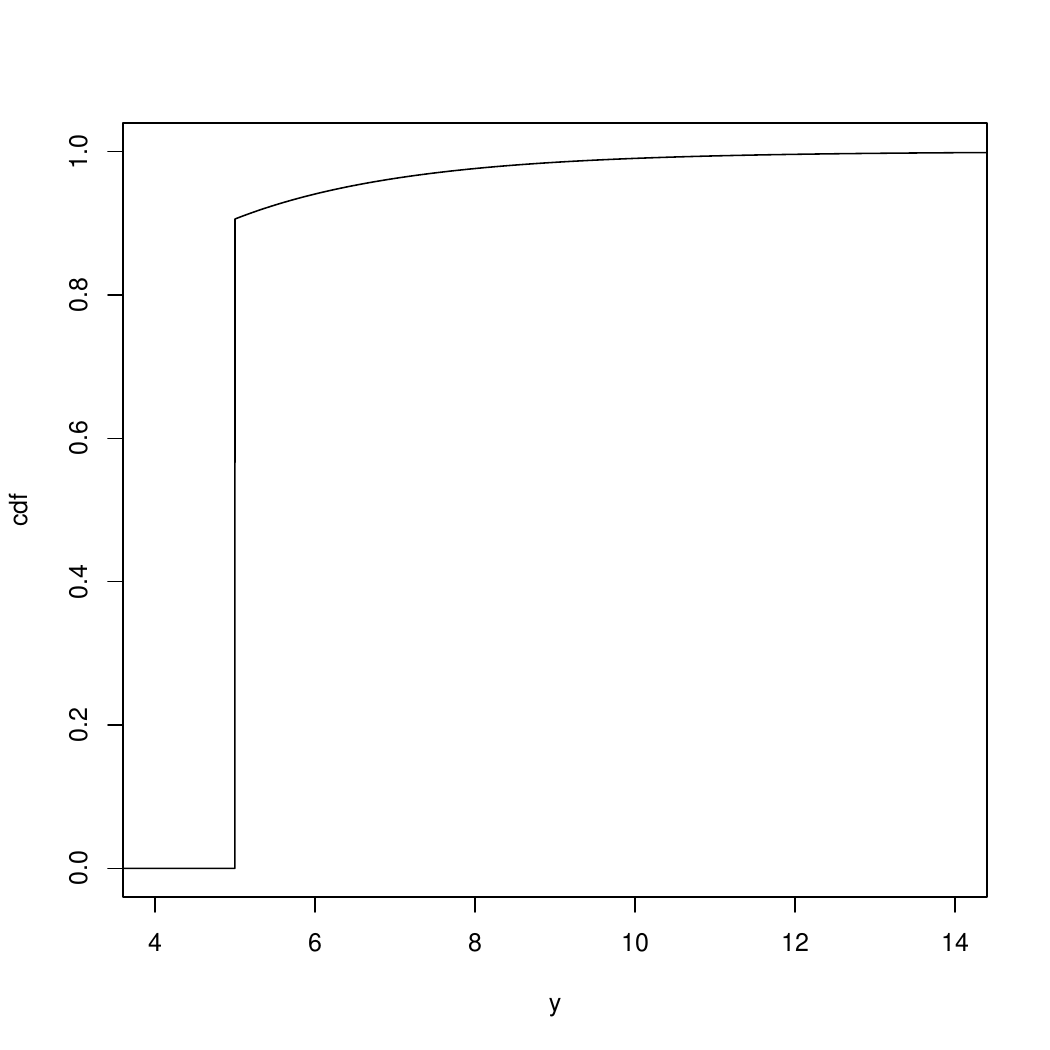}&
\includegraphics[width=4cm]{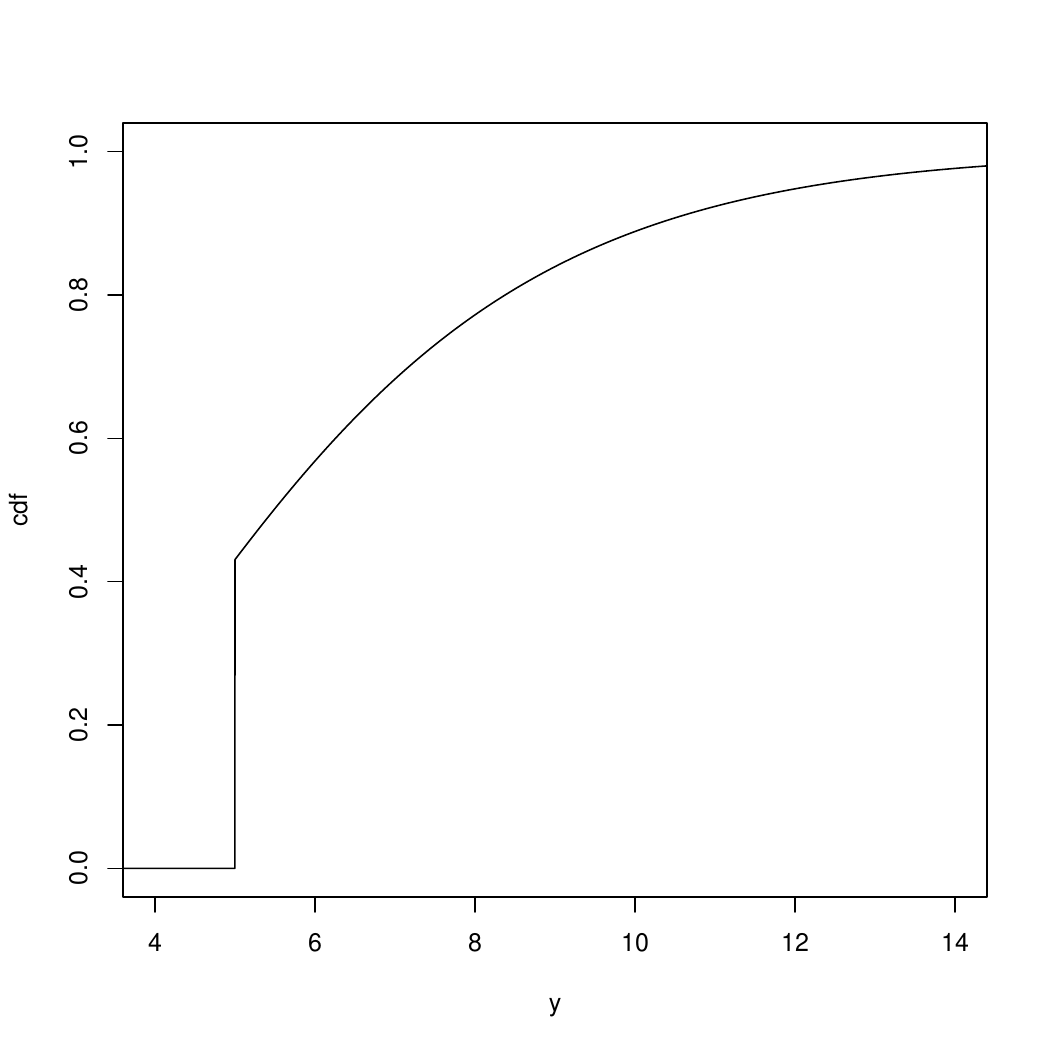}&\includegraphics[width=4cm]{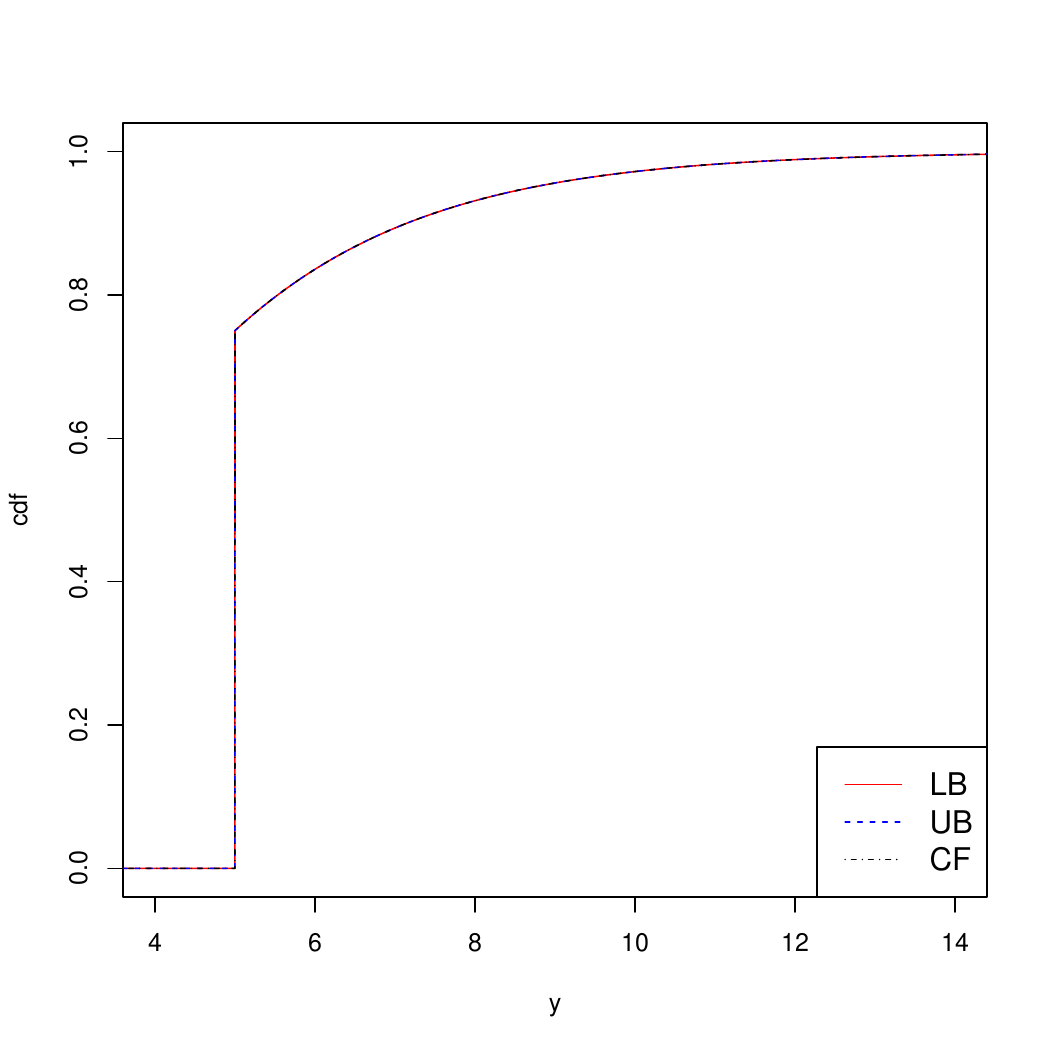}
\\
\end{tabular}}
\end{figure}

\begin{figure}[H]\caption{Numerical Example II: Left-censoring, $c_0=c_1=5$, $k_0=3$, $k_1=5$}
\footnotesize{\begin{tabular}{cccc}
\\
Panel A. $F_{Y_{00}|D=0}$ &Panel B. $F_{Y_{10}|D=0}$&Panel C. $F_{Y_{00}|D=1}$ & Panel D. CS Bounds on $F_{Y_{10}|D=1}$\\
\includegraphics[width=4cm]{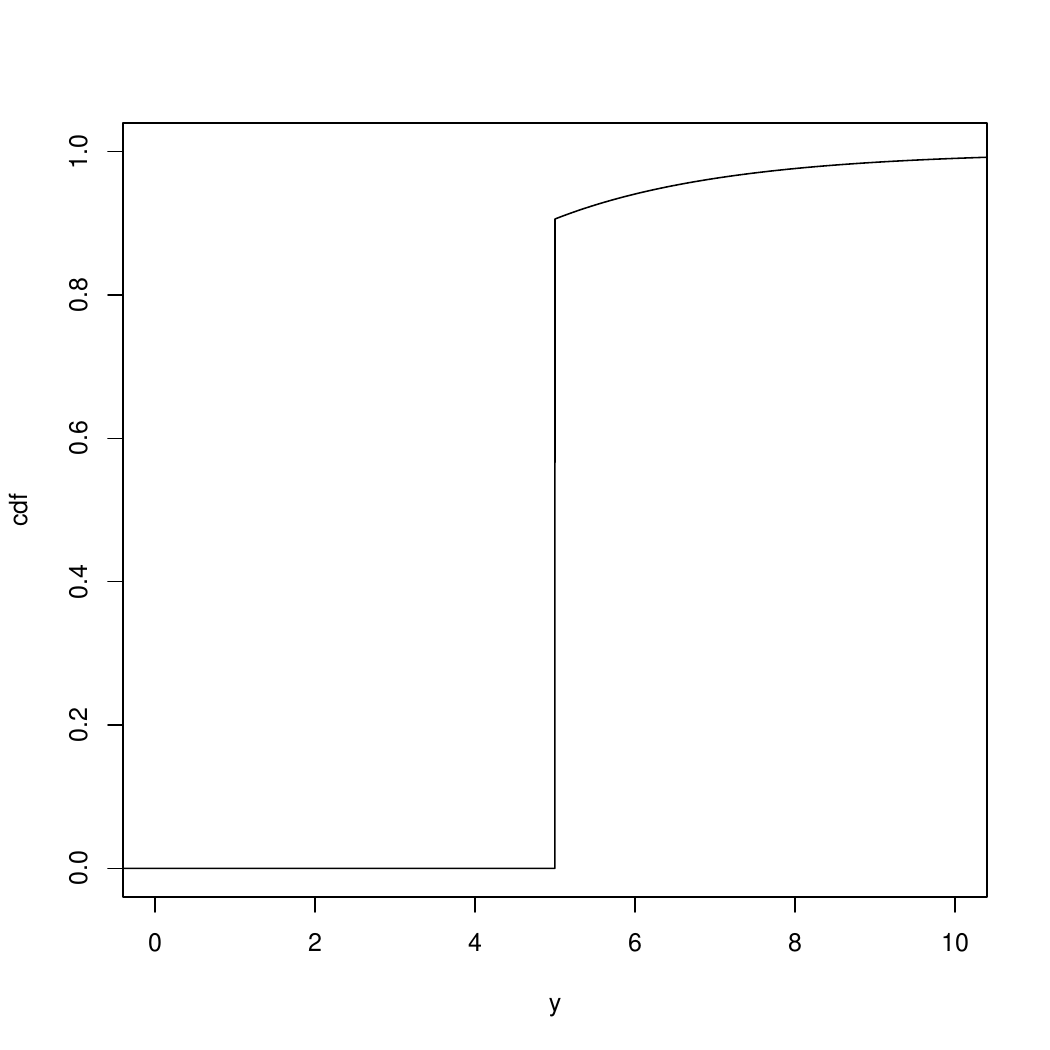}&\includegraphics[width=4cm]{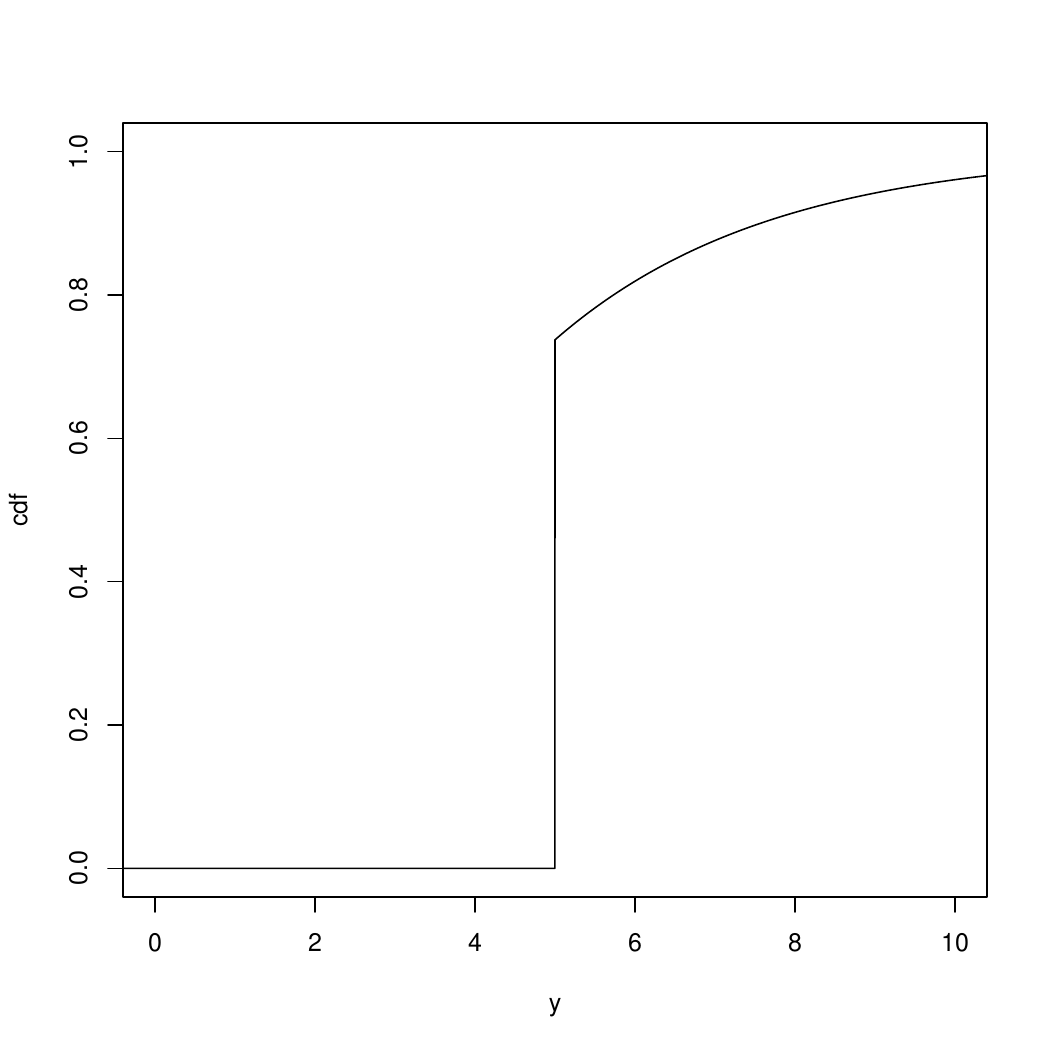}&
\includegraphics[width=4cm]{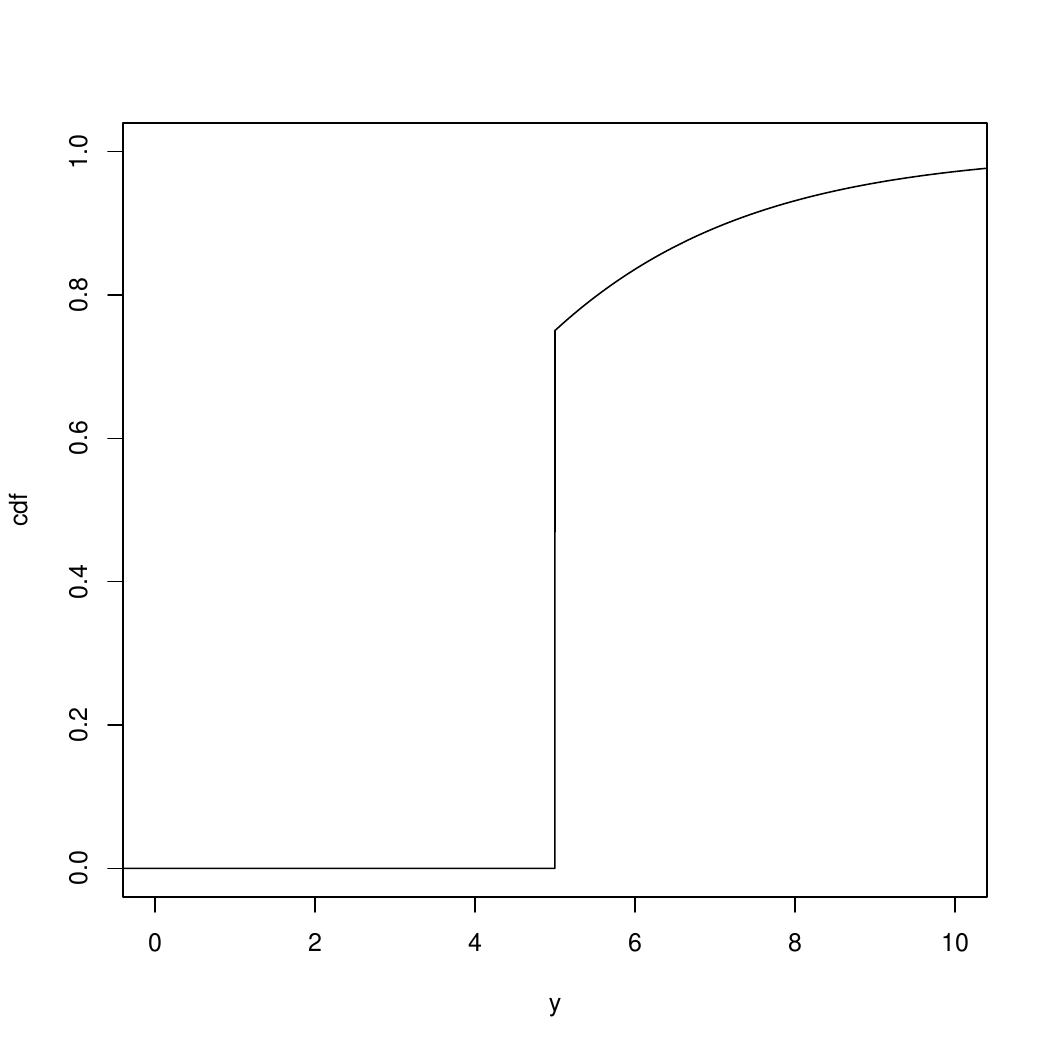}&\includegraphics[width=4cm]{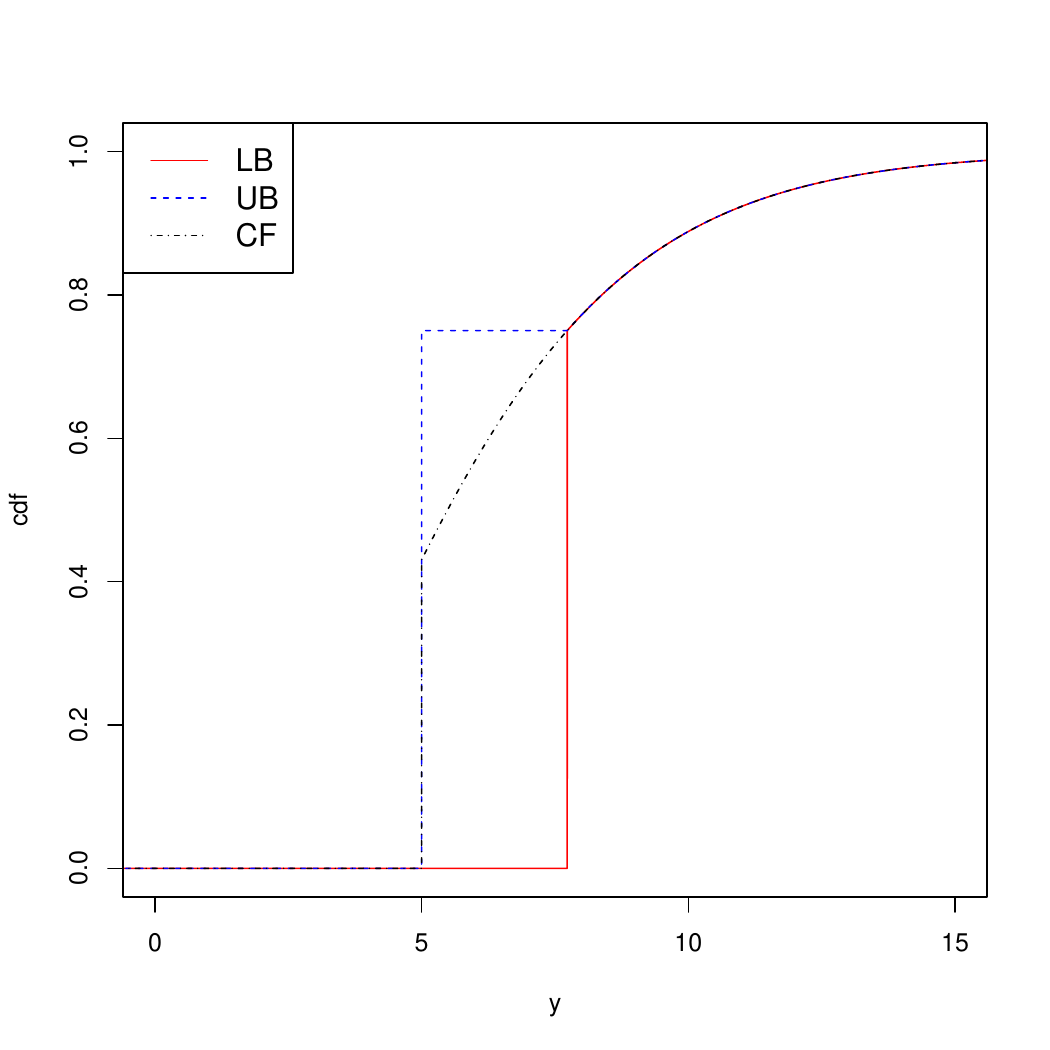}\\

\end{tabular}}
\end{figure}

\begin{figure}[htbp]\caption{Numerical Example III: Right-censoring, $c_0=5$, $c_1=10$, $k_0=3$, $k_1=5$}\label{fig:right_censoring}\footnotesize{
\begin{tabular}{cccc}
\\
Panel A. $F_{Y_{00}|D=0}$ &Panel B. $F_{Y_{10}|D=0}$&Panel C. $F_{Y_{00}|D=1}$ & Panel D. CS Bounds on $F_{Y_{10}|D=1}$\\
\includegraphics[width=4cm]{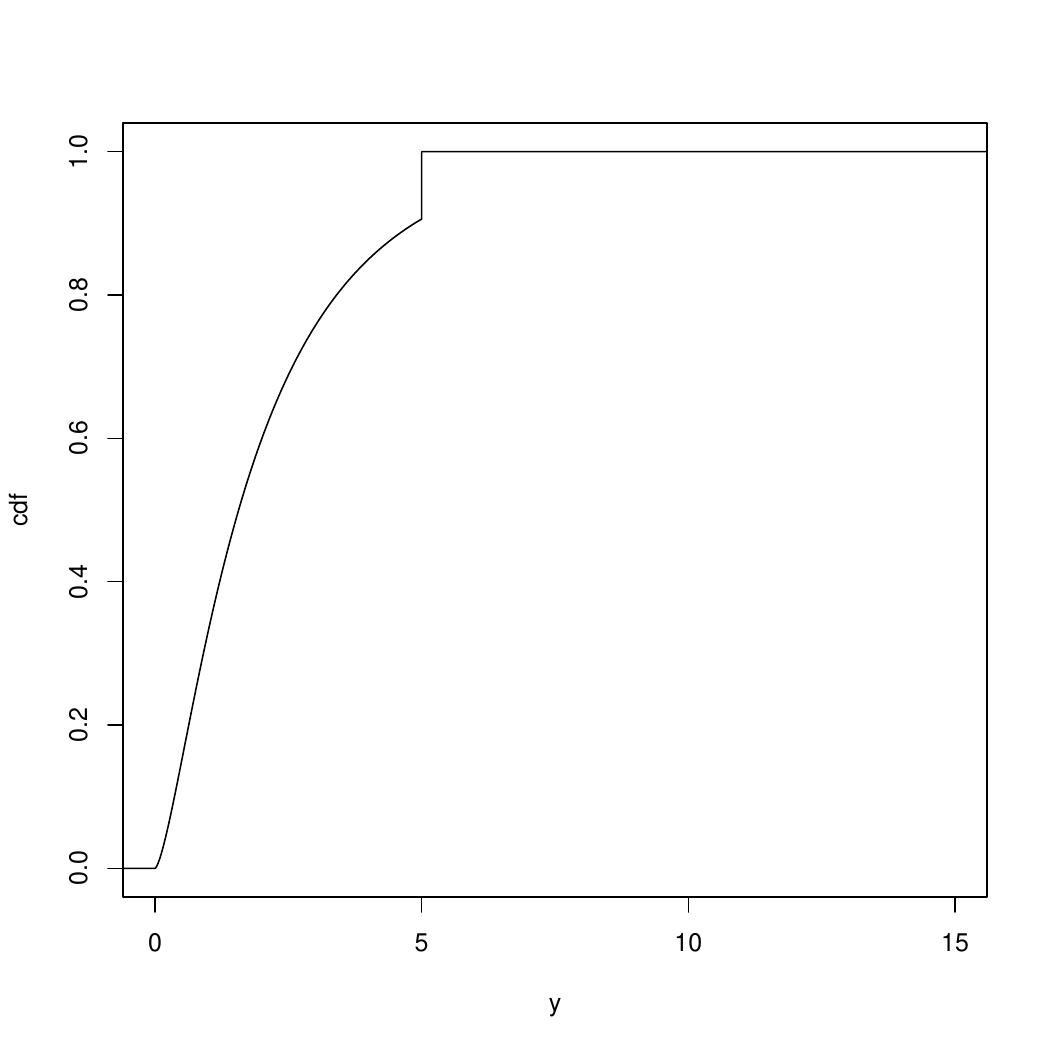}&\includegraphics[width=4cm]{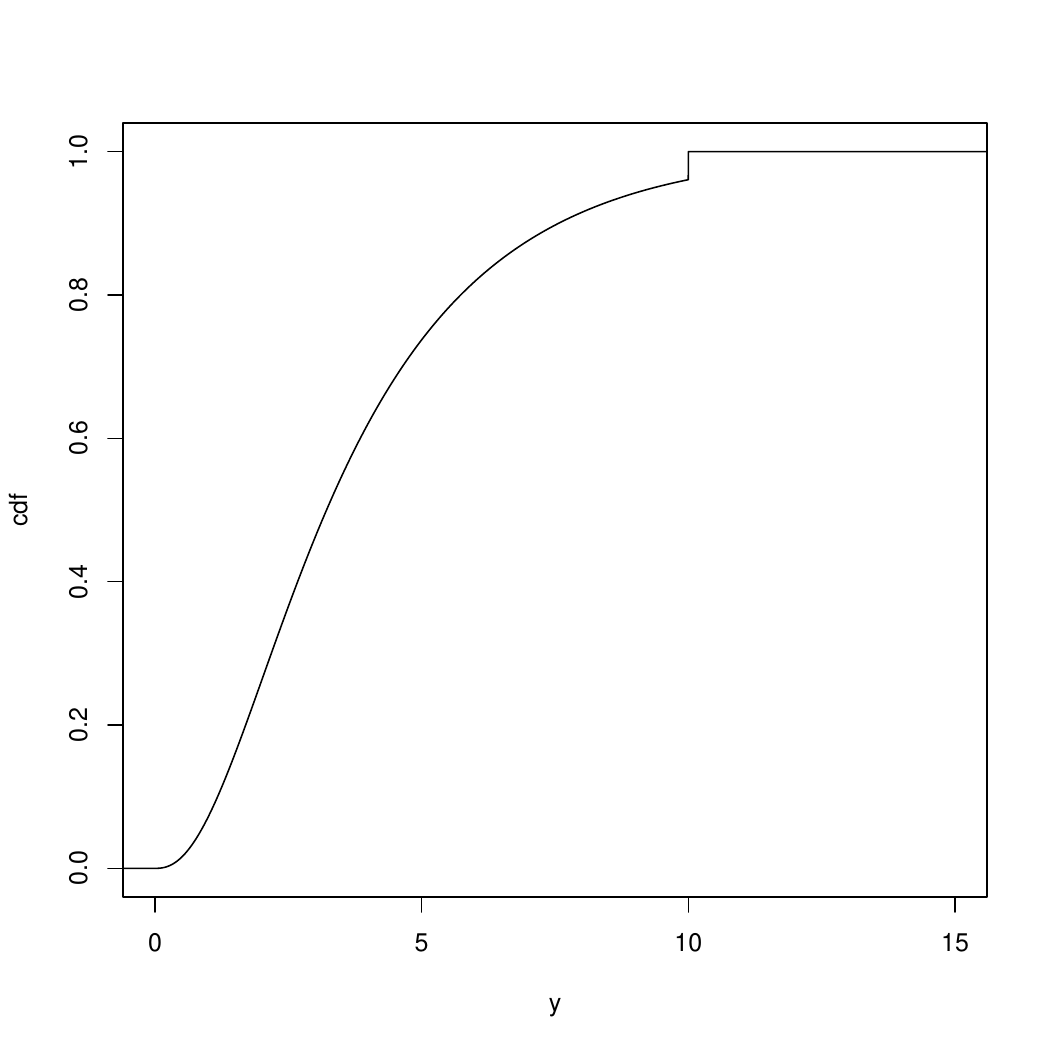}&
\includegraphics[width=4cm]{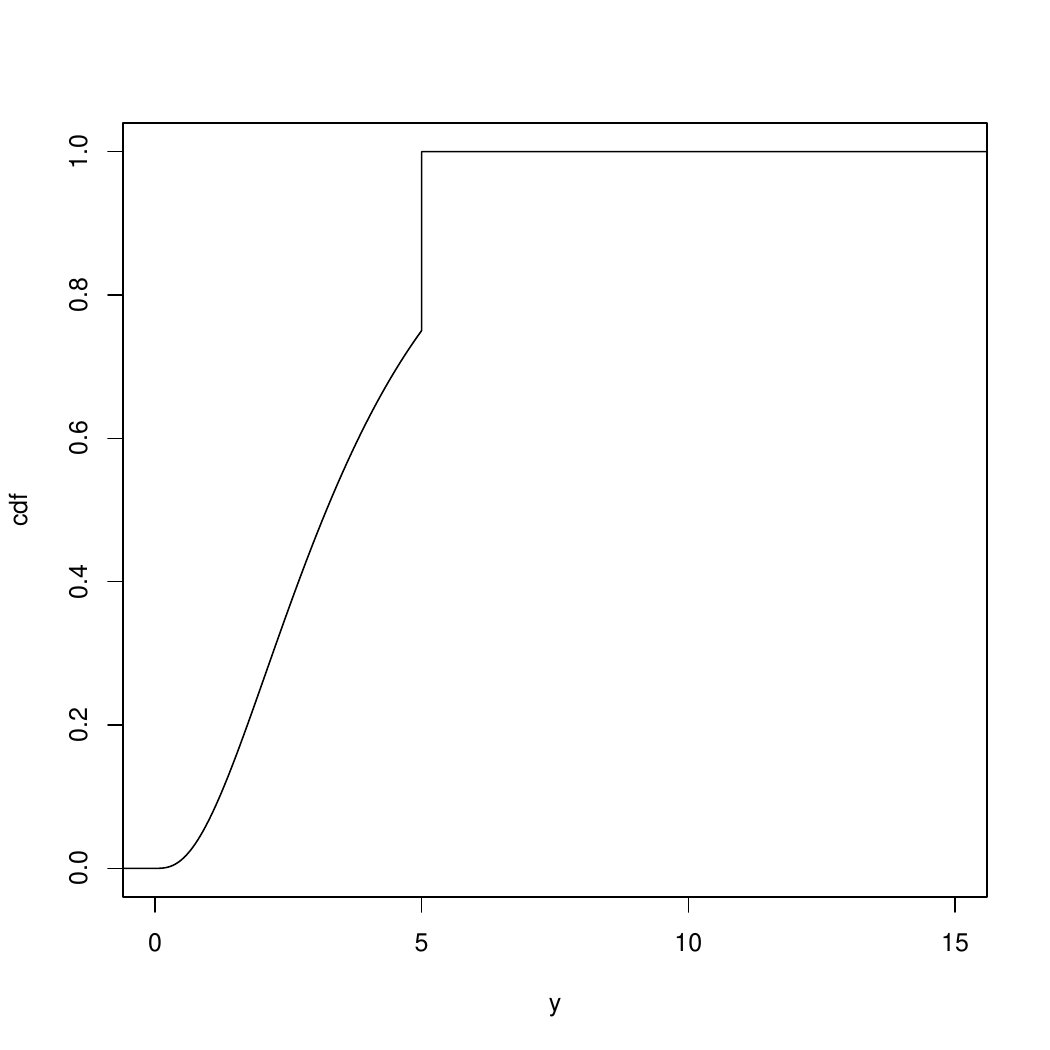}&\includegraphics[width=4cm]{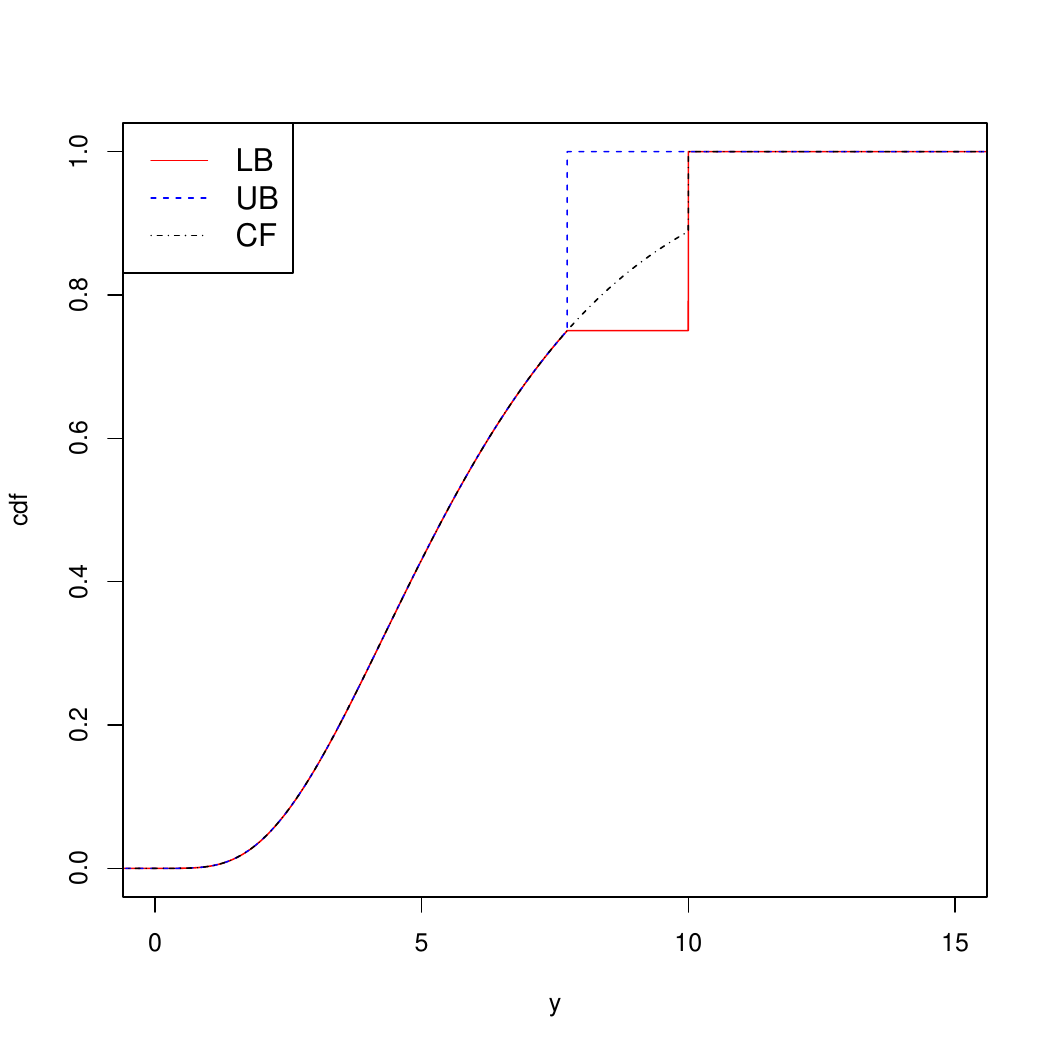}\\
\multicolumn{4}{c}{\parbox{\textwidth}{}}
\end{tabular}}\label{fig:left_censoring_2}
\end{figure}
\vspace{-1.5cm}
\begin{figure}[H]\caption{Numerical Example IV: Outcome Distribution with Bunching}\footnotesize{
\begin{tabular}{cccc}
\\
Panel A. $F_{Y_{00}|D=0}$ &Panel B. $F_{Y_{10}|D=0}$&Panel C. $F_{Y_{00}|D=1}$ & Panel D. CS Bounds on $F_{Y_{10}|D=1}$\\
\includegraphics[width=4cm]{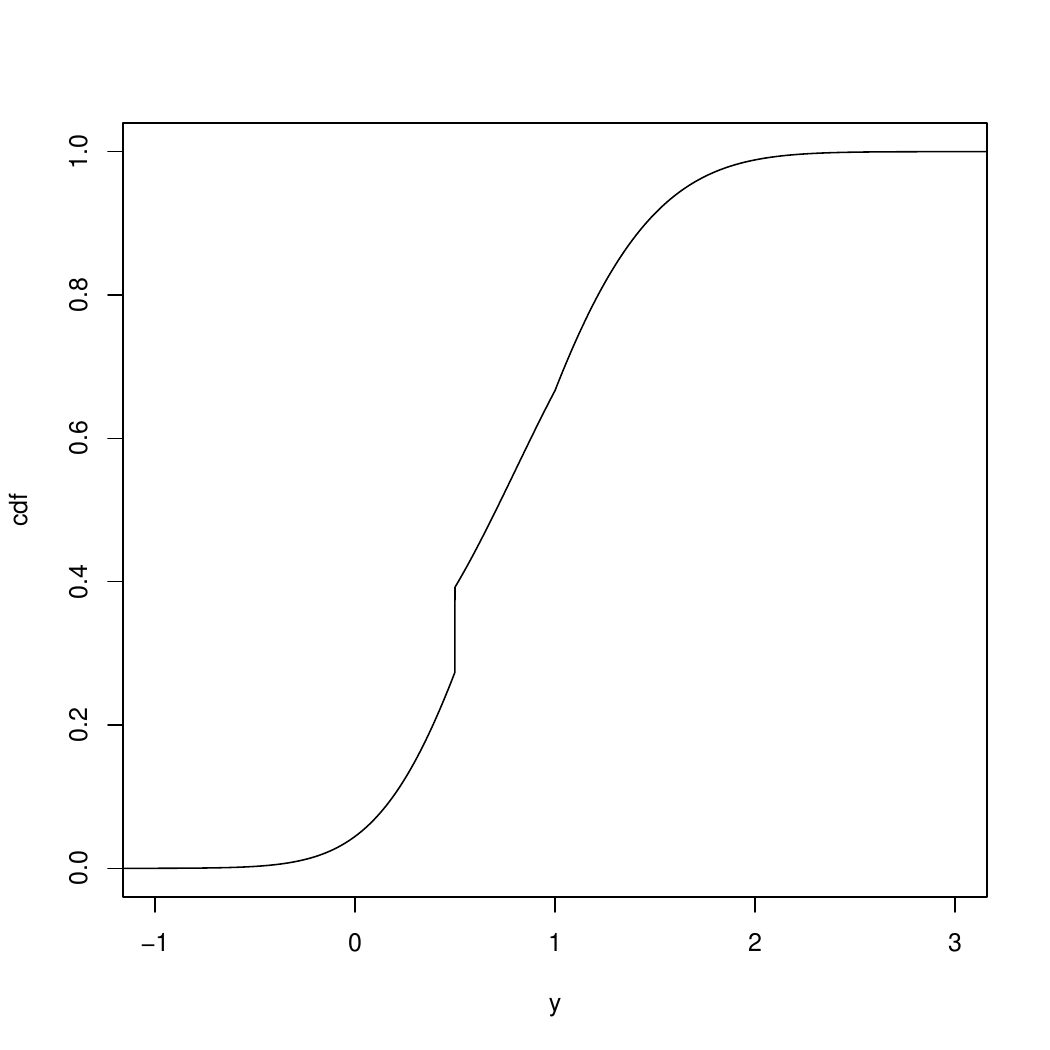}&\includegraphics[width=4cm]{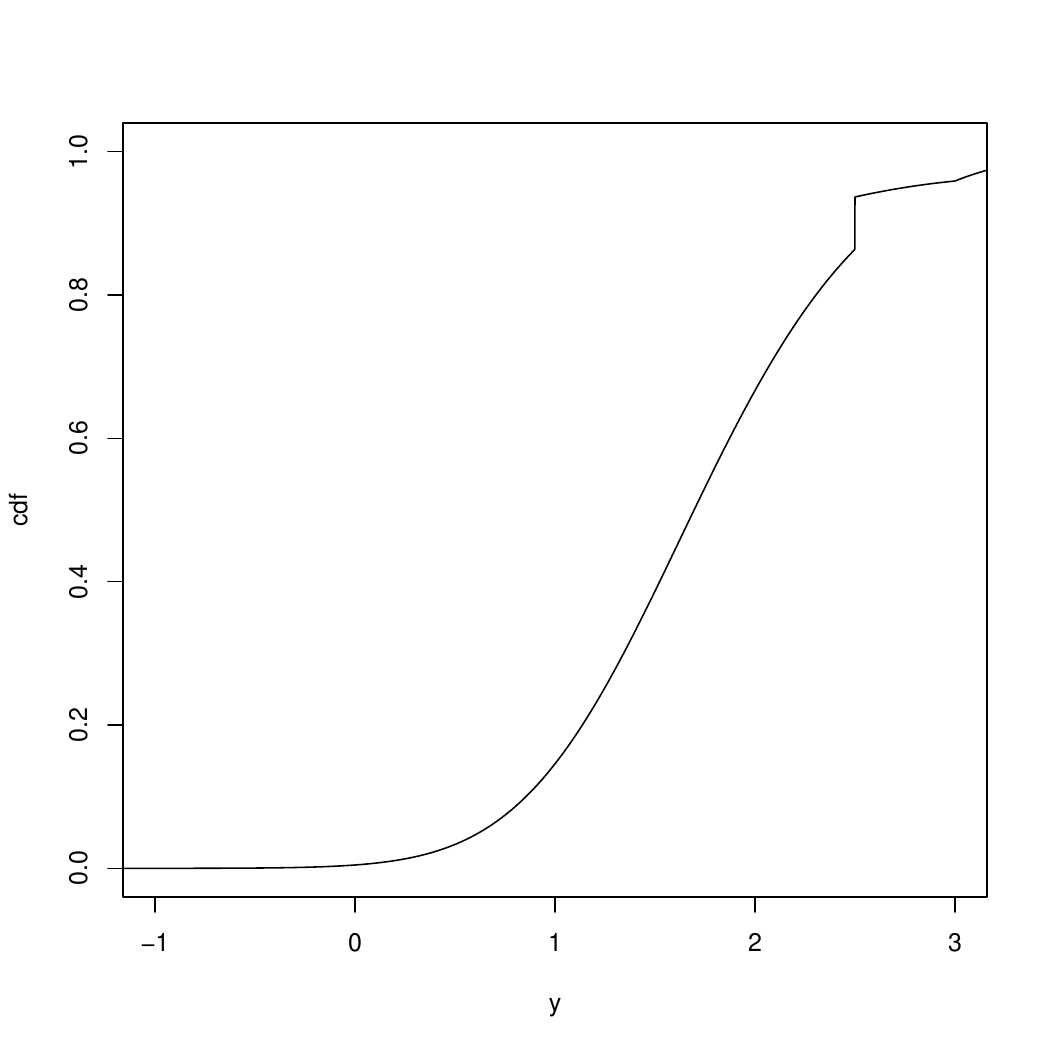}&
\includegraphics[width=4cm]{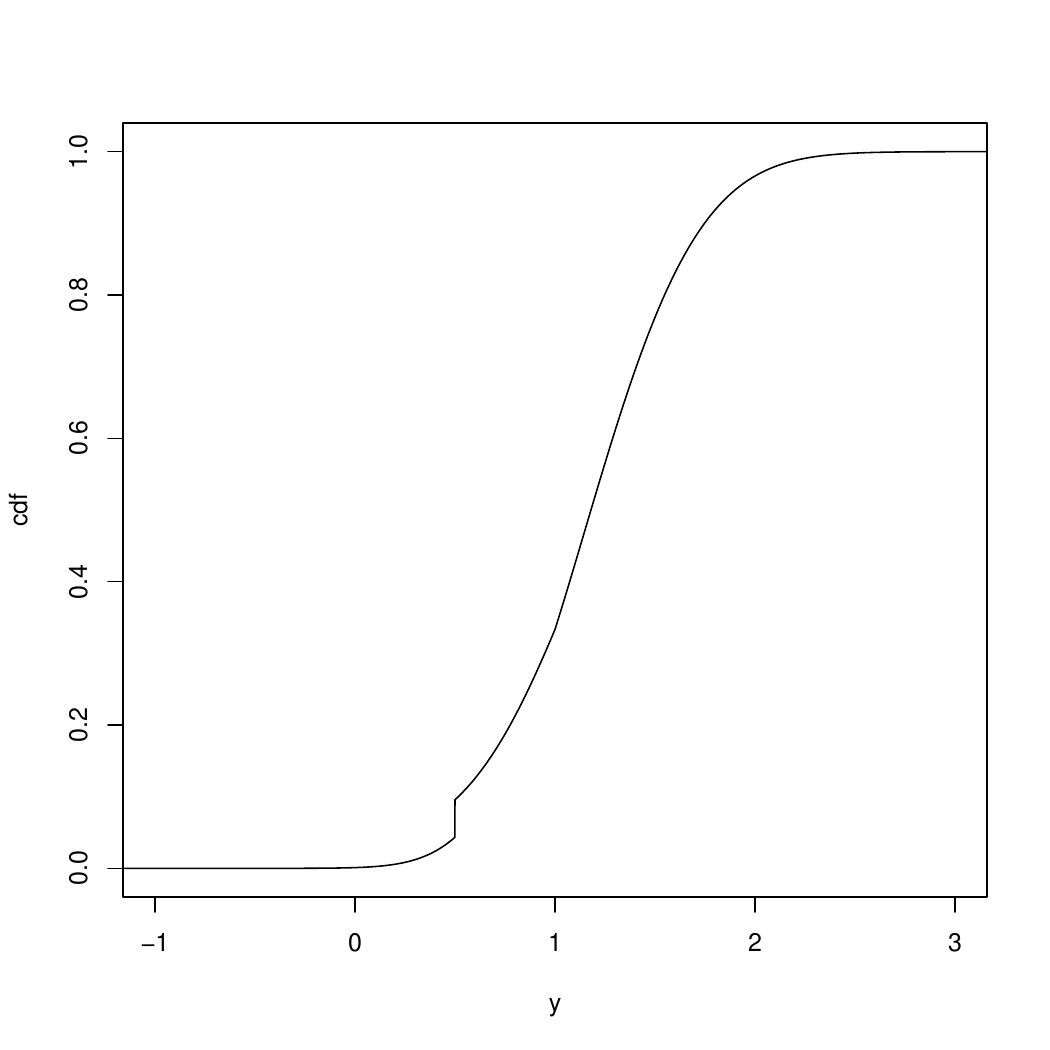}&\includegraphics[width=4cm]{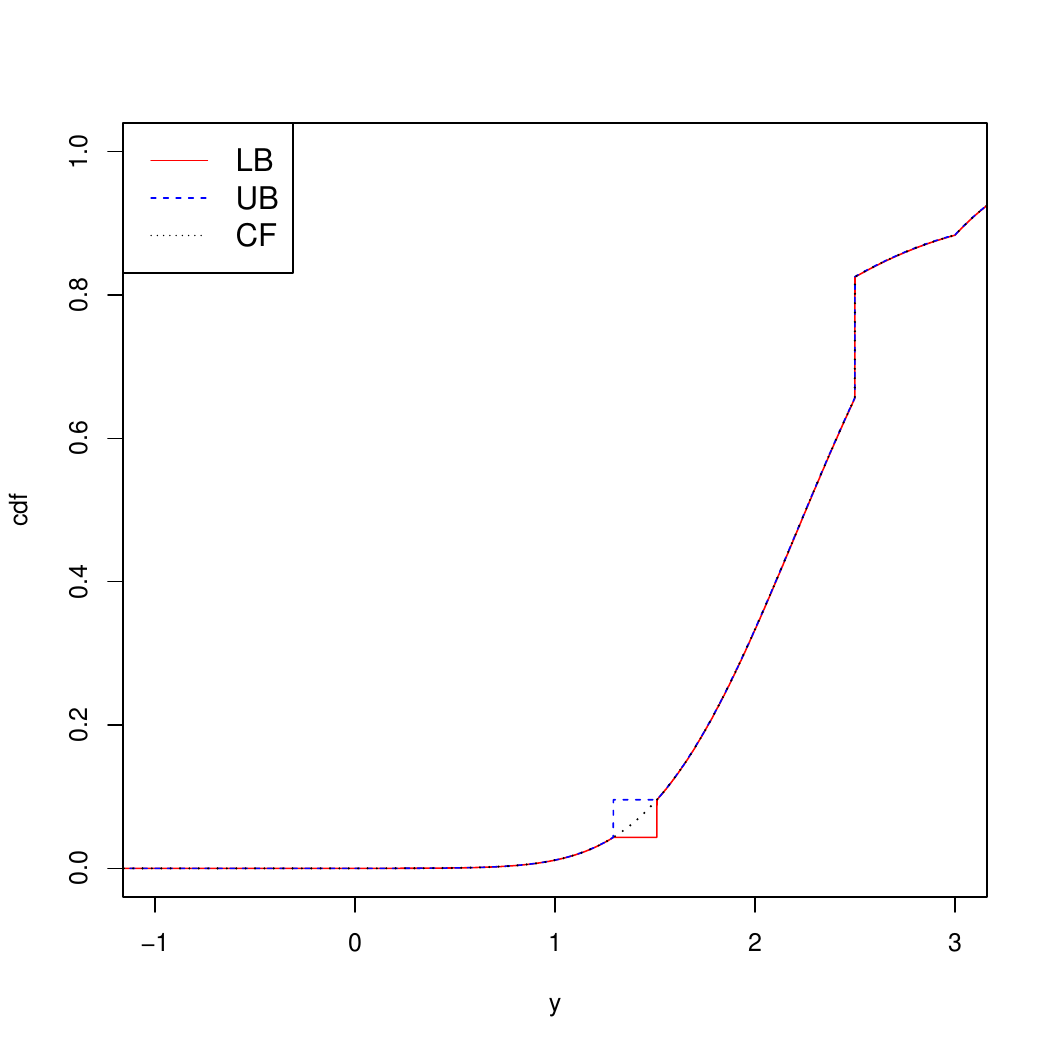}\\
\multicolumn{4}{c}{\parbox{\textwidth}{\scriptsize{\emph{Notes}: The figures are generated by numerically evaluating the conditional potential outcome distribution for the bunching example (IV) in Table \ref{tab:num_examples} with $c_0=0.5$, $w_0=1$, $c_1=2.5$, $w_1=3$, $b_0=0.25$, and $b_1=0.75$. }}}\\
\end{tabular}}\label{fig:bunching}
\end{figure}

\end{singlespace}

\section{Point-identification for any outcome with multiple pre-treatment periods}
We provide a condition under which we can achieve point-identification when multiple pre-treatment periods are available. 

\begin{corollary}\label{cor:pointid_binary}
Suppose that $\mathbb Y_{t0|1} \subseteq \mathbb Y_{t0|0}$ for $t \in \{-T_0,\ldots,0\}$, and Assumptions \ref{inc} and \ref{Gstab} hold. 
Suppose
\begin{eqnarray*}
    F_t^{LB}(s)&=&F_{Y_t|D=1}\left(Q^{\mathbb R,+}_{Y_t|D=0}\left(F_{Y_{1|D=0}}(s)\right)\right)\; \\\; F_t^{UB}(s)&=&F_{Y_t|D=1}\left(Q^{\mathbb R,-}_{Y_t|D=0}\left(F_{Y_{1|D=0}}(s)\right)\right).
\end{eqnarray*}
Suppose there exists $t_0 \in \{-T_0,\ldots,0\}$ such that $F_{Y_{t_0}|D=0}(s)=F_{Y_{1}|D=0}(s)$ for all $s$.\footnote{This can be easily verified by the researcher from the data.} Then, 
    $F_{Y_{10}|D=1}(y)=F_{Y_{t_0}|D=1}(y)$.
\end{corollary}

\begin{proof}
We have
\begin{eqnarray*}
\min_{t \in \{-T_0,\dots,0\}} F^{UB}_t(s) &\leq& F^{UB}_{t_0}(s) = F_{Y_{t_0}|D=1}\left(Q^{\mathbb R,-}_{Y_{t_0}|D=0}\left(F_{Y_{1|D=0}}(s)\right)\right),\\
&=& F_{Y_{t_0}|D=1}\left(Q^{\mathbb R,-}_{Y_{t_0}|D=0}\left(F_{Y_{t_0|D=0}}(s)\right)\right),\\
&\leq& F_{Y_{t_0}|D=1}(s),
\end{eqnarray*}
where the second equality holds because $F_{Y_{t_0}|D=0}(s)=F_{Y_{1}|D=0}(s)$, and second inequality holds because $F_{Y_{t_0}|D=1}$ is nondecreasing, and $Q^{\mathbb R,-}_{Y_{t_0}|D=0}\left(F_{Y_{t_0|D=0}}(s)\right) \leq s$ by definition of $Q^{\mathbb R,-}$. Similarly,
\begin{eqnarray*}
\max_{t \in \{-T_0,\dots,0\}} F^{LB}_t(s) &\geq& F^{LB}_{t_0}(s) = F_{Y_{t_0}|D=1}\left(Q^{\mathbb R,+}_{Y_{t_0}|D=0}\left(F_{Y_{1|D=0}}(s)\right)\right),\\
&=& F_{Y_{t_0}|D=1}\left(Q^{\mathbb R,+}_{Y_{t_0}|D=0}\left(F_{Y_{t_0|D=0}}(s)\right)\right),\\
&\geq& F_{Y_{t_0}|D=1}(s),
\end{eqnarray*}
where the second equality holds because $F_{Y_{t_0}|D=0}(s)=F_{Y_{1}|D=0}(s)$, and the second inequality holds because $F_{Y_{t_0}|D=1}$ is nondecreasing, and $Q^{\mathbb R,+}_{Y_{t_0}|D=0}\left(F_{Y_{t_0|D=0}}(s)\right) \geq s$ by definition of $Q^{\mathbb R,+}$. Hence, 
$$\min_{t \in \{-T_0,\dots,0\}} F^{UB}_t(s) \leq \max_{t \in \{-T_0,\dots,0\}} F^{UB}_t(s).$$
From Corollary \ref{result:fals-test}, we must have under our identifying assumptions
$$\max_{t \in \{-T_0,\dots,0\}} F^{LB}_t(s) \leq \min_{t \in \{-T_0,\dots,0\}} F^{UB}_t(s).$$
Therefore, the following equality holds. 
$$\max_{t \in \{-T_0,\dots,0\}} F^{LB}_t(s) = \min_{t \in \{-T_0,\dots,0\}} F^{UB}_t(s).$$
Hence,
$$\max_{t \in \{-T_0,\dots,0\}} F^{LB}_t(s) = \min_{t \in \{-T_0,\dots,0\}} F^{UB}_t(s)=F^{UB}_{t_0}(s)=F^{LB}_{t_0}(s)=F_{Y_{t_0}|D=1}(s).$$

\end{proof}

\section{Structural underpinnings of the copula stability assumption}

Consider a policymaker who wants to implement a policy in a specific region, i.e. introduction/increase of a minimum wage.
The policymaker decides to implement a policy if the gain in social welfare under the policy is higher than the gain in social welfare without the policy. The gain is evaluated by the policymaker given her information set $\mathcal I$. This decision rule is modeled as:
\begin{eqnarray*}
D&=&\mathbbm 1\left\{\mathbb E[W(Y_{11})-W(Y_{00})|\mathcal I] > \mathbb E[W(Y_{10})-W(Y_{00})|\mathcal I]\right\},\\
&=&\mathbbm 1\left\{\mathbb E[W(Y_{11})|\mathcal I] > \mathbb E[W(Y_{10})|\mathcal I]\right\}
\end{eqnarray*}
where $\mathcal I$ is the sigma-algebra characterizing
the decision maker  information set at the time of the decision, $Y_{td}$ for $t,d \in \{0,1\}$ are $\mathcal I$ measurable.
$W(.)$ is a measurable function that depends on the type of social welfare the policymaker wants to use. $W(.)$ can be specified to capture various types of societal welfare, like those discussed in the previous subsection.

To mimic our empirical illustration, we are considering the case where the outcomes of interest are mixed random variables because of the pre-existing minimum wage. We consider  a general case where a minimum wage $c_0$ exists in the pre-treatment period and the policymaker is considering an increase in this minimum wage, i.e. $c_1>c_0$.
\begin{eqnarray*}Y_{t0}&=&Y^{\ast}_{t0}\mathbbm{1}\{Y^{\ast}_{t0}>c_0\}+c_0\mathbbm{1}\{Y^{\ast}_{t0}\leq c_0\},~t=0,1,\\
Y_{11}&=&Y^{\ast}_{11}\mathbbm{1}\{Y^{\ast}_{11}>c_1\}+c_1\mathbbm{1}\{Y^{\ast}_{11}\leq c_1\}.\end{eqnarray*}

Assume that $Z$ is a vector of random variables that is measurable with respect to the policymaker information $\sigma$-algebra $\mathcal I$, and $\mathbb E[W(Y_{1d})|\mathcal{I}]=\psi_{1d}(Z)+ V_{1d}$, with $\mathbb E[V_{1d}\vert Z]=0$, where $V_{1d}$ for $d \in \{0,1\}$ are the prediction errors made by the policymaker given her information set. $Z$ could have a degenerate distribution and in such a case, the policymaker does not have additional information based on which she can form expectations. When $Z$ is observed by the econometrician, all our results hold conditional on $Z$.   

In the following, we assume that $\zeta\equiv V_{10}-V_{11}$ and the latent variables have continuous distributions. Our model simplifies to:
\begin{eqnarray}\label{eq:example1}
\left\{ \begin{array}{lcl}
     Y_0 &=& Y_{00}\\ \\
     Y_{1} &=& Y_{11} D+ Y_{10}(1-D) \\ \\
     D &=& \mathbbm{1}\left\{\psi_{11}(Z)-\psi_{10}(Z)\geq \zeta\right\}
     \end{array} \right.
\end{eqnarray}
Let $C_{Y^{\ast}_{t0}, \zeta\vert Z=z}(u,v;\rho_t(z))$ be the conditional copula that captures the dependence between $Y^{\ast}_{t0}$ and  $\zeta$. Suppose that $C_{Y^{\ast}_{t0}, \zeta\vert Z=z}(u,v;\rho_t(z))$ belongs to the class of totally ordered copulas.\footnote{$\{C_\theta\}$ is a totally strictly ordered family of copula if either $C_{\theta}(u,v)< C_{\theta'}(u,v)$ for all $(u,v) \in [0,1]^2$ whenever $\theta < \theta'$ for any $\theta, \theta'$ in the parameter space or $C_{\theta}(u,v) >C_{\theta'}(u,v)$ for all $(u,v) \in [0,1]^2$ when $\theta < \theta'$ for any $\theta, \theta'$ in the parameter space.} 
Therefore, it can be shown that 
if $\rho_0(z)=\rho_1(z)$ ---meaning that the dependence  between the policymaker prediction errors  $\zeta$ and $Y^{\ast}_{00}$ is the same as the dependence between $\zeta$ and $Y^{\ast}_{10}$, then the copula stability assumption holds conditional on $Z=z$, $C_{Y_{00},D|Z=z}(u,q)=C_{Y_{10},D|Z=z}(u,q)$ for all $u\in [0,1]$. A special case of this result is imposing a joint normal distribution on all the latent variables in the model such as 
$$\left(\begin{array}{c}Y^{\ast}_{00}\\ Y^{\ast}_{10} \\ \zeta \end{array}\right)\vert Z=z\sim N(0,\Sigma),~ \Sigma=\left(\begin{array}{ccc}\sigma_0^2(z) & \delta(z) \sigma_0(z) \sigma_1(z) &  \rho_0(z) \sigma_0(z)\\ \delta(z) \sigma_0(z) \sigma_1(z) & \sigma_1^2(z) &  \rho_1(z) \sigma_1(z)\\ \rho_0(z) \sigma_0(z)&  \rho_1(z) \sigma_1(z) & 1\end{array}\right).$$\\ \\
In this case, copula stability conditional on $Z$ is equivalent to 
$\rho_0(z)=\rho_1(z) \Leftrightarrow   Corr(\zeta, Y^{\ast}_{00}\vert Z=z)=Corr(\zeta,Y^{\ast}_{10}\vert Z=z)$, since the Gaussian copula belongs to the family of the strictly totally ordered copula. In this special case, our assumption is valid when the error of predictions made by the policymakers is correlated with the latent outcomes in the same way over time. 

\begin{proof}By definition, $Y_{t0}=Y^*_{t0}\mathbbm{1}\{Y^*_{t0} > c_0\}+c_0\mathbbm{1}\{Y^*_{t0} \leq c_0\}$. Take $y > c_0$. In the following, all arguments are conditional on $Z=z$:
\begin{eqnarray*}
\mathbb P(Y_{t0} \leq y, D \leq 0) &=& \mathbb P(Y_{t0} \leq y, \zeta \leq \psi(z)),\\
&=& \mathbb P(Y_{t0} \leq y, \zeta \leq \psi(z), Y^*_{t0} > c_0)+\mathbb P(Y_{t0} \leq y, \zeta \leq \psi(z), Y^*_{t0} \leq c_0),\\
&=&  \mathbb P(c_0 < Y^*_{t0} \leq y, \zeta \leq \psi(z)) + \mathbb P(Y^*_{t0} \leq c_0, \zeta \leq \psi(z)),\\
&=& \mathbb P(Y^*_{t0} \leq y, \zeta \leq \psi(z))
\end{eqnarray*}
Making the conditioning on $Z=z$ explicit, we have
$\mathbb P(Y_{t0} \leq y, D \leq 0|Z=z)=\mathbb P(Y^*_{t0} \leq y, \zeta \leq \psi(Z)|Z=z)=\Phi_2\left(\frac{y}{\sigma_t},\psi(z);\rho_t(z)\right)$,
which implies $C_{Y_{t0},D|Z=z}(u,q)=\Phi_2\left(\frac{Q^{\mathbb R, -}_{t0}(u)}{\sigma_t},Q^{\mathbb R,-}_D(q);\rho_t(z)\right)=\Phi_2(\Phi^{-1}(u),0;\rho_t(z))$. Therefore, it follows that $$C_{Y_{00},D|Z=z}(u,q)=C_{Y_{10},D|Z=z}(u,q) \Longleftrightarrow \rho_0(z)=\rho_1(z).$$
\end{proof}

\end{appendix}

\end{document}